\documentclass[reqno]{amsart}
\makeatletter
\def\l@subsection{\@tocline{2}{0pt}{2.5pc}{5pc}{}}
\def\l@subsubsection{\@tocline{2}{0pt}{5pc}{7.5pc}{}}
\makeatother

\usepackage{times}

\usepackage{frcursive}
\usepackage{yfonts}
\usepackage{graphicx}
\usepackage{empheq}
\usepackage{mathtools}
\usepackage{tikz}
\usepackage{amsbsy}
\usepackage{avant}
\usepackage{amsmath}
\usepackage{amsthm}
\usepackage{amssymb}
\usepackage{mathrsfs}
\usepackage{amsfonts}
\usepackage{newlfont}
\usepackage[sans]{dsfont}
\usepackage{dcolumn}
\usepackage{bm}
\usepackage{bbm}
\usepackage{stmaryrd}
\usepackage{bbm}
\usepackage{relsize}
\usepackage{euscript}
\usepackage{eufrak}
\usepackage{enumerate}
\topmargin - .23 in

\numberwithin{equation}{section}
\newtheorem{thm}{Theorem}[section]

\newtheorem{cor}[thm]{Corollary}
\newtheorem{lem}[thm]{Lemma}
\newtheorem{prop}[thm]{Proposition}
\newtheorem{defn}[thm]{Definition}
\newtheorem{rem}[thm]{Remark}

\newtheorem{gen}[thm]{General Proposition}
\begin{document}
\allowdisplaybreaks{
\title[]{A TURBULENT FLUID MECHANICS VIA NONLINEAR MIXING OF SMOOTH VELOCITY FLOWS WITH REYNOLDS-WEIGHTED RANDOM FIELDS}
\author{Steven D Miller}\email{stevendm@ed-alumni.net}
\address{}
\maketitle
\begin{abstract}
We consider a finite-volume domain $\bm{\mathfrak{D}}\subset\mathbb{R}^{3}$ of size $\mathrm{Vol}(\bm{\mathfrak{D}})\sim \mathrm{L}^{3}$, containing a viscous fluid of kinematic viscosity $\nu$ and velocity field $U_{a}(x,t)$ satisfying the Navier--Stokes equations with prescribed boundary data. Let ${\mathscr{B}}(x)$ be a zero-centred, homogeneous–isotropic Gaussian spatial random field on $\bm{\mathfrak{D}}$ with Bargmann–Fock correlation $\bm{\mathbb{E}}\langle {\mathscr{B}}(x)\otimes{\mathscr{B}}(y)\rangle=\mathsf{C}\exp(-\|x-y\|^{2}\lambda^{-2})$, where $\lambda\le \mathrm{L}$. For a volume-averaged Reynolds number
$\bm{\mathrm{Re}}(\bm{\mathfrak{D}},t)=\left(|\mathrm{Vol}(\bm{\mathfrak{D}})|^{-1} \int_{\bm{\mathfrak{D}}}\|U_{a}(x,t)\|\,d\mu(x)\right)\mathrm{L}/\nu $
let $\bm{\mathrm{Re}}_{c}(\bm{\mathfrak{D}})$ denote the critical threshold for turbulence. We introduce an exploratory Reynolds-weighted mixing ansatz for a turbulent random velocity field ${\mathscr{U}}_{a}(x,t)$,
\[
{\mathscr{U}}_{a}(x,t)=U_{a}(x,t)
+{\bm{\alpha}}\,U_{a}(x,t)\,
\psi\!\left(|\bm{\mathrm{Re}}(\bm{\mathfrak{D}},t)-\bm{\mathrm{Re}}_{c}(\bm{\mathfrak{D}})|\right)
{\mathbb{I}}_{\mathcal{S}}\!\left[\bm{\mathrm{Re}}(\bm{\mathfrak{D}},t)\right]
{\mathscr{B}}(x),
\]
where ${\bm{\alpha}}\ge 1$, $\psi$ is any monotone-increasing functional, and the indicator ${\mathbb{I}}_{\mathcal{S}}$ activates mixing only when $\bm{\mathrm{Re}}(\bm{\mathfrak{D}},t)>\bm{\mathrm{Re}}_{c}(\bm{\mathfrak{D}})$. The construction preserves the mean flow, $\bm{\mathbb{E}}\langle {\mathscr{U}}_{a}(x,t)\rangle=U_{a}(x,t)$, while allowing turbulence intensity to grow with the control parameter $\bm{\mathrm{Re}}$. This yields a tentative stochastic closure for Navier--Stokes, enabling estimates of Reynolds-type velocity correlations ${\bm{\mathsf{T}}}_{ab}(x,y;t)=\bm{\mathbb{E}}\langle {\mathscr{U}}_{a}(x,t)\otimes{\mathscr{U}}_{b}(y,t)\rangle$ and associated higher-order moments. For test functions $f$ and curves $\Im\subset\bm{\mathfrak{D}}$, we also formulate a Hopf-like functional integral
\[
\mathbb{H}[\mathscr{U}_{a}(x),t]
=\bm{\mathbb{E}}\!\left\langle
\exp\!\left(i\oint_{\Im} f(x,t)\,{\mathscr{U}}_{a}(x,t)\,dx^{a}\right)
\right\rangle,
\]
providing a compact description of tangled vortices if expanded. This framework is proposed as a tentative model for turbulence onset, mixing, and correlation structure in finite-volume domains.
\end{abstract}

\raggedbottom
\maketitle
\tableofcontents
\section{Introduction}

\begin{center}
\emph{``There's a magic in water that draws men to it..."} \\
--- Ishmael, \textit{Moby Dick}, Herman Melville \\[1em]

\emph{``I try all things, I achieve what I can."} \\
--- Ishmael, \textit{Moby Dick}, Herman Melville \\[0.5em]
\end{center}
Turbulence is an infamously difficult subject---perhaps even impossible, perhaps even beyond mathematics itself---and this paper is an attempt to experience that difficulty firsthand. Starting from minimal prior knowledge, the author experiments with stochastic analysis and the mathematics of random fields, and records the results. Many derivations are heuristic and messy, calculations are rough, and errors are inevitable, but the goal is to \emph{learn by doing}: to develop intuition, confront the subtle complexities of turbulent fluid mechanics, and experience for oneself just how difficult it all is. Later, it may be possible to do something better and more refined. These notes therefore serve as a personal record of an autodidactic exploration and may provide a starting point for others to explore similar ideas for themselves. Nonetheless, the idea of using \emph{Reynolds-weighted} random fields to describe turbulent flows appears to be original.

The mathematical analysis of the transition of a laminar flow to a turbulent flow is of great interest in fluid mechanics, but also a problem of very considerable difficulty which is still not well understood, developed or rigorously established and which continues to resist in-depth mathematical analysis--turbulence remains a great technical challenge. There is by now a vast literature on fluid mechanics and turbulence going back well over a century: in applied mathematics, in mathematical physics, in physics and in engineering. [See $\mathbf{[1-52]}$ and references therein]. However, although the general criteria under which the Navier-Stokes[NS] and Euler equations hold are well established, much remains unknown about these nonlinear partial differential equations, and also the complex, and indeed mysterious, nature of turbulence and incompressible turbulent flows. As such, these PDEs have become of increasing interest to mathematicians who aim to elevate the many heuristic physical results of fluid mechanics to a higher state of mathematical rigour; for example, it remains a major challenge to prove global regularity for the N-S equations. On the other hand, within applied mathematics, engineering and climate science there has been extensive development of practical numerical and computational methods.[REFS]

While there is no precise and universal definition of turbulence, a turbulent flow is characterised by complex and random structure at some scale or range of scales $\ell\le \lambda\le \mathrm{L}$ of dynamical significance in a continuous medium, usually a fluid. Turbulence also occurs in gases and plasmas of sufficient density. Key properties of turbulence are:
\begin{enumerate}[\bfseries I.]
\item Turbulent flows are far from equilibrium and also tend to be highly chaotic. A key property of turbulence is that is enhances and facilitates transport processes. It tends to be highly rotational and it mixes and transports energy, mass, and momentum very efficiently.
\item The presence of eddies and vortices over a very large range of length and times scales. The extent of this range is essentially determined by the Reynolds number, which is the ratio of the inertial to viscous forces. For large Reynolds numbers, the inertial term or nonlinearity dominates resulting in eddies/vortices of many scales being created and so the flow is turbulent; for example in the ocean and atmosphere, vortices can range from hundreds of kilometers to a few millimetres. But as the viscosity is increased turbulence tends to be suppressed; for example, turbulence is suppressed in flowing blood and highly suppressed or eliminated in flowing honey.
\item The nonlinearity of the Navier-Stokes equations result in strong coupling between scales so that energy and momentum are continuously exchanged between eddies or vortices of various sizes. Turbulence is both damped and driven with energy flowing into and out of the system. In the steady state, energy injection and dissipation are equal due to conservation of energy but the scales at which these two process operate can be vastly different. The 'stirring forces' that initiate turbulence--for example, solar heating gradients in the atmosphere or a rotating blade mixer immersed in a fluid, tend to create eddies/vorticity at the largest scales. Viscosity then dissipates turbulent kinetic energy but is generally too weak to dampen large eddies, acting instead at much smaller scales. A crucial feature of turbulence, as exemplified by Kolmogorov, is a continuous transport or cascade of energy from the largest to decreasingly smaller eddies, until finally viscosity dissipates the energy as heat. To quote Von Neumann:"...transport of a fixed flow of energy from sources in the low frequencies to sinks in the higher frequencies of the Fourier space". This is encapsulated in the famous
    $5/3$ law of Kolmogorov over an inertial range of wavenumbers.
\item There exists correlations between turbulent fluid motions or velocities at different points and times. These are binary, triple and all higher-order correlations. These correlations decay with increasing separation with the existence of correlation lengths and times.
\end{enumerate}
Turbulence is also ubiquitous in nature and plays a crucial role in virtually all phenomena involving liquids or gases: in the oceans and atmosphere of course where it affects weather and climate; in the atmospheres of gas giant planets such as Jupiter and Neptune; in stellar atmospheres and in the cores of stars where turbulent mixing can facilitate burning or fusion of nuclei, and mix the heavier elements within the star; turbulence can also facilitate thermal, conduction and radiative transport process in stars \textbf{[53]}, and ionized hydrogen at many millions of degrees within stars and nebulae will invariably be very turbulent; turbulence is also relevant to marine biology \textbf{[54]} and turbulence may even have played a role in the origin of life on Earth and its subsequent evolution, since it originated within liquid water.\textbf{[55]}. Turbulence is also an important consideration in very many practical engineering problems. Many of the key issues concerning turbulence, discussed by Von Neumann in a review paper from 1949 remain relevant.\textbf{[1]}
\subsection{Historical development}
Historically, statistical models of turbulence were developed in various phases:
\begin{enumerate}[\bfseries I.]
\item \textbf{The works of Burgers from circa 1923-1948.} The viscous 1D Burgers equation is
\begin{align}
 {\partial}_{t}U({x},t)+U(x,t){\partial}_{{x}}U({x},t)=\nu{\partial}_{\mathbf{xx}}U({x},t)
\end{align}
which is the NS equation in 1D with the pressure set to zero. It describes weakly compressible 1D flows. Von Neumann [1] also advocated this simple equation as the basis for a toy model of turbulence as it retains the salient features of the full Navier-Stokes equations. The BE has found a surprisingly large number of applications \textbf{[18], [48-52]}.
\item \textbf{The pioneering work of Reynolds and the later classic works of Taylor, Batchelor etc. from circa 1935-1937.} See \textbf{[5],[47],[34]} The total Navier-Stokes flow or fluid velocity is split into a 'mean part' $\overline{U(x,t)}$ and a 'fluctuating part' $\widetilde{U(x,t)}$ so that the total turbulent flow is
\begin{align}
{\bm{{U}}(x,t)}=\overline{{U}(x,t)} +\widetilde{U(x,t)}
\end{align}
and the statistical average is
\begin{align}
\big\langle U(x,t)\big\rangle=\overline{U(x,t)} + \big\langle\widehat{U(x,t)}\big\rangle=\overline{U(x,t)}
\end{align}
with $\big\langle\widehat{U(x,t)}\big\rangle=0$, and the flow being incompressible so that $\bm{D}.U=0$. The averages $\big\langle\bullet\big\rangle$ are taken to be either space or time averages or ensemble averages. For example, the space and time averages in a domain ${\bm{\mathfrak{D}}}\subset\mathbb{R}^{3}$ of volume $\mathrm{Vol}({\bm{\mathfrak{D}}})$ over a time scale $[0,T]$ are
\begin{align}
\langle\bullet\rangle_{V}=\frac{1}{\mathrm{Vol}(\bm{\mathfrak{D}})}{\int}_{\bm{\mathfrak{D}}}\bullet d\mu(x)
\end{align}
\begin{align}
\langle\bullet\rangle_{T}=\frac{1}{\mathrm{T}}{\int}_{0}^{T}\bullet d\tau
\end{align}
Taylors hypothesis (and ergodic considerations) take the time and space averages to be equivalent. However, there remain issues with rigorously establishing and mathematically defining the appropriate averages to be taken. The spatial averages require the flows to be statistically homogeneous on scales larger than the scales where turbulence occurs whereas time averages require stationarity of the flow. (Although experimentally TAs are more practical to implement.) Often, the ensemble average [EA] or time derivative is instead utilised. If $\lbrace
U_{a}^{[\xi]}(x,t)\rbrace_{\xi=1}^{N}$ is a sequence of realisations of a function $U_{a}(x,t)$ for all
$x\in{\bm{\mathfrak{D}}}$ and $t\in[0,T]$ then the ensemble average is
\begin{align}
\langle U_{a}(x,t)\rangle=\frac{1}{N}\sum_{\xi=1}^{N}U_{a}^{[\xi]}(x,t)
\end{align}
Again, one encounters mathematical difficulties dealing with convergence and the limit $N\rightarrow\infty$. At best, the EA converges very slowly. EAs are also difficult to implement within experimental or computational settings where $\mathbf{N}$ is necessarily finite. The procedure $\big\langle\bullet\big\rangle$ will usually represent a 'generic averaging process'. At any rate, the statistical averaging procedure giving rise to Reynolds-averaged Navier-Stokes equations, is given in the following theorem
\begin{thm}
Let $U_{a}(x,t)$ be the velocity of a fluid of viscosity $\nu$ filling a domain $\bm{\mathfrak{D}}$ or the entire space $\mathbb{R}^{3}$ and satisfying the Navier-Stokes equations
\begin{align}
&\frac{\partial}{\partial t}U_{a}(x,t)-\nu \Delta U_{a}(x,t)+U^{b}(x,t)\nabla_{b}U_{a}(x,t)+\nabla_{a}P(x,t)\nonumber\\&
\equiv \frac{\partial}{\partial t}U_{a}(x,t)-\nu \Delta U_{a}(x,t)+\nabla_{b}\big(U_{a}(x,t)U^{b}(x,t)\big)+\nabla_{a}\mathbf{P}(x,t)=0
\end{align}
since $\nabla_{b}U^{b}(x,t)=0$ for an incompressible fluid. If the flow is split into mean and fluctuating contributions as in (1.2) then
$U_{a}(x,t)=\overline{U_{a}(x,t)}+\widetilde{U_{a}(x,t)}$, the averaged N-S equations become
\begin{align}
&\langle\frac{\partial}{\partial t}U_{a}(x,t)-\nu \Delta U_{a}(x,t)+U^{b}(x,t)\nabla_{b}U_{a}(x,t)+\nabla_{a}\mathbf{P}(x,t)\rangle\nonumber\\&
=\frac{\partial}{\partial t}\overline{U_{a}(x,t)}-\nu \Delta\overline{U_{a}(x,t)}
-\nabla_{b}\big(\overline{U_{a}(x,t)}\overline{U^{b}(x,t)}\big)+\nabla_{a}\mathbf{P}(x,t)+
{\bm{\mathrm{R}}}_{ab}(x,t)\delta^{ij}
\end{align}
where ${\mathbf{R}}_{ab}(x,t)\delta^{ij}$ is the 'induced' Reynolds stress tensor.
\end{thm}
\begin{proof}
The Navier-Stokes equations are expanded out as
\begin{align}
&\frac{\partial}{\partial t}U_{a}(x,t)-\nu \Delta U_{a}(x,t)+U^{b}(x,t)\nabla_{b}U_{a}(x,t)+\nabla_{a}\mathbf{P}(x,t)\nonumber\\&
\equiv\frac{\partial}{\partial t}U_{a}(x,t)-\nu \Delta U_{a}(x,t)+\nabla_{b}\big(U_{a}(x,t)U^{b}(x,t)\big)+\nabla_{a}\mathbf{P}(x,t)\nonumber\\&=\frac{\partial}{\partial t}
\big(\overline{U_{a}(x,t)}+\widetilde{U_{a}(x,t)}\big)
-\nu \Delta\big(\overline{U_{a}(x,t)}+\widetilde{U_{a}(x,t)}\nonumber\\&+
\nabla_{b}\big(\overline{U_{a}(x,t)}+\widetilde{U_{a}(x,t)})\big(\overline{U^{b}(x,t)}+\widetilde{U_{a}(x,t)}\big)+
\nabla_{a}\mathbf{P}(x,t)\nonumber\\&
=\frac{\partial}{\partial t}\overline{U_{a}(x,t)}+\frac{\partial}{\partial t}\widetilde{U_{a}(x,t)}-\nu \Delta\overline{U_{a}(x,t)}-
\nu \Delta\widetilde{U^{b}(x,t))}\nonumber\\&+
\nabla_{b}\overline{U_{a}(x,t)}~\overline{U^{b}(x,t)}+\nabla_{b}\widetilde{U_{a}(x,t)}
\overline{U^{b}(x,t)}\nonumber\\&+\nabla_{b}\overline{U_{a}(x,t)}
\widetilde{U_{a}(x,t)}+\nabla_{b}\big(\widetilde{U_{a}(x,t)}\widetilde{U_{a}(x,t)}\big)+\nabla_{a}\mathbf{P}(x,t)
\end{align}
Taking the average $\langle\bullet\rangle$ gives
\begin{align}
&\langle\frac{\partial}{\partial t}\big(\overline{U_{a}(x,t)}+\widetilde{U_{a}(x,t)}-\nu \Delta\big(\overline{U_{a}(x,t)}
+\widetilde{U^{b}(x,t)}\nonumber\\&+\nabla_{b}\big(\overline{U_{a}(x,t)}+\widetilde{U_{a}(x,t)}\big)
\big(\overline{U^{b}(x,t)}+\widetilde{U_{a}(x,t)}\big)\rangle+
\nabla_{a}\mathbf{P}(x,t)\nonumber\\&
=\frac{\partial}{\partial t}\overline{U_{a}(x,t)}+\underbrace{\frac{\partial}{\partial t}\langle\widetilde{U_{a}(x,t)}\rangle}-\nu \Delta\overline{U_{a}(x,t)}
-\underbrace{\nu \Delta\langle\widetilde{U_{a}(x,t)}^{b})\rangle}\nonumber\\&+
\nabla_{b}\overline{U_{a}(x,t)}~\overline{U^{b}(x,t)}+\underbrace{\nabla_{b}\langle\widetilde{U_{a}(x,t)}
\rangle\overline{U^{b}}
(x,t)}\nonumber\\&+\underbrace{\nabla_{b}\overline{U_{a}(x,t)}\langle\widetilde{U^{b}(x,t)}\rangle}+
\langle\widetilde{U_{a}(x,t)}\widetilde{U_{b}(x,t)}\rangle+\nabla_{a}\mathbf{P}(x,t)
\end{align}
Linear terms (underbracketed) involving $\langle \widetilde{U_{a}(x,t)}\rangle$ vanish leaving
\begin{align}
&\langle\frac{\partial}{\partial t}\big(\overline{U_{a}(x,t)}+\widetilde{U^{b}(x,t)})-\nu \Delta\big(\overline{U_{a}(x,t)}
+\widetilde{U^{a}(x,t)}\rangle\nonumber\\&+
\langle \nabla_{b}\big(\overline{U_{a}(x,t)}+\widetilde{U^{b}(x,t)}\big)\big(\overline{U^{b}(x,t)}
+\widetilde{U_{a}(x,t)}\rangle+
\nabla_{a}\mathbf{P}(x,t)\nonumber\\&
=\frac{\partial}{\partial t}\overline{U_{a}(x,t)}-\nu \Delta\overline{U_{a}(x,t)}-\nabla_{b}\big(\overline{U_{a}(x,t)}\overline{U^{b}x,t)}\big)\nonumber\\&+
\nabla_{b}\langle\widetilde{U_{a}(x,t)}\widetilde{U_{b}(x,t)}\rangle+\nabla_{a}\mathbf{P}(x,t)\nonumber\\&
=\frac{\partial}{\partial t}\overline{U_{a}(x,t)}-\nu \Delta\overline{U_{a}(x,t)}-\nabla_{b}\big(\overline{U_{a}(x,t)}
\overline{U^{b}(x,t)}\big)\nonumber\\&+\nabla_{a}\mathbf{P}(x,t)+
\nabla_{b}\bm{\mathsf{T}}_{ab}(x,t)\delta^{ij}
\end{align}
\end{proof}
The Reynolds stress tensor
\begin{align}
\bm{\mathsf{T}}_{ab}(x,t)=\langle\widetilde{U^{b}(x,t)}\widetilde{U^{b}(x,t)}\rangle
\end{align}
is a non-vanishing term representing turbulence, that arises from the nonlinearity of the Navier-Stokes equations; that is the
nonlinear convective term $U^{b}\nabla_{b}U_{a}(x,t)=\nabla_{b}(U^{b}(x,t)U_{a}(x,t)$. There is no analogy of this for a linear PDE. For example, if the pressure is set to zero and the nonlinear convective term is removed the NS equations will reduce to a linear heat or diffusion-type equation $\frac{\partial}{\partial t}U_{a}(x,t)-\nu \Delta U_{a}(x,t)=0$. Substituting and averaging then gives
\begin{align}
&\langle\frac{\partial}{\partial t}U_{a}(x,t)\rangle -\nu\langle \Delta U_{a}(x,t)\rangle=
\frac{\partial}{\partial t}\overline{U_{a}(x,t)}-\nu \Delta\overline{U_{a}(x,t)}=0
\end{align}
No new terms arise upon averaging since this equation is linear. One can also try and estimate binary and triple velocity correlations of the form
\begin{align}
\langle\widetilde{U_{a}(x,t)}\widetilde{U_{b}(y,t)}\rangle,~~~\langle\widetilde{U_{a}(x,t)}
\widetilde{U_{b}(y,t)}\widetilde{U_{c}(\mathbf{z},t)}\rangle
\end{align}
for any points $( x,y,\mathbf{z})$, and also higher-order correlations.  However, in the basic form of RANS we also encounter the problem of turbulence closure. In simple words, the number of unknowns are more than the number of equations.
\item\textbf{The works of Kolmogorov and Heisenberg from circa 1941-1948.} Heisenberg developed a model of turbulence \textbf{[46]} and it was the subject of his Phd thesis. The K41 theory was developed by Kolmogorov in several papers in 1941.\textbf{[40, 41]} Although highly cited, the work remains incomplete and not fully understood. The assumptions are also not entirely clear but the central results have remained robust and are correct within the confines of the these underlying assumptions \textbf{[33, 35]}. At very high, but not infinite, Reynolds number, all of the small-scale statistical properties are assumed to be uniquely and universally determined by the length scale $\bm{\ell}$, the mean dissipation rate (per unit mass) $\epsilon$ and the viscosity $\nu$. Despite its conjectural status from the perspective of mathematical rigour, with some heuristic assumptions on statistical properties (homogeneity, isotropy, monofractal scaling), Kolmogorov \textbf{[40, 33, 35]} made a key prediction about the structure of turbulent velocity fields for incompressible viscous fluids at high Reynolds number, namely that
\begin{align}
&\mathbf{S}_{p}[U]=\langle|\widetilde{U_{a}(x+\bm{\ell},t)}-\widetilde{U_{a}(x,t)}|^{p}\rangle\nonumber\\&
=\mathrm{C}_{p}{\epsilon}^{p/3}\bm{\ell}^{p/3},~~~{\bm{\ell}}_{K}=(\nu^{4}/{\epsilon})^{1/4}\le {\bm{\ell}} \le \mathrm{L}
\end{align}
where $\mathrm{C}_{p}$ is a constant. For $p=2$ this gives the famous $2/3$ law or equivalently in Fourier space, the $5/3$ for the energy spectrum. The length $\bm{\ell}_{K}=(\nu^{4}/\epsilon)^{1/4}$ known as the Kolmogorov scale, represents a small scale dissipative cutoff or the size of the smallest eddies, and the integral scale L represents the size of the largest eddy in the flow. At this scale, viscosity dominates and the kinetic energy is dissipated into heat. The range $\bm{\ell}_{K}\le \ell\le \mathrm{L}$ over which the scaling law holds is known as the inertial range and there is a 'cascade' process whereby energy is transferred from the largest scales/eddies to the Kolmogorov scale. The objects $\mathbf{S}_{p}[U]$ are the pth-order longitudinal structure functions. A correction was made to this law in 1962 (the K62 theory) to incorporate the effects of intermittency \textbf{[41]}.
\begin{rem}
A key aspect of the Kolmogorov theory is that turbulence seems to be describable by isotropic spatio-temporal Gaussian 'random fields' and that the randomness or random structure within the fluid grows with increasing Reynolds number.
\end{rem}
\end{enumerate}
These are described in many texts and in the classic papers themselves.
\subsection{Basic results from fluid mechanics}
Some basic background results from smooth or deterministic 'laminar' fluid mechanics are briefly given \textbf{[2],[3],[9]}. In the absence of turbulence, we consider a set of \textbf{\textit{smooth  and deterministic solutions}} $(U(x,t),\rho(x),\mathbf{P}(x))$ of the steady state Navier-Stokes or Euler equations. Here $U(x)$ or $U_{a}(x)$ is the steady state fluid velocity at $x\in\bm{\mathbb{R}}^{n}$, $\mathbf{P}(x)$ is the pressure and $\rho$ is the density. This steady deterministic and smooth fluid flow will be 'mixed' with a Gaussian random (scalar) field to create a new random (vector) field or turbulent flow. For the general time-dependent Navier-Stokes equations, let $\bm{\mathfrak{D}}\subset\bm{\mathbb{R}}^{3}$ be a compact bounded domain with $x\in\bm{\mathfrak{D}}$ and filled with a fluid of density $\rho:[0,T]\times\bm{\mathbb{R}}^{3}\rightarrow{\mathbb{R}}_{\ge}$, pressure $\mathbf{P}:[0,T]\times {\mathbb{R}}^{3}\rightarrow{\mathbb{R}}_{\ge}$ and velocity $U:[0,T]\times\bm{\mathbb{R}}^{3}\rightarrow{\mathbb{R}}^{3}$ where $U(x)=(U_{a}(x,t))_{1\le i\le 3}$ so that $\mathbf{P}=\mathbf{P}(x,t),\rho=\rho(x,t)$ and $U_{a}=U_{a}(x,t)$. The continuity and Navier-Stokes equations are then
\begin{align}
&\frac{\partial}{\partial t}\rho+{D}.(\rho U)=0, x\in\bm{\mathfrak{D}},t\ge 0\\&
\bm{\mathrm{D}}_{m}U-\eta \Delta U+\Delta \mathrm{P}\equiv\frac{\partial}{\partial t}U-\eta \Delta U+(U.D)U+\Delta P=0,  x\in\bm{\mathfrak{D}},t\ge 0
\end{align}
where $\bm{\mathrm{D}}_{m}=\frac{\partial}{\partial t}+(U.D)$ is a material derivative, and $D.U=0$ for incompressible fluids.
The viscosity is $\eta$. In component form with $(i,j)=1,2,3$
\begin{align}
&\bm{\mathbf{D}}_{m}{U}_{a}(x,t)-\nu \Delta U_{a}(x,t)+\nabla_{a}\mathbf{P}(x,t)\nonumber\\&\equiv\frac{\partial}{\partial t}{U}_{a}(x,t)
-\nu \Delta{U}_{a}(x,t)+U^{b}(x,t)\nabla_{b}U_{a}(x,t)+\nabla_{a}\mathbf{P}(x,t)=0,x\in\bm{\mathfrak{D}},t\ge 0
\end{align}
with the incompressibility condition $\nabla_{a}U^{a}=0$.
\begin{prop}
The following will also apply:
\begin{enumerate}[(a)]
\item The smooth initial Cauchy data is $U(x,0)=U_{o}(x)$. One could also impose periodic boundary conditions if $\bm{\mathfrak{D}}$ is a cube or box of with sides of length $\mathrm{L}$ such that $U_{a}(x+\mathrm{L},t)=U_{a}(x,t)$, or no-slip BCs
$U_{a}(x,t)=0, \forall~x\in\partial\bm{\mathfrak{D}}$
\item By a \textbf{\textit{smooth deterministic flow}}, we mean a $U_{a}(x,t)$ which is deterministic and non-random  and evolves \textit{predictably} by the NS equations from some initial Cauchy data $U(x,0)=U_{o}(x)$. For example, a simple laminar flow with $U_{a}(x,t)=U_{a}
=const$ A generic smooth flow will be differentiable to at least 2nd order so that $\nabla_{b}U_{a}(x,t)$ and $\nabla_{a}\nabla_{b}U_{a}(x,t)$. This is preferably a strong solution.
\item The Euclidean and $L_{p}$ norms of $U_{a}$ are
\begin{align}
&\|U_{a}(x,t)\|_{E_{p}(\bm{\mathfrak{D}})}=\left(\sum_{a}|U_{a}(x,t)|^{p}\right)^{1/p}\\&
\|U_{a}(\bullet,t\|_{L_{p}(\bm{\mathfrak{D}})}={\int}_{\bm{\mathfrak{D}}}|U_{a}(x,t)|^{p}d^{3}x
\equiv{\int}_{\bm{\mathfrak{D}}}|U_{a}(x,t)|^{p}d\mu(\bm{\mathfrak{D}})
\end{align}
\item For some $C,K>0$, the initial data will also satisfy a bound of the typical form
\begin{align}
|\nabla_{x}^{\alpha}{U}_{o}(x)|\le C(\alpha,K)(1+|x|)^{-K}
\end{align}
with some boundary conditions $U(x,t)=0;~\mathbf{P}(x,t)=\rho(x,t)=0, x\in\partial\bm{\mathfrak{D}}$.
\item The fluid velocity $(x,t)$ is a divergence-free vector field that should be physically reasonable: that is, the solution should not
grow too large or blow up for $|x|\rightarrow \infty$ so that within the entire space $\mathbb{R}^{3}$ one must have $\lim_{|x|\rightarrow\infty}\|U_{a}(x,t)\|\le \mathbf{C}$ for any $t>0$.
\item There is a general energy bound of the form
\begin{align}
\sup_{t}\left(\|U_{c}(\bullet,t)\|^{2}_{L^{2}(\bm{\mathfrak{D}})}+{\int}_{0}^{t}\|DU_{c}(\bullet,\tau)\|^{2}_{L^{2}(\bm{\mathfrak{D}})}d\tau\right)\le {\psi}
\end{align}
In pioneering work Leray used this energy bound to prove existence of weak global solutions to the NS Cauchy problem when the initial
data is in $L_{2})(\bm{\mathfrak{D}})$. For strong solutions, equality holds.
\item The basic energy balance equation for a viscous fluid is obeyed such that
\begin{align}
{\partial}_{t}{\int}_{\bm{\mathfrak{D}}}U_{a}(x,t)U^{a}(x,t)d^{3}x= -
\nu{\int}_{\bm{\mathfrak{D}}}|\nabla^{a}U_{a}(x,t)|^{2}d^{3}x
\end{align}
or ${\partial}_{t}{{\bm{\mathcal{E}}}}[U]=-\nu{{{\psi}}}[U]$.
\end{enumerate}
\end{prop}
\section{Turbulence and turbulent flows via mixing of smooth Navier-Stokes flows with Bargmann-Fock random fields}
\subsection{Motivation}
The paper is motivated by the following:
\begin{enumerate}[\bfseries I.]
\item There remains much opportunity (and an ongoing need) to apply new and established mathematical ideas, tools and methods to the problem of developed turbulence: these include stochastic and statistical geometry, random fields and stochastic PDE.
\item As discussed, a central issue within fully developed turbulence is how to define and calculate Reynolds stress and velocity correlations. Established methods are mostly heuristic and it is very difficult to rigorously define or mathematically formalise the required spatial, temporal or ensemble averages $\big\langle\bullet\big\rangle$ in a useful manner. Rigorously defining statistical averages in conventional statistical hydrodynamics is fraught with technical difficulties and limitations, as well as having a limited scope of physical applicability.
\item A key insight of Kolmogorov's work is that turbulent flows are essentially (Gaussian) random fields. In this paper, we instead consider modelling a random flow or "turbulent fluid" using classical Gaussian random fields defined rigorously with respect to a probability space or probability triplet $[\Omega,{\psi},\mathbb{P}]$. Stochastic averages or expectations are then defined with respect to the probability space so that the expectation or average of a stochastic/random field is denoted ${\bm{\mathbb{E}}}\big\langle\bullet\big\rangle$. The smooth deterministic flow $U_{a}(x,t)$, obeying the NS equations, is then 'mixed' with a classical Gaussian random field ${{\mathscr{F}}}(x,t)$ in particular Bargmann-Fock random fields which have Gaussian binary correlations or decay kernels and are automatically regulated and differentiable. The correlation length is $\lambda$ with $\lambda\le \mathrm{L}$ or $|\lambda|^{3}\le \mathrm{Vol}(\bm{\mathfrak{D}}$.
The result of 'mixing' the smooth (vector) and random (scalar) fields is a new random (vector) field ${{\mathscr{U}}}_{a}(x,t)$ describing a random or turbulent flow. Stochastic averages can then be taken. The Reynolds number or viscosity is also utilised as a 'control parameter', whereby the degree of randomness or turbulence depends on the size of the Reynolds number. (Precise and more rigorous definitions are given later.) The random field ${{\mathscr{U}}}_{a}(x,t)$ is shown to be a solution of stochastically averaged Navier-Stokes equations.
\item This will be interpreted more as a tentative and (hopefully original) mathematical description or application of mathematical ideas/tools, rather than a hard physical model; that is, 'physically inspired mathematics' rather than physical fluid mechanics. Various aspects of turbulence are not purely Gaussian in nature. As such it might only describes some form of mathematically idealised or fictional 'turbulent fluid', not necessarily a real physical one, although it may capture some salient features of a real turbulent fluid under certain conditions.
\end{enumerate}
\subsection{Classical Gaussian random fields}
To attempt to describe turbulence as a spatio-temporal random field or random geometry, it is necessary to first briefly define Gaussian random fields [GRFs] and their properties.
Classical random fields correspond naturally to structures, and properties of systems, that are varying randomly in time and/or space. They have found many useful applications in mathematics and applied science: in the statistical theory or turbulence, in geoscience, medical science, engineering, imaging, computer graphics, statistical mechanics and statistics, biology and cosmology $\mathbf{[56-71]}$. Gaussian random fields (GRFs) are of special significance as they can occur spontaneously in systems with a larger number of degrees of freedom via the central limit theorem. GRFs also arise in the study of complex systems like spin glasses, optimization problems and protein folding \textbf{[73]}
and the Ising model in a random potential \textbf{[86]}. Coupling random fields or noise to ODEs or PDEs is also a useful methodology in studying turbulence, chaos, random systems, pattern formation etc. $\mathbf{[72-75],[ 79-84]}$. The study of stochastic partial differential equations (SPDEs),arising from the coupling of random fields/noises to PDEs is also a rapidly growing research field $\mathbf{[76-79]}$. Such SPDES can potentially model the propagation of heat, diffusions or waves in random medias or randomly fluctuating medias. Many dynamical systems are affected or influenced by (regulated) noise. For many years now there has also been a burst of activity to devise stochastic representations of fluid dynamics. These types of models are strongly motivated by climate and weather forecasting issues \textbf{[42-45]} and geophysical dynamical models \textbf{[62]}.

The GRF is defined with respect to a probability space/triplet as follows:
\begin{defn}(\textbf{Formal definition of Gaussian random fields}\newline
Let $(\bm{\Omega},\mathfrak{F},{\mathbb{P}})$ be a probability space. Within the probability triplet, $(\bm{\Omega},\mathfrak{F})$ is a \textbf{measurable space}, where $\mathfrak{F}$ is the $\sigma$-algebra (or Borel field) that should be interpreted as being comprised of all reasonable subsets of the state space $\bm{\Omega}$. Then:
\begin{enumerate}[(a)]
 \item ${{\mathbb{P}}}$ is a function such that ${{\mathbb{P}}}:\mathcal{F}\rightarrow [0,1]$, so that for all $\mathfrak{B}\in\mathfrak{F}$, there is an associated probability ${{\mathbb{P}}}(\mathfrak{B})$. The measure is a probability measure when ${\mathbb{P}}(\bm{\Omega})=1$.
\item Let $x_{a}\subset{{\bm{\mathfrak{D}}}}\subset\bm{\mathbb{R}}^{n}$ be Euclidean coordinates and let
$(\bm{\Omega},{\mathfrak{F}},{{\mathbb{P}}})$ be a probability space. Let ${{\mathscr{R}}}(x;\omega)$ be a random scalar function that depends on the coordinates $x\subset{\bm{\mathfrak{D}}}\subset{\mathbb{R}}^{n}$ and also $\omega\in\bm{\Omega}$.
     \item Given any pair $(x,\omega)$ there $\bm{\exists}$ map $\bm{\mathfrak{D}}:{\mathbb{R}}^{n}\times\bm{\Omega}\rightarrow{\mathbb{R}}$ such that
\begin{align}
\bm{\mathfrak{D}}:(\omega,x)\longrightarrow{{\mathscr{R}}}(x;\omega)
\end{align}
so that ${{{\mathscr{R}}}(x;\omega)}$ is a \textbf{random variable or field} on $\bm{\mathfrak{D}}\subset\mathbb{R}^{n}$ with respect to the probability space $(\bm{\Omega},\mathfrak{F},{{\mathbb{P}}})$.
\item A random field is then essentially a family of random variables $\lbrace{{\mathscr{R}}}(x;\omega)\rbrace$ defined with respect to the space $(\bm{\Omega},\mathcal{F},{\mathbb{P}})$ and ${\mathbb{R}}^{n}$.
\item The fields can also include a time variable $t\in{\mathbb{ R}}^{+}$ so that given any triplet
$(x,t,\omega)$ there is a mapping $\bm{\mathfrak{D}}:{\mathbb{R}}\times{\mathbb{R}}^{n}\times\bm{\Omega}\rightarrow {\mathbb{R}}$ such that $\bm{\mathfrak{D}}:(x,t,\omega)\hookrightarrow {\mathscr{R}}(x,t;\omega)$ is a \textbf{spatio-temporal random field}.Normally, the field will be expressed in the form ${{\mathscr{R}}}(x,t)$ or ${{\mathscr{R}}}(x)$ with $\omega$ dropped.
\item The random field ${\mathscr{R}}(x)$ will have the following bounds and continuity properties [REFs]
\begin{align}
&{{\mathbb{P}}}[\sup_{x\in\bm{\mathfrak{D}}}|{{\mathscr{R}}}(x)|~~<~~\infty]~=+1\\&
{{\mathbb{P}}}[\lim_{x\rightarrow y}\big|{{\mathscr{R}}}(x)-{{\mathscr{R}}}(x)\big|=0,
~\forall~(x,y)\in\bm{\mathfrak{D}}]=1
\end{align}
\end{enumerate}
\end{defn}
\begin{lem}
The random field is at the least, mean-square differentiable in that \textbf{[56], [62], [64]}
\begin{align}
\nabla_{b}{{\mathscr{R}}}(x)=\frac{\partial}{\partial x_{b}}{{\mathscr{R}}}(x)= \lim_{\mathbf{h}\rightarrow 0} \frac{{{\mathscr{R}}}(x+|\mathbf{h}|{\bm{\mathbb{E}}}_{b})-{{\mathscr{R}}}(x)}{|\mathbf{h}|}
\end{align}
where ${\bm{\mathbb{E}}}_{b}$ is a unit vector in the $j^{th}$ direction. For a Gaussian field, sufficient conditions for differentiability can be given in terms
of the covariance or correlation function, which must be regulated at $x=y$ The derivatives of the field $\nabla_{a}{{\mathscr{R}}},\nabla_{a}\nabla_{b}{{\mathscr{R}}}(x)$ exist at least up to 2nd order and do line, surface and volume integrals
${\int}_{\bm{\Omega}}{{\mathscr{R}}}(x,t)d\mu(x)$ with respect to domain $\bm{\mathfrak{D}}$.(See Appendix C.)
The derivatives or integrals of a random field are also a random field.
\end{lem}
\begin{defn}
The stochastic expectation $\bm{\mathbb{E}}\langle\bullet\rangle $) and binary correlation with respect to the space $(\Omega,{\mathcal{F}},{{\mathbb{P}}})$ is defined as follows, with $(\omega,\zeta)\in\mathbf{\Omega}$
\begin{align}
&{\bm{\mathbb{E}}}\langle\bullet\rangle={\int}_{\omega}\bullet~d{{\mathbb{P}}}[\omega]\\&
{\bm{\mathbb{E}}}\langle\bullet{{{\otimes}}}\bullet\rangle={\int}\!\!\!\!{\int}_{\Omega}
\bullet{{\otimes}}\bullet~d{{\mathbb{P}}}[\omega]
d{{\mathbb{P}}}[\zeta]
\end{align}
For Gaussian random fields $\bullet={{\mathscr{R}}}(x,t)$ only the binary correlation is required so that
\begin{align}
&{\bm{\mathbb{E}}}\langle{{\mathscr{R}}}(x,t)\rangle={\int}_{\omega}{{\mathscr{R}}}(x,t;\omega)~d\mathbb{P}[\omega]=0\\&
{\bm{\mathbb{E}}}\langle{{\mathscr{R}}}(x,t){{\otimes}}{{\mathscr{R}}}(y,s)\rangle=
{\int}\!\!\!\!{\int}_{\Omega}{{\mathscr{R}}}(x,t;\omega){{\otimes}}
{{\mathscr{R}}}(y,t;\zeta)~d\mathbb{P}[\omega]d{{\mathbb{P}}}[\zeta]\nonumber\\&
={\bm{\Xi}}(x,y;\lambda)\bm{\varphi}(t,s)
\end{align}
and regulated at $x=y$ for all $(x,y)\in\bm{\mathfrak{D}}$ and $t\in[0,\infty)$ if
\begin{align}
{\bm{\mathbb{E}}}\langle{{\mathscr{R}}}(x,t){{\otimes}}
{{\mathscr{R}}}(x,t)\rangle=\mathsf{C}<\infty
\end{align}
\end{defn}
\begin{defn}
The correlation between two or more fields is denoted by the operation ${{\otimes}}$. We say that two random fields $({\mathscr{R}}(x),{\mathscr{R}}(y))$ defined for any $(x,y)\in\bm{\mathfrak{D}}$ are correlated or uncorrelated if
\begin{empheq}[right=\empheqrbrace]{align}
&{\bm{\mathbb{E}}}\langle{{\mathscr{R}}}(x,t){{\otimes}}{{\mathscr{R}}}(y,t)\rangle\ne 0 \nonumber\\&
{\bm{\mathbb{E}}}\langle{{\mathscr{R}}}(x,t){{\otimes}}{{\mathscr{R}}}(y,t)\rangle= 0 \nonumber
\end{empheq}
\end{defn}
\begin{defn}(\textbf{Multivariate Gaussian random fields})\newline
Let $\lbrace x^{(\alpha)}\rbrace_{\alpha=1}^{N} $ be a set of discrete and independent points within $\bm{\mathfrak{D}}$ so that
\begin{align}
\lbrace x^{(\alpha)}\rbrace_{\alpha=1}^{N}=\lbrace x^{(1)},...,x^{(\alpha)},...x^{(N)}\rbrace~{\in}~\bm{\mathfrak{D}}\nonumber
\end{align}
and let
\begin{align}
\lbrace {{\mathscr{R}}}^{(\alpha)}(x_{(\alpha)}\rbrace_{\alpha=1}^{N}=\lbrace {{\mathscr{R}}}^{(1)}(x^{(1)}),...,{{\mathscr{R}}}^{(\alpha)}(x^{(\alpha)}),...{{\mathscr{R}}}^{(N)}{x}^{(N)}\rbrace\nonumber
\end{align}
be the set of random fields at these points. (Note that $x\equiv (x_{1},x_{2},x_{3})\equiv x_{a}$, with i=1,2,3 in $\mathbb{R}^{3}$ is a single point
whereas $x^{(\alpha)}$ are a set of discrete and independent points.) Then
\begin{align}
\mathbb{P}[{{\mathscr{R}}}^{(\alpha)}]d{{\mathscr{R}}}(x^{(\alpha)})=\mathbf{Prob}
\big[{{\mathscr{R}}}(x^{(\alpha)})~{\in}~{{\mathscr{R}}}(x^{(\alpha)})+d{{\mathscr{R}}}(x^{(\alpha)})\big]
\end{align}
The N-point joint distribution function is then multivariate Gaussian and of the form (REFs)
\begin{align}
&\mathbb{P}\big[{{\mathscr{R}}}^{(1)}...{{\mathscr{R}}}^{(N)}\big]
d{{\mathscr{R}}}^{(1)}(x^{(1)})...d{{\mathscr{R}}}(x^{(N)})\nonumber\\&=\frac{1}{(2\pi)^{N}(det(\mathbf{M})^{1/2}}
\exp\left(-\frac{1}{2}\sum_{\alpha,\beta}{\mathscr{R}}^{(\alpha)}(x^{(\alpha)})(\mathbf{M})^{-1}_{\alpha\beta}{\mathscr{R}}^{(\beta)}(x^{(\beta)})\right)
d{{\mathscr{R}}}^{(1)}(x^{(1)})...d{{\mathscr{R}}}(x^{(N)})
\end{align}
where $\mathbf{M}^{\alpha\beta}=\bm{\mathbb{E}}\langle {\mathscr{R}}^{\alpha}(x^{\alpha}){{\otimes}}{\mathscr{R}}^{\beta}(x^{\beta})\rangle$.
\end{defn}
The binary 2-point function nor covariance fully determines all its properties, but the key advantages of GRSF are that a GRSF is Gaussian distributed and can be classified purely by its first and second moments, and all high-order moments and cumulants can be ignored. GRFs also tend to be more tractable.
The same defintion can be applied to spatio-temporal fields.
\begin{defn}
If ${{\mathscr{R}}}(x,t)$ is a GRSF existing for all $x\in\bm{\mathbb{R}}^{n}$ or $x\in\bm{\mathfrak{D}}\subset{\mathbb{R}}^{3}$ then
\begin{align}
&\bm{\mathbb{E}}\langle{{\mathscr{R}}}(x,t)\rangle=0\\&
\bm{\mathbb{E}}\langle{{{\mathscr{R}}}}(x,t){{\otimes}}~{{\mathscr{R}}}(y,t)
\rangle={\bm{\Xi}}(x,y;\lambda)\bm{\varphi}(s,t)
\end{align}
for any 2 points $(x,y)\in\bm{\mathfrak{D}}$, and with a correlation length $\lambda$. It is regulated if $\bm{\mathbb{E}}\langle({{\mathscr{R}}}(x){{\otimes}}~{{\mathscr{R}}}(x)\big\rangle=\mathfrak{N}<\infty$
For a white-in-space Gaussian noise or random field $\bm{\mathbb{E}}\langle{{\mathscr{R}}}(x){{\otimes}}~{{\mathscr{R}}}(y)\rangle=
\mathsf{C}\delta^{n}(x-y)$ and is unregulated. This paper utilises only regulated GRSFs. The binary covariance is then
\begin{align}
&\bm{\mathbb{K}}\langle{{\mathscr{R}}}(x,t){{\otimes}}~{{\mathscr{R}}}(y,t)\rangle\equiv
\bm{\mathbb{E}}\langle{{\mathscr{R}}}(x,t){{\otimes}}~{{\mathscr{R}}}(y,t)\rangle
+\bm{\mathbb{E}}\langle{{\mathscr{R}}}(x,t)\rangle\bm{\mathbb{E}}\langle{{\mathscr{R}}}(y,t)\rangle\nonumber\\&
=\bm{\mathbb{E}}\langle{{\mathscr{R}}}(x,t){{\otimes}}~{{\mathscr{R}}}(y,t)\rangle
={\bm{\Xi}}(x,y;\lambda)\bm{\varphi}(s,t)
\end{align}
\end{defn}
Given a GRF and a smooth vector field defined on $\bm{\mathfrak{D}}$ which evolves via some PDE, a tentative and nonlinear \textbf{'mixing formula'} can prescribed which
generates a class of new random vector fields within $\bm{\mathfrak{D}}$.
\begin{prop}\textbf{Random vector fields via a nonlinear mixing of a smooth deterministic vector field with
a scalar Gaussian random field}\newline
Let ${\mathscr{R}}(x)$ be a GRF existing~$\forall~x\in\bm{\mathfrak{D}}$ with stochastic expectation $\bm{\mathbb{E}}\big\langle{\mathscr{R}}(x)\big\rangle=0$ and $\bm{\mathbb{E}}\big\langle{\mathscr{R}}(x){{\otimes}}{\mathscr{R}}(x)\big\rangle=\mathsf{C}$. Let $X_{a}(x,t)$ be a \textbf{smooth deterministic vector field} existing for all $x\in\bm{\mathfrak{D}}$ and $t>0$. A 'smooth vector field' is defined here as one for which at least the first and second derivatives $\nabla_{b}X_{a}(x,t),\nabla_{a}\nabla_{b}X_{a}(x,t)$ exist. 'Deterministic' means that the vector field evolves in space and time predictably via some PDE and from some initial smooth Cauchy data $X_{a}(x,0)=X_{a}^{o}(x)$, so that
\begin{align}
{\partial}_{t}{Y}_{a}(x,t)={\mathcal{O}}\big(\nabla_{a},\nabla_{a}\nabla^{a}\big){Y}_{a}(x,t)
\end{align}
or
\begin{align}
{\partial}_{t}{Y}_{a}(x,t)={\mathcal{N}}\big(X_{a},\nabla_{a},\nabla_{a}\nabla^{a}\big){Y}_{a}(x,t)
\end{align}
where $\mathcal{O}$ and $\mathcal{N}$ are linear or nonlinear operators. The Euclidean norm of $X_{a}$ is $\|X_{a}(x,t)\|_{E_{2}(\bm{\mathfrak{D}}}=\sqrt{\sum_{a}|X_{a}(x,t)|^{2}}$. Let ${\psi}:\mathbb{R}\rightarrow\mathbb{R}$ be a smooth 'test' function such that ${\psi}={\psi}(\|X_{a}(x,t)\|)$, a functional of the norm. Suppose $\mathbf{C}$ is a constant such that ${\psi}(\|X_{a}(x,t)\|)=0$
if $\|X_{a}(x,t)\|_{E_{2}(\bm{\mathfrak{D}}}\le\mathbf{C}$ for some $x\in\bm{\mathfrak{D}},t>0$, and define a set $\mathfrak{S}$ such that
\begin{align}
{\mathcal{S}}=\big\lbrace\text{set of all}\|X_{a}(x,t)\|_{E_{2}(\bm{\mathfrak{D}}}>\mathbf{C}\big\rbrace
\end{align}
An 'indicator function' or 'switch function' $\mathbb{I}_{\mathcal{S}}[\|X_{a}(x,t)\|]$ is then defined as
\begin{align}
&{\mathbb{I}}_{\mathcal{S}}~[\|X_{a}(x,t)\|]=1~if~\|X_{a}(x,t)\|{\in}{\mathcal{S}}\nonumber\\&
{\mathbb{I}}_{\mathcal{S}}~[\|X_{a}(x,t)\|]=0~if~\|X_{a}(x,t)\|{\notin}{\mathcal{S}}
\end{align}
A new random vector field within $\bm{\mathfrak{D}}$ can then be defined by nonlinearly 'mixing' the smooth field
$X_{a}(x,t)$ with the GRF ${\mathscr{R}}(x)$ such that
\begin{align}
{{\mathscr{X}}}_{a}(x,t)=X_{a}(x,t)+{\bm{\alpha}}X_{a}(x,t)
{{\psi}}\big(\|X_{a}(x,t)\|\big)\mathbb{I}_{\mathcal{S}}~
\big[\|X_{a}(x,t)\|\big]{{\mathscr{R}}}(x)
\end{align}
where $\bm{\alpha}$ is an arbitrary mixing parameter or constant. The functional ${\psi}$ is arbitrary but if ${\psi}$ was a power law for example, then ${\psi}(\|X_{a}(x,t)\|=\alpha\big(\|X_{a}(x,t)\|-\mathbf{C}\big)^{\kappa}$ for
$\alpha>0$ and $\kappa>0$ giving
\begin{align}
{{\mathscr{X}}}_{a}(x,t)=X_{a}(x,t)+{\bm{\alpha}}X_{a}(x,t)
\alpha\|X_{a}(x,t)\|-\mathbf{C}^{\kappa}\mathbb{I}_{\mathcal{S}}~
\big[\|X_{a}(x,t)\|\big]{{\mathscr{R}}}(x)
\end{align}
It follows that:
\begin{enumerate}
\item The stochastic expectation is then
\begin{align}
{\bm{\mathbb{E}}}\langle{{{\mathscr{X}}}}_{a}(x,t)\rangle=X_{a}(x,t)+{\bm{\alpha}}X_{a}(x,t)
{{\psi}}\big(\|X_{a}(x,t)\|\big)\mathbb{I}_{\mathcal{S}}~
\big[\|X_{a}(x,t)\|\big]{\bm{\mathbb{E}}}\langle{{\mathscr{R}}}(x)\rangle
=X_{a}(x,t)
\end{align}
\item The random field reduces back to a smooth deterministic field for $\|X_{a}(x,t)\|\le \mathbf{C}$ so that
\begin{align}
{{\mathscr{X}}}_{a}(x,t)={Y}_{a}(x,t),~~if~\|X_{a}(x,t)\|\le \mathbf{C}~or~\big\|{Y}_{a}(x,t)\big\|\notin{\mathcal{S}}
\end{align}
\item The derivatives are
\begin{align}
\nabla_{a}{{{\mathscr{X}}}}_{a}(x,t)&=\nabla_{a}X_{a}(x,t)+\bm{\alpha} \nabla_{a}X_{a}(x,t){{\psi}}\big(\|X_{a}(x,t)\|\big)\mathbb{I}_{\mathcal{S}}~
\big[\|X_{a}(x,t)\|\big]{{\mathscr{R}}}(x)\nonumber\\&+\bm{\alpha} X_{a}(x,t)\nabla_{a}{{\psi}}\big(\|X_{a}(x,t)\|\big)\mathbb{I}_{\mathcal{S}}~
\big[\|X_{a}(x,t)\|\big]{{\mathscr{R}}}(x)\nonumber\\&+\bm{\alpha} X_{a}(x,t){{\psi}}\big(\|X_{a}(x,t)\|\big)\mathbb{I}_{\mathcal{S}}~
\big[\|X_{a}(x,t)\|\big]\nabla_{a}{{\mathscr{R}}}(x)
\end{align}
\begin{align}
{\partial}_{t}{{{\mathscr{X}}}}_{a}(x,t)&={\partial}_{t}X_{a}(x,t)+\bm{\alpha} \frac{\partial}{\partial t}X_{a}(x,t){{\psi}}\big(\|X_{a}(x,t)\|\big)\mathbb{I}_{\mathcal{S}}~
\big[\|X_{a}(x,t)\|\big]{{\mathscr{R}}}(x)\nonumber\\&+\bm{\alpha} X_{a}(x,t){\partial}_{t}{{\psi}}\big(\|X_{a}(x,t)\|\big)\mathbb{I}_{\mathcal{S}}~
\big[\|X_{a}(x,t)\|\big]{{\mathscr{R}}}(x)
\end{align}
with averages ${\bm{\mathbb{E}}}\big\langle \nabla_{a}{{\mathscr{X}}}_{a}(x,t)\big\rangle=\nabla_{a}X_{a}(x,t)$ and  ${\bm{\mathbb{E}}}\big\langle \frac{\partial}{\partial t}{{\mathscr{X}}}_{a}(x,t)\big\rangle=\frac{\partial}{\partial t}X_{a}(x,t)$
\item If the field is constant throughout $\bm{\mathfrak{D}}$ then $\|X_{a}(x,t)\|=\|X_{a}\|=const$ so that ${\psi}(\|X_{a}\|)$ is constant giving the random field
\begin{align}
{{\mathscr{X}}}_{a}(x)=X_{a}+{\bm{\alpha}}X_{a}{\psi}(\|X_{a}\|)\mathbb{I}_{\mathcal{S}}~
\big[\|X_{a}\|\big]{{\mathscr{R}}}(x)
\end{align}
\end{enumerate}
\end{prop}
An alternative definition or ansatz is as follows.
\begin{prop}
Let the scenario of \textbf{Prop (2.8)} hold. Given the smooth vector field $X_{a}(x,t)$ define a spatial or volume-averaged field over domain $\bm{\mathfrak{D}}$ at each $t>0$ by
\begin{align}
{\mathcal{X}}_{a}(\bm{\mathfrak{D}},t)=\frac{1}{\mathrm{Vol}(\bm{\mathfrak{D}})}{\int}_{\bm{\mathfrak{D}}}
X_{a}(x,t)d^{3}x
\end{align}
with norm
\begin{align}
\big\|{\mathcal{X}}_{a}(\bm{\mathfrak{D}},t)\|=\left\|\frac{1}{\mathrm{Vol}(\bm{\mathfrak{D}})}{\int}_{\bm{\mathfrak{D}}}
X_{a}(x,t)d^{3}x\right\|
\end{align}
For a constant field $X_{a}(x,t)=X_{a}$
\begin{align}
{\mathcal{X}}_{a}(\bm{\mathfrak{D}})=\frac{1}{\mathrm{Vol}(\bm{\mathfrak{D}})}{\int}_{\bm{\mathfrak{D}}}
X_{a}d^{3}x=\frac{1}{\mathrm{Vol}(\bm{\mathfrak{D}})}X_{a}{\int}_{\bm{\mathfrak{D}}}d^{3}x=\frac{1}{\mathrm{Vol}(\bm{\mathfrak{D}})}X_{a}
\mathrm{Vol}(\bm{\mathfrak{D}})=X_{a}
\end{align}
Then the mixing ansatz gives a random vector field of the form
\begin{align}
{{\mathscr{X}}}_{a}(x,t)=X_{a}(x,t)+{\bm{\alpha}}X_{a}(x,t)
{{\psi}}\big(\|X_{a}(\bm{\mathfrak{D}},t)\|\big)\mathbb{I}_{\mathcal{S}}~
\big[\|{\mathcal{X}}_{a}(\bm{\mathfrak{D}},t)\|\big]{{\mathscr{R}}}(x)
\end{align}
The functional ${\psi}$ now vanishes if $\|X_{a}(\bm{\mathfrak{D}},t)\le \mathbf{C}$ and the spatial derivative $\nabla_{b}{\psi}=0$. This proposition in particular, will be applied to fluid mechanics and turbulence later in the paper.
\end{prop}
\begin{prop}(\textbf{Stochastic Euclidean norm of a random vectorial field}\newline
Given the random vector field ${\mathscr{X}}_{a}(x,t)$ or a generic RVF, the \textbf{$p^{th}$-order stochastic Euclidean norm} will be defined as
\begin{align}
&\big\|\!\big\|{{\mathscr{X}}}_{a}(x,t)\big\|\!\big\|_{SE_{p}(\bm{\mathfrak{D}}))}=\left({\bm{\mathbb{E}}}\left\langle\sum_{i=1}^{3}
|{{\mathscr{X}}}_{a}(x,t)|^{p}\right\rangle\right)^{1/p};~~~~x~~{\in}\bm{\mathfrak{D}},~t>0\\&
\big\|\!\big\|{{\mathscr{X}}}_{a}(x,t)\big\|\!\big\|^{p}_{SE_{p}(\bm{\mathfrak{D}}))}={\bm{\mathbb{E}}}\left\langle\sum_{i=1}^{3}
|{{\mathscr{X}}}_{a}(x,t)|^{p}\right\rangle\equiv\sum_{i=1}^{3}{\bm{\mathbb{E}}}\left\langle|{{\mathscr{X}}}_{a}(x,t)|^{p}\right\rangle  ;~~~~x~~{\in}\bm{\mathfrak{D}},~t>0
\end{align}
Then:
\begin{enumerate}
\item For p=1
\begin{align}
\big\|\!\big\|{{\mathscr{X}}}_{a}(x,t)\big\|\!\big\|_{SE_{p}(\bm{\mathfrak{D}}))}={\bm{\mathbb{E}}}\left\langle\sum_{i=1}^{3}
|{{\mathscr{X}}}_{a}(x,t)|\right\rangle\equiv\sum_{i=1}^{3}{\bm{\mathbb{E}}}\left\langle|{{\mathscr{X}}}_{a}(x,t)|\right\rangle  =0 ;~~~~x~~{\in}\bm{\mathfrak{D}},~t>0
\end{align}
\item For $p=2$
\begin{align}
&\big\|\!\big\|{{\mathscr{X}}}_{a}(x,t)\big\|\!\big\|_{SE_{2}(\bm{\mathfrak{D}}))}=\left({\bm{\mathbb{E}}}\left\langle\sum_{i=1}^{3}
|{{\mathscr{X}}}_{a}(x,t)|^{2}\right\rangle\right)^{1/2};~~~~x~~{\in}\bm{\mathfrak{D}},~t>0\\&
\big\|\!\big\|{{\mathscr{X}}}_{a}(x,t)\big\|\!\big\|^{2}_{SE_{2}(\bm{\mathfrak{D}}))}={\bm{\mathbb{E}}}\left\langle\sum_{i=1}^{3}
\big|{{\mathscr{X}}}_{a}(x,t)|^{2}\right\rangle\equiv\sum_{i=1}^{3}{\bm{\mathbb{E}}}
\left\langle|{{\mathscr{X}}}_{a}(x,t)|^{2}\right\rangle  ;~~~~x~~{\in}\bm{\mathfrak{D}},~t>0
\end{align}
\item For a random scalar field we can also use the notation
\begin{align}
&\big\|\!\big\|{{\mathscr{R}}}(x,t)\big\|\!\big\|_{SE_{1}(\bm{\mathfrak{D}}))}\equiv {\bm{\mathbb{E}}}\langle{{\mathscr{R}}}(x,t)\rangle\\&
\big\|\!\big\|{{\mathscr{R}}}(x,t){{\otimes}}{{\mathscr{R}}}(x,t)\big\|\!\big\|_{SE_{1}(\bm{\mathfrak{D}}))}\equiv {\bm{\mathbb{E}}}\langle{{\mathscr{R}}}(x,t){{\otimes}}{{\mathscr{R}}}(x,t)\rangle
\end{align}
\end{enumerate}
\end{prop}
\subsection{Bargmann-Fock random fields}
In this paper, we will utilise Bargmann-Fock random scalar fields, which have a Gaussian-decaying correlation function or kernel.
This will facilitate various computations and estimates.
\begin{defn}(\textbf{{Bargmann-Fock random fields}})\newline
Let ${{\mathscr{R}}}(x,t)$ be a centred, isotropic classical spatio-temporal Gaussian random field defied with respect to a probability triplet. Then define the following binary 'super Gaussian' correlation kernel for all $(x,y)\in\bm{\mathfrak{D}}$  and $(s,t)>0$ such that
\begin{align}
\mathbf{Cov}(x,y;t,s)=\bm{\mathbb{E}}\langle{{\mathscr{R}}}(x,t){{\otimes}}~{{\mathscr{R}}}(y,s)\rangle
={\bm{\Xi}}(x,y;\lambda)\bm{\varphi}(t,s)
=\mathsf{C}\exp\left(-\frac{\|x-y\|^{\kappa}}{\lambda^\kappa}\right)\bm{\varphi}(t,s;\tau)
\end{align}
with $\mathbb{E}\langle{{\mathscr{R}}}(x,t)\rangle=0$, and where $\kappa\in\mathbb{R}^{+}$ and $\varphi(t,t)=1$.
For example, $\varphi(s,t;\tau)= \exp(-|t-s|^{\kappa}/\tau^{\kappa})$. For a spurely spatial random field ${\mathscr{R}}(x)$
\begin{align}
\mathbf{Cov}(x,y)=\bm{\mathbb{E}}\langle{{\mathscr{R}}}(x){{\otimes}}~{{\mathscr{R}}}(y)\rangle
={\bm{\Xi}}(x,y;\lambda)=\mathsf{C}\exp\left(-\frac{\|x-y\|^{\kappa}}{\lambda^\kappa}\right)
\end{align}
This kernel is then regulated, smooth and differentiable and it is fast decaying with respect to the correlation length $\lambda$.
\begin{enumerate}[(a)]
\item For $\kappa>2$ this is a 'super Gaussian' correlation.
\item For $\kappa=1$, ${{\mathscr{R}}}(x)={{\mathscr{C}}}(x)$ this is a 'colored noise' so that
\begin{align}
{\bm{\mathbb{E}}}\langle{{\mathscr{C}}}(x){{\otimes}}~{{\mathscr{C}}}(y)\rangle
={\bm{\Xi}}(x,y;\lambda)=\mathsf{C}\exp\left(-\frac{\|x-y\|}{\lambda}\right)
\end{align}
\item for $\kappa=2$,let ${{\mathscr{R}}}(x)={{\mathscr{B}}}(x)$ or ${{\mathscr{B}}}(x)$ . These regulated fields have the following fast-decaying Gaussian binary correlation kernels
\begin{align}
&\big\|\!\big\|{{\mathscr{B}}}(x)\big\|\!\big\|_{SE_{1}(\bm{\mathfrak{D}})}=\bm{\mathbb{E}}\langle{{\mathscr{B}}}(x)\rangle=0
\\&\|\!\|{{\mathscr{B}}}((x){{\otimes}}~{{\mathscr{B}}}((y)\|\!\|_{SE_{1}(\bm{\mathfrak{D}})}=
{\bm{\mathbb{E}}}\langle{{\mathscr{B}}}(x){{\otimes}}~{{\mathscr{B}}}(y)\rangle
=\bm{\Xi}(x,y;\lambda)=\mathsf{C}\exp\left(-\frac{\|x-y\|^{2}}{\lambda^{2}}\right)
\end{align}
Then ${{\mathscr{B}}}(x,t)$ and ${{\mathscr{B}}}(x)$ are essentially \textbf{Bargmann-Fock (BF) random fields} with correlation length $\ell$ and $\mathsf{C}$ is a (dimensionless) constant. The correlation does not need to be normalised. For this paper we consider only the spatial dependence so that ${{\mathscr{B}}}(x,t)={{\mathscr{B}}}(x)$.
\item The BF field in $\mathbb{R}^{d}$ can also be realized as the following series, where $\alpha_{i_{1}....i_{d}}$ are iid standard Gaussian variables so that
\begin{align}
{{\mathscr{B}}}(x)=\exp\left(-\frac{\|x\|^{2}}{\lambda^{2}}\right)\sum_{1_{1}...i_{d}}\alpha_{i_{1}....i_{d}}
\frac{x_{1}^{i_{1}},...x_{d}^{i_{d}}}{i_{1}!...i_{d}!}
\end{align}
In \textbf{[85]} it is explained (with $\lambda^{-2}=\frac{1}{2})$ how the field arises within algebraic geometry when considering random homogeneous polynomials.
\item The field is regulated so that for $x=y$
\begin{align}
&\|\!\|{{\mathscr{B}}}(x,t){{\otimes}}~{{\mathscr{B}}}(x,t)\|\!\|_{SE_{1}(\bm{\mathfrak{D}}}
={\bm{\mathbb{E}}}\langle{{\mathscr{B}}}(x,t){{\otimes}}~{{\mathscr{B}}}(x,t)\rangle
={\mathsf{C}}={\mathsf{C}}\\&\big\|\!\big\|{{\mathscr{B}}}(x,t)\big\|\!\big\|_{SE_{2}(\bm{\mathfrak{D}})}
={\bm{\mathbb{E}}}\langle{{\mathscr{B}}}(x,t)\rangle=0\\&
\|\!\|{{\mathscr{B}}}(x){{\otimes}}~{{\mathscr{B}}}(x)\|\!\|_{SE_{2}(\bm{\mathfrak{D}})}^{2}
={\bm{\mathbb{E}}}\langle{{\mathscr{B}}}(x){{\otimes}}~{{\mathscr{B}}}(x)\rangle=\mathsf{C}
\end{align}
\item The BF field is isotropic and homogenous so that for any $\ell>0$
\begin{align}
&{\bm{\mathbb{E}}}\langle{{\mathscr{B}}}(x+\bm{\ell}){{\otimes}}~{{\mathscr{B}}}(y+\mathbf{\ell})\rangle=
{\bm{\mathbb{E}}}\langle{{\mathscr{B}}}(x){{\otimes}}~{{\mathscr{B}}}(y)\rangle\nonumber\\&
={\bm{\Xi}}(x+\mathbf{\ell},y+\mathbf{\ell};\lambda)={{\bm{\Xi}}}(x,y;\lambda)
={\mathsf{C}}\exp\left(-\frac{\|x-y\|^{2}}{\lambda^{2}}\right)
\end{align}
\end{enumerate}
\end{defn}
Physically, such a field has also been used to model a particles in an (impenetrable) spherical 'box' with a random potential, and the 1-dimensional Ising model at zero temperature with a random potential \textbf{[86]}. For $Z\in\mathbb{R}$, consider the Hamiltonian
\begin{align}
{\mathrm{H}}(Z)=\frac{1}{2}\mu|Z|^{2}+{\mathscr{V}}(Z)
\end{align}
where ${\mathscr{V}}(Z)$ is a random function or random potential given by ${\mathscr{V}}(Z)=f{\mathscr{B}}(Z)$, with
\begin{align}
\langle{{\mathscr{B}}}(Z){{\otimes}}{{\mathscr{B}}}(Z^{\prime})\rangle=\frac{1}{\sqrt{2\pi}{\xi}}
\exp\left(-\frac{|Z-Z^{\prime}}{\xi^{2}}\right)
\end{align}
and $\langle{{\mathscr{B}}}(Z){{\otimes}}{{\mathscr{B}}}(Z^{\prime})\rangle=0$, and where $\langle\bullet\rangle$ is here an ensemble-type average and $f$ is a 'field strength'.

The BF field has the advantage that
\begin{align}
\bm{\mathbb{E}}\langle{{\mathscr{B}}}(x){{\otimes}}~\nabla_{b}{{\mathscr{B}}}(x)\rangle=0
\end{align}
vanishes which will simplify most computations, estimates and derivations.
(See Appendix.A)
Note that the Bargman-Fock field ${{\mathscr{B}}}(x)$ can also be considered as a 'smeared out' white noise ${\mathscr{W}}(x)$ since the delta function can be smeared into a very highly peaked Gaussian of very narrow width $\lambda$ so that
\begin{align}
&\lim_{\lambda\uparrow 0}={\bm{\mathbb{E}}}\langle{{\mathscr{B}}}(x){{\otimes}}~{{\mathscr{B}}}(y)\rangle=
{\bm{\mathbb{E}}}\langle{\mathscr{W}}(x){{\otimes}}~{\mathscr{W}}(y)\rangle\nonumber\\&
=\lim_{\xi\uparrow 0}{\mathsf{C}}\exp\left(-\frac{\|x-y\|^{2}}{\lambda^{2}}\right)={\mathsf{C}}\delta^{3}(x-y)
\end{align}
Since the BF field is regulated, the derivatives $\nabla_{a}{{\mathscr{B}}}(x),
\nabla_{a}\nabla_{b}{{\mathscr{B}}}(x),\Delta{{\mathscr{B}}}(x)$ exist and are also random Gaussian fields.

\begin{lem}[Gaussian kernel --- Mean-square differentiability, integrability, and applications]
Let $\mathscr{B}(x)$ be a mean-zero Gaussian process on an interval
$I \subset \mathbb{R}$ with squared-exponential covariance
\begin{align}
    \bm{\Xi}(x,y;\lambda)
    &= \rho^{2}\exp\!\left(-\frac{(x-y)^{2}}{2\lambda^{2}}\right),
    & \rho^{2}>0,\ \lambda>0. \label{eq:covariance}
\end{align}

\begin{enumerate}
    \item \textbf{Mean-square differentiability.}
    The mean-square derivative $\nabla_{x}\mathscr{B}(x)$ exists for all $x\in I$.
    Its covariance is
    \begin{align}
        \mathbb{E}\!\left[\nabla_x \mathscr{B}(x)\,\nabla_y \mathscr{B}(y)\right]
        &= \nabla_x \nabla_y \bm{\Xi}(x,y) \label{eq:first_derivative_cov} \\
        &= \mathsf{C}
          \left(\frac{1}{\lambda^{2}} - \frac{(x-y)^{2}}{\lambda^{4}}\right)
          \exp\!\left(-\frac{(x-y)^{2}}{2\lambda^{2}}\right), \nonumber
    \end{align}
    and the variance on the diagonal is
    \begin{align}
        \operatorname{Var}\!\left[\nabla_x\mathscr{B}(x)\right] = \frac{\mathsf{C}}{\lambda^{2}}. \label{eq:first_derivative_var}
    \end{align}

    \item \textbf{Higher-order mean-square differentiability.}
    For integers $p,q\ge 0$,
    \begin{align}
        \mathbb{E}\!\left[\nabla_x^{p}\mathscr{B}(x)\,\nabla_y^{q}\mathscr{B}(y)\right]
        &= \nabla_x^{p}\nabla_y^{q}\bm{\Xi}(x,y). \label{eq:higher_order_cov}
    \end{align}
    Using $s=x-y$, we have
    \begin{align}
        \nabla_x = \nabla_s, \qquad \nabla_y = -\nabla_s, \nonumber
    \end{align}
    so that
    \begin{align}
        \mathbb{E}\!\left[\nabla_x^{p}\mathscr{B}(x)\,\nabla_y^{q}\mathscr{B}(y)\right]
        = (-1)^q \mathsf{C} \frac{d^{\,p+q}}{ds^{\,p+q}}
        \exp\!\left(-\frac{s^{2}}{2\lambda^{2}}\right) \Big|_{s=x-y}. \label{eq:hermite_form}
    \end{align}
    On the diagonal ($x=y$), the variance is
    \begin{align}
        \operatorname{Var}\!\left[\nabla_x^{p}\mathscr{B}(x)\right]
            = \nabla_x^{p}\nabla_y^{p}\bm{\Xi}(x,y)\Big|_{y=x}. \label{eq:higher_order_var}
    \end{align}
     \item \textbf{Mean-square integrability.}
    Define the line integral
    \begin{align}
        \mathscr{H}(x) := \int_{0}^{x} \mathscr{B}(u)\,du. \label{eq:line_integral}
    \end{align}
    Its covariance exists in the mean-square sense:
    \begin{align}
        \mathbb{E}\!\left[\mathscr{H}(x)\mathscr{H}(y)\right]
            = \int_{0}^{x}\int_{0}^{y} \bm{\Xi}(u,v)\,du\,dv < \infty. \label{eq:integral_cov}
    \end{align}
\end{enumerate}
\textit{Summary.} Mean-square differentiability requires a smooth covariance ($\lambda>0$).
Mean-square integrability holds whenever the covariance is finite.
In the white-noise limit $\lambda\to 0$, derivatives fail but integrals remain well-defined.
\end{lem}
\subsection{Spectral-Fourier representation}
The BF random fields have a Fourier representation in $\mathbf{k}$-space.
\begin{defn}
If $\mathfrak{F}:\mathbb{R}\rightarrow\mathbb{K}$ is a Fourier transform then a generic random
Gaussian scalar field ${\mathscr{R}}(x)$ is said to be 'harmonisable' if
\begin{align}
{{\mathscr{R}}}(x)={\int}_{\mathbb{R}}\exp(i\mathbf{k}_{a}x^{a}){{\mathscr{R}}}(k)d\mu(\mathbf{k})
\end{align}
and $d\mu(\mathbf{k})\equiv d^{3}k$.
\end{defn}
\begin{lem}(\textbf{\textrm{Spectral-Fourier~representation of the correlation}})\newline
Let ${{\mathscr{R}}}(x)$ be an arbitrary harmonisable Gaussian random scalar field existing for all $x\in{\mathbb{R}}^{3}$. Given the basic Fourier representation of the binary correlation
\begin{align}
{\bm{\mathbb{E}}}\langle{{\mathscr{R}}}(x){{\otimes}}~{{\mathscr{R}}}(y)
\rangle={\int}_{\mathbb{K}^{3}}d^{3}\mathbf{k}\bm{\Phi}(\mathbf{k})\exp(i\mathbf{k}_{a}(x-y)^{a})
\end{align}
where $\Phi(\mathbf{k})$ is a spectral function, then for $x=y$
\begin{align}
{\bm{\mathbb{E}}}\langle{{\mathscr{R}}}(x){{\otimes}}~{{\mathscr{R}}}(y)
\rangle={\int}_{\mathbf{K}^{3}}\bm{\Phi}(k)d^{3}k
\end{align}
\begin{enumerate}[(a)]
\item For $\bm{\Phi}(\mathbf{k})=\mathsf{C}=const.$, one recovers an unregulated white noise with
\begin{align}
{\bm{\mathbb{E}}}\langle{{\mathscr{R}}}(x){{\otimes}}~{{\mathscr{R}}}(y)
={\bm{\mathbb{E}}}\langle{\mathscr{W}}(x){{\otimes}}~{\mathscr{W}}(y)\rangle
=\mathsf{C}\delta^{3}(x-y)
\end{align}
\item For $\bm{\Phi}(k)=\frac{\mathsf{C}}{\mathbf{k}^{2}}\exp(-\frac{1}{4}\lambda^{2}\mathbf{k}^{2})$ one recovers a Bargmann-Fock random field such that
\begin{align}
\bm{\mathbb{E}}\langle{\mathscr{R}}(x){{\otimes}}~{\mathscr{R}}(y)\rangle\equiv
\bm{\mathbb{E}}\langle{\mathscr{B}}(x){{\otimes}}~{\mathscr{B}}(y)\rangle=\mathsf{C}
\exp\left(-\frac{\|x-y\|^{2}}{\lambda^{2}}\right)
\end{align}
\end{enumerate}
\end{lem}
The proof is given in Appendix A.
\begin{lem}(\textbf{Moments of BF fields}\newline
Given the Gaussian correlation $\bm{\mathbb{E}}\big\langle{\mathscr{B}}(x){{\otimes}}{\mathscr{B}}(y)\big\rrbracket=\mathsf{C}\exp(-\|x-y\|\lambda^{-2})$ and
$\bm{\mathbb{E}}\big\langle{\mathscr{B}}(x)\big\rangle=0$ then higher-order correlations at the same point $y=x$ have the form
\begin{align}
&{\bm{\mathbb{E}}}\langle{{\mathscr{B}}}(x){{\otimes}}{{\mathscr{B}}}(x)\rangle=\mathsf{C}\\&
{\bm{\mathbb{E}}}\langle{{\mathscr{B}}}(x){{\otimes}}{\mathscr{B}}(x){{\otimes}}{{\mathscr{B}}}(x)\rangle=
{\bm{\mathbb{E}}}\langle{{\mathscr{B}}}(x){{\otimes}}{{\mathscr{B}}}(x)\rangle
{\bm{\mathbb{E}}}\langle{\mathscr{B}}(x)\rangle=0\\&
{\bm{\mathbb{E}}}\langle{{\mathscr{B}}}(x){{\otimes}}{{\mathscr{B}}}(x){{\otimes}}{\mathscr{B}}(x)
{{\otimes}}{{\mathscr{B}}}(x)\rangle=
{\bm{\mathbb{E}}}\langle{{\mathscr{B}}}(x){{\otimes}}{{\mathscr{B}}}(x)\rangle
{\bm{\mathbb{E}}}\langle{{\mathscr{B}}}(x){{\otimes}}{{\mathscr{B}}}(x)\rangle
=\mathsf{C}^{2}
\end{align}
and so on, where all odd moments vanish. The general $N^{th}$-order moments are then
\begin{align}
{\bm{\mathbb{E}}}\langle{|{{\otimes}}{\mathscr{B}}}(x)|^{N}\rangle=\frac{1}{2}\big[\mathsf{C}^{N/2}+(-1)^{N}\mathsf{C}^{N/2}]
\end{align}
\end{lem}
\subsection{Binary correlation for derivatives of BF random fields}
We will also require the binary correlations for  ${\partial}_{t}{{\mathscr{B}}}(x,t)$ and derivatives $\nabla_{a}{{\mathscr{B}}}(x,t)$ with themselves and with fields
${{\mathscr{B}}}(x,t)$, at pairs of points $(x,y)\in\bm{\mathfrak{D}}$ and at equal points when $x=y$ such as
\begin{align}
&{\bm{\mathbb{E}}}\langle{{\mathscr{B}}}(x){{\otimes}} \nabla_{a}{{\mathscr{B}}}(y)\rangle\\&
={\bm{\mathbb{E}}}\langle{{\mathscr{B}}}(x){{\otimes}}~\nabla_{a}{{\mathscr{B}}}(x)\rangle
=\lim_{y\uparrow x}\bm{\mathbb{E}}\langle{{\mathscr{B}}}(x){{\otimes}}~\nabla_{a}{{\mathscr{B}}}(y)\rangle
\\&={\bm{\mathbb{E}}}\langle \nabla_{a}^{(x)}{{\mathscr{B}}}(x){{\otimes}}~\nabla_{b}^{(x)}{{\mathscr{B}}}(x)\rangle
=\lim_{y\rightarrow x}{\bm{\mathbb{E}}}\langle \nabla_{a}^{(x)}{{\mathscr{B}}}(x){{\otimes}}~\nabla_{b}^{(y)}
{{\mathscr{B}}}(y)\rangle
\end{align}
where $\nabla_{a}^{(x)}=\frac{\partial}{\partial x}$ and  $\nabla_{b}^{(y)}=\frac{\partial}{\partial y}$. For spatio-temporal random fields ${\mathscr{B}}(x,t)$ we have the following correlations
\begin{align}
&{\bm{\mathbb{E}}}\langle{{\mathscr{B}}}(x,t){{\otimes}} \nabla_{a}{{\mathscr{B}}}(y,s)\rangle\\&
{\bm{\mathbb{E}}}\langle{{\mathscr{B}}}(x,t){{\otimes}}~\nabla_{a}{{\mathscr{B}}}(x,t)\rangle=
\lim_{y\rightarrow x}\lim_{t\uparrow s}\bm{\mathbb{E}}\langle {{\mathscr{B}}}(x,t){{\otimes}}~\nabla_{a}{{\mathscr{B}}}(y,s)\rangle
\\&{\bm{\mathbb{E}}}\langle \nabla_{a}^{(x)}{{\mathscr{B}}}(x,t){{\otimes}}~\nabla_{b}^{(x)}
{{\mathscr{B}}}(x,t)\rangle=\lim_{y\rightarrow x}\lim_{t\uparrow s}
{\bm{\mathbb{E}}}\langle \nabla_{a}^{(x)}{{\mathscr{B}}}(x,s){{\otimes}}~\nabla_{b}^{(y)}{{\mathscr{B}}}(y,t)\rangle\\&
{\bm{\mathbb{E}}}\langle{{\mathscr{B}}}(x,t){{\otimes}}~\nabla_{a}{{\mathscr{B}}}(x,s)\rangle=
\lim_{y\rightarrow x}{\bm{\mathbb{E}}}\langle\ {{\mathscr{B}}}(x,t){{\otimes}}~\nabla_{a}{{\mathscr{B}}}(y,s)\rangle
\\&{\bm{\mathbb{E}}}\langle \nabla_{a}^{(x)}{{\mathscr{B}}}(x,t){{\otimes}}~\nabla_{b}^{(x)}{{\mathscr{B}}}(x,s)\rangle
=\lim_{y\rightarrow x}{\bm{\mathbb{E}}}\langle \nabla_{a}^{(x)}{{\mathscr{B}}}(x,s){{\otimes}} \nabla_{b}^{(y)}
{{\mathscr{B}}}(y,t)\rangle
\end{align}
\begin{lem}
Given a BF spatio-temporal random field ${{\mathscr{B}}}(x,t)$ with the full binary covariance
\begin{align}
&{\bm{\mathbb{E}}}\langle{{\mathscr{B}}}(x,t){{\otimes}}~{{\mathscr{B}}}(y,t)\rangle\nonumber
={\bm{\Xi}}(x,y;\lambda)\bm{\varphi}(t,s;\eta)\nonumber\\&
=\mathsf{C}\exp(-\|x-y\|^{2}\lambda^{-2})\exp(-|t-s|^{2}\eta^{2})
\end{align}
then
\begin{align}
{\bm{\mathbb{E}}}\langle{{\mathscr{B}}}(x,t){{\otimes}}~{\partial}_{t}{{\mathscr{B}}}(x,t)\rangle=0
\end{align}
\end{lem}
\begin{proof}
The time derivative of the correlation is
\begin{align}
&{\bm{\mathbb{E}}}\langle{{\mathscr{B}}}(x,t){{\otimes}}{{\mathscr{B}}}(y,t)\rangle\nonumber\\&
={\bm{\mathbb{E}}}\langle{{\mathscr{B}}}(x,s){{\otimes}}~{\partial}_{t}{{\mathscr{B}}}(y,t)\rangle+
{\bm{\mathbb{E}}}\langle{\partial}_{t}{{\mathscr{B}}}(x,s){{\otimes}}~{{\mathscr{B}}}(y,t)\rangle
\nonumber\\&
={\bm{\mathbb{E}}}\langle{{\mathscr{B}}}(x,t){{\otimes}}\frac{\partial}{\partial t}{{\mathscr{B}}}(y,t)\rangle
\nonumber\\&={\bm{\Xi}}(x,y;\lambda){\partial}_{t}\bm{\varphi}(t,s;\eta)
={\mathsf{C}}\exp(-\|x-y\|^{2}\lambda^{-2}){\partial}_{t}\exp(-|t-s|^{2}/\eta^{2})
\end{align}
Then
\begin{align}
&{\bm{\mathbb{E}}}\langle{{\mathscr{B}}}(x,t){{\otimes}}~{\partial}_{t}{{\mathscr{B}}}(x,t)\rangle=\lim_{s\rightarrow t}{\bm{\mathbb{E}}}\langle{{\mathscr{B}}}(y,s){{\otimes}}~{\partial}_{t}{{\mathscr{B}}}(x,t)\rangle\nonumber\\&=\lim_{s\rightarrow t}-2{\bm{\Xi}}(x,y;\lambda)|t-s|\eta^{-2}\exp(-|t-s|^{2}\eta^{-2})=0
\end{align}
which vanishes as $s\rightarrow t$.
\end{proof}
The binary correlations of the derivatives of the purely spatial BF fields are
\begin{lem}
Let $(x,y)\in\bm{\mathfrak{D}}$  and let ${{\mathscr{B}}}(x)$ be a Gaussian BF random field with statistics
\begin{align}
\bm{\mathbb{E}}\langle{{\mathscr{B}}}(x){{\otimes}}~{{\mathscr{B}}}(y)\rangle={\bm{\Xi}}(x,y;\lambda)
=\mathsf{C}\exp(-\|x-y\|^{2}\lambda^{-2})
\end{align}
and $\bm{\mathbb{E}}\big\langle{{\mathscr{B}}}(x)\big\rangle=0$. Let $\nabla_{a}^{(x)}=\frac{\partial}{\partial x}$ and $\nabla_{b}^{(y)}=\frac{\partial}{\partial y}$ such that $\nabla_{a}^{(x)}{{\mathscr{B}}}(y)=0$ and  $\nabla_{b}^{(y)}{{\mathscr{B}}}(x)=0$.
The binary correlation of the derivatives is then
\begin{align}
&\nabla_{a}^{x}\nabla_{b}^{(y)}\bm{\mathbb{E}}\langle{{\mathscr{B}}}(x){{\otimes}}~{{\mathscr{B}}}(y)\rangle\equiv
{\bm{\mathbb{E}}}\langle \nabla_{a}^{(x)}{{\mathscr{B}}}(x){{\otimes}}~\nabla_{b}^{(y)}
{{\mathscr{B}}}(y)\rangle\nonumber\\&
=\nabla_{a}^{(x)}\nabla_{b}^{(y)}{\bm{\Xi}}(x,y;\lambda)=\frac{18\delta_{ab}}{\lambda^{2}}{\bm{\Xi}}(x,y;\lambda)-
\frac{36(x-y)_{a}(x-y)_{b}}{\lambda^{2}}{\bm{\Xi}}(x,y;\lambda)
\end{align}
The proof is given in Appendix B.
\end{lem}
The correlation of the derivatives of the field at the same point $x$ follows in the limit as $x\rightarrow y$.
\begin{cor}
\begin{align}
&{\bm{\mathbb{E}}}\langle \nabla_{a}^{(x)}{{\mathscr{B}}}(x){{\otimes}} \nabla_{a}^{(x)}{{\mathscr{B}}}(x)\rangle
=\lim_{y\uparrow x}\nabla_{a}^{(x)}\nabla_{b}^{(y)}\bm{\mathbb{E}}\langle{{\mathscr{B}}}(x){{\otimes}}~{{\mathscr{B}}}(y)\rangle\nonumber\\&\equiv
\lim_{y\uparrow x}{\bm{\mathbb{E}}}\langle \nabla_{a}^{(x)}{{\mathscr{B}}}(x){{\otimes}}~\nabla_{b}^{(y)}{{\mathscr{B}}}(y)\rangle
=\lim_{y\uparrow x}\nabla_{a}^{(x)}\nabla_{b}^{(y)}{\bm{\Xi}}(x,y;\lambda)\nonumber\\&=\lim_{y\uparrow x}\frac{18\delta_{ab}}{\lambda^{2}}{\bm{\Xi}}(\mathbf{x,y};\lambda)-\frac{36(x-y)_{a}(x-y)_{b}}{\lambda^{2}}=\frac{18\mathsf{C}\delta_{ab}}{\lambda^{2}}
\end{align}
\end{cor}
\begin{lem}
Given the Bargmann-Fock random field ${{\mathscr{B}}}(x)$ and its derivative $\nabla_{a}{{\mathscr{B}}}(x)$ then
\begin{align}
{\bm{\mathbb{E}}}\langle{{\mathscr{B}}}(x){{\otimes}}~\nabla_{a}{{\mathscr{B}}}(x)\rangle=0
\end{align}
\end{lem}
\begin{proof}
Since $\bm{\mathbb{E}}\big\langle{{\mathscr{B}}}(x){{\otimes}}{{\mathscr{B}}}(y)\big\rangle
={\bm{\Xi}}(x,y;\lambda)=\mathsf{C}\exp(-\|x-y\|^{2}\lambda^{-2})$ then taking the 1st derivative gives
\begin{align}
&\nabla_{a}^{(x)}\bm{\mathbb{E}}\langle{{\mathscr{B}}}(x){{\otimes}}~{{\mathscr{B}}}(y)\rangle
\equiv{\bm{\mathbb{E}}}\langle{{\mathscr{B}}}(y){{\otimes}}~\nabla_{a}^{(x)}{{\mathscr{B}}}(x)\rangle\nonumber\\&=
\nabla_{a}^{(x)}{\bm{\Xi}}(x,y;\lambda)=-\frac{6(x-y)_{a}}{\lambda^{2}}{\bm{\Xi}}(x,y;\lambda)
\end{align}
Then taking the limit gives the binary correlation for the field and its derivative at the same point
\begin{align}
&{\bm{\mathbb{E}}}\langle{{\mathscr{B}}}(x){{\otimes}}~\nabla_{a}{{\mathscr{B}}}(x)\rangle=\lim_{y\uparrow x}\equiv{\bm{\mathbb{E}}}\langle{{\mathscr{B}}}(y){{\otimes}}~\nabla_{a}^{(x)}{{\mathscr{B}}}(x)\rangle\nonumber\\&=
-\lim_{y\uparrow x}\frac{6(x-y)_{a}}{\lambda^{2}}{\bm{\Xi}}(x,y;\lambda)=0
\end{align}
\end{proof}
\begin{cor}
We also have the following identities
\begin{align}
\nabla_{a}{\bm{\mathbb{E}}}\langle{{\mathscr{B}}}(x)&{{\otimes}}~{{\mathscr{B}}}(x)\rangle=
{\bm{\mathbb{E}}}\langle \nabla_{a}{{\mathscr{B}}}(x){{\otimes}}~{{\mathscr{B}}}(x)
\rangle\nonumber\\&=\bm{\mathbb{E}}\langle{{\mathscr{B}}}(x){{\otimes}} \nabla_{a}{{\mathscr{B}}}(x)+
{\bm{\mathbb{E}}}\langle \nabla_{a}{{\mathscr{B}}}(x){{\otimes}}{{\mathscr{B}}}(x)\rangle=0\\&
\nabla_{a}\nabla_{b}\bm{\mathbb{E}}\langle{{\mathscr{B}}}(x){{\otimes}}~{{\mathscr{B}}}(x)\rangle
=\bm{\mathbb{E}}\langle \nabla_{a}\nabla_{a}{{\mathscr{B}}}(x){{\otimes}}{{\mathscr{B}}}(x)
\rangle\nonumber\\&
={\bm{\mathbb{E}}}\langle \nabla_{a}\nabla_{b}{{\mathscr{B}}}(x){{\otimes}}~{{\mathscr{B}}}(x)\rangle
+{\bm{\mathbb{E}}}\langle \nabla_{a}{{\mathscr{B}}}(x){{\otimes}}~\nabla_{b}{{\mathscr{B}}}(x)\rangle\nonumber\\&
={\bm{\mathbb{E}}}\langle{{\mathscr{B}}}(x){{\otimes}}~\nabla_{a}\nabla_{b}{{\mathscr{B}}}(x)\rangle
+{\bm{\mathbb{E}}}\langle \nabla_{b}{{\mathscr{B}}}(x){{\otimes}}~\nabla_{a}{{\mathscr{B}}}(x)\rangle\nonumber\\&
={\bm{\mathbb{E}}}\langle \nabla_{a}{{\mathscr{B}}}(x){{\otimes}}~\nabla_{b}{{\mathscr{B}}}(x)\rangle
+{\bm{\mathbb{E}}}\langle \nabla_{b}{{\mathscr{B}}}(x){{\otimes}}~\nabla_{a}{{\mathscr{B}}}(x)\rangle=\frac{36\delta_{ab}C}{\lambda^{2}}
\end{align}
\end{cor}
\subsection{Reynolds~number~and~the~transition~of laminar flow to~turbulent flow}
To quote Von Neumann in (ref):"\textit{The transition from 'laminar' flow $\mathfrak{L}$ to fully turbulent flow $\mathfrak{T}$ is best defined by a critical value of the Reynolds number than by any other geometric quantity}". We will therefore define a suitable Reynolds number which can either vary through a finite domain or be spatially averaged over the domain. A mathematical description of the transition $\mathfrak{L}\rightarrow\mathfrak{T}$ is fraught with mathematical difficulties and may even be intractable. No truly rigorous theory or description exists. The transition to turbulence in the fluid typically occurs for ${\mathbf{Re}}\sim 2000$. In the 1940s and even into the 1970s, the accepted theory was that of Landau and Hopf \textbf{[87]}. They (independently) proposed a 'branching theory' whereby a smooth laminar flow essentially undergoes an 'infinity of transitions', during which an additional frequencies (or wavenumbers) arise due to flow instabilities, leading to complex turbulent motion. However, this (very) heuristic theory has been shown to be untenable in virtually all turbulence scenarios and was never experimentally observed or verified, and is based on linearised approximations. Nevertheless, it has an interesting feature in that \textbf{the amplitude grows and turbulence evolves as a power law of the difference of the Reynolds number with the critical Reynolds number.}

In the Landau-Hopf theory a laminar flow $\mathfrak{L}$ with velocity $U^{o}(x,t)$ will remain a laminar and stable if $\mathbf{Re}<\mathbf{Re}_{c}$. If the Reynolds number increases slightly beyond the critical value then $\mathbf{Re}\ge\mathbf{Re}_{c}$ and within a linear approximation to the NS equations, the laminar flow becomes unstable to small perturbations of the form $\widehat{U^{(1)}(x,t)}=U^{(0)}(x,t)+U^{(1)}(x,t)$ where $U^{(1)}(x,t)=\mathcal{A}(t)f(x)=f(x)\exp(-i\omega_{1}(\mathbf{Re})t+\phi)$, where $\omega_{1}(\mathbf{Re})$ is a frequency that depends on $\mathbf{Re}$ and $\phi$ is an arbitrary phase. Via the linearised NS equations, the evolution of the amplitude $\mathcal{A}(t)$ is described by the equation (REFs)
\begin{align}
\frac{d\mathcal{A}(t)}{dt}=\alpha|\mathcal{A}(t)|^{2}-\bm{\bm{\alpha}}|\mathcal{A}(t)|^{4}
\end{align}
where $\alpha\sim |\mathbf{Re}-\mathbf{Re}_{c}|$ and $\bm{\bm{\alpha}}>0$. This has the solution
\begin{align}
|\mathcal{A}(t)|^{2}=\frac{\alpha|\mathcal{A}(0)\exp(\alpha t)}{\alpha-\bm{\bm{\alpha}}|\mathcal{A}(0)|^{2}(1-\exp(\alpha t))}
\end{align}
Then
\begin{empheq}[right=\empheqrbrace]{align}
&\lim_{t\rightarrow\infty}|\mathcal{A}(t)|^{2}=0,~for~\mathbf{Re}<\mathbf{Re}_{c}\nonumber\\&
\lim_{t\rightarrow\infty}|\mathcal{A}(t)|^{2}=\left|\frac{\alpha}{\bm{\bm{\alpha}}}\right|,~for~\mathbf{Re}>\mathbf{Re}_{c}\nonumber
\end{empheq}
The amplitude of oscillation is then proportional to the square root of the difference of the Reynolds number and critical Reynolds number
\begin{align}
\mathcal{A}(t)\sim \sqrt{|\mathbf{Re}-\mathbf{Re}_{c}|}
\end{align}
One can proceed indefinitely for increasing $\mathbf{Re}$ so that $U^{(n)}(x,t)=U^{(n-1)}(x,t)+U^{(n)}(x,t)$ and
have a turbulent flow of the form
\begin{align}
U(x,t)=\sum_{n=-\infty}^{\infty}f(x)\exp(-i\omega_{n}t-i\phi_{n})
\end{align}
with many frequencies and phases.

Although this theory is untenable, one can be inspired to try and define a random BF field representing a turbulent flow that grows more in amplitude and randomness with increasing Reynolds number above some critical value.
\begin{prop}($\underline{\bm{Volume-averaged~Reynolds~number~within~a~domain}}$)\newline
Let $\bm{\mathfrak{D}}\subset{\mathbb{R}}^{3}$ be a "subdomain" with volume $\mathrm{Vol}(\bm{\mathfrak{D}})\sim \mathrm{L}^{3}$ and $\mathrm{Vol}(\bm{\mathfrak{D}})\sim \mathcal{L}^{3}$ with $\mathrm{L}\ll\mathcal{L}^{3}$. Let $\bm{\mathfrak{D}}$ contain a fluid of viscosity $\nu$ and velocity $U_{a}(x,t)$ for $(x,t)\in\bm{\mathfrak{D}}\times[0,T]$. The fluid velocity evolves according to the Navier-Stokes equations for initial Cauchy data $U_{a}(x,0)=U_{a}^{(0)}$ and some boundary conditions on $\partial\bm{\mathfrak{D}}$, and is at least weakly Leray Hopf and bounded so that $\|U_{a}(x,t)\|\le \mathrm{K}$ for $(x,t)\in\bm{\mathfrak{D}}\times[0,T]$. We consider the subdomain $\bm{\mathfrak{D}}$ and how turbulence might evolves within this subdomain as the Reynolds number increases. The domain $\bm{\mathfrak{D}}$ can be considered as the union of a set of subdomains $\lbrace\bm{\mathfrak{D}}_{\xi}\rbrace$ so that $\bm{\mathfrak{D}}=\bigcup_{\xi=1}^{K}\bm{\mathfrak{D}}_{\xi}$ with volume
\begin{align}
\mathrm{Vol}(\bm{\mathfrak{D}})=\mathrm{Vol}\bigcup_{\xi=1}^{K}{\bm{\mathfrak{D}}}_{\xi}\equiv\bigcup_{\xi=1}^{K}\mathrm{Vol}(\bm{\mathfrak{D}}_{\xi})\sim\mathrm{L}^{3}
\end{align}
This can be a cube, sphere or cylinder. As in \textbf{Prop(2.8)} one can now apply a volume average over $\bm{\mathfrak{D}}$ so that
\begin{align}
\big\|{\bm{\mathfrak{U}}}_{a}(\bm{\mathfrak{D}},t)\big\|
=\left\|\frac{1}{Vol[\bm{\mathfrak{D}}]}{\int}_{\bm{\mathfrak{D}}}U_{a}(x,t)d\mu(\bm{\mathfrak{D}})\right\|
\end{align}
where the integration measure is $d\mu(\bm{\mathfrak{D}})=d^{3}x$.
\end{prop}
\begin{defn}
The averaged Reynolds number within $\bm{\mathfrak{D}}$ is then denoted
\begin{align}
{\bm{\mathrm{Re}}}(\bm{\mathfrak{D}},t)=\frac{{\int}_{\bm{\mathfrak{D}}}\left\|U_{a}(x,t)\right\|d\mu(\bm{\mathfrak{D}})
\mathrm{L}}{\mathrm{Vol}(\bm{\mathfrak{D}}\mu}
=\frac{\|{{{\mathfrak{u}}}}_{a}(t)\|\mathrm{L}}{\nu}
\end{align}
\emph{This is the averaged Reynolds number within $\bm{\mathfrak{D}}$ at any time $t\in[0,T]$.} The derivatives are then
\begin{align}
&\frac{\partial}{\partial t}\bm{\mathrm{Re}}(\bm{\mathfrak{D}},t)=\frac{1}{\nu}\frac{\partial}{\partial t}\left\|\frac{1}{\mathrm{Vol}[\bm{\mathfrak{D}}]}{\int}_{\bm{\Omega}}U_{a}(x,t)d\mu(\bm{\mathfrak{D}})\right\|\mathrm{L}
=\frac{\frac{\partial}{\partial t}\|\bm{\mathfrak{U}}_{a}(t)\|\mathrm{L}}{\nu}\\&
\nabla_{b}\bm{\mathrm{Re}}(\bm{\mathfrak{D}},t)=\frac{1}{\nu}\frac{\partial}{\partial t}\left\|\frac{1}{\mathrm{Vol}[\bm{\mathfrak{D}}]}{\int}_{\bm{\Omega}}U_{a}(x,t)d\mu(\bm{\mathfrak{D}})\right\|\mathrm{L}
=\frac{\nabla_{b}\|\bm{\mathfrak{U}}_{a}(t)\|\mathrm{L}}{\nu}=0
\end{align}
but only the temporal derivative is non zero. For a constant fluid velocity field $U_{a}(x,t)=U_{a}$
\begin{align}
{\mathcal{U}}_{a}(\bm{\mathfrak{D}})=\frac{1}{\mathrm{Vol}(\bm{\mathfrak{D}})}{\int}_{\bm{\mathfrak{D}}}
U_{a}d^{3}x=\frac{1}{\mathrm{Vol}(\bm{\mathfrak{D}})}U_{a}{\int}_{\bm{\mathfrak{D}}}d^{3}x
=\frac{1}{\mathrm{Vol}(\bm{\mathfrak{D}})}U_{a}\mathrm{Vol}(\bm{\mathfrak{D}})=U_{a}
\end{align}
\end{defn}
\begin{figure}[htb]
\begin{center}
\includegraphics[height=2.0in,width=2.0in]{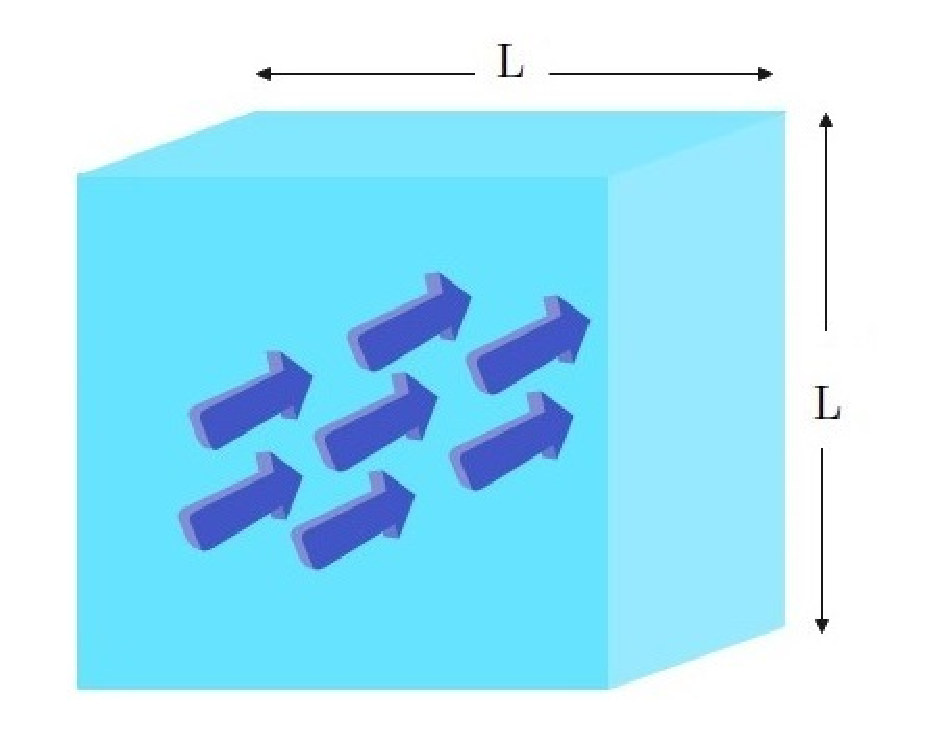}
\caption{Volume-averaged Reynolds number within domain $\bm{\mathfrak{D}}$, with smooth fluid flow $U_{a}(x,t)$ and viscosity $\nu$}
\end{center}
\end{figure}
Since $\mathrm{L},\|\mathfrak{A}_{a}(x,t)\|,\nu)$ have the dimensions $cm,~cms^{-1}~and~cm^{2}s^{-1}$ it follows that $\bm{\mathrm{Re}}(\bm{\mathfrak{D}},t)$ is dimensionless.
\begin{prop}
The \textbf{critical or threshold Reynolds number $\bm{\mathrm{Re}}_{c}(\bm{\mathfrak{D}})$ } will be taken to be a critical or threshold value at which the flow makes a transition from 'laminar' to 'turbulent' flow. The transition will be taken to be 'sharp' or near instantaneous.
\end{prop}
This is highly idealised so some further remarks on the Reynolds number are in order. The Landau-Hopf theory of the turbulence transition and its limitations was discussed.  It is clear that laminar flow in a tube for example, can exist in a 'metastable(' state just above a critical value $\bm{\mathrm{Re}}_{c}^{(1)}(\bm{\mathfrak{D}})$. Turbulent flow cannot exist below this critical value. There is then a range of Reynolds number over which there is some uncertainty. Near $\bm{\mathrm{Re}}_{c}^{(1)}(\bm{\mathfrak{D}})$, there are also large fluctuations in pressure and the flow can fluctuate between laminar and turbulent states. Some systems also make a transition to a laminar flow in a progressive fashion rather than as a sharp transition. Turbulence tends not to be uniformly distributed throughout the flow but instead forms turbulent regions/structures that are interspersed by regions of laminar flow. This leads to the concept of 'intermittancy'.

For a sharp transition at $\bm{\mathrm{Re}}_{c}^{(1)}(\bm{\mathfrak{D}})$ this can be considered as a kind of 'phase transition', although not an equilibrium one. However, it is well established feature of turbulent flows that as Reynolds number increases about $\bm{\mathrm{Re}}_{c}^{(1)}(\bm{\mathfrak{D}})$, the range of scales involved increases with increasing RN. This is also well established experimentally.
\begin{rem}
There are essentially two key crit
ical values of the RN that are important:
\begin{enumerate}[(a)]
\item A lower critical RN denoted $\bm{\mathrm{Re}}_{c}^{(1)}(\bm{\mathfrak{D}})$, below which turbulent flow cannot exist. If the RN is below this value then the flow is laminar. Then laminar flow occurs for
    \begin{align}
    \bm{\mathrm{Re}}(\bm{\mathfrak{D}},t)<\bm{\mathrm{Re}}_{c}^{(1)}(\bm{\mathfrak{D}})
    \end{align}
\item A higher critical value $\bm{\mathrm{Re}}_{c}^{(2)}(\bm{\mathfrak{D}})>\bm{\mathrm{Re}}_{c}^{(1)}(\bm{\mathfrak{D}})$ which designates the onset of inertial-range behavior with
\begin{align}
    \bm{\mathrm{Re}}_{c}^{(2)}(\bm{\mathfrak{D}})\le \bm{\mathrm{Re}}(\bm{\mathfrak{D}},t)
\end{align}
This is the inertial range over which the classic Kolmogorov scaling laws should hold.
\item The non-inertial range is the set of RNs for which
\begin{align}
\bm{\mathrm{Re}}_{c}^{(1)}(\bm{\mathfrak{D}})\le \bm{\mathrm{Re}}(\bm{\mathfrak{D}},t)\le\bm{\mathrm{R}}_{c}^{(2)}(\bm{\mathfrak{D}})
\end{align}
However, we expect $|\bm{\mathrm{Re}}_{c}^{(2)}(\bm{\mathfrak{D}})-\bm{\mathrm{Re}}_{c}^{(1)}(\bm{\mathfrak{D}})|\ll 1$ or very small so that
$\bm{\mathrm{Re}}_{c}^{(1)}(\bm{\mathfrak{D}})$ and $\bm{\mathrm{Re}}_{c}^{(2)}(\bm{\mathfrak{D}})$ are closely separated.
\end{enumerate}
\end{rem}
The volume-averaged Reynolds number $\bm{\mathrm{R}e}(\bm{\mathfrak{D}},t)$ defined throughout the domain $\bm{\mathfrak{D}}\subset{\mathbb{R}}^{3}$ can be considered a 'control parameter' whereby $\mathrm{L}$ is fixed but $U_{a}(x,t)$ or the viscosity $\nu$ is varied. The following functional will be utilised  in defining a turbulent flow and the transition to a turbulent flow.
\begin{prop}(\textbf{Functional of the averaged Reynolds number})\newline
Let $\bm{\mathrm{R}e}(\bm{\mathfrak{D}},t)=\|\bm{\mathfrak{U}}_{a}(x,t)\|\mathrm{L}\nu^{-1}$ be the ensemble or volume-averaged Reynolds number for fluid with viscosity $\mu$ and smooth (laminar) velocity $U_{a}(x,t)$ within a domain $\bm{\mathfrak{D}}\subset{\mathbb{R}}^{3}$ for all $(x,t)\in\bm{\mathfrak{D}}\times [0\,\infty)$, where $\mathrm{Vol}(\bm{\mathfrak{D}})\sim \mathrm{L}^{3}$ is a fixed volume. Let $\bm{\mathrm{R}e}_{c}(\bm{\mathfrak{D}})$ be the critical Reynolds number at which the fluid makes a an \textbf{sharp} transition from laminar to turbulent flow. Define the following (arbitrary) functional ${\psi}$ of the separation
$|\big|\bm{\mathrm{R}e}(\bm{\mathfrak{D}},t))-\bm{\mathrm{R}e}_{c}(\bm{\mathfrak{D}})|$
\begin{align}
{{\psi}}\big(\big|\bm{\mathrm{R}e}(\bm{\mathfrak{D}},t))-\bm{\mathrm{R}e}_{c}(\bm{\mathfrak{D}})\big|\big)\equiv {{\psi}}\left(\left|\frac{\|\bm{\mathfrak{U}}_{a}(\bm{\mathfrak{D}},t)\|\mathrm{L}}{\nu},\bm{\mathrm{R}e}(\bm{\mathfrak{D}})\right|\right)
\end{align}
Then this functional vanishes as $\bm{\mathrm{R}e}(\bm{\mathfrak{D}},t)\rightarrow \bm{\mathrm{R}e}_{c}(\bm{\mathfrak{D}})$ so that
\begin{align}
\lim_{\bm{\mathrm{R}e},t)\hookrightarrow\bm{\mathrm{R}e}_{c}}{{\psi}}\big(\bm{\mathrm{R}e}(\bm{\mathfrak{D}},t)-\bm{\mathrm{R}e}_{c}(\bm{\mathfrak{D}})\big)=0
\end{align}
or ${{\psi}}\big(\bm{\mathrm{R}e}_{c}(\bm{\mathfrak{D}})-\bm{\mathrm{R}e}_{c}(\bm{\mathfrak{D}})\big)=0$. If $\bm{\mathrm{R}e}(\bm{\mathfrak{D}},t)<\bm{\mathrm{R}e}_{c}(\bm{\mathfrak{D}})$ for all $t>0$ then ${{\psi}} \big(\bm{\mathrm{R}e}(\bm{\mathfrak{D}},t),\bm{\mathrm{R}e}_{c}(\bm{\mathfrak{D}})\big)=0$ for all $t>0$. This functional is also strictly monotone increasing so that for any $\bm{\mathrm{R}e}(\bm{\mathfrak{D}},t)>\bm{\mathrm{R}e}_{(C)}(\bm{\mathfrak{D}},t)$
\begin{align}
{{\psi}}\big(\bm{\mathrm{R}e}^{(2)}(\bm{\mathfrak{D}},t),\bm{\mathrm{R}e}_{c}(\bm{\mathfrak{D}})\big)>
{{\psi}}\big(\big|\bm{\mathrm{R}e}^{(1)}(\bm{\mathfrak{D}},t),\bm{\mathrm{R}e}_{c}(\bm{\mathfrak{D}})\big|\big)
\end{align}
The functional can also be bounded from above so that
\begin{align}
\sup_{t}{{\psi}}\big(\big|\bm{\mathrm{R}e}^{(2)}(\bm{\mathfrak{D}},t)-\bm{\mathrm{R}e}_{c}(\bm{\mathfrak{D}}\big|\big)={\Re}
\end{align}
\end{prop}
\begin{lem}
The derivatives are
\begin{align}
&{\partial}_{t}{{{\psi}}}\big(\bm{\mathrm{R}e}(\bm{\mathfrak{D}},t))-\bm{\mathrm{R}e}_{c}(\bm{\mathfrak{D}})\big)\equiv \frac{\partial}{\partial t}{{\psi}}\left(\left|\frac{\|{{\bm{\mathfrak{U}}}_{a}(\bm{\mathfrak{D}},t)}\|\mathrm{L}}{\nu}-\bm{\mathrm{R}e}_{c}\right|\right)\ne 0\\&\nabla_{b}{{\psi}}\big(\bm{\mathrm{R}e}(\bm{\mathfrak{D}},t))-\bm{\mathrm{R}e}_{c}(\bm{\mathfrak{D}})\big)\equiv
\nabla_{b}{{\psi}}\left(\frac{\|\bm{\mathfrak{U}}_{a}\bm{\mathfrak{D}},t)\|\mathrm{L}}{\nu}-
\bm{\mathrm{R}e}_{c}(\bm{\mathfrak{D}})\right)\\&=\nabla_{b}{{\psi}}\left(\left|\frac{\mathrm{L}}{\nu}\left\|\frac{1}{\mathrm{Vol}(\bm{\mathfrak{D}})}
{\int}_{\mathcal{H}}U_{a}(x,t)d\mu(x)\right\|-\bm{\mathrm{R}e}_{c}
(\bm{\mathfrak{D}})\right|\right)=0
\end{align}
so that the spatial derivatives of $\bm{\mathrm{R}e}(\bm{\mathfrak{D}},t)$ vanish throughout $\bm{\mathfrak{D}}$ for all $t>0$.
\end{lem}
For example, for a power law ansatz
\begin{align}
{{\psi}}\big(\bm{\mathrm{R}e}(\bm{\mathfrak{D}},t))-\bm{\mathrm{R}e}_{c}(\bm{\mathfrak{D}})\big)=
\alpha\big|\bm{\mathrm{R}e}(\bm{\mathfrak{D}},t))-\bm{\mathrm{R}e}_{c}(\bm{\mathfrak{D}})\big|^{\kappa}
\end{align}
with $\kappa>0$ then
\begin{align}
&\frac{\partial}{\partial t}{{\psi}}\big(\bm{\mathrm{R}e}(\bm{\mathfrak{D}},t))-\bm{\mathrm{R}e}_{c}(\bm{\mathfrak{D}})\big)=\alpha{\kappa}\big|
\bm{\mathrm{R}e}(\bm{\mathfrak{D}},t))-\bm{\mathrm{R}e}_{c}(\bm{\mathfrak{D}})\big|^{{\kappa}-1}
\frac{\partial}{\partial t}\bm{\mathrm{R}e}(\bm{\mathfrak{D}},t))\nonumber\\&
=\alpha\kappa\big|\bm{\mathrm{R}e}(\bm{\mathfrak{D}},t))-\bm{\mathrm{R}e}_{c}(\bm{\mathfrak{D}})\big|^{\kappa-1}\frac{\frac{\partial}{\partial t}\|
\bm{\mathfrak{U}}_{a}(\bm{\mathfrak{D}},t)\|\mathrm{L}}{\nu}
\end{align}
and
\begin{align}
&\nabla_{b}{{{\psi}}}\big(\big|\bm{\mathrm{R}e}(\bm{\mathfrak{D}},t))-\bm{\mathrm{R}e}_{c}(\bm{\mathfrak{D}})\big)=\kappa\big|
\bm{\mathrm{R}e}(\bm{\mathfrak{D}},t)),\bm{\mathrm{Re}}_{c}\big|^{\kappa-1}\nabla_{b}\bm{\mathrm{R}e}(\bm{\mathfrak{D}},t))\nonumber\\&
=\kappa\big|\bm{\mathrm{R}e}(\bm{\mathfrak{D}},t))-\bm{\mathrm{R}e}_{c}(\bm{\mathfrak{D}})\big|^{\kappa-1}
\frac{\nabla_{b}\|\bm{\mathfrak{U}}_{a}(\bm{\mathfrak{D}},t)\|\mathrm{L}}{\nu}=0
\end{align}
If ${{\psi}}$ is an exponential with
\begin{align}
&{{\psi}}(\big(\bm{\mathrm{Re}}(\bm{\mathfrak{D}},t),\bm{\mathrm{Re}}_{c}(\bm{\mathfrak{D}})\big|\big)
=\exp(\alpha|\bm{\mathrm{Re}}(\bm{\mathfrak{D}})-\bm{\mathrm{Re}}_{c}(\bm{\mathfrak{D}})|\big)\nonumber \\&
\equiv\exp\left(\alpha\left(\frac{\|\bm{\mathfrak{U}}_{a}(\bm{\mathfrak{D}},t)\|\mathrm{L}}{\nu}-\bm{\mathrm{Re}}_{c}(\bm{\mathfrak{D}})\right)\right)
\end{align}
then
\begin{align}
&{\partial}_{t}{{\psi}}\big(\bm{\mathrm{Re}}(\bm{\mathfrak{D}},t),\bm{\mathrm{Re}}_{c}(\bm{\mathfrak{D}})\big)
={\partial}_{t}\bm{\mathrm{Re}}(\bm{\mathfrak{D}},t)\exp(|\alpha\bm{\mathrm{Re}}(\bm{\mathfrak{D}},t)-\bm{\mathrm{Re}}_{c}(\bm{\mathfrak{D}})|\big)\nonumber\\&
\equiv\frac{\frac{\partial}{\partial t}\|\bm{\mathfrak{U}}_{a}(\bm{\mathfrak{D}},t)\|\mathrm{L}}{\nu}\exp\left(\alpha
\left(\frac{\|\bm{\mathfrak{U}}_{a}(\bm{\mathfrak{D}},t)\|\mathrm{L}}{\nu}-
\bm{\mathrm{Re}}_{c}(\bm{\mathfrak{D}})\right)\right)
\end{align}
Note that ${{\psi}}\big(\big|\bm{\mathrm{Re}}(\bm{\mathfrak{D}},t))-\bm{\mathrm{Re}}_{c}(\bm{\mathfrak{D}})|)$ is always dimensionless. For a steady state laminar flow with $\frac{\partial}{\partial t}U_{a}(x)$=0, the time derivative vanishes so that
$\frac{\partial}{\partial t}{{\psi}}\big(\big|\bm{\mathrm{Re}}(\bm{\mathfrak{D}},t)),\bm{\mathrm{Re}}_{c}(\bm{\mathfrak{D}}))=0$.
\begin{prop}
One could also define a functional derivative
\begin{align}
\frac{\delta}{\delta\bm{\mathrm{Re}}(\bm{\mathfrak{D}},t)}{{\psi}}\big(\bm{\mathrm{Re}}(\bm{\mathfrak{D}},t)-\bm{\mathrm{Re}}_{c}(\bm{\mathfrak{D}})\big)
\end{align}
For a power law
\begin{align}
&\frac{\delta}{\delta\bm{\mathrm{Re}}(\bm{\mathfrak{D}},t)}{{\psi}}\big(\bm{\mathrm{Re}}(\bm{\mathfrak{D}},t),\bm{\mathrm{Re}}_{c}(\bm{\mathfrak{D}})\big)=
\frac{\delta}{\delta\bm{\mathrm{Re}}(\bm{\mathfrak{D}},t)}\big|\bm{\mathrm{Re}}(\bm{\mathfrak{D}},t)-\bm{\mathrm{Re}}_{c}(\bm{\mathfrak{D}})|^{\bm{\bm{\alpha}}}\nonumber\\&
=\bm{\bm{\alpha}}\big|\bm{\mathrm{Re}}(\bm{\mathfrak{D}},t)-\bm{\mathrm{Re}}_{c}(\bm{\mathfrak{D}})|^{\bm{\bm{\alpha}}-1}\frac{\delta}{\delta\bm{\mathrm{Re}}(\bm{\mathfrak{D}},t)}\bm{\mathrm{Re}}(\bm{\mathfrak{D}},t)\nonumber\\&
=\bm{\bm{\alpha}}\big|\bm{\mathrm{Re}}(\bm{\mathfrak{D}},t)-\bm{\mathrm{Re}}_{c}(\bm{\mathfrak{D}})|^{\bm{\bm{\alpha}}-1}
\end{align}
\end{prop}
The next step is to construct a spatio-temporal random velocity field which can potentially describe a turbulent velocity flow in $\bm{\mathfrak{D}}$ when $\bm{\mathrm{Re}}(\bm{\mathfrak{D}},t)>\bm{\mathrm{Re}}_{c}(\bm{\mathfrak{D}})$. To quote from the classic Landau and Liftshitz text on fluid mechanics:
\emph{Turbulent flow at fairly large Reynolds numbers is characterised by the presence of an extremely irregular variation of the velocity with time at each
point. This is called fully developed turbulence. The velocity continually fluctuates about some mean value, and it should be noted that the amplitude
of this variation is in general not small in comparison with the magnitude of the velocity itself. A similar irregular variation of the velocity exists between
points in the flow at a given instant. The paths of the fluid particles in turbulent flow are extremely complicated, resulting in an extensive mixing of the
fluid.}

To try and formulate such a random field--and inspired by \textbf{Propositions(2.7)} and \textbf{(2.8)}-- we make the following general propositions. We take a smooth flow which at least a weak solution of the NS equations and nonlinearly 'mix' this flow with a classical random field, specifically a BF random field to create a new random field or random turbulent flow, utilising the averaged Reynolds number as a 'control parameter' to adjust the degree or magnitude of randomness or turbulence.
\begin{prop}(\textbf{\textrm{Construction of a random vector field representing a  turbulent flow}})
Let $\bm{\mathfrak{D}}\subset\bm{\mathrm{R}}^{3}$ be a domain with volume $Vol[\bm{\mathfrak{D}}]\sim L^{3}$ and
let $\bm{\mathfrak{D}}$ contain an incompressible fluid of velocity $U_{a}({x},t)$ that is smooth or laminar for all $(x,t)\in\bm{\mathfrak{D}}\times[0,\infty)$ so that $\nabla_{a}U^{a}=0$. The fluid has viscosity $\nu$. Here, $U_{a}$ is a smooth bounded solution of the Navier-Stokes equations. We are interested in how turbulence will develop within a region or domain $\bm{\mathfrak{D}}$. The Euclidean norm of the fluid velocity is $\|U_{a}(x,t)\|_{E_{2}(\bm{\mathfrak{D}}}=(\sum_{a}^{3}U_{a}|(x,t)|^{2})^{1/2}$ so that at each $(x,t)$ one can define a volume-averaged Reynolds number throughout $\bm{\mathfrak{D}}$ as $\bm{\mathrm{Re}}(\bm{\mathfrak{D}},t)=\|{\mathscr{U}}_{a}(\bm{\mathfrak{D}},t)\|L\nu^{-1}$. Now let ${{\mathscr{R}}}(x)={{\mathscr{B}}}(x)$ be a spatial Bargmann-Fock Gaussian random field with expectation $\bm{\mathbb{E}}\big\langle{{\mathscr{B}}}(x)\big\rangle=0$ and binary correlation
\begin{align}
&\big\|\!\big\|{{\mathscr{B}}}(x){{\otimes}}{{\mathscr{B}}}(y)\big\|\!\big\|_{E_{1}(\bm{\mathfrak{D}})}
={\bm{\mathbb{E}}}\langle{{\mathscr{B}}}(x)~{{\otimes}}~{{\mathscr{B}}}(y)\rangle\nonumber\\&
={\bm{\Xi}}(x,y;\lambda)={\mathsf{C}}\exp(-\|x-y\|^{2}\lambda^{-2})
\end{align}
where $(\mathsf{C})>0$ is a constant. The derivatives of the Bargmann-Fock field ${{\mathscr{B}}}(x,t)$ are $\frac{\partial}{\partial t}{{\mathscr{B}}}(x)$,$\nabla_{a}{{\mathscr{B}}}_{a}(x),\newline \nabla_{a}\nabla_{b}{{\mathscr{B}}}_{b}(y,t)$
and exist. The field ${{\mathscr{B}}}(x)$ can also be interpreted as a 'smeared out' white-in-space noise ${\mathscr{W}}(x)$.

Let $\bm{\mathrm{R}e}_{c}(\bm{\mathfrak{D}})$ be a critical Reynold's number at which the smooth flow makes a (sharp) transition to a turbulent flow. The smooth N-S flow $U_{a}(x,t)$ is then 'mixed' with the random field ${{\mathscr{B}}}(x,t)$ to produce a random or turbulent flow ${{\mathscr{U}}}_{a}(x,t)$ which is also a random field. The turbulent flow ${{\mathscr{U}}}_{a}(x,t)$ then grows with the averaged Reynolds number $\bm{\mathrm{R}e}(\bm{\mathfrak{D}},t)$ throughout $\bm{\mathfrak{D}}$. The following are then proposed:
\begin{enumerate}[(a)]
\item The 'mixing of the fields $U_{a}(x,t)$ and ${{\mathscr{B}}}(x)$ at any Reynolds number $\bm{\mathrm{R}e}(\bm{\mathfrak{D}},t)>\bm{\mathrm{R}e}_{c}(\bm{\mathfrak{D}})$ is denoted by
\begin{align}
{{\mathcal{M}}}:\big\lbrace{U}_{a}(x),\bm{\mathrm{R}e}(\bm{\mathfrak{D}},t),{{\mathscr{B}}}(x)\big\rbrace
\hookrightarrow{{\mathscr{U}}}_{a}(x,t)
\end{align}
The tentative 'mixing formula' ansatz is then
\begin{align}
&{{\mathscr{U}}}_{a}(x,t)=U_{a}(x,t)+{\bm{\bm{\alpha}}}U_{a}(x,t)
[{{\psi}}\big(\big|{\mathrm{Re}}(\bm{\mathfrak{D}},t)-\bm{\mathrm{Re}}_{c})(\bm{\mathfrak{D}})\big)]
{\mathbb{I}}_{\mathcal{S}}[\bm{\mathrm{R}e}(\bm{\mathfrak{D}},t)]\big|{{\mathscr{B}}}(x)\nonumber\\&
=U_{a}(x,t)1+{\bm{\bm{\alpha}}}{{\psi}}\big(|\bm{\mathrm{R}e}(\bm{\mathfrak{D}},t)-
\bm{\mathrm{R}e}_{c}(\bm{\mathfrak{D}}))|\big){\mathbb{I}}_{\mathcal{S}}[\bm{\mathrm{R}e}(\bm{\mathfrak{D}},t)]
{{\mathscr{B}}}(x))\nonumber\\&
=U_{a}(x,t)+{{\mathscr{C}}}_{a}(x,t)\equiv U_{a}(x,t)+U_{a}(x,t)
{{\mathscr{R}}}_{a}(x,t){\mathbb{I}}_{\mathcal{S}}~[\bm{\mathrm{R}e}(\bm{\mathfrak{D}},t)]
\end{align}
where $\mathscr{Q}_{a}(x,t)=1+{\mathscr{R}}_{a}(x,t)=1+{\bm{\bm{\alpha}}}U_{a}(x,t)
{{\psi}}\big(\bm{\mathrm{R}e}(\bm{\mathfrak{D}},t),\bm{\mathrm{R}e}_{c}(\bm{\mathfrak{D}}){{\mathscr{B}}}(x)$ are random vector fields.
Again, ${\mathbb{I}}_{\mathcal{S}}[\bm{\mathrm{R}e}(\bm{\mathfrak{D}},t)]$ is the indicator or switch function such that
\begin{empheq}[right=\empheqrbrace]{align}
&{\mathbb{I}}_{\mathcal{S}}~\big[\bm{\mathrm{R}e}(\bm{\mathfrak{D}},t)\big]=1,~if~\bm{\mathrm{R}e}(\bm{\mathfrak{D}},t)>\mathbf{Re}_{c}(\bm{\mathfrak{D}})~or~~
\bm{\mathrm{R}e}(\bm{\mathfrak{D}},t)~{\in}{{\mathcal{S}}}~\nonumber\\&
{\mathbb{I}}_{\mathcal{S}}~\big[\bm{\mathrm{R}e}(\bm{\mathfrak{D}},t)\big]=0,~
if~\bm{\mathrm{R}e}(\bm{\mathfrak{D}},t)\le\mathbf{Re}_{c}(\bm{\mathfrak{D}})~or~~
\bm{\mathrm{R}e}(\bm{\mathfrak{D}},t)~{\notin}{{\mathcal{S}}}\nonumber
\end{empheq}
and
\begin{align}
{\mathcal{S}}=\lbrace set~of~ ~all~\mathbf{Re}(\bm{\mathfrak{D}},t)~>~\mathbf{Re}_{c}(\bm{\mathfrak{D}})\rbrace
\end{align}
Also
\begin{align}
&{{\mathscr{U}}}_{a}(x,t)\equiv U_{a}(x,t)+ {\bm{\bm{\alpha}}}U_{a}(x,t) {{\psi}}\left(|\frac{\|\bm{\mathfrak{U}}_{a}(x,t)\|{L}}{\nu}-\bm{\mathrm{R}e}_{c}(\bm{\mathfrak{D}})|\right)
{\mathbb{I}}_{\mathcal{S}}[\bm{\mathrm{R}e}(\bm{\mathfrak{D}},t)]
{{\mathscr{B}}}(x,t)\nonumber\\&
\equiv U_{a}(x,t)+{\bm{\bm{\alpha}}}U_{a}(x,t) {{\psi}}\left(\left|\frac{\mathrm{L}}{\nu}\left\|\frac{1}{Vol(\bm{\mathfrak{D}})}{\int}_{\bm{\mathfrak{D}}}U_{a}(x,t)d\mu(x)\right\|
-\bm{\mathrm{R}e}_{c}(\bm{\mathfrak{D}})\right|\right)
{\mathbb{I}}_{\mathcal{S}}[\bm{\mathrm{R}e}(\bm{\mathfrak{D}},t)]
\end{align}
and where ${\bm{\bm{\alpha}}}\ge 1$ is an arbitrary constant 'mixing parameter' which contributes to the relative magnitude of the turbulent term. For a spatio-temporal field ${\mathscr{B}}(x,t)$
\begin{align}
&{{\mathscr{U}}}_{a}(x,t)=U_{a}(x,t)+{\bm{\bm{\alpha}}} U_{a}(x,t){{\psi}}\big(\bm{\mathrm{R}e}(\bm{\mathfrak{D}},t)-\bm{\mathrm{R}e}_{c}(\bm{\mathfrak{D}})\big)
\mathbb{I}_{\mathcal{S}}[\bm{\mathrm{R}e}(\bm{\mathfrak{D}},t)]{{\mathscr{B}}}(x,t))\nonumber\\&
=U_{a}(x,t)(1+{\bm{\bm{\alpha}}}{{\psi}}
\big(\bm{\mathrm{R}e}(\bm{\mathfrak{D}},t)-\bm{\mathrm{R}e}_{c}(\bm{\mathfrak{D}})\mathbb{I}_{\mathcal{S}}[\bm{\mathrm{R}e}(\bm{\mathfrak{D}},t)]{{\mathscr{B}}}(x,t))\nonumber\\&
=U_{a}(x,t)+{{\mathscr{C}}}_{a}(x,t)\equiv U_{a}(x,t)+U_{a}(x,t){\mathscr{Q}}_{a}(x,t)
\end{align}
However, from here we will consider only random fields ${\mathscr{B}}(x)$.
\item Keeping $(\mathrm{L},{\bm{\bm{\alpha}}})$ fixed, the viscosity $\nu$ or the averaged Reynolds number $\bm{\mathrm{R}e}(\bm{\mathfrak{D}},t)$ at time $t>0$ within the domain $\bm{\mathfrak{D}}$ is then a 'control parameter' which be adjusted. This is a nonlinear 'mixing' of the underlying smooth flow $U_{a}(x,t)$ with the random BF field. Representations of the transition to turbulence within $\bm{\mathfrak{D}}$ are given in the figures.
\begin{figure}[htb]
\begin{center}
\includegraphics[height=3.0in,width=3.0in]{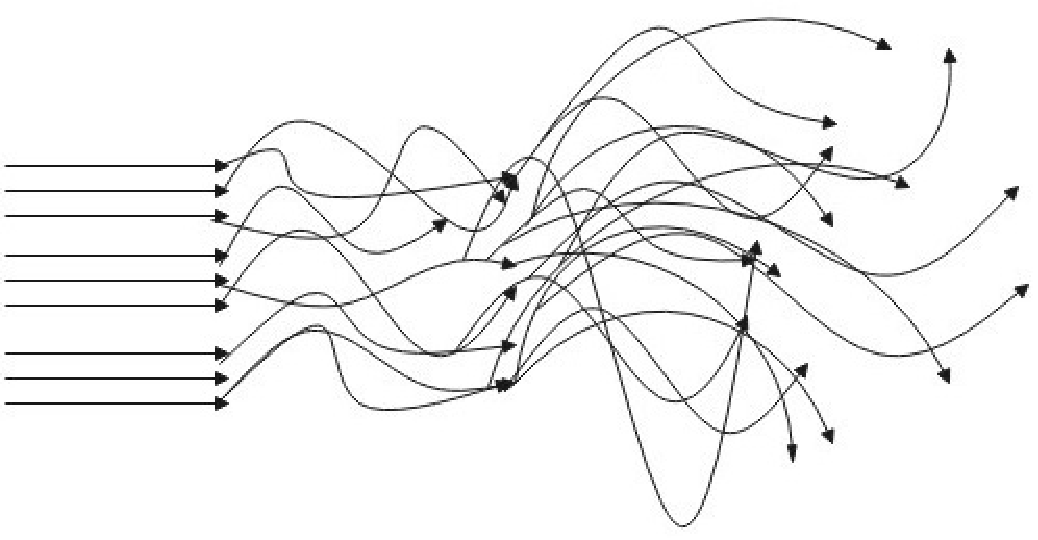}
\caption{Transition of a smooth flow $U_{a}(x,t)$ to a turbulent flow ${\mathscr{U}}_{a}(x,t)$ as $\bm{\mathrm{R}e}(\bm{\mathfrak{D}},t)$ exceeds $\bm{\mathrm{R}e}_{c}(\bm{\mathfrak{D}})$}
\end{center}
\end{figure}
\begin{figure}[htb]
\begin{center}
\includegraphics[height=2.0in,width=5.0in]{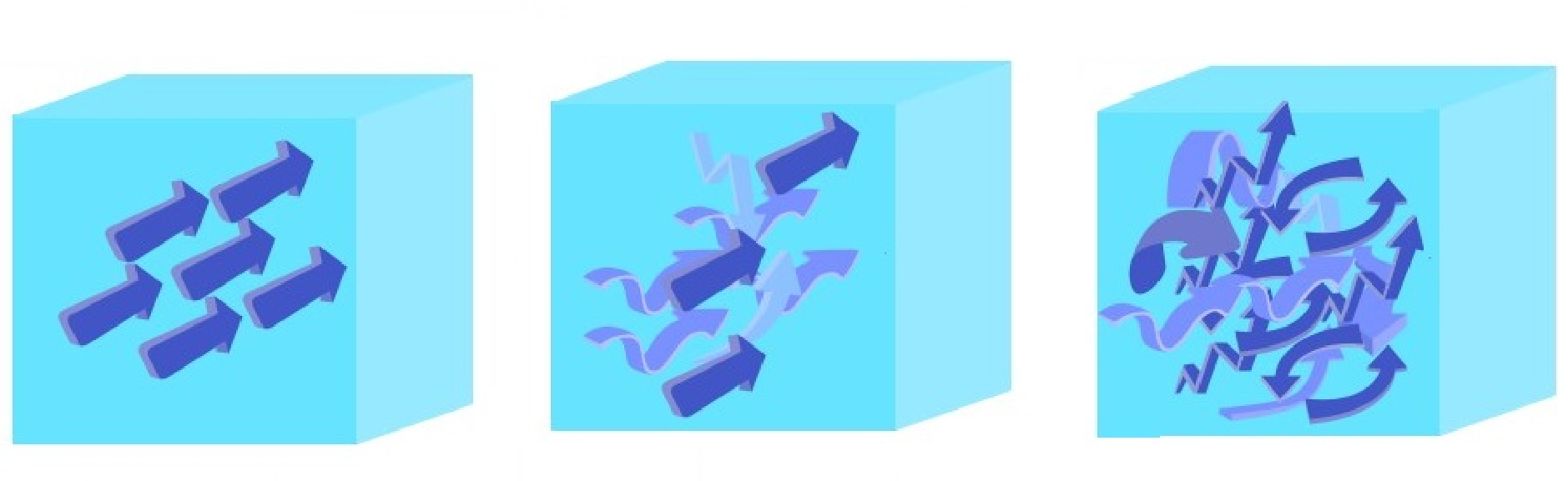}
\caption{Gradual transition of a smooth flow $U_{a}(x,t)$ to a turbulent flow ${\mathscr{U}}_{a}(x,t)$ within domain $\bm{\mathfrak{D}}$ as $\bm{\mathrm{R}e}(\bm{\mathfrak{D}},t)$ exceeds $\bm{\mathrm{R}e}_{c}(\bm{\mathfrak{D}})$}
\end{center}
\end{figure}
\item The total flow is then a sum of contributions from laminar and turbulent flow, with the turbulent contribution increasing or dominating as the Reynolds number and velocity increase and/or the viscosity decreases.
\begin{align}
&\underbrace{{{\mathscr{U}}}_{a}(x,t)}_{\textbf{total~turbulent~flow}}=\underbrace{U_{a}(x,t)}_{\textbf{laminar~flow}}+
\underbrace{{\bm{\bm{\alpha}}}U_{a}(x,t){{\psi}}\big(\big|\bm{\mathrm{Re}}(\bm{\mathfrak{D}},t))-\bm{\mathrm{Re}}_{c}(\bm{\mathfrak{D}})|)
\mathbb{I}_{\mathcal{S}}[\bm{\mathrm{R}e}(\bm{\mathfrak{D}},t)]{{\mathscr{B}}}(x,t)}_{\textbf{random/turbulent~flow}}
\end{align}
and in the limit of high viscosity and/or small $\bm{\mathrm{R}e}(\bm{\mathfrak{D}},t)$ near or at $\bm{\mathrm{R}e}_{c}(\bm{\mathfrak{D}})$, the turbulent term vanishes or becomes very small.
\begin{align}
&\lim_{\bm{\mathrm{Re}}(\bm{\mathfrak{D}},t)\rightarrow\bm{\mathrm{Re}}_{c}}\underbrace{{{\mathscr{U}}}_{a}(x,t)}_{total~turbulent~flow}\\&
=\underbrace{U_{a}(x,t)}_{\textbf{laminar~flow}}+\lim_{\bm{\mathrm{Re}}(\bm{\mathfrak{D}},t)\rightarrow\bm{\mathrm{Re}}_{c}}
\underbrace{{\bm{\bm{\alpha}}}U_{a}(x,t){{\psi}}\big(\big|\bm{\mathrm{Re}}(\bm{\mathfrak{D}},t))-\bm{\mathrm{Re}}_{c}(\bm{\mathfrak{D}})|)
\mathbb{I}_{\mathcal{S}}[\bm{\mathrm{R}e}(\bm{\mathfrak{D}},t)]
{{\mathscr{B}}}(x,t)}_{\textbf{random/turbulent flow}}=U_{a}(x,t)
\end{align}
\item The random field ${{\mathscr{U}}}_{a}(x,t)$ describes a 3D randomly fluctuating or stochastic 'fluid geometry' existing within domain $\bm{\mathfrak{D}}$. For a constant laminar flow $U_{a}(x,t)=U_{a}=(U,U,U)$, and for a laminar flow along the z-axis for example $U_{a}=(0,0,U)$. Then the turbulent flow is
\begin{align}
{{\mathscr{U}}}_{a}(x,t)=U+{\bm{\bm{\bm{\alpha}}}}U{{\psi}}\big(\big|\bm{\mathrm{Re}}(\bm{\mathfrak{D}},t))
-\bm{\mathrm{Re}}_{c}(\bm{\mathfrak{D}})|)
\mathbb{I}_{\mathcal{S}}[\bm{\mathrm{R}e}(\bm{\mathfrak{D}},t)]{{\mathscr{B}}}(x)
\end{align}
\item Since $\bm{\mathbb{E}}\big\langle{{\mathscr{U}}}_{a}(x,t)\big\rangle=0$. The mean velocity or stochastic averaged flow is then
\begin{align}
&\big\|\!\big\|{{\mathscr{U}}}_{a}(x,t)\big\|\!\big\|_{SE_{1}(\bm{\mathfrak{D}})}=
{\bm{\mathbb{E}}}\langle{{\mathscr{U}}}_{a}(x,t)\rangle\nonumber\\&
=U_{a}(x,t)+{\bm{\bm{\bm{\alpha}}}}U_{a}(x,t){{\psi}}\big(\bm{\mathrm{Re}}(\bm{\mathfrak{D}},t))
-\bm{\mathrm{Re}}_{c}(\bm{\mathfrak{D}})|)\mathbb{I}_{\mathcal{S}}[\bm{\mathrm{R}e}(\bm{\mathfrak{D}},t)]\bm{\mathbb{E}}\langle
{{\mathscr{B}}}(x,t)\rangle\nonumber\\&
=U_{a}(x,t)+{\bm{\bm{\alpha}}}U_{a}(x,t)
{{\psi}}\big(\big|\bm{\mathrm{Re}}(\bm{\mathfrak{D}},t))-\bm{\mathrm{Re}}_{c}(\bm{\mathfrak{D}})|)
\mathbb{I}_{\mathcal{S}}[\bm{\mathrm{R}e}(\bm{\mathfrak{D}},t)]\bm{\mathbb{E}}\langle
{{\mathscr{B}}}(x,t)\rangle=U_{a}(x,t)
\end{align}
\item The function is a generic monotone increasing function that increase with Reynolds number and vanishes at the critical Reynolds number. For example, if the function ${{\psi}}$ is an exponential then the turbulence would grow exponentially with Reynolds number and vanish at $
\bm{\mathrm{Re}}_{c}^{(1)}(\bm{\mathfrak{D}},t)$ so that
    \begin{align}
&{{\mathscr{U}}}_{a}(x,t)=U_{a}(x,t)+ {\bm{\bm{\alpha}}}~U_{a}(x,t)
\exp\big(\big|\bm{\mathrm{Re}}(\bm{\mathfrak{D}},t)-\bm{\mathrm{Re}}_{c}(\bm{\mathfrak{D}})\big|\big)-1\big)
\mathbb{I}_{\mathcal{S}}[\bm{\mathrm{R}e}(\bm{\mathfrak{D}},t)]
{{\mathscr{B}}}(x)
\end{align}
For a more realistic fluid ${{\psi}}$ might grow as some power law so that
\begin{align}
&{{\mathscr{U}}}_{a}(x,t)=U_{a}(x,t)+
{\bm{\bm{\alpha}}}U_{a}(x,t)\big|\big(\bm{\mathrm{Re}}(\bm{\mathfrak{D}},t)-\bm{\mathrm{Re}}_{c}(\bm{\mathfrak{D}})\big)^{\kappa}
\rbrace\mathbb{I}_{\mathcal{S}}[\bm{\mathrm{R}e}(\bm{\mathfrak{D}},t)]
{{\mathscr{B}}}(x)
\end{align}
where $\alpha > 0$. If the randomness grows as the square root of the difference between the Reynolds number and the critical Reynolds number then
\begin{align}
&{{\mathscr{U}}}_{a}(x,t)=U_{a}(x,t)+
{\bm{\bm{\alpha}}}U_{a}(x,t)\sqrt{\big|\big(\bm{\mathrm{Re}}(\bm{\mathfrak{D}},t)-\bm{\mathrm{Re}}_{c}(\bm{\mathfrak{D}})\big)}
\mathbb{I}_{\mathcal{S}}[\bm{\mathrm{R}e}(\bm{\mathfrak{D}},t)]
{{\mathscr{B}}}(x)
\end{align}
\end{enumerate}
\end{prop}
\begin{figure}[htb]
\begin{center}
\includegraphics[height=3.0in,width=3.0in]{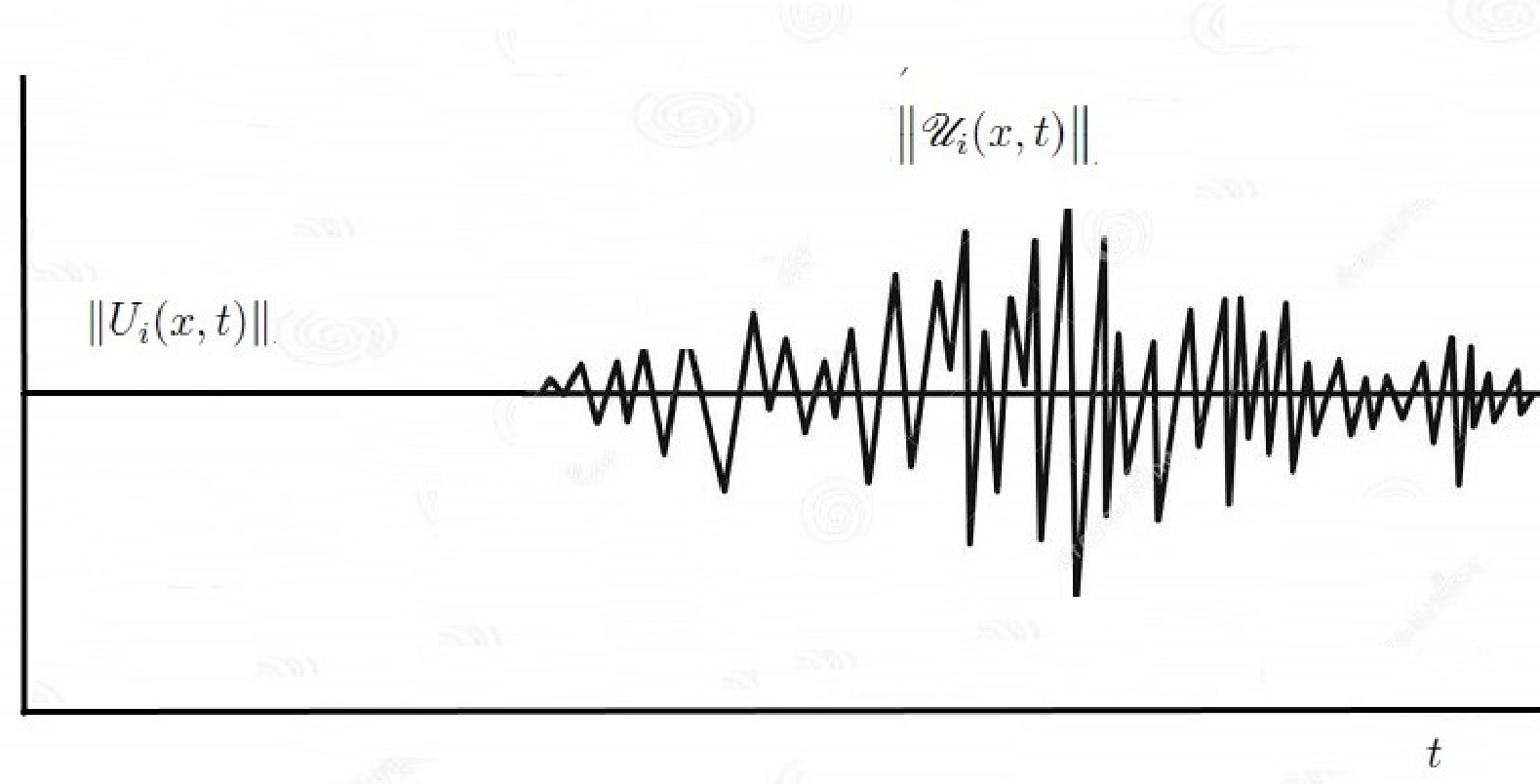}
\caption{Transition of a smooth flow $U_{a}(x,t)$ to a turbulent flow ${\mathscr{U}}_{a}(x,t)$ as $\bm{\mathrm{R}e}(\bm{\mathfrak{D}},t)$ exceeds $\bm{\mathrm{R}e}_{c}(\bm{\mathfrak{D}})$. The turbulent velocity fluctuates randomly about the mean.}
\end{center}
\end{figure}
A representation of the evolution of a turbulent flow from a smooth flow as $\bm{\mathrm{Re}}(\bm{\mathfrak{D}},t)$ increases in time beyond the critical value
$\bm{\mathrm{Re}}_{c}(\bm{\mathfrak{D}})$ is given in figures 2 and 3. The evolution of a turbulent flow ${\mathscr{U}}_{a}(x,t)$  at some $(x,t)\in\bm{\mathfrak{D}}\times\infty)$ from laminar flow $U_{a}(x,t)$, once $\bm{\mathrm{Re}}(\bm{\mathfrak{D}},t)>\bm{\mathrm{Re}}_{c}(\bm{\mathfrak{D}})$ is also illustrated in Figure 3.

Generally, either the deterministic or random contribution will dominate the contribution to the random flow ${{\mathscr{U}}}_{a}(x,t)$ depending on the values of the variables $({\bm{\bm{\alpha}}},\nu,\bm{\mathrm{Re}}(\bm{\mathfrak{D}}))$. Bargman-Fock random fields with Gaussian correlations are utilised since they are regulated and differentiable and have the useful property that
\begin{align}
\bm{\mathbb{E}}\big\langle{{\mathscr{B}}}(x,t){{\otimes}}~\nabla_{a}{{\mathscr{B}}}(x,t)\big\rangle=0
\end{align}
This will simplify the derivations of many estimates and computations.
\begin{prop}
Given ${{\mathscr{U}}}_{a}(x,t)$, define a 'stochastic material derivative' ${\bm{\mathrm{D}}}_{m}$ as
\begin{align}
{\bm{\mathrm{D}}}_{m}=\frac{\partial}{\partial t}+{{\mathscr{U}}}^{b}(x,t)\nabla_{b}
\end{align}
Since $U_{a}(x,t)$ is a solution of the NS equations then ${{\mathscr{U}}}_{a}(x,t)$ should be a solution of some stochastically averaged NS equations
\begin{align}
&{\bm{\mathbb{E}}}\langle\bm{\mathrm{D}}_{m}(x,t)-\nu \Delta{{\mathscr{U}}}_{a}(x,t)\rangle\nonumber\\&
={\bm{\mathbb{E}}}\langle
{\partial}_{t}{{\mathscr{U}}}_{a}(x,t)-\nu \Delta{{\mathscr{U}}}_{a}(x,t)+
{{\mathscr{U}}}^{b}(x,t)\nabla_{b}{{\mathscr{U}}}_{a}(x,t)+\nabla_{a}\mathbf{P}(x,t)
\rangle\nonumber\\&
={\partial}_{t}U_{a}(x,t)-\nu \Delta U_{a}(x,t)+
U^{b}(x,t)\nabla_{b}U_{a}(x,t)+\nabla_{a}\mathbf{P}(x,t)+\textit{"new induced terms"}
\end{align}
New terms should arise via stochastic averaging since the NS equations are \textbf{nonlinear}. This is explored in detail in \textbf{Section 3}.
\end{prop}
\begin{rem}
An alternative 'mixing' ansatz is to utilise \textbf{Lemma (2.7)}, and take the Reynolds number at each $(x,t)\in\bm{\mathfrak{D}}\times [0,\infty)$ rather than the volume-averaged Reynolds number, namely $\bm{\mathrm{Re}}(x,t)$ given by
\begin{align}
\bm{\mathrm{Re}}(x,t)={\|U_{a}(x,t)\|L}/\nu,~~~(x,t)\in\bm{\mathfrak{D}}\times[0,\infty)
\end{align}
where $U_{a}(x,t)$ evolves by the Navier-Stokes equations. Then ${{\psi}}:\mathbb{R}^{+}\rightarrow\mathbb{R}^{+}$ is a monotone increasing functional ${{\psi}}(\bm{\mathrm{Re}}(x,t),\bm{\mathrm{Re}}_{c})$ which vanishes for
$\bm{\mathrm{Re}}(x,t)\le \bm{\mathrm{Re}}_{c}$, and ${{\psi}}(\bm{\mathrm{Re}}(x,t_{1}),\bm{\mathrm{Re}}_{c})<{{\psi}}(\bm{\mathrm{Re}}(x,t_{2}),\bm{\mathrm{Re}}_{c})$ if $\bm{\mathrm{Re}}(x,t_{1})<\bm{\mathrm{Re}}(x,t_{2})$. The turbulent flow is then the random vector field of the form (2.19)
\begin{align}
{{\mathscr{U}}}_{a}(x,t)=U_{a}(x,t)+{\bm{\bm{\alpha}}} U_{a}(x,t){{\psi}}(\bm{\mathrm{Re}}(x,t),\bm{\mathrm{Re}}_{c})
\mathbb{I}_{\mathcal{S}}[\bm{\mathrm{R}e}(\bm{\mathfrak{D}},t)]
{{\mathscr{B}}}(x)
\end{align}
However, now $\nabla_{b}\bm{\mathrm{Re}}(x,t)\ne 0$ and $\nabla_{b}{{\psi}}(\bm{\mathrm{Re}}(x,t),\mathrm{Re}_{c})\ne 0$. This leads to an extra term when one takes the derivative $\nabla_{b}{\mathscr{U}}_{a}(x,t)$ and vastly complicates various derivations and estimates that are to be made.
\end{rem}
\begin{prop}
The notation can be simplified somewhat by defining
\begin{align}
{{\psi}}(\bm{\mathrm{Re}}(\bm{\mathfrak{D}},t);{\mathbb{I}})\equiv
{{\psi}}\big(\big|\bm{\mathrm{Re}}(\bm{\mathfrak{D}},t))-\bm{\mathrm{Re}}_{c}(\bm{\mathfrak{D}})\big|\big)
\mathbb{I}_{\mathcal{S}}[\bm{\mathrm{R}e}(\bm{\mathfrak{D}},t)]
\end{align}
so that
\begin{empheq}[right=\empheqrbrace]{align}
&{{\psi}}(\bm{\mathrm{Re}}(\bm{\mathfrak{D}},t);\mathbb{I})>0,~if~\mathbf{Re}(\bm{\mathfrak{D}},t)>\mathbf{Re}_{c}(\bm{\mathfrak{D}})~or~\bm{\mathrm{Re}}(\bm{\mathfrak{D}},t){\in}{\mathcal{S}}
\nonumber\\&
{{\psi}}(\bm{\mathrm{Re}}(\bm{\mathfrak{D}},t);\mathbb{I})=0,~if~\mathbf{Re}(\bm{\mathfrak{D}},t)\le\mathbf{Re}_{c}(\bm{\mathfrak{D}})~or~\bm{\mathrm{Re}}(\bm{\mathfrak{D}},t){\notin}{\mathcal{S}}\nonumber
\end{empheq}
Then
\begin{align}
{{\mathscr{U}}}_{a}(x,t)=U_{a}(x,t)+
{\bm{\bm{\alpha}}}U_{a}(x,t){{\psi}}(\bm{\mathrm{Re}}(\bm{\mathfrak{D}},t);\mathbb{I}_{\mathcal{S}})[\bm{\mathrm{R}e}(\bm{\mathfrak{D}},t)]
{{\mathscr{B}}}(x)
\end{align}
\end{prop}
\subsection{Derivatives of the random field ${{\mathscr{U}}}_{a}(x,t)$ and their expectations}
Since the derivatives of the BF fields ${{\mathscr{B}}}(x)$ exist then so do the derivatives and expectations of the turbulent velocity flow ${{\mathscr{U}}}_{a}(x,t)$.
\begin{lem}
Given the random fluid flow ${{\mathscr{U}}}_{a}(x,t)$ then the mean or average values within $\bm{\mathfrak{D}}\subset{\mathbb{R}}^{3}$ are
\begin{align}
&\big\|\!\big\|\frac{\partial}{\partial t}{{\mathscr{U}}}_{a}(x,t)\big\|\!\big\|_{SE_{1}(\bm{\mathfrak{D}})}
\equiv{\bm{\mathbb{E}}}\langle \frac{\partial}{\partial t}{{\mathscr{U}}}_{a}(x,t)\rangle=\frac{\partial}{\partial t}U_{a}\\&\big\|\!\big\|\nabla_{a}{{\mathscr{U}}}_{a}(x,t)
\big\|\!\big\|_{SE_{1}(\bm{\mathfrak{D}})}\equiv{\bm{\mathbb{E}}}\langle \nabla_{a}{{\mathscr{U}}}_{a}(x,t)\rangle
=\nabla_{a}U_{a}(x,t)=0\\&\big\|\!\big\|\nabla_{b}
{{\mathscr{U}}}_{a}(x,t)\big\|\!\big\|_{SE_{1}(\bm{\mathfrak{D}})}\equiv{\bm{\mathbb{E}}}\langle \nabla_{b}{{\mathscr{U}}}_{a}(x,t)\rangle
=\nabla_{b}U_{a}(x,t)\ne\\&\big\|\!\big\|\nabla_{a}\nabla_{b}{{\mathscr{U}}}_{a}(x,t))\big\|\!\big\|_{E_{1}(\bm{\mathfrak{D}})}
\equiv{\bm{\mathbb{E}}}\langle
\nabla_{a}\nabla_{b}{{\mathscr{U}}}_{a}(x,t)\rangle=\nabla_{a}\nabla_{b}U_{a}(x,t)
\end{align}
\end{lem}
\begin{proof}
The time derivative is
\begin{align}
\frac{\partial}{\partial t}{{\mathscr{U}}}_{a}(x,t)&=\frac{\partial}{\partial t}U_{a}(x,t)+{\bm{\bm{\alpha}}} \frac{\partial}{\partial t}U_{a}(x,t){{\psi}}(\bm{\mathrm{Re}}(\bm{\mathfrak{D}}),t);\mathbb{I})){{\mathscr{B}}}(x)\nonumber\\&
+{\bm{\bm{\alpha}}}U_{a}(x,t)\frac{\partial}{\partial t}{{\psi}}(\bm{\mathrm{Re}}(\bm{\mathfrak{D}}),t);\mathbb{I}))
{{\mathscr{B}}}(x)
\end{align}
For a spatio-temporal BF field ${\mathscr{B}}(x,t)$, there would be an additional term in (2.132)
\begin{align}
{\bm{\bm{\alpha}}}U_{a}(x,t)
{{\psi}}(\bm{\mathrm{Re}}(\bm{\mathfrak{D}}),t);\mathbb{I}))~\frac{\partial}{\partial t}{{\mathscr{B}}}(x)
\end{align}
Now taking the expectation or stochastic average, the underbraced terms vanish
\begin{align}
&{\bm{\mathbb{E}}}\langle \frac{\partial}{\partial t}{{\mathscr{U}}}_{a}(x,t)\rangle=\frac{\partial}{\partial t}U_{a}(x,t)+\underbracket{{\bm{\bm{\alpha}}} \frac{\partial}{\partial t}U_{a}(x,t)\langle{{\psi}}(\bm{\mathrm{Re}}(\bm{\mathfrak{D}}),t);\mathbb{I}))\rangle
{\bm{\mathbb{E}}}\langle{{\mathscr{B}}}(x,t)\rangle}\nonumber\\&
+\underbracket{{\bm{\bm{\alpha}}}U_{a}(x,t)\frac{\partial}{\partial t}{{\psi}}(\bm{\mathrm{Re}}(\bm{\mathfrak{D}}),t);\mathbb{I}))
{\bm{\mathbb{E}}}\langle{{\mathscr{B}}}(x,t)\rangle}+\underbrace{{\bm{\bm{\alpha}}}U_{a}(x,t){{\psi}}(\bm{\mathrm{Re}}(\bm{\mathfrak{D}}),t);\mathbb{I}))
{\bm{\mathbb{E}}}\langle \frac{\partial}{\partial t}{{\mathscr{B}}}(x,t)\rangle}
\nonumber\\&
\equiv\frac{\partial}{\partial t}U_{a}(x,t)+\underbracket{{\bm{\bm{\alpha}}}\frac{\partial}{\partial t}U_{a}(x,t)
{{\psi}}(\bm{\mathrm{Re}}(\bm{\mathfrak{D}}),t);\mathbb{I}))\big\|\!\big\|
{{\mathscr{B}}}(x)\big\|\!\big\|_{E_{1}(\bm{\mathfrak{D}})}}\nonumber\\&
+\underbracket{{\bm{\bm{\alpha}}}U_{a}(x,t)\frac{\partial}{\partial t}{{\psi}}(\bm{\mathrm{Re}}(\bm{\mathfrak{D}}),t);\mathbb{I}))\big\|\!\big\|
{{\mathscr{B}}}(x)\big\|\!\big\|_{E_{1}({\psi})}}\nonumber\\&+
\underbracket{{\bm{\bm{\alpha}}}U_{a}(x,t){{\psi}}(\bm{\mathrm{Re}}(\bm{\mathfrak{D}}),t);\mathbb{I}))\big\|\!\big\|
\frac{\partial}{\partial t}{{\mathscr{B}}}(x)\big\|\!\big\|_{SE_{1}(\bm{\mathfrak{D}})}}=\frac{\partial}{\partial t}U_{a}(x,t)
\end{align}
since $\bm{\mathbb{E}}\big\langle{{\mathscr{B}}}(x,t)\big\rangle=0$. Similarly, the stochastic average of the gradient or spatial derivative is
\begin{align}
{\bm{\mathbb{E}}}\langle \nabla_{b}{{\mathscr{U}}}_{a}(x,t)\rangle&={D}_{b}U_{a}(x,t)+
\underbracket{{\bm{\bm{\alpha}}}\frac{\partial}{\partial t}U_{a}(x,t)
{{\psi}}(\bm{\mathrm{Re}}(\bm{\mathfrak{D}}),t);\mathbb{I}))\bm{\mathbb{E}}\langle
{{\mathscr{B}}}(x,t)\rangle}\nonumber\\&
+\underbracket{\langle{{\psi}}(\bm{\mathrm{Re}}(\bm{\mathfrak{D}}),t);\mathbb{I}))\rangle
{\bm{\mathbb{E}}}\langle \nabla_{b}{{\mathscr{B}}}(x,t)\rangle}
\nonumber\\&\equiv \nabla_{b}U_{a}(x,t)+\underbracket{x\frac{\partial}{\partial t}U_{a}(x,t)
{{\psi}}(\bm{\mathrm{Re}}(\bm{\mathfrak{D}}),t);\mathbb{I}))
\big\|\!\big\|{{\mathscr{B}}}(x,t)\big\|\!\big\|_{SE_{1}(\bm{\mathfrak{D}})}}
\nonumber\\&+\underbracket{\mathbf{}U_{a}(x,t){{\psi}}(\bm{\mathrm{Re}}(\bm{\mathfrak{D}}),t);\mathbb{I}))\big\|\!\big\|
{{\mathscr{B}}}(x,t)\big\|\!\big\|_{SE_{1}(\bm{\mathfrak{D}})}}\nonumber\\&+
\underbracket{U_{a}(x,t){{\psi}}(\bm{\mathrm{Re}}(\bm{\mathfrak{D}}),t);\mathbb{I}))
\big\|\!\big\|{{\mathscr{B}}}(x,t)\big\|\!\big\|_{SE_{1}(\bm{\mathfrak{D}})}}
=\nabla_{b}U_{a}(x,t)
\end{align}
Since the fluid is incompressible then ${{\mathds{E}}}\big\langle \nabla_{a}{{\mathscr{U}}}_{a}(x,t)\big\rangle=\nabla_{a}U_{a}(x,t)=0$. The 2nd derivative is
\begin{align}
\nabla_{a}\nabla_{b}{{\mathscr{U}}}_{a}(x,t)&=\nabla_{a}\nabla_{b}U_{a}(x,t)+
{\bm{\bm{\alpha}}}\nabla_{a}\nabla_{b}U_{a}(x,t){{\psi}}(\bm{\mathrm{Re}}(\bm{\mathfrak{D}},t);
\mathbb{I}{{\mathscr{B}}}(x,t)\nonumber\\&
+{\bm{\bm{\alpha}}}\nabla_{b}U_{a}(x,,t){{\psi}}(\bm{\mathrm{Re}}(\bm{\mathfrak{D}}),t);\mathbb{I})\nabla_{a}{{\mathscr{B}}}(x,t)
+{\bm{\bm{\alpha}}}U_{a}(x,t)\nabla_{a}\nabla_{b}{{\psi}}(\bm{\mathrm{Re}}(\bm{\mathfrak{D}},t);\mathbb{I}){{\mathscr{B}}}(x,t)\nonumber\\&
+{\bm{\bm{\alpha}}}U_{a}(x,t){{\psi}}(\bm{\mathrm{Re}}(\bm{\mathfrak{D}}),t;\mathbb{I})\nabla_{a}{D}_{b}{{\mathscr{B}}}(x,t)
+{\bm{\bm{\alpha}}}\nabla_{a}U_{a}(x,t){{\psi}}(\bm{\mathrm{Re}}(\bm{\mathfrak{D}},t);\mathbb{I})\nabla_{b}{{\mathscr{B}}}(x,t)\nonumber\\&
+{\bm{\bm{\alpha}}}U_{a}(x,t){{\psi}}(\bm{\mathrm{Re}}(\bm{\mathfrak{D}},t);\mathbb{I})\nabla_{a}\nabla_{b}
{{\mathscr{B}}}(x,t)
\end{align}
Again, taking the expectation, all underbraced terms vanish so that
\begin{align}
&\bm{\mathbb{E}}\langle \nabla_{a}\nabla_{b}{{\mathscr{U}}}_{a}(x,t)\rangle=\nabla_{a}\nabla_{b}U_{a}(x,t)+
\underbrace{{\bm{\bm{\alpha}}}\nabla_{a}\nabla_{b}{{\psi}}(\bm{\mathrm{Re}}(\bm{\mathfrak{D}},t);\mathbb{I})\big\|\!\big\|
{{\mathscr{B}}}(x,t)\big\|\!\big\|_{SE_{1}(\bm{\mathfrak{D}})}}\nonumber\\&
+\underbracket{{\bm{\bm{\alpha}}}\nabla_{a}U_{a}(x,t)\nabla_{b}{{\psi}}(\bm{\mathrm{Re}}(\bm{\mathfrak{D}},t);\mathbb{I})
\big\|\!\big\|{{\mathscr{B}}}(x,t)\big\|\!\big\|_{SE_{1}(\mathfrak{H})}}
+\underbracket{{\bm{\bm{\alpha}}} U_{a}(x,t){{\psi}}(\bm{\mathrm{Re}}(\bm{\mathfrak{D}},t);\mathbb{I})\big\|\!\big\|
{{\mathscr{B}}}(x,t)\big\|\!\|_{SE_{1}(\bm{\mathfrak{D}})}}
\nonumber\\&+\underbracket{{\bm{\bm{\alpha}}} U_{a}(x,t){{\psi}}(\bm{\mathrm{Re}}(\bm{\mathfrak{D}},t);\mathbb{I})\big)\nabla_{a}\nabla_{b}\big\|\!\big\|
{{\mathscr{B}}}(x,t)\big\|\!\big\|_{SE_{1}(\bm{\mathfrak{D}})}}
+\underbrace{{\bm{\bm{\alpha}}}\nabla_{a}U_{a}(x,t)
{{\psi}}(\bm{\mathrm{Re}}(\bm{\mathfrak{D}},t);\mathbb{I})\big\|\!\big\|\nabla_{b}
{{\mathscr{B}}}(x,t)\big\|\!\big\|_{SE_{1}(\bm{\mathfrak{D}})}}
\nonumber\\&
\equiv \nabla_{a}\nabla_{b}U_{a}(x,t)+\underbrace{{\bm{\bm{\alpha}}}
{{\psi}}(\bm{\mathrm{Re}}(\bm{\mathfrak{D}},t);\mathbb{I})\big\|\!\big\|
\nabla_{a}\nabla_{b}{{\mathscr{B}}}(x,t)\big\|\!\big\|_{SE_{1}(\bm{\mathfrak{D}})}}\nonumber\\&
+\underbrace{{\bm{\bm{\alpha}}}\nabla_{a}U_{a}(x,t)\nabla_{b}
{{\psi}}(\bm{\mathrm{Re}}(\bm{\mathfrak{D}},t);\mathbb{I})\big\|\!\big\|
{{\mathscr{B}}}(x,t)\big\|\!\big\|_{SE_{1}(\bm{\mathfrak{D}})}}
+\underbracket{{\bm{\bm{\alpha}}}U_{a}(x,t){{\psi}}(\bm{\mathrm{Re}}(\bm{\mathfrak{D}},t);\mathbb{I})\big\|\!\big\|
{{\mathscr{B}}}(x,t)\big\|\!\|_{SE_{1}(\bm{\mathfrak{D}})}}\nonumber\\&+\underbrace{{\bm{\bm{\alpha}}} U_{a}(x,t){{\psi}}(\bm{\mathrm{Re}}(\bm{\mathfrak{D}},t);\mathbb{I})\big\|\!\big\|\nabla_{a}\nabla_{b}
{{\mathscr{B}}}(x,t)\big\|\!\big\|_{SE_{1}(\bm{\mathfrak{D}})}}\underbrace{\bm{\bm{\alpha}}\nabla_{a}U_{a}(x,t)
{{\psi}}(\bm{\mathrm{Re}}(\bm{\mathfrak{D}},t);\mathbb{I})
\big\|\!\big\|\nabla_{b}{{\mathscr{B}}}(x,t)\big\|\!\big\|_{SE_{1}(\bm{\mathfrak{D}})}}
\end{align}
\end{proof}
\begin{rem}
If the ansatz (2.125) was used, then there would be an extra term in the spatial derivative so that
\begin{align}
\nabla_{b}{{\mathscr{U}}}_{a}(x,t)&=\nabla_{b}U_{a}(x,t)
+{\bm{\bm{\alpha}}}\nabla_{b}U_{a}(x,t){{\psi}}(\bm{\mathrm{Re}}
(x,t),\bm{\mathrm{Re}}_{c}){{\mathscr{B}}}(x)\nonumber\\&
+{\bm{\bm{\alpha}}}U_{a}(x,t)\nabla_{b}{{\psi}}(\bm{\mathrm{Re}}(x,t),\bm{\mathrm{Re}}_{c}){{\mathscr{B}}}(x)
{\bm{\bm{\alpha}}}U_{a}(x,t){\psi}(\bm{\mathrm{Re}}(x,t),\bm{\mathrm{Re}}_{c})\nabla_{b}{{\mathscr{B}}}(x)
\end{align}
plus three extra terms when one takes the second derivative or Laplacian. This greatly complicates derivations and estimates that are to be made, so the ansatz (2.109) is utilised throughout.
\end{rem}
Finally, once can (tentatively) propose a functional derivative of the random field or flow at any $t>0$
\begin{prop}
Given the random field ${{\mathscr{U}}}_{a}(x,t)$, the functional derivatives are
\begin{align}
&\frac{\delta{\mathscr{U}}_{a}(x,t)}{\delta \mathbf{Re}(\bm{\mathfrak{D}},t)}
={\bm{\bm{\alpha}}}U_{a}(x,t)\frac{\delta}{\delta \mathbf{Re}(\bm{\mathfrak{D}},t)}{{\psi}}(\bm{\mathrm{Re}}(\bm{\mathfrak{D}},t);\mathbb{I}){{\mathscr{B}}}(x)\\&
\frac{\delta^{2}{\mathscr{U}}_{a}(x,t)}{\delta \mathbf{Re}(\bm{\mathfrak{D}},t)^{2}}
={\bm{\bm{\alpha}}}U_{a}(x,t)\frac{\delta^{2}}{\delta\mathbf{Re}(\bm{\mathfrak{D}},t)^{2}}
{{\psi}}(\bm{\mathrm{Re}}(\bm{\mathfrak{D}},t);\mathbb{I}){{\mathscr{B}}}(x)
\end{align}
with the expectations
\begin{align}
&{\bm{\mathbb{E}}}\left\langle\frac{\delta{\mathscr{U}}_{a}(x,t)}{\delta \mathbf{Re}(\bm{\mathfrak{D}},t)}\right\rangle
={\bm{\bm{\alpha}}}U_{a}(x,t)\frac{\delta}{\delta \mathbf{Re}(\bm{\mathfrak{D}},t)}{{\psi}}(\bm{\mathrm{Re}}(\bm{\mathfrak{D}},t);\mathbb{I})
{\bm{\mathbb{E}}}\langle{{\mathscr{B}}}(x)\rangle=0\\&{\bm{\mathbb{E}}}\left\langle\frac{\delta^{2}{\mathscr{U}}_{a}(x,t)}{\delta \mathbf{Re}(\bm{\mathfrak{D}},t)^{2}}\right\rangle
={\bm{\bm{\alpha}}}U_{a}(x,t)\frac{\delta^{2}}{\delta \mathbf{Re}(\bm{\mathfrak{D}},t)^{2}}{{\psi}}(\bm{\mathrm{Re}}(\bm{\mathfrak{D}},t);\mathbb{I})
{\bm{\mathbb{E}}}\langle{{\mathscr{B}}}(x)\rangle=0
\end{align}
\end{prop}
\subsection{Reynolds stresses, stochastic averages and correlations}
A crucial issue of central importance in any theory of turbulence is how to define, model or calculate Reynolds stresses and other statistical correlations. However, such averages are difficult to define and formulate rigorously. As previously stated in the introduction, one can utilise either long time averages, volume averages or ensemble averages.
Here, instead, we consider binary correlations of the random fields or turbulent flows ${{\mathscr{U}}}_{a}(x,t)$ and ${{\mathscr{U}}}_{b}(y,t)$ at different positions $(x,y)\in\bm{\mathfrak{D}}\subset{\mathbb{R}}^{3}$ at any time $t>0$. The expectations or averages are taken with respect to
$\mathbb{E}\langle\bullet\rangle$. One can also consider any high-order correlations.
\begin{gen}Given the turbulent flow ${{\mathscr{U}}}_{a}(x,t)=U_{a}(x,t)+{\mathscr{T}}_{a}(x,t)$ within $\bm{\mathfrak{D}}$, the following correlations and stochastic averages/means can be computed for scales $\lambda\le L$ if the statistics of the regulated isotropic and homogenous random field ${{\mathscr{B}}}_{a}(x,t)$ are known.
\begin{enumerate}[(a)]
\item The binary velocity correlations for all $(x,y)\in\bm{\mathfrak{D}}$ and $(t,s)\in[0,T]$ or $(t,s)\in[0,T]$
\begin{align}
{\bm{\mathsf{T}}}_{ab}(x,y)=
{\bm{\mathbb{E}}}\langle{{\mathscr{U}}}_{a}(x,t){{\otimes}}{{\mathscr{U}}}_{b}(y,s)\rangle
\end{align}
For any triplet $(x,y,\mathbf{z})\in\bm{\mathfrak{D}} $, the triple velocity correlations
\begin{align}
{\bm{\mathsf{T}}}_{abc}(x,y,\mathbf{z})=
{\bm{\mathbb{E}}}\langle{{\mathscr{U}}}_{a}(x,t){{\otimes}}{{\mathscr{U}}}_{b}(y,s){{\otimes}}
{{\mathscr{U}}}_{c}(\mathbf{z},t)\rangle
\end{align}
The p-point correlations for points $\lbrace x_{i_{1}},x_{2},t),...,x_{i_{p}}(x,t)\rbrace\in\bm{\mathfrak{D}}$ such that
\begin{align}
&{\bm{\mathsf{T}}}_{1_{1}...i_{M}}(x_{i_{1}},...,x_{i_{M}})\nonumber\\&
={\bm{\mathbb{E}}}\langle{{\mathscr{U}}}_{i_{1}}(x_{1},t){{\otimes}}{{\mathscr{U}}}_{i_{2}}(x_{2},t){{\otimes}}...
{{\otimes}}{{\mathscr{U}}}_{i_{M-1}}(x_{M},t)
{{\otimes}}{{\mathscr{U}}}_{i_{p}}(x_{M},t)\rangle
\end{align}
\item The mean velocity, the rms velocity and the p-order moments
\begin{align}
&\big\|\!\big\|{{\mathscr{U}}}_{a}(x,t)\big\|\!\big\|_{SE_{2}(\bm{\mathfrak{D}})}^{2}\equiv
{\bm{\mathbb{E}}}\langle{{\mathscr{U}}}_{a}(x,t){{\otimes}}
{{\mathscr{U}}}^{a}(x,t)\rangle)
={\bm{\mathbb{E}}}\langle\sum_{i=1}^{3}|{{\mathscr{U}}}_{a}(x,t)|^{2}\rangle\equiv
{\bm{\mathbb{E}}}\big\|{{\mathscr{U}}}(x,t)\big\|^{2}
\\&\big\|\!\big\|{{\mathscr{U}}}_{a}(x,t)\big\|\!\big\|=\sqrt{{\bm{\mathbb{E}}}\big({{\mathscr{U}}}_{a}(x,t)
{{\otimes}}{{\mathscr{U}}}^{a}(x,t)\big)}=\sqrt{{\bm{\mathbb{E}}}
\langle\sum_{i=1}^{3}|{{\mathscr{U}}}_{a}(x,t)|^{2}
\rangle}\equiv\sqrt{\bm{\mathbb{E}}\big[\big\|{{\mathscr{U}}}(x,t)\big\|^{2}\big]}\\&
\big\|\!\big\|{{\mathscr{U}}}_{a}(x,t)\big\|\!\big\|_{SE_{p}(\bm{\mathfrak{D}})}=\left({\bm{\mathbb{E}}}\big\|{{\mathscr{U}}}(x,t)|\big\|^{p}_{SE(\bm{\mathfrak{D}}}\right)^{1/p}
=\left({\bm{\mathbb{E}}}\big[\sum_{i=1}^{3}|{{\mathscr{U}}}_{a}(x,t)|^{p}\big]\right)^{1/p}
\end{align}
\item Structure functions or correlations of the form
\begin{align}
&{\bm{\mathsf{S}}}_{p}[{{\mathscr{U}}}]=\big\|\!\big\|{{\mathscr{U}}}_{a}(x+\mathbf{\ell},t)-{{\mathscr{U}}}_{a}(x,t)\big\|\!\big\|_{E_{p}(\bm{\mathfrak{D}})}^{p}\equiv
{\bm{\mathbb{E}}}\langle|{{\mathscr{U}}}_{a}(x+\mathbf{\ell},t)-{{\mathscr{U}}}_{a}(x,t)|^{p}\rangle\nonumber\\& {\bm{\mathbb{E}}}\langle\sum_{i=1}^{3}\big|{{\mathscr{U}}}_{a}(x+\mathbf{\ell},t)-{{\mathscr{U}}}_{a}(x,t)\big|^{p}\rangle
\end{align}
for all $p\ge 2$.
\item Binary correlations among gradients such as
\begin{align}
{\bm{\mathbb{E}}}\langle \nabla_{a}^{(x)}{{\mathscr{U}}}_{a}(x,t){{\otimes}} \nabla_{b}^{(y)}{{\mathscr{U}}}_{b}(y,t)\rangle
\end{align}
or when $y=x$
\begin{align}
\big\|\!\big\|D{i}{{\mathscr{U}}}_{b}(x,t)\big\|\!\big\|^{2}
={\bm{\mathbb{E}}}\langle \nabla_{a}^{(x)}{{\mathscr{U}}}_{a}(x,t){{\otimes}} \nabla_{b}^{(x)}
{{\mathscr{U}}}_{b}(y,t)\rangle
\end{align}
'Off-diagonal' binary correlation such as
\begin{align}
&{\bm{\mathbb{E}}}\langle{{\mathscr{U}}}_{a}(x,t){{\otimes}} \nabla_{b}^{(y)}{{\mathscr{U}}}^{a}(x,t)\rangle\\&
{\bm{\mathbb{E}}}\langle{{\mathscr{U}}}_{a}(x,t){{\otimes}} \Delta{{\mathscr{U}}}^{a}(x,t)\rangle
\end{align}
\item Stochastic averages of the turbulent energy and enstrophy integrals
\begin{align}
{\bm{\mathbb{E}}}\langle{\mathcal{E}}[{{\mathscr{U}}}]\rangle \equiv{\int}_{\bm{\mathfrak{D}}}\big\|\!\big\|{{\mathscr{U}}}_{a}(x,t) \big\|\!\big\|^{2}=
{\bm{\mathbb{E}}}\langle{\int}_{\bm{\mathfrak{D}}}{{\mathscr{U}}}_{a}(x,t)
{{\otimes}}{{\mathscr{U}}}^{a}(x,t)d\mu(x)\rangle
\end{align}
\begin{align}
{\bm{\mathbb{E}}}\langle{{{\psi}}}[{{\mathscr{U}}}]\rangle&=\equiv\int_{\bm{\mathfrak{D}}}\big\|\!\big\|\nabla^{a}{{\mathscr{U}}}_{a}(x,t) \big\|\!\big\|^{2}={\bm{\mathbb{E}}}\langle{\int}_{\bm{\mathfrak{D}}}\nabla_{a}{{\mathscr{U}}}^{a}(x,t)
{{\otimes}}~\nabla^{a}{{\mathscr{U}}}_{a}(x,t)
d\mu(x)\rangle
\end{align}
\item The averaged or mean nonlinear convective term of the NS equations
\begin{align}
\frac{1}{2}{\bm{\mathbb{E}}}
\langle \nabla_{b}({{\mathscr{U}}}_{a}(x,t){{\otimes}}{{\mathscr{U}}}_{b}(x,t))\rangle
\end{align}
and the full stochastically averaged Navier-Stokes equations
\begin{align}
&\|\!\|\frac{\partial}{\partial t}{\mathscr{U}}_{a}(x,t)+\frac{1}{2}\nabla_{b}({\mathscr{U}}_{a}(x,t)
{{\otimes}}{\mathscr{U}}_{b}(x,t))
-\nu \Delta{\mathscr{U}}_{a}(x,t)-\nabla_{a}
\mathbf{P}(x,t)\|\!\|_{SE_{1}(\bm{\mathfrak{D}})}\nonumber\\&\equiv\langle \frac{\partial}{\partial t}{\mathscr{U}}_{a}(x,t)+\frac{1}{2}\nabla_{b}
({\mathscr{U}}_{a}(x,t){{\otimes}}{\mathscr{U}}_{b}(x,t))-\nu \Delta{\mathscr{U}}_{a}(x,t)
-\nabla_{a}{\mathbf{P}}(x,t)\rangle
\end{align}
or for steady state flow
\begin{align}
&\|\!\|\frac{1}{2}\nabla_{b}({{\mathscr{U}}}_{a}(x,t){{\otimes}}{{\mathscr{U}}}_{b}(x,t))-\Delta{{\mathscr{U}}}_{a}(x,t)
-\nabla_{a}{\mathscr{P}}(x,t)\|\!\|_{SE_{1}(\mathcal{H})}\nonumber\\&\equiv
{\bm{\mathbb{E}}}\langle\frac{1}{2}
\nabla_{b}({{\mathscr{U}}}_{a}(x,t){{\otimes}}{{\mathscr{U}}}_{b}(x,t))
-\Delta{{\mathscr{U}}}_{a}(x,t)-\nabla_{a}{\mathbf{P}}(x,t)\rangle
\end{align}
\end{enumerate}
\end{gen}
\section{The stochastically averaged Navier-Stokes equations for the turbulent flow I: Neglecting the pressure term}
We now present the main theorem, which gives the stochastically averaged Navier-Stokes equations satisfied by the turbulent fluid flow ${{\mathscr{U}}}_{a}(x,t)$.
The stochastic averaging leads to an extra non-vanishing term via the nonlinear convective term in the Navier-Stokes equations. Traditionally, statistical
averaging or Reynolds averaging of NS equations is done with respect to either spatial or ensemble averages as discussed in the introduction. This has physical limitations and can't be rigorously defined. Here, instead the stochastic averaging or expectations ${\bm{\mathbb{E}}}\big\langle\bullet\big\rangle$ is done with respect to a well-defined theory of classical random fields, with which the underlying flow is 'mixed' to produce a turbulent fluid flow. Here, random fluctuations induced within the pressure gradient term $\nabla_{a}\mathbf{P}(x,t)$ are ignored for now to minimise complications. There are incorporated in detail in \textbf{Section 5}.
\begin{thm}(\textrm{\textbf{Stochastically averaged Navier-Stokes equations for which ${\mathscr{U}}_{a}(x,t)$ is a solution)}}\newline
Let $U_{a}(x,t)$ be a smooth deterministic flow satisfying the NS equations within a domain $\bm{\mathfrak{D}}\subset\bm{\mathbb{R}}^{3}$ with $t\in[0,T]$ or $t\in[0,\infty)$ and $\mathrm{Vol}(\bm{\mathfrak{D}})\sim L^{3}$. The flow $U_{a}(x,t)$ is at least smooth enough to differentiable to 2nd order and satisfying some suitable BCs and initial data on $\bm{\mathfrak{D}}$. The fluid has viscosity $\nu$ and is compressible so that $\nabla^{a}U_{a}=0$. The (sharp) critical transitional Reynolds number at which a smooth flow starts to become turbulent is $\bm{\mathrm{Re}}_{c}(\bm{\mathfrak{D}})$. Now let ${{\mathscr{B}}}(x)$ be a Bargman-Fock Gaussian random field existing for all $x\in\bm{\mathfrak{D}}\subset\mathbb{R}^{3}$ and having the statistical properties previously defined so that
\begin{align}
&\big\|\!\big\|\frac{\partial}{\partial t}{{\mathscr{B}}}(x)\big\|\!\big\|_{SE_{1}(\bm{\mathfrak{D}})}\equiv{\bm{\mathbb{E}}}\langle \frac{\partial}{\partial t}{{\mathscr{B}}}(x)\rangle
=0\nonumber\\&\big\|\!\big\|{{\mathscr{B}}}(x)\|\!\big\|_{SE_{1}(\bm{\mathfrak{D}})}\equiv{\bm{\mathbb{E}}}\langle{{\mathscr{B}}}(x)\rangle=0\nonumber\\&
\big\|\!\big\|\nabla_{a}{{\mathscr{B}}}(x)\big\|\!\big\|_{SE_{1}(\bm{\mathfrak{D}})}\equiv{\bm{\mathbb{E}}}\langle \nabla_{a}{{\mathscr{B}}}(x)\rangle=0\nonumber\\&
\big\|\!\big\|\nabla_{a}\nabla_{b}{\mathscr{B}}(x)\big\|\!\big\|_{SE_{1}(\bm{\mathfrak{D}})}
\equiv{\bm{\mathbb{E}}}\langle \nabla_{a}\nabla_{b}{{\mathscr{B}}}(x,t)\rangle=0\nonumber\\&
{\bm{\mathbb{E}}}\langle{{\mathscr{B}}}(x){{\otimes}}~{{\mathscr{B}}}(y)\rangle
={\bm{\Xi}}(x,y;\lambda)
={\mathsf{C}}\exp(-\|x-y\|^{2}\lambda^{-2})\nonumber\\&
\big\|\!\big\|{{\mathscr{B}}}(x,t)\big\|\!\big\|_{SE_{2}(\bm{\mathfrak{D}})}^{2}=
{\bm{\mathbb{E}}}\langle{{\mathscr{B}}}(x){{\otimes}}~{{\mathscr{B}}}(x)\rangle
={\bm{\Xi}}(x,x;\lambda)=\mathsf{C}\nonumber\\&
{\bm{\mathbb{E}}}\langle{{\mathscr{B}}}(x){{\otimes}}~\nabla_{a}{{\mathscr{B}}}(x)\rangle=0
\end{align}
where $\lambda$ is the correlation length. Let $\bm{\mathrm{Re}}(\bm{\mathfrak{D}})=\|\bm{\mathfrak{U}}_{a}(x,t)\|\mathrm{L}/\nu >\bm{\mathrm{Re}}_{c}^{(1)}(\bm{\mathfrak{D}},t)$ for all $(x,t)\in\bm{\mathfrak{D}}\times[0,T]$. As before, the smooth flow $U_{a}(x,t)$ is 'mixed' with the BF random field ${\mathscr{B}}(x)$ to give a turbulent flow of the form
\begin{align}
&{{\mathscr{U}}}_{a}(x,t)=U_{a}(x,t)+{\bm{\bm{\alpha}}}
U_{a}(x,t){{\psi}}(\bm{\mathrm{Re}}(\bm{\mathfrak{D}}),t);\mathbb{I})
{{\mathscr{B}}}(x,t)
\end{align}
so that the mean or averaged flow is ${\bm{\mathbb{E}}}\langle{{\mathscr{U}}}_{a}(x,t)\big\rangle=U_{a}(x,t)$. The pressure fluctuations (for now) are not considered. The turbulent flow ${{\mathscr{U}}}_{a}(x,t)$ is then a solution of the following stochastically averaged Navier-Stokes equations
\begin{align}
&\|\!\|\frac{\partial}{\partial t}{{\mathscr{U}}}_{a}(x,t)-\nu \Delta{{\mathscr{U}}}_{a}(x,t)
+{{\mathscr{U}}}^{b}(x,t){{\otimes}}~\nabla_{b}{{\mathscr{U}}}_{a}(x,t)+\nabla_{a}{\mathbf{P}}(x,t)
\|\!\|_{SE_{1}(\bm{\mathfrak{D}})}\nonumber\\&={\bm{\mathbb{E}}}\langle \frac{\partial}{\partial t}{{\mathscr{U}}}_{a}(x,t)-\nu \Delta{{\mathscr{U}}}_{a}(x,t)
+{{\mathscr{U}}}^{b}(x,t){{\otimes}}~\nabla_{b}{{\mathscr{U}}}_{a}(x,t)+\nabla_{a}{\mathbf{P}}(x,t)\rangle\nonumber\\&
=\frac{\partial}{\partial t}U_{a}(x,t)-\nu \Delta U_{a}(x,t)+U^{b}(x,t)\nabla_{b}U_{a}(x,t)+\nabla_{a}\mathbf{P}(x,t)\nonumber\\&
+{\bm{\bm{\alpha}}}^{2}U^{b}(x,t)\nabla_{b}U_{a}(x,t)
{{\psi}}(\bm{\mathrm{Re}}(\bm{\mathfrak{D}},t);{\mathbb{I}})^{2}
{\mathsf{C}}\nonumber\\&
=\frac{\partial}{\partial t}U_{a}(x,t)-\nu \Delta U_{a}(x,t)+U^{b}(x,t)\nabla_{b}U_{a}(x,t)1+\bm{\mathcal{Y}}(t)+\nabla_{a}\mathbf{P}(x,t)\nonumber\\&=\frac{\partial}{\partial t}
U_{a}(x,t)-\nu \Delta U_{a}(x,t)+U^{b}(x,t)\nabla_{b}U_{a}(x,t)\bm{\mathcal{Y}}(t)+\nabla_{a}\mathbf{P}(x,t)
\end{align}
where $\bm{\mathsf{{\psi}}}_{1}(t),\bm{\mathsf{{\psi}}}_{2}(t)$ are dimensionless.
\begin{align}
&\bm{\mathsf{{\psi}}}_{1}(t)={\bm{\bm{\alpha}}}^{2}{{\psi}}(\bm{\mathrm{Re}}(\bm{\mathfrak{D}},t);{\mathbb{I}})^{2}\mathsf{C}\\&
\bm{\mathsf{{\psi}}}_{2}(t)= 1 + \bm{\mathsf{{\psi}}}_{1}(t)
\end{align}
More succinctly, The turbulent flow ${\mathscr{U}}_{a}(x,t)$ is a solution of the stochastically averaged NS equations where the nonlinear term is now modified by a factor $\bm{\mathsf{{\psi}}}_{2}(t)$.
\begin{align}
&{\bm{\mathbb{E}}}\langle \frac{\partial}{\partial t}{{\mathscr{U}}}_{a}(x,t)-\nu \Delta{{\mathscr{U}}}_{a}(x,t)
+{{\mathscr{U}}}^{b}(x,t){{\otimes}}~\nabla_{b}{{\mathscr{U}}}_{a}(x,t)+D{\mathscr{P}}(x,t)\rangle\nonumber\\&
=\frac{\partial}{\partial t}U_{a}(x,t)-\nu \Delta U_{a}(x,t)+\bm{\mathsf{{\psi}}}(t)_{2}U^{b}(x,t)\nabla_{b}U_{a}(x,t)+\nabla_{a}\mathbf{P}(x,t)
\end{align}
\end{thm}
\begin{rem}
Because there is a nonlinear term in the NS equations, this introduces a correlation between the random field ${\mathscr{U}}^{b}(x,t)$ and its derivative
$\nabla_{b}{\mathscr{U}}_{a}(x,t)$, that is ${\mathscr{U}}^{b}(x,t){{\otimes}} \nabla_{b}{\mathscr{U}}_{a}(x,t)$ with (nonvanishing) expectation\newline
$\bm{\mathbb{E}}\langle{\mathscr{U}}^{b}(x,t){{\otimes}} \nabla_{b}{\mathscr{U}}_{a}(x,t)\rangle\ne 0$.
\end{rem}
\begin{proof}
Substituting ${\mathscr{U}}_{a}(x,t)$ into the Navier-Stokes equations and taking the derivatives gives
\begin{align}
\frac{\partial}{\partial t}{\mathscr{U}}_{a}(x,t)&-\nu \Delta {\mathscr{U}}_{a}(x,t)
+{\mathscr{U}}^{b}(x,t)\nabla_{b}{\mathscr{U}}_{a}(x,t)+\nabla_{b}\bm{\mathrm{P}}(x,t)\nonumber\\&
=\overbracket{\frac{\partial}{\partial t}U_{a}(x,t)}+{\bm{\bm{\alpha}}}\frac{\partial}{\partial t}U_{a}(x,t){{\psi}}(\bm{\mathrm{Re}}(\bm{\mathfrak{D}}),t);{\mathbb{I}})
{{\mathscr{B}}}(x)\nonumber\\&+U_{a}(x)
{\bm{\bm{\alpha}}}{{\psi}}(\bm{\mathrm{Re}}(\bm{\mathfrak{D}}),t);{\mathbb{I}})\frac{\partial}{\partial t}{{\mathscr{B}}}(x)\nonumber\\&
-\overbracket{\nu\delta^{ab}\nabla_{a}\nabla_{b}U_{a}(x,t)}-\nu\delta^{ab}{\bm{\bm{\alpha}}}\nabla_{a}\nabla_{b}U_{a}(x,t)
{{\psi}}(\bm{\mathrm{Re}}(\bm{\mathfrak{D}}),t);{\mathbb{I}}){{\mathscr{B}}}(x)\nonumber\\&
-\nu\delta^{ab}{\bm{\bm{\alpha}}}\nabla_{b}U_{a}(x,t){{\psi}}(\bm{\mathrm{Re}}(\bm{\mathfrak{D}}),t);{\mathbb{I}})
\nabla_{a}{{\mathscr{B}}}(x,t)\nonumber\\&
-\nu\delta^{ab}{\bm{\bm{\alpha}}}\times\underbrace{\nabla_{a}U_{a}(x,,t)}_{=0}\times{{\psi}}(\bm{\mathrm{Re}}(\bm{\mathfrak{D}}),t);{\mathbb{I}})
\nabla_{b}{{\mathscr{B}}}(x)\nonumber\\&
-\nu\delta^{ab}{\bm{\bm{\alpha}}}U_{a}(x,t){{\psi}}(\bm{\mathrm{Re}}(\bm{\mathfrak{D}}),t);{\mathbb{I}})\nabla_{a}\nabla^{b}
{{\mathscr{B}}}(x)\nonumber\\&
+\overbracket{U^{b}\nabla_{b}U_{a}(x,t)}+U^{b}(x,t)\nabla_{b}U^{a}(x,t){\bm{\bm{\alpha}}}
{{\psi}}(\bm{\mathrm{Re}}(\bm{\mathfrak{D}}),t);{\mathbb{I}}){{\mathscr{B}}}(x)\nonumber\\&+
+{\bm{\bm{\alpha}}}U^{b}(x,t)\nabla_{b}U_{a}(x,t){{\psi}}(\bm{\mathrm{Re}}(\bm{\mathfrak{D}}),t);{\mathbb{I}})
{{\mathscr{B}}}(x)\nonumber\\&+U^{b}(x,t)|{\bm{\bm{\alpha}}}|^{2}\nabla_{b}U_{a}(x,t)
|{{\psi}}\big(\bm{\mathrm{Re}}(\bm{\mathfrak{D}},t);{\mathbb{I}})|^{2}|{{\mathscr{B}}}(x){{\otimes}}{{\mathscr{B}}}(x)\nonumber\\&
+|{\bm{\bm{\alpha}}}|^{2}U^{b}(x,t)U_{a}(x,t){{\psi}}(\bm{\mathrm{Re}}(\bm{\mathfrak{D}}),t);{\mathbb{I}})^{2}\nabla_{b}{{\mathscr{B}}}(x)
{{\otimes}}{{\mathscr{B}}}(x)+\overbracket{\nabla_{a}{\mathbf{P}}(x,t)}
\end{align}
The pieces of the underlying deterministic Navier-Stokes equations are emphasised with an overbracket and the extra terms represent the random field contributions from turbulence. Now taking the stochastic average ${\bm{\mathbb{E}}}\big\langle\bullet\big\rangle$,or equivalently the stochastic norm, all underbraced (linear) terms vanish upon using (3.1)
\begin{align}
{\bm{\mathbb{E}}}\langle \frac{\partial}{\partial t}{\mathscr{U}}_{a}(x,t)&-\nu \Delta {\mathscr{U}}_{a}(x,t)
+{{\mathscr{U}}}^{b}(x,t)\nabla_{b}{\mathscr{U}}_{a}(x,t)+\nabla_{b}\bm{\mathrm{P}}(x,t)\rangle\nonumber\\&
=\overbracket{\frac{\partial}{\partial t}U_{a}(x,t)}+\underbrace{{\bm{\bm{\alpha}}}\frac{\partial}{\partial t}U_{a}(x,t){{\psi}}(\bm{\mathrm{Re}}(\bm{\mathfrak{D}}),t);{\mathbb{I}})^{2}
{\bm{\mathbb{E}}}\langle{{\mathscr{B}}}(x)\rangle}\nonumber\\&+\underbrace{U_{a}(x)
{\bm{\bm{\alpha}}}{{\psi}}(\bm{\mathrm{Re}}(\bm{\mathfrak{D}}),t);{\mathbb{I}}){\bm{\mathbb{E}}}
\langle \frac{\partial}{\partial t}{{\mathscr{B}}}(x)\rangle}\nonumber\\&
-\overbracket{\nu\delta^{ab}\nabla_{a}\nabla_{b}U_{a}(x,t)}-\underbrace{\nu\delta^{ab}{\bm{\bm{\alpha}}}\nabla_{a}\nabla_{b}U_{a}(x,t)
{{\psi}}(\bm{\mathrm{Re}}(\bm{\mathfrak{D}}),t);{\mathbb{I}}){\bm{\mathbb{E}}}\langle{{\mathscr{B}}}(x)\rangle}\nonumber\\&
-\underbrace{\nu\delta^{ab}{\bm{\bm{\alpha}}}\nabla_{b}U_{a}(x,t){{\psi}}(\bm{\mathrm{Re}}(\bm{\mathfrak{D}}),t);{\mathbb{I}})
{\bm{\mathbb{E}}}\langle \nabla_{a}{{\mathscr{B}}}(x,t)\rangle}\nonumber\\&
-\underbrace{\nu\delta^{ab}{\bm{\bm{\alpha}}}\times\underbrace{\nabla_{a}U_{a}(x,,t)}_{=0}\times{{\psi}}(\bm{\mathrm{Re}}(\bm{\mathfrak{D}}),t);{\mathbb{I}})
{\bm{\mathbb{E}}}\langle \nabla_{b}{{\mathscr{B}}}(x)\rangle}\nonumber\\&
-\underbrace{\nu\delta^{ab}{\bm{\bm{\alpha}}}U_{a}(x,t){{\psi}}(\bm{\mathrm{Re}}(\bm{\mathfrak{D}}),t);{\mathbb{I}})
{\bm{\mathbb{E}}}\langle \nabla_{a}\nabla^{b}{{\mathscr{B}}}(x)\rangle}\nonumber\\&
+\overbracket{U^{b}\nabla_{b}U_{a}(x,t)}+\underbrace{U^{b}(x,t)\nabla_{b}U^{a}(x,t){\bm{\bm{\alpha}}}
{{\psi}}(\bm{\mathrm{Re}}(\bm{\mathfrak{D}}),t);{\mathbb{I}}){\bm{\mathbb{E}}}\langle{{\mathscr{B}}}(x)\rangle}\nonumber\\&
+\underbrace{{\bm{\bm{\alpha}}}U^{b}(x,t)\nabla_{b}U_{a}(x,t){{\psi}}(\bm{\mathrm{Re}}(\bm{\mathfrak{D}}),t);{\mathbb{I}})
{\bm{\mathbb{E}}}\langle{{\mathscr{B}}}(x)\rangle}\nonumber\\&+{\bm{\bm{\alpha}}}^{2}U^{b}(x,t)\nabla_{b}U_{a}(x,t)
{{\psi}}\big(\bm{\mathrm{Re}}(\bm{\mathfrak{D}},t);{\mathbb{I}})^{2}{\bm{\mathbb{E}}}
\langle{{\mathscr{B}}}(x){{\otimes}}{{\mathscr{B}}}(x)\rangle\nonumber\\&
+\underbrace{{\bm{\bm{\alpha}}}^{2}U^{b}(x,t)U_{a}(x,t){{\psi}}(\bm{\mathrm{Re}}(\bm{\mathfrak{D}}),t);{\mathbb{I}})^{2}
{\bm{\mathbb{E}}}\langle \nabla_{b}{{\mathscr{B}}}(x)
{{\otimes}}{{\mathscr{B}}}(x)\rangle}+\overbracket{\nabla_{a}{\mathbf{P}}(\mathbf{x,t})}
\end{align}
This leaves extra non-vanishing terms arising from the non-linear (convective) term of the NS equations so that
\begin{align}
{\bm{\mathbb{E}}}\langle{{\mathscr{U}}}_{a}(x,t)&-\nu \Delta{ {\mathscr{U}}}_{a}(x,t)
+{{\mathscr{U}}}^{b}(x,t){{\otimes}}~\nabla_{b}{{\mathscr{U}}}_{a}(x,t)+\nabla_{b}\mathbf{P}(x,t)\rangle\nonumber\\&
=\overbracket{\frac{\partial}{\partial t}U_{a}(x,t)}-\overbracket{\nu \Delta U_{a}(x,t)}+\overbracket{U^{b}\nabla_{b}U_{a}(x,t)}
+\overbracket{\nabla_{a}\mathbf{P}(x,t)}\nonumber\\& +{\bm{\bm{\alpha}}}^{2}U^{b}(x,t)\nabla_{b}
U_{a}(x,t){{\psi}}(\bm{\mathrm{Re}}(\bm{\mathfrak{D}}),t);{\mathbb{I}})^{2}
{\bm{\mathbb{E}}}\langle{{\mathscr{B}}}(x,t)
{{\otimes}}~{{\mathscr{B}}}(x,t)\rangle
\nonumber\\&=\frac{\partial}{\partial t}U_{a}(x,t)-\nu \Delta U_{a}(x,t)+U^{b}(x,t)\nabla_{b}U_{a}(x,t)+\nabla_{a}\mathbf{P}(x,t)\nonumber\\& +{\mathsf{\bm{\bm{\alpha}}}}^{2}U^{b}(x,t)\nabla_{b}
U_{a}(x,t){{\psi}}(\bm{\mathrm{Re}}(\bm{\mathfrak{D}}),t);{\mathbb{I}})^{2}{\mathsf{C}}
\nonumber\\&\equiv\frac{\partial}{\partial t}U_{a}(x,t)-\nu \Delta U_{a}(x,t)+\nabla_{a}\mathbf{P}(x,t)\nonumber\\& +U^{b}(x,t)\nabla_{b}U_{a}(x,t)1+{\bm{\bm{\alpha}}}^{2}{{\psi}}(\bm{\mathrm{Re}}(\bm{\mathfrak{D}}),t);{\mathbb{I}})^{2}{\mathsf{C}}
\nonumber\\&\equiv\frac{\partial}{\partial t}U_{a}(x,t)-\nu \Delta U_{a}(x,t)+\nabla_{a}\mathbf{P}(x,t)\nonumber\\& +U^{b}(x,t)\nabla_{b}U_{a}(x,t)1+{{\bm{\mathsf{{\psi}}}}}_{1}(t)
\nonumber\\&\equiv\frac{\partial}{\partial t}U_{a}(x,t)-\nu \Delta U_{a}(x,t)+\nabla_{a}\mathbf{P}(x,t)+U^{b}(x,t)\nabla_{b}U_{a}(x,t)
\bm{\mathcal{Y}}(t)
\end{align}
\end{proof}
\begin{cor}
In the limit that the volume-averaged Reynolds number $\bm{\mathrm{Re}}(\bm{\mathfrak{D}},t)$ over $\bm{\mathfrak{D}}$ at any $t>0$ is reduced to or below the critical Reynolds number $\bm{\mathrm{Re}}_{c}(\bm{\mathfrak{D}})$, the extra term $\bm{\mathcal{Y}}(t)$ vanishes and $\bm{\mathcal{Y}}(t)=1$ the standard Navier-Stokes PDEs are recovered so that
\begin{align}
{\bm{\mathbb{E}}}\langle{{\mathscr{U}}}_{a}(x,t)&-\nu \Delta{ {\mathscr{U}}}_{a}(x,t)
+{{\mathscr{U}}}^{b}(x,t){{\otimes}}~\nabla_{b}{{\mathscr{U}}}_{a}(x,t)+\nabla_{b}\mathbf{P}(x,t)\rangle
\nonumber\\&\equiv\frac{\partial}{\partial t}U_{a}(x,t)-\nu \Delta U_{a}(x,t)+U^{b}(x,t)\nabla_{b}U_{a}(x,t)+\nabla_{a}\mathbf{P}(x,t),~~{\forall}~\bm{\mathrm{Re}}(\bm{\mathfrak{D}},t)\le
\bm{\mathrm{Re}}_{c}(\bm{\mathfrak{D}})
\end{align}
\end{cor}
\subsection{Equivalence of the averaged NS equations to a set of transformed equations}
The stochastically avergaged NS equations (3.3) can be shown to be equivalent to a set of transformed NS equations
\begin{lem}
Let $U_{a}(x,t)$ be a fluid flow with viscosity $\nu$ evolving according to the NS equations from initial data, for all ${\mathbb{R}}\times [0,\infty)$.
Let $\bm{\mathfrak{U}}(t)$ be the volume-averaged velocity within $\bm{\mathfrak{D}}$ so that the averaged Reynolds number
anywhere within $\bm{\mathfrak{D}}$ at any $t\in[0,\infty)$ is $\bm{\mathrm{Re}}(\bm{\mathfrak{D}},t)=\|\bm{\mathfrak{U}}_{a}(x,t)\|\mathrm{L}/\nu$, and as before
$\mathrm{Vol}[\bm{\mathfrak{D}}]\sim\mathrm{L}^{3}$. Define the following 'boost' transform $\bm{\mathfrak{B}}:\mathbb{R}^{(+)}\rightarrow{\mathbb{R}}^{(+)}$ on the velocity and pressure such that
\begin{empheq}[right=\empheqrbrace]{align}
&U_{a}(x,t)\hookrightarrow{\bm{\mathfrak{B}}}(t)U_{a}(x,t)={\mathfrak{U}}_{a}(x,t)\nonumber\\&
\frac{\partial}{\partial t}\hookrightarrow \frac{\partial}{\partial t}-{{\frac{\frac{\partial}{\partial t}\bm{\mathfrak{B}}(t)}{\bm{\mathfrak{B}}(t)}}}={\mathbf{D}}_{t}\nonumber\\&
\nabla_{b}\hookrightarrow \nabla_{b}\nonumber\\&
\mathbf{P}(x,t)\hookrightarrow{\bm{\mathfrak{B}}}(t)\mathbf{P}(x,t)={\mathfrak{P}}(x,t)
\end{empheq}
where ${\bm{\mathfrak{B}}}(t)$ is a smooth function with derivative $\frac{\partial}{\partial t}{\bm{\mathfrak{B}}}(t)$. The boosted Navier-Stokes equations then have the form
\begin{align}
\frac{\partial}{\partial t}U_{a}(x,t)-\nu \Delta U_{a}(x,t)+{\bm{\mathfrak{B}}}(t)U^{b}(x,t)\nabla_{b}U_{a}(x,t)+\nabla_{a}\mathbf{P}(x,t)=0
\end{align}
so that the non-linear convective term is boosted or rescaled by a factor ${\bm{\mathfrak{B}}}(t)$. The averaged Reynolds number $\bm{\mathrm{Re}}(\bm{\mathfrak{D}},t)$ is also boosted by a factor ${\bm{\mathfrak{B}}}(t)$ so that
\begin{align}
{\Re}(\bm{\mathfrak{D}},t)={\bm{\mathfrak{B}}}(t)\bm{\mathrm{Re}}(\bm{\mathfrak{D}},t)=
{\bm{\mathfrak{B}}}(t)\|\bm{\mathfrak{U}}_{a}(t)\|\mathrm{L}\nu^{-1}
\end{align}
\end{lem}
\begin{proof}
The transformed or 'boosted' NS equations are
\begin{align}
{\bm{\mathfrak{D}}}_{t}{\bm{\mathfrak{U}}}_{a}(x,t)-\nu \Delta{\bm{\mathfrak{U}}}_{a}(x,t)+
{\bm{\mathfrak{U}}}^{b}(x,t)\nabla_{b}{\bm{\mathfrak{U}}}_{a}(x,t)+\nabla_{a}
{\bm{\mathfrak{P}}}(x,t)=0
\end{align}
which upon applying (3.11) become
\begin{align}
\frac{\partial}{\partial t}&-{\frac{\frac{\partial}{\partial t}{\bm{\mathfrak{B}}}(t)}{{\bm{\mathfrak{B}}}(t)}}{\bm{\mathfrak{B}}}(t)U_{a}(x,t)
-\nu{\bm{\mathfrak{B}}}(t)\Delta U_{a}(x,t)+{\bm{\mathfrak{B}}}(t)^{2}U^{b}(x,t)\nabla^{b}U_{b}(x,t)+
{\bm{\mathfrak{B}}}(t)\nabla_{a}\mathbf{P}(x,t)\nonumber\\&
=\big(\frac{\partial}{\partial t}{\bm{\mathfrak{B}}}(t)\big)U_{a}(x,t)+(\frac{\partial}{\partial t}U_{a}(x,t)){\bm{\mathfrak{B}}}(t)-(\frac{\partial}{\partial t}
{\bm{\mathfrak{B}}}(t)U_{a}(x,t)\nonumber\\&
-\nu{\bm{\mathfrak{B}}}(t)(t)\Delta U_{a}(x,t)+{\bm{\mathfrak{B}}}(t)^{2}U^{b}(x,t)\nabla_{b}U_{a}(x,t)
+{\bm{\mathfrak{B}}}(t)(t)\nabla_{a}\mathbf{P}(x,t)\nonumber\\&
=(\frac{\partial}{\partial t}U_{a}(x,t)){\bm{\mathfrak{B}}}(t)-\nu{\bm{\mathfrak{B}}}(t)\Delta U_{a}(x,t)+{\bm{\mathfrak{B}}}(t)^{2}U^{b}(x,t)\nabla_{b}U_{a}(x,t)+
{\bm{\mathfrak{B}}}(t)(t)\nabla_{a}\mathbf{P}(x,t)=0
\end{align}
Dividing out by ${\bm{\mathfrak{B}}}(t)$
\begin{align}
{\partial}_{t}U_{a}(x,t)-\nu \Delta U_{a}(x,t)+{\bm{\mathfrak{B}}}(t)U^{b}(x,t)\nabla_{b}U_{a}(x,t)+\nabla_{a}\mathbf{P}(x,t)=0
\end{align}
\end{proof}
The next theorem establishes that the stochastically averaged NS equations are equivalent to a set of 'boosted' or transformed deterministic NS equations.
\begin{thm}
Let the scenario and conditions of previous theorems/lemmas hold for a turbulent fluid flow of the form
\begin{align}
{{\mathscr{U}}}_{a}(x,t)=U_{a}(x,t)+
{\bm{\bm{\alpha}}}U_{a}(x,t){{\psi}}(\bm{\mathrm{Re}}(\bm{\mathfrak{D}}),t);{\mathbb{I}})
{{\mathscr{B}}}(x,t)=U_{a}(x)
+{\bm{\bm{\alpha}}}{\mathscr{D}}_{a}(x,t)
\end{align}
where ${\bm{\mathbb{E}}}\big\langle{\mathscr{B}}(x,t){{\otimes}}~{\mathscr{B}}(x,t)\big\rangle={\mathsf{C}}$. The Reynolds number throughout $\bm{\mathfrak{D}}$ at any $t>0$ is $\bm{\mathrm{Re}}(\bm{\mathfrak{D}},t)$ and the critical Reynolds number is $\bm{\mathrm{Re}}(\bm{\mathfrak{D}})$. As before, ${{\psi}}(\bm{\mathrm{Re}}(\bm{\mathfrak{D}}),t);{\mathbb{I}})$ is an arbitrary monotone increasing function and vanishes at $\bm{\mathrm{Re}}(\bm{\mathfrak{D}},t)=\bm{\mathrm{Re}}_{c}(\bm{\mathfrak{D}})$. Then the stochastically averaged NS equations
\begin{align}
&\bm{\mathbb{E}}\langle \frac{\partial}{\partial t}{{\mathscr{U}}}_{a}(x,t)-\nu \Delta{{\mathscr{U}}}_{a}(x,t)
+{{\mathscr{U}}}^{b}(x,t){{\otimes}}~\nabla_{b}{{\mathscr{U}}}_{a}(x,t)+\nabla_{b}\bm{\mathrm{P}}(x,t)\rangle\nonumber\\&
=\frac{\partial}{\partial t}U_{a}(x,t)-\nu \Delta U_{a}(x,t)+1+{\bm{\bm{\alpha}}}{{\psi}}(\bm{\mathrm{Re}}(\bm{\mathfrak{D}}),t;{\mathbb{I}}){\mathsf{C}}
U^{b}(x,t)\nabla_{b}U_{a}(x,t)+\nabla_{a}\mathbf{P}(x,t)\nonumber\\&
\equiv\frac{\partial}{\partial t}U_{a}(x,t)-\nu \Delta U_{a}(x,t)+
1+\bm{\mathcal{Y}}(t)U^{b}(x,t)\nabla_{b}U_{a}(x,t)+\nabla_{a}\mathbf{P}(x,t)
\end{align}
are equivalent to the NS equations which will arise from the linear boost transformations of the form
\begin{align}
&U_{a}(x,t)\longrightarrow{\bm{\mathfrak{B}}}(t)U_{a}(x,t)={\bm{\mathfrak{U}}}_{a}(x,t)
\nonumber\\&=1+{\bm{\bm{\alpha}}}~~
{{\psi}}(\bm{\mathrm{Re}}(\bm{\mathfrak{D}}),t);{\mathbb{I}})^{2}{\mathsf{C}}U_{a}(x,t)=
1+\bm{\mathcal{Y}}(t)U_{a}(x,t)\nonumber\\&
\frac{\partial}{\partial t}\longrightarrow\frac{\partial}{\partial t}-\frac{\frac{\partial}{\partial t}\bm{\mathfrak{B}}(t)}{\bm{\mathfrak{B}}(t)}=\frac{\partial}{\partial t}-\frac{\frac{\partial}{\partial t}1+{\bm{\bm{\alpha}}}{\mathsf{C}}^{2}~
{{\psi}}(\bm{\mathrm{Re}}(\bm{\mathfrak{D}},t);{\mathbb{I}}^{2}}{1+{\bm{\bm{\alpha}}}{\mathsf{C}}^{2}~
{{\psi}}(\bm{\mathrm{Re}}(\bm{\mathfrak{D}},t);{\mathbb{I}}^{2}}=\frac{\partial}{\partial t}-\frac{\frac{\partial}{\partial t}1+\bm{\mathcal{Y}}(t)
}{1+\bm{\mathcal{Y}}(t)}
\nonumber\\&\mathbf{P}(x,t)\longrightarrow {\bm{\mathfrak{B}}}(t)\mathbf{P}(x,t)={\bm{\mathfrak{P}}}(x,t)=1+{\bm{\bm{\alpha}}}{\mathsf{C}}^{2}~
{{\psi}}(\bm{\mathrm{Re}}(\bm{\mathfrak{D}},t);{\mathbb{I}}^{2}\mathbf{P}(x,t)\nonumber\\&
=1+\bm{\mathcal{Y}}(t)\mathbf{P}(x,t)\equiv \bm{\mathcal{Y}}(t)\mathbf{P}(x,t)
\end{align}
In the limit as $\bm{\mathrm{Re}}(\bm{\mathfrak{D}},t)\rightarrow\bm{\mathrm{Re}}_{c}(\bm{\mathfrak{D}})$ one has
\begin{empheq}[right=\empheqrbrace]{align}
&U_{a}(x,t)\hookrightarrow U_{a}(x,t)\nonumber\\&
\frac{\partial}{\partial t}\hookrightarrow\frac{\partial}{\partial t}\nonumber\\&
\nabla_{b}\hookrightarrow \nabla_{b}\nonumber\\&
\mathbf{P}(x,t)\hookrightarrow\mathbf{P}(x,t)
\end{empheq}
\end{thm}
\begin{proof}
The transformed or 'boosted' NS equations are
\begin{align}
{\bm{\mathfrak{D}}}_{t}{\bm{\mathfrak{U}}}_{a}(x,t)-\nu \Delta{\bm{\mathfrak{U}}}_{a}(x,t)+
{\bm{\mathfrak{U}}}^{b}(x,t)\nabla_{b}{\bm{\mathfrak{U}}}_{a}(x,t)+\nabla_{a}
{\bm{\mathfrak{P}}}(x,t)=0
\end{align}
which upon applying (3.17) become
\begin{align}
\frac{\partial}{\partial t}&-\frac{\frac{\partial}{\partial t}1+{\bm{\bm{\alpha}}}{{\psi}}(\bm{\mathrm{Re}}(\bm{\mathfrak{D}}),t);\mathbb{I})^{2}{\mathsf{C}}
}{1+{\bm{\bm{\alpha}}}{{\psi}}(\bm{\mathrm{Re}}(\bm{\mathfrak{D}}),t);\mathbb{I})^{2}{\mathsf{C}}
}1+{\bm{\bm{\alpha}}}{{\psi}}(\bm{\mathrm{Re}}(\bm{\mathfrak{D}}),t);\mathbb{I})^{2}{\mathsf{C}}
U_{a}(x,t)\nonumber\\&
-\nu1+{\bm{\bm{\alpha}}}{{\psi}}(\bm{\mathrm{Re}}(\bm{\mathfrak{D}}),t);\mathbb{I})^{2}{\mathsf{C}}
(t)\Delta U_{a}(x,t)+1+{\bm{\bm{\alpha}}}{{\psi}}(\bm{\mathrm{Re}}(\bm{\mathfrak{D}}),t);\mathbb{I})^{2}{\mathsf{C}}
^{2}U^{b}(x,t)\nabla^{b}U_{b}(x,t)\nonumber\\&+1+{\bm{\bm{\alpha}}}{{\psi}}(\bm{\mathrm{Re}}(\bm{\mathfrak{D}}),t);{\mathbb{I}})^{2}{\mathsf{C}}
\nabla_{a}\mathbf{P}(x,t)\nonumber\\&
=\underbrace{\frac{\partial}{\partial t}1+{\bm{\bm{\alpha}}}{{\psi}}(\bm{\mathrm{Re}}(\bm{\mathfrak{D}}),t);\mathbb{I})^{2}{\mathsf{C}}
\big)U_{a}(x,t)}+(\frac{\partial}{\partial t}U_{a}(x,t))1+{\bm{\bm{\alpha}}}{{\psi}}(
\bm{\mathrm{Re}}(\bm{\mathfrak{D}}),t);\mathbb{I})^{2}{\mathsf{C}}
\nonumber\\&-\underbrace{\frac{\partial}{\partial t}1+{\bm{\bm{\alpha}}}{{\psi}}(\bm{\mathrm{Re}}(\bm{\mathfrak{D}}),t);\mathbb{I})^{2}{\mathsf{C}}
U_{a}(x,t)}-\nu1+{\bm{\bm{\alpha}}}{{\psi}}(\bm{\mathrm{Re}}(\bm{\mathfrak{D}}),t);\mathbb{I})^{2}{\mathsf{C}}
\Delta U_{a}(x,t)\nonumber\\&
+1+{\bm{\bm{\alpha}}}{{\psi}}(\bm{\mathrm{Re}}(\bm{\mathfrak{D}}),t);\mathbb{I})^{2}{\mathsf{C}}
^{2}U^{b}(x,t)\nabla_{b}U_{a}(x,t)+1+{\bm{\bm{\alpha}}}{{\psi}}(\bm{\mathrm{Re}}(\bm{\mathfrak{D}}),t);{\mathbb{I}})^{2}{\mathsf{C}}
\nabla_{a}\mathbf{P}(x,t)\nonumber\\&
=(\frac{\partial}{\partial t}U_{a}(x,t))1+{\bm{\bm{\alpha}}}{{\psi}}(\bm{\mathrm{Re}}(\bm{\mathfrak{D}}),t);\mathbb{I})^{2}{\mathsf{C}}
-\nu1+{\bm{\bm{\alpha}}}{{\psi}}(\bm{\mathrm{Re}}(\bm{\mathfrak{D}}),t);\mathbb{I})^{2}{\mathsf{C}}
\Delta U_{a}(x,t)\nonumber\\&
+1+{\bm{\bm{\alpha}}}{{\psi}}(\bm{\mathrm{Re}}(\bm{\mathfrak{D}}),t);\mathbb{I})^{2}{\mathsf{C}}
^{2}U^{b}(x,t)\nabla_{b}U_{a}(x,t)+1+{\bm{\bm{\alpha}}}{{\psi}}(\bm{\mathrm{Re}}(\bm{\mathfrak{D}}),t);\mathbb{I})^{2}{\mathsf{C}}
\nabla_{a}\mathbf{P}(x,t)
\end{align}
Dividing out by $1+{\bm{\bm{\alpha}}}{{\psi}}(\bm{\mathrm{Re}}(\bm{\mathfrak{D}}),t);\mathbb{I})^{2}{\mathsf{C}}
$ gives
\begin{align}
&\frac{\partial}{\partial t}U_{a}(x,t)-\nu \Delta U_{a}(x,t)+1+
{\bm{\bm{\alpha}}}{{\psi}}(\bm{\mathrm{Re}}(\bm{\mathfrak{D}}),t);\mathbb{I})^{2}{\mathsf{C}}
U^{b}(x,t)\nabla_{b}U_{a}(x,t)+\mathbf{P}(x,t)=0
\end{align}
which is exactly (3.3).
Equivalently
\begin{align}
\frac{\partial}{\partial t}&-\frac{\frac{\partial}{\partial t}1+\bm{\mathcal{Y}}(t)}
{1+\bm{\mathcal{Y}}(t))}1+\bm{\mathcal{Y}}(t))U_{a}(x,t)\nonumber\\&
-\nu1+\bm{\mathcal{Y}}(t)
\Delta U_{a}(x,t)+1+\bm{\mathcal{Y}}(t))^{2}U^{b}(x,t)\nabla^{b}U_{b}(x,t)\nonumber\\&
+1+\bm{\mathcal{Y}}(t)\nabla_{a}\mathbf{P}(x,t)\nonumber\\&
=\underbrace{\frac{\partial}{\partial t}1+\bm{\mathcal{Y}}(t)\big)U_{a}(x,t)}+(\frac{\partial}{\partial t}U_{a}(x,t))
1+\bm{\mathcal{Y}}(t)
\nonumber\\&-\underbrace{\frac{\partial}{\partial t}1+\bm{\mathcal{Y}}(t)
U_{a}(x,t)}-\nu1+\bm{\mathcal{Y}}(t)
\Delta U_{a}(x,t)\nonumber\\&
+1+\bm{\mathcal{Y}}(t)
^{2}U^{b}(x,t)\nabla_{b}U_{a}(x,t)+1+\bm{\mathcal{Y}}(t)
\nabla_{a}\mathbf{P}(x,t)\nonumber\\&
=(\frac{\partial}{\partial t}U_{a}(x,t))1+\bm{\mathcal{Y}}(t)
-\nu1+\bm{\mathcal{Y}}(t)
\Delta U_{a}(x,t)\nonumber\\&
+1+\bm{\mathcal{Y}}(t)
^{2}U^{b}(x,t)\nabla_{b}U_{a}(x,t)+1+\bm{\mathcal{Y}}(t)
\nabla_{a}\mathbf{P}(x,t)
\end{align}
Dividing out by $1+\bm{\mathcal{Y}}(t)$ gives
\begin{align}
&\frac{\partial}{\partial t}U_{a}(x,t)-\nu \Delta U_{a}(x,t)+1+\bm{\mathcal{Y}}(t)
U^{b}(x,t)\nabla_{b}U_{a}(x,t)+\mathbf{P}(x,t)=0
\end{align}
\end{proof}
\subsection{Stochastically averaged Navier-Stokes equations for a turbulent fluid flow arisingaleg from a steady state or laminar flow}
\begin{thm}
As before, let $\bm{\mathfrak{D}}\subset{\mathbb{R}}^{3}$ contain a fluid of viscosity $\nu$ with constant underlying laminar velocity
$U_{a}(x,t)=U_{a}=(0,0,U)$ throughout $\bm{\mathfrak{D}}$ when the Reynolds number is at or below the critical value $\bm{\mathrm{Re}}_{c}^{(1)}(\bm{\mathfrak{D}})$. The domain has volume $\mathrm{Vol}(\bm{\mathfrak{D}})\sim \mathrm{L}^{3}$. As before, the volume-averaged Reynolds number is homogenous or constant through out $\bm{\mathfrak{D}}$ so that $\bm{\mathrm{Re}}(\bm{\mathfrak{D}})=\|\bm{\mathfrak{U}}_{a}\|\mathrm{L}/\nu$. The turbulent fluid flow ${{\mathscr{U}}}_{a}(x,t)$ within $\bm{\mathfrak{D}}$ for any $\bm{\mathrm{Re}}(\bm{\mathfrak{D}})>\bm{\mathrm{Re}}_{c}^{(1)}(\bm{\mathfrak{D}})$ is again the random field
\begin{align}
&{{\mathscr{U}}}_{a}(x,t)=U_{a}+{\bm{\bm{\alpha}}}U_{a}
{{\psi}}(\bm{\mathrm{Re}}(\bm{\mathfrak{D}}),t);\mathbb{I}){{\mathscr{B}}}(x)\nonumber\\&
=U_{a}+{\bm{\bm{\alpha}}}U_{a}{{\psi}}(\bm{\mathrm{Re}}(\bm{\mathfrak{D}}),t);{\mathbb{I}})
{{\mathscr{B}}}(x)\equiv U_{a}+{\mathscr{Q}}_{a}(x,t)
\end{align}
so that for very small $\bm{\mathrm{Re}}(\bm{\mathfrak{D}},t)$ below the critical value
\begin{align}
{{\mathscr{U}}}_{a}(x,t)\sim U_{a}
\end{align}
where the regulated BF random field ${{\mathscr{B}}}(x)$ has the usual properties
\begin{align}
&\bm{\mathbb{E}}\langle{{\mathscr{B}}}(x)\rangle=0,~\bm{\mathbb{E}}\langle \nabla_{b}{{\mathscr{B}}}(x)
\rangle=0\nonumber\\&
\bm{\mathbb{E}}\langle{{\mathscr{B}}}(x){{\otimes}} \nabla_{b}{{\mathscr{B}}}(x)\rangle=0
\end{align}
Then the turbulent flow $\mathcal{{\mathscr{U}}}_{a}(x,t)$ is a solution of the stochastically averaged Navier-Stokes equations
\begin{align}
&\bm{\mathbb{E}}\langle{\bm{\mathrm{D}}}_{t}{{\mathscr{U}}}_{a}(x,t)-\nu \Delta{{\mathscr{U}}}_{a}(x,t)\rangle\\&\equiv
\bm{\mathbb{E}}\langle \frac{\partial}{\partial t}{{\mathscr{U}}}_{a}(x,t)-\nu \Delta{{\mathscr{U}}}_{a}(x,t)
+{{\mathscr{U}}}^{b}(x,t){{\otimes}}~\nabla_{b}{{\mathscr{U}}}_{a}(x,t)\rangle=0
\end{align}
\end{thm}
\begin{proof}\renewcommand{\qedsymbol}{}
The derivatives of the field ${{\mathscr{U}}}_{a}(x,t)$ are
\begin{align}
&\nabla_{b}{{\mathscr{U}}}_{a}(x,t)={\bm{\bm{\alpha}}}U_{a}{{\psi}}(\bm{\mathrm{Re}}(\bm{\mathfrak{D}}),t);\mathbb{I})
~\nabla_{b}{{\mathscr{B}}}(x,t)\\&\frac{\partial}{\partial t}{{\mathscr{U}}}_{a}(x,t)={\bm{\bm{\alpha}}}\frac{\partial}{\partial t}
{{\psi}}(\bm{\mathrm{Re}}(\bm{\mathfrak{D}},t);\mathbb{I})\nabla_{b}
{{\mathscr{B}}}(x,t)\nonumber\\&+{\bm{\bm{\alpha}}}{{\psi}}(\bm{\mathrm{Re}}(\bm{\mathfrak{D}}),t);\mathbb{I})
\frac{\partial}{\partial t}{{\mathscr{B}}}(x,t)\nonumber\\&={\bm{\bm{\alpha}}}U_{a}{{\psi}}
(\bm{\mathrm{Re}}(\bm{\mathfrak{D}},t);\mathbb{I}){{{\psi}}}
(\bm{\mathrm{Re}}(\bm{\mathfrak{D}}),t);\mathbb{I})
\frac{\partial}{\partial t}{{\mathscr{B}}}(x,t)\nonumber\\&
\Delta{{\mathscr{U}}}_{a}(x,t)={\bm{\bm{\alpha}}}U^{b}
{{\psi}}(\bm{\mathrm{Re}}(\bm{\mathfrak{D}},t);\mathbb{I})
\Delta{{\mathscr{B}}}(x,t)
\end{align}
since $\nabla_{a}{{\psi}}(\bm{\mathrm{Re}}(\bm{\mathfrak{D}}),t);\mathbb{I})=0$, and $\frac{\partial}{\partial t} {{\psi}}(\bm{\mathrm{Re}}(\bm{\mathfrak{D}}),t);\mathbb{I})=0$ for a steady state flow. Then the nonlinear convective
term is simply
\begin{align}
&{{\mathscr{U}}}^{b}(x,t){{\otimes}} \nabla_{b}{{\mathscr{U}}}_{a}(x,t)={\bm{\bm{\alpha}}}U^{b}U_{a}
{{\psi}}(\bm{\mathrm{Re}}(\bm{\mathfrak{D}}),t);\mathbb{I})\nabla_{b}{{\mathscr{B}}}(x,t)\nonumber\\&+ {\bm{\bm{\alpha}}}U_{a}U^{b}{{\psi}}(\bm{\mathrm{Re}}(\bm{\mathfrak{D}}),t);\mathbb{I})
~{{\mathscr{B}}}(x,t){{\otimes}}~\nabla_{b}{{\mathscr{B}}}(x,t)
\end{align}
with expectation
\begin{align}
&\bm{\mathbb{E}}\langle{{\mathscr{U}}}^{b}(x,t){{\otimes}}~\nabla_{b}{{\mathscr{U}}}_{a}(x,t)\rangle={\bm{\bm{\alpha}}}U^{b}
{{\psi}}(\bm{\mathrm{Re}}(\bm{\mathfrak{D}}),t);\mathbb{I})\rangle\bm{\mathbb{E}}\langle \nabla_{b}{{\mathscr{B}}}(x,t)
\nonumber\\&+{\bm{\bm{\alpha}}}U_{a}U^{b}|{{\psi}}(\bm{\mathrm{Re}}(\bm{\mathfrak{D}}),t);\mathbb{I})^{2}
\bm{\mathbb{E}}\langle{{\mathscr{B}}}(x,t)
{{\otimes}}{{\mathscr{B}}}(x,t)\rangle
\end{align}
The SA Navier-Stokes equations are then
\begin{align}
\bm{\mathbb{E}}\langle{\bm{\mathrm{D}}}_{m}{{\mathscr{U}}}_{a}(x,t)&-\nu \Delta{{\mathscr{U}}}_{a}(x,t)\rangle\nonumber\\&\equiv
\bm{\mathbb{E}}\langle \frac{\partial}{\partial t}{{\mathscr{U}}}_{a}(x,t)-\nu \Delta{{\mathscr{U}}}_{a}(x,t)
+{{\mathscr{U}}}^{b}(x,t){{\otimes}}~\nabla_{b}{{\mathscr{U}}}_{a}(x,t)\rangle\nonumber\\&
={\bm{\bm{\alpha}}}U^{b}{{\psi}}(\bm{\mathrm{Re}}(\bm{\mathfrak{D}}),t);\mathbb{I})
\bm{\mathbb{E}}\langle
\frac{\partial}{\partial t}{{\mathscr{B}}}(x,t)\rangle\nonumber\\&-\nu {\bm{\bm{\alpha}}}U^{b}[{{\psi}}(\bm{\mathrm{Re}}(\bm{\mathfrak{D}}),t);\mathbb{I})]
\bm{\mathbb{E}}\langle \Delta{{\mathscr{B}}}(x,t)\rangle\nonumber\\&+U^{b}{\bm{\bm{\alpha}}}
[{{\psi}}(\bm{\mathrm{Re}}(\bm{\mathfrak{D}},t);\mathbb{I})]
\bm{\mathbb{E}}\langle \nabla_{b}{{\mathscr{B}}}(x,t)\rangle\nonumber\\&+U_{a}U^{b}{\bm{\bm{\alpha}}}
{{\psi}}(\bm{\mathrm{Re}}(\bm{\mathfrak{D}},t);\mathbb{I})^{2}
\bm{\mathbb{E}}\langle{{\mathscr{B}}}(x,t){{\otimes}}~\nabla_{b}{{\mathscr{B}}}(x,t)\rangle\\&
=U^{b}{\bm{\bm{\alpha}}}{{\psi}}(\bm{\mathrm{Re}}(\bm{\mathfrak{D}}),t);\mathbb{I})\bm{\mathbb{E}}\langle \frac{\partial}{\partial t}
{{\mathscr{B}}}(x,t)\rangle\nonumber\\&-\nu U^{b}{\bm{\bm{\alpha}}}{{\psi}}(\bm{\mathrm{Re}}(\bm{\mathfrak{D}}),t);\mathbb{I})
\bm{\mathbb{E}}\langle \Delta{{\mathscr{B}}}(x,t)\rangle\\&+{\bm{\bm{\alpha}}}U^{b}U_{a}
{{\psi}}(\bm{\mathrm{Re}}(\bm{\mathfrak{D}},t);\mathbb{I})\bm{\mathbb{E}}\langle \nabla_{b}{{\mathscr{B}}}(x,t)
\rangle\nonumber\\&+{\bm{\bm{\alpha}}}U_{a}U^{b}{{\psi}}(\bm{\mathrm{Re}}(\bm{\mathfrak{D}}),t);\mathbb{I})^{2}
\bm{\mathbb{E}}\langle{{\mathscr{B}}}(x,t){{\otimes}} \nabla_{b}{{\mathscr{B}}}(x,t)\rangle=0\nonumber
\end{align}
\end{proof}
\section{Binary and triple velocity correlations}
The next theorems deal with estimates for the binary, triple and high-order velocity correlations for the turbulent fluid flow ${{\mathscr{U}}}_{a}(x,t)$.
\begin{thm}
As before, let $\bm{\mathfrak{D}}\subset\bm{\mathbb{R}}^{3}$, with $\mathrm{Vol}(\bm{\mathfrak{D}})\sim L^{3}$  be the region where the random flow occurs so that for all $x\in\bm{\mathfrak{D}}$ and $t\in[0,\infty)$ the turbulent flow is the random field ${{\mathscr{U}}}_{a}(x,t)=U_{a}(x,t)+{\bm{\bm{\alpha}}}U_{a}(x,t)
{{{\psi}}}(\bm{\mathrm{Re}}(\bm{\mathfrak{D}},t);{\mathbb{I}})
{{\mathscr{B}}}(x,t)$, where $\bm{\mathrm{Re}}(x)\equiv\bm{\mathrm{Re}}(\bm{\mathfrak{D}})=\|\bm{\mathfrak{U}}_{a}(x,t)\|\mathrm{L}/\nu$  is the ensemble or spatially averaged Reynolds number in $\bm{\mathfrak{D}}$. The Bargman-Fock random field ${{\mathscr{B}}}(x)$ has the usual Gaussian distribution and correlation properties so that for any $(x,y,\mathbf{z})\in\bm{\mathfrak{D}}$ the expectation vanishes $\bm{\mathbb{E}}\big\langle{{\mathscr{B}}}(x,t)\big\rangle=
\bm{\mathbb{E}}\big\langle{{\mathscr{B}}}_{a}(x,t)\big\rangle=\bm{\mathbb{E}}\big\langle{{\mathscr{B}}}(x,t)\big\rangle=0$ and the following binary correlations hold with Gaussian decay and correlation length scale $\lambda\le L$.
\begin{align}
&\big\|\!\big\|{{\mathscr{B}}}(x,t){{\otimes}}~{{\mathscr{B}}}(y,t)\big\|\!\big\|_{SE_{1}(\bm{\mathfrak{D}}))}
={\bm{\mathbb{E}}}\bm{\langle}{{\mathscr{B}}}(x,t){{\otimes}}~{{\mathscr{B}}}(y,t)\rangle\nonumber\\&
={\mathsf{\bm{\Xi}}}(x,y;\lambda)=\mathsf{C}\exp(-\|x-y\|\lambda^{-2})\\&
\big\|\!\big\|{{\mathscr{B}}}(x,t){{\otimes}}~{{\mathscr{B}}}(\mathbf{z},t)\big\|\!\big\|_{SE_{1}(\bm{\mathfrak{D}}))}
={\bm{\mathbb{E}}}\langle{{\mathscr{B}}}(x,t){{\otimes}}~{{\mathscr{B}}}(z,t)\rangle\nonumber\\&
={\mathsf{\bm{\Xi}}}(x,\mathbf{z};\lambda)={\mathsf{C}}\exp(-\|x-\mathbf{z}\|^{2}\lambda^{-2})\\&
\big\|\!\big\|{{\mathscr{B}}}(y,t){{\otimes}}~{{\mathscr{B}}}(\mathbf{z},t)\big\|\!\big\|_{SE_{1}(\bm{\mathfrak{D}}))}=
{\bm{\mathbb{E}}}\langle{{\mathscr{B}}}(y,t){{\otimes}}~{{\mathscr{B}}}(\mathbf{z},t)\rangle\nonumber\\&
={\mathsf{\bm{\Xi}}}(y,\mathbf{z};\lambda)={\mathsf{C}}\exp(-\|y-\mathbf{z}\|^{2}\lambda^{-2})
\end{align}
For $(x,y,\mathbf{z})\in\bm{\mathfrak{D}}$ then the binary and triple velocity correlations are defined as the Reynolds tensors.
\begin{align}
&{\bm{\mathsf{T}}}_{ab}(x,t)=\big\|\!\big\|{{\mathscr{U}}}_{a}(x,t){{\otimes}}
~{{\mathscr{U}}}_{b}(y,t)\big\|\!\big\|_{SE_{1}(\bm{\mathfrak{D}})}
\equiv{\bm{\mathbb{E}}}\langle{{\mathscr{U}}}_{a}(x,t){{\otimes}}{{\mathscr{U}}}_{b}(x,t)\rangle\\&
{\bm{\mathsf{T}}}_{abc}(x,y,\mathbf{z})=\big\|\!\big\|{{\mathscr{U}}}_{a}(x,t){{\otimes}}~{{\mathscr{U}}}_{b}(y,t)
{{\otimes}}~{{\mathscr{U}}}_{c}(x,t)\big\|\!\big\|_{SE_{1}(\bm{\mathfrak{D}})}\nonumber\\&\equiv {\bm{\mathbb{E}}}\langle{{\mathscr{U}}}_{a}(x,t){{\otimes}}~{{\mathscr{U}}}_{b}(x,y,t)
{{\otimes}}~{{\mathscr{U}}}_{c}(\mathbf{z},t)\rangle
\end{align}
with the 2nd and 3rd-order moments being the tensors
\begin{align}
&{\bm{\mathsf{T}}}_{ab}(x)=\big\|\!\big\|{{\mathscr{U}}}_{a}(x,t){{\otimes}}~{{\mathscr{U}}}_{b}(x,t)\big\|\!
\big\|_{SE_{1}(\bm{\mathfrak{D}})}
=\lim_{y\rightarrow x}\bm{\mathbb{E}}\langle{{\mathscr{U}}}_{a}(x,t){{\otimes}}~{{\mathscr{U}}}_{b}(x,t)\rangle\nonumber\\&=\lim_{y\hookrightarrow x}{\bm{\mathsf{T}}}_{ab}(x,y)=\lim_{y\rightarrow x}\bm{\mathbb{E}}\langle{{\mathscr{U}}}_{a}(x,t){{\otimes}}~{{\mathscr{U}}}_{b}(y,t)\rangle\\&
{\bm{\mathsf{T}}}_{abc}(x)=\|\!\|{{\mathscr{U}}}_{a}(x,t){{\otimes}}~{{\mathscr{U}}}_{b}(x,t){{\otimes}}~
{{\mathscr{U}}}_{c}(x,t)\|\!\|_{SE_{1}(\bm{\mathfrak{D}})}\nonumber\\&=\lim_{y\rightarrow x}={\bm{\mathbb{E}}}\langle{{\mathscr{U}}}_{a}(x,t){{\otimes}}~{{\mathscr{U}}}_{b}(x,t){{\otimes}}~
{{\mathscr{U}}}_{c}(x,t)\rangle\nonumber\\&=\lim_{(z,y)\rightarrow x}\bm{\mathsf{T}}_{abc}(x,y)=\lim_{(z,y)\rightarrow x}{\bm{\mathbb{E}}}\langle{{\mathscr{U}}}_{a}(x,t){{{\otimes}}}~{{\mathscr{U}}}_{b}(y,t){{\otimes}}~
{{\mathscr{U}}}_{c}(x,t)\rangle
\end{align}
\begin{enumerate}[(a)]
\item \textbf{Then~the~binary~velocity~correlations~are:}
\begin{align}
&{\bm{\mathsf{T}}}_{ab}(x,t)_{ab}(x,y)
=\bm{\mathbb{E}}\langle{{\mathscr{U}}}_{a}(x,t){{\otimes}}~{{\mathscr{U}}}_{b}(y,t)\rangle\nonumber\\&
=U_{a}(x,t)U_{b}(y,t)\left(1+{{\psi}}(\bm{\mathrm{Re}}(\bm{\mathfrak{D}},t);{\mathbb{I}})^{2}
\mathsf{C}\exp(-\|x-y\|^{2}\lambda^{-2})\right)
\end{align}
and
\begin{align}
&{\bm{\mathsf{T}}}_{ab}(x,t)=\bm{\mathbb{E}}\langle{{\mathscr{U}}}_{a}(x,t){{\otimes}}~{{\mathscr{U}}}_{b}(x,t)\rangle=
\lim_{y\uparrow x}{\bm{\mathbb{E}}}\langle{{\mathscr{U}}}_{a}(x,t){{\otimes}}~{{\mathscr{U}}}_{b}(y,t)\rangle\nonumber\\&
=\lim_{y\uparrow x}U_{a}(x,t)U_{b}(x,t)\left(1+{\bm{\bm{\alpha}}}^{2}{{\psi}}(\bm{\mathrm{Re}}(\bm{\mathfrak{D}},t);{\mathbb{I}})^{2}
{\mathsf{C}}\exp(-\|x-y\|^{2}\lambda^{-2})\right)\nonumber\\&
=\lim_{y\uparrow x}U_{a}(x,t)U_{b}(x,t)\left(1+{\bm{\bm{\alpha}}}^{2}{{\psi}}(\bm{\mathrm{Re}}(\bm{\mathfrak{D}},t);{\mathbb{I}})^{2}
\mathsf{C}\right)
\end{align}
\item \textbf{The~triple~correlations~are}
\begin{align}
{\bm{\mathsf{T}}}_{abc}(x,y,\mathbf{z})&={\bm{\mathbb{E}}}\langle{{\mathscr{U}}}_{a}(x,t)
{{\otimes}}~{{\mathscr{U}}}_{b}(y,t)
{{\otimes}}~{{\mathscr{U}}}_{c}(z,t)\rangle=U_{a}(x,t)U_{b}(y,t)U_{c}(\mathbf{z},t)\nonumber\\&+
U_{a}(x,t)U_{b}(y,t)U_{c}(\mathbf{z},t)\times{\bm{\bm{\alpha}}}^{3}
{{\psi}}(\bm{\mathrm{Re}}(\bm{\mathfrak{D}},t);{\mathbb{I}})^{3}
\mathsf{C}\exp(-\|y-\mathbf{z}\|^{2}\lambda^{-2})
\nonumber\\&
U_{a}(x,t)U_{b}(y,t)U_{c}(\mathbf{z},t)\times{\bm{\bm{\alpha}}}^{3}
{{\psi}}(\bm{\mathrm{Re}}(\bm{\mathfrak{D}},t);{\mathbb{I}})^{3}
\mathsf{C}\exp(-\|x-\mathbf{z}\|^{2}\lambda^{-2})
\nonumber\\&\times
U_{a}(x,t)U_{b}(y,t)U_{c}(\mathbf{z},t)\times{\bm{\bm{\alpha}}}^{3}
{{\psi}}(\bm{\mathrm{Re}}(\bm{\mathfrak{D}},t);{\mathbb{I}})^{3}
\mathsf{C}\exp(-\|x-y\|^{2}\lambda^{-2})
\end{align}
\item
For $y\rightarrow x$ and $\mathbf{z}\rightarrow x$
\begin{align}
{\bm{\mathsf{T}}}_{ab}(x)&={\bm{\mathbb{E}}}\langle{{\mathscr{U}}}_{a}(x,t){{\otimes}}
~{{\mathscr{U}}}_{b}(x,t){{\otimes}}
{{\mathscr{U}}}_{c}(x,t)\rangle=
U_{a}(x,t)U_{b}(x,t)U_{c}(x,t)\nonumber\\&+
U_{a}(x,t)U_{b}(x,t)U_{c}(x,t)\times\times{\bm{\bm{\alpha}}}^{3}{{\psi}}(\bm{\mathrm{Re}}
(\bm{\mathfrak{D}},t);{\mathbb{I}})^{3}
\nonumber\\&
U_{a}(x,t)U_{b}(x,t)U_{c}(x,t)\times{\bm{\bm{\alpha}}}^{3}
{{\psi}}(\bm{\mathrm{Re}}(\bm{\mathfrak{D}},t);{\mathbb{I}})^{3}
\nonumber\\&\times
U_{a}(x,t)U_{b}(x,t)U_{c}(x,t)\times{\bm{\bm{\alpha}}}^{3}
{{\psi}}(\bm{\mathrm{Re}}(\bm{\mathfrak{D}},t);{\mathbb{I}})^{3}\nonumber\\&=
U_{a}(x,t)U_{b}(x,t)U_{c}(x,t)
1+3{\bm{\bm{\alpha}}}^{3}{{\psi}}(\bm{\mathrm{Re}}(\bm{\mathfrak{D}},t);{\mathbb{I}})^{3}
\mathsf{C}
\end{align}
\end{enumerate}
\end{thm}
\begin{proof}
The turbulent flows at $x$ and $y$ are the random fields
\begin{align}
&{{\mathscr{U}}}_{a}(x,t)=U_{a}(x,t)+\times{\bm{\bm{\alpha}}}
U_{a}(x,t){{\psi}}(\bm{\mathrm{Re}}(\bm{\mathfrak{D}},t);{\mathbb{I}})
{{\mathscr{B}}}(x,t))\nonumber\\&{{\mathscr{U}}}_{a}(y,t)
=U_{a}(y,t)+\times{\bm{\bm{\alpha}}}U_{a}(y,t)
{{\psi}}(\bm{\mathrm{Re}}(\bm{\mathfrak{D}},t);{\mathbb{I}}){{\mathscr{B}}}(y,t))
\end{align}
The product of the fields is then
\begin{align}
&{{\mathscr{U}}}_{a}(x,t){{\otimes}}~{{\mathscr{U}}}_{b}(y,t)=U_{a}(x,t)U_{b}(y,t)\nonumber\\&+
{\bm{\bm{\alpha}}}U_{a}(x,t)U_{b}(y,t){{\psi}}(\bm{\mathrm{Re}}(\bm{\mathfrak{D}},t);{\mathbb{I}})
{{\mathscr{B}}}(x,t)\nonumber\\&+{\bm{\bm{\alpha}}}U_{a}(x,t)U_{b}(y,t)
{{\psi}}(\bm{\mathrm{Re}}(\bm{\mathfrak{D}},t);{\mathbb{I}}){{\mathscr{B}}}(y,t)\nonumber\\&
+{\bm{\bm{\alpha}}}^{2}U_{a}(x,t)U_{b}(y,t){{\psi}}(\bm{\mathrm{Re}}(\bm{\mathfrak{D}},t);{\mathbb{I}})^{2}
{{\mathscr{B}}}(x,t){{\otimes}}~{{\mathscr{B}}}(y,t)
\end{align}
Taking the stochastic expectation then give a Reynolds-type tensor, with underbraced terms vanishing
\begin{align}
{\bm{\mathsf{T}}}_{ab}(x,y)&
={\bm{\mathbb{E}}}\langle{{\mathscr{U}}}_{a}(x,t){{\otimes}}~{{\mathscr{U}}}_{b}(y,t)\rangle
=U_{a}(x,t)U_{b}(y,t)\nonumber\\&+\underbrace{{\bm{\bm{\alpha}}}U_{a}(x,t)
U_{b}(y,t){{\psi}}(\bm{\mathrm{Re}}(\bm{\mathfrak{D}},t);{\mathbb{I}})
{\bm{\mathbb{E}}}\langle{{\mathscr{B}}}(x,t)\rangle}\nonumber\\&
+\underbrace{{\bm{\bm{\alpha}}}U_{a}(x,t)U_{b}(y,t)
{{\psi}}(\bm{\mathrm{Re}}(\bm{\mathfrak{D}},t);{\mathbb{I}})
{\bm{\mathbb{E}}}\langle{{\mathscr{B}}}(y,t)}\rangle\nonumber\\&
+{\bm{\bm{\alpha}}}^{3}U_{a}(x,t)U_{b}(y,t)
{{\psi}}(\bm{\mathrm{Re}}(\bm{\mathfrak{D}},t);{\mathbb{I}})^{2}{\bm{\mathbb{E}}}\langle{{\mathscr{B}}}(x,t)
{{\otimes}}~{{\mathscr{B}}}(y,t)\rangle\nonumber\\&
=U_{a}(x,t)U_{b}(y,t)1+|{\bm{\bm{\alpha}}}|^{2}
{{\psi}}(\bm{\mathrm{Re}}(\bm{\mathfrak{D}},t);{\mathbb{I}})^{2}{\mathsf{C}}\exp(-\|x-y\|{\xi}^{-2})\nonumber\\&
=U_{a}(x,t)U_{b}(y,t)1+|{\bm{\bm{\alpha}}}|^{2}
[{{\psi}}(\bm{\mathrm{Re}}(\bm{\mathfrak{D}},t);{\mathbb{I}})]^{2}{\mathsf{\bm{\Xi}}}(x,y;\lambda)
\end{align}
Then
\begin{align}
&{\bm{\mathsf{T}}}_{ab}(x)=\lim_{x\hookrightarrow y}
{\bm{\mathbb{E}}}\langle{{\mathscr{U}}}_{a}(x,t){{\otimes}}~
{{\mathscr{U}}}_{b}(y,t)\rangle
={\bm{\mathbb{E}}}\langle{{\mathscr{U}}}_{a}(x,t){{\otimes}}
~{{\mathscr{U}}}_{b}(x,t)\rangle\nonumber\\&
=U_{a}(x,t)U_{b}(y,t)1+{\bm{\bm{\alpha}}}^{2}
{{\psi}}(\bm{\mathrm{Re}}(\bm{\mathfrak{D}},t);\mathbb{I})^{2}{\mathsf{C}}
\end{align}
The averaged NS equations can then be expressed in terms of this stress tensor.
\begin{lem}
The averaged NS equations are
\begin{align}
&{\bm{\mathbb{E}}}\langle{{\mathscr{U}}}_{a}(x,t)-\nu \Delta{ {\mathscr{U}}}_{a}(x,t)
+{{\mathscr{U}}}^{b}(x,t){{\otimes}}~\nabla_{b}{{\mathscr{U}}}_{a}(x,t)+\nabla_{b}\mathbf{P}(x,t)\rangle\nonumber\\&
=\frac{\partial}{\partial t}U_{a}(x,t)-\nu \Delta U_{a}(x,t)+U^{b}(x,t)\nabla_{b}U_{a}(x,t)+\nabla_{a}\mathbf{P}(x,t)\nonumber\\& +{\bm{\bm{\alpha}}}^{2}U^{b}(x,t)\nabla_{b}
U_{a}(x,t){{\psi}}(\bm{\mathrm{Re}}(\bm{\mathfrak{D}},t);{\mathbb{I}})^{2}
{\mathsf{C}}\nonumber\\&
=\frac{\partial}{\partial t}U_{a}(x,t)-\nu \Delta U_{a}(x,t)+\nabla_{a}\mathbf{P}(x,t)\nonumber\\&
+U^{b}(x,t)\nabla_{b}U_{a}(x,t)(1+{\bm{\bm{\alpha}}}^{2}
{{\psi}}(\bm{\mathrm{Re}}(\bm{\mathfrak{D}},t);{\mathbb{I}})^{2}{\mathsf{C}}
\end{align}
Using (4.15) this is then
\begin{align}
&\bm{\mathbb{E}}\langle{{\mathscr{U}}}_{a}(x,t)-\nu \Delta{ {\mathscr{U}}}_{a}(x,t)
+{{\mathscr{U}}}^{b}(x,t){{\otimes}}~\nabla_{b}{{\mathscr{U}}}_{a}(x,t)+\nabla_{b}
{\mathbf{P}}(x,t)\rangle\nonumber\\&
=\frac{\partial}{\partial t}U_{a}(x,t)-\nu \Delta U_{a}(x,t)+\nabla_{a}\mathbf{P}(x,t)+\nabla^{b}
{\bm{{\mathsf{R}}}}_{ab}(x)
\end{align}
In terms of the random vector field $\mathscr{Q}_{a}(x,t)$, the turbulent flow is
\begin{align}
{{\mathscr{U}}}_{a}(x,t)=U_{a}(x,t)+{\bm{\bm{\alpha}}}U_{a}(x,t)
{{\psi}}(\bm{\mathrm{Re}}(\bm{\mathfrak{D}},t);{\mathbb{I}}){{\mathscr{B}}}(x)\equiv U_{a}(x,t)+{\mathscr{Q}}_{a}(x,t)
\end{align}
If we write $\bm{\mathsf{T}}_{ab}(x,x)=\bm{\mathbb{E}}\langle {\mathscr{Q}}_{a}(x,t){{\otimes}} {\mathscr{Q}}_{b}(x,t)\rangle$ then
\begin{align}
&{\bm{\mathbb{E}}}\langle{{\mathscr{U}}}_{a}(x,t)-\nu \Delta{{\mathscr{U}}}_{a}(x,t)
+{{\mathscr{U}}}^{b}(x,t){{\otimes}}~\nabla_{b}{{\mathscr{U}}}_{a}(x,t)+\nabla_{b}{\mathbf{P}}(x,t)\rangle\nonumber\\&
=\frac{\partial}{\partial t}U_{a}(x,t)-\nu \Delta U_{a}(x,t)+\nabla_{a}\mathbf{P}(x,t)+U^{b}(x,t)\nabla_{a}U_{b}(x,t)+
\nabla^{b}\bm{\mathsf{T}}_{ab}(x,x)
\end{align}
Hence, the extra term which is induced can be expressed as a Reynolds-type stress tensor.
\end{lem}
\subsection{Triple velocity correlations}
For the triple velocity correlations for any $(x,y,\mathbf{z})\in\bm{\mathfrak{D}}$, the random/turbulent flows are the random fields
\begin{align}
&{{\mathscr{U}}}_{a}(x,t)=U_{a}(x,t)+{\bm{\bm{\alpha}}}U_{a}(x,t)
{{\psi}}(\bm{\mathrm{Re}}(\bm{\mathfrak{D}},t);{\mathbb{I}})
{{\mathscr{B}}}(x,t))\\&
{{\mathscr{U}}}_{a}(y,t)=U_{a}(y,t)+{\bm{\bm{\alpha}}}U_{a}(y,t)
{{\psi}}(\bm{\mathrm{Re}}(\bm{\mathfrak{D}},t);\mathbb{I})
{{\mathscr{B}}}(y,t))\\&
{{\mathscr{U}}}_{a}(\mathbf{z},t)=U_{a}(\mathbf{z},t)+{\bm{\bm{\alpha}}}U_{a}(\mathbf{z},t)
{{\psi}}(\bm{\mathrm{Re}}(\bm{\mathfrak{D}},t);{\mathbb{I}})
{{\mathscr{B}}}(\mathbf{z},t))
\end{align}
The product of these random fields is then
\begin{align}
&{{\mathscr{U}}}_{a}(x,t){{\otimes}}~{{\mathscr{U}}}_{b}(y,t){{\otimes}}~{{\mathscr{U}}}_{c}(\mathbf{z},t)
=U_{a}(x,t)U_{b}(y,t)U_{c}(\mathbf{z},t)\nonumber\\&+{\bm{\bm{\alpha}}}^{3}U_{a}(x,t)
U_{b}(y,t)U_{c}(\mathbf{z},t){{\psi}}(\bm{\mathrm{Re}}(\bm{\mathfrak{D}},t);{\mathbb{I}})^{3}
{{\mathscr{B}}}(z,t)\nonumber\\&
+{\bm{\bm{\alpha}}}^{3}U_{a}(x,t)U_{b}(y,t)U_{c}(\mathbf{z},t)
{{\psi}}(\bm{\mathrm{Re}}(\bm{\mathfrak{D}},t);{\mathbb{I}})^{3}
{{\mathscr{B}}}(y,t)\nonumber\\&
+{\bm{\bm{\alpha}}}^{3}U_{a}(x,t)U_{b}(y,t)U_{c}(\mathbf{z},t)
{{\psi}}(\bm{\mathrm{Re}}(\bm{\mathfrak{D}},t);{\mathbb{I}})^{3}
{{\mathscr{B}}}(x,t)\nonumber\\&
+{\bm{\bm{\alpha}}}^{3}U_{a}(x,t)U_{b}(y,t)U_{c}(\mathbf{z},t)
{{\psi}}(\bm{\mathrm{Re}}(\bm{\mathfrak{D}},t);{\mathbb{I}})^{3}
{{\mathscr{B}}}(x,t){{\otimes}}~{{\mathscr{B}}}(\mathbf{z},t)\nonumber\\&
+{\bm{\bm{\alpha}}}^{3}U_{a}(x,t)U_{b}(y,t)U_{c}(\mathbf{z},t)
{{\psi}}(\bm{\mathrm{Re}}(\bm{\mathfrak{D}},t);{\mathbb{I}})^{3}
{{\mathscr{B}}}(x,t){{\otimes}}~{{\mathscr{B}}}(y,t)\nonumber\\&
+{\bm{\bm{\alpha}}}^{3}U_{a}(x,t)U_{b}(y,t)U_{c}(\mathbf{z},t)
{{\psi}}(\bm{\mathrm{Re}}(\bm{\mathfrak{D}},t);{\mathbb{I}})^{3}
{{\mathscr{B}}}(y,t){{\otimes}}{{\mathscr{B}}}(\mathbf{z},t)\nonumber\\&
+{\bm{\bm{\alpha}}}^{3}{{\psi}}\big(\bm{\mathrm{Re}}(\bm{\mathfrak{D}},t),{\mathbb{I}})^{3}
{{\mathscr{B}}}(x,t){{\otimes}}~{{\mathscr{B}}}(y,t){{\otimes}}~{{\mathscr{B}}}(\mathbf{z},t)
\end{align}
Taking the expectation or average, then all underbraced terms vanish giving the correlation tensor
\begin{align}
\bm{\mathsf{T}}_{ab}(x,&y,\mathbf{z})
=\bm{{\mathbb{E}}}\langle{{\mathscr{U}}}_{a}(x,t){{\otimes}}~{{\mathscr{U}}}_{b}(y,t)
{{\otimes}}~{{\mathscr{U}}}_{c}(z,t)\rangle
=U_{a}(x,t)U_{b}(y,t)U_{c}(\mathbf{z},t)\nonumber\\&+
\underbrace{\bm{\bm{\alpha}}^{3}U_{a}(x,t)U_{b}(y,t)U_{c}(\mathbf{z},t)
{{\psi}}\big(\bm{\mathrm{Re}}(\bm{\mathfrak{D}},t);{\mathbb{I}})^{3}
\bm{{\mathbb{E}}}\langle{{\mathscr{B}}}(z,t)}\rangle\nonumber\\&
+\underbrace{\bm{\bm{\alpha}}^{3}U_{a}(x,t)U_{b}(y,t)U_{c}(\mathbf{z},t)
+U_{a}(x,t)U_{b}(y,t)U_{c}(\mathbf{z},t)
{{\psi}}\big(\bm{\mathrm{Re}}(\bm{\mathfrak{D}},t);{\mathbb{I}})^{3}
{\bm{\mathbb{E}}}\langle{{\mathscr{B}}}(x,t)\rangle}\nonumber\\&
+\underbrace{\bm{\bm{\alpha}}^{3}U_{a}(x,t)U_{b}(y,t)U_{c}(\mathbf{z},t)+U_{a}(x,t)U_{b}(y,t)
U_{c}(\mathbf{z},t){{\psi}}\big(\bm{\mathrm{Re}}(\bm{\mathfrak{D}},t),{\mathbb{I}})^{3}
{\bm{\mathbb{E}}}\langle{{\mathscr{B}}}(y,t)\rangle}~\nonumber\\&
+{\bm{\bm{\alpha}}}^{3}U_{a}(x,t)U_{b}(y,t)U_{c}(\mathbf{z},t)
]{{\psi}}\big(\bm{\mathrm{Re}}(\bm{\mathfrak{D}},t);{\mathbb{I}})^{3}
\bm{{\mathbb{E}}}\langle{{\mathscr{B}}}(x,t){{\otimes}}~{{\mathscr{B}}}(\mathbf{z},t)\rangle
\nonumber\\&
+{\bm{\bm{\alpha}}}^{3}U_{a}(x,t)U_{b}(y,t)U_{c}(\mathbf{z},t){{\psi}}\big(\bm{\mathrm{Re}}(\bm{\mathfrak{D}},t);{\mathbb{I}})^{3}
{\bm{\bm{\alpha}}}\bm{\mathbb{E}}\langle{{\mathscr{B}}}(x,t){{\otimes}}~{{\mathscr{B}}}(y,t)\rangle\nonumber\\&+|
{\bm{\bm{\alpha}}}^{3}U_{a}(x,t)U_{b}(y,t)U_{c}(\mathbf{z},t){{\psi}}\big(\bm{\mathrm{Re}}(\bm{\mathfrak{D}},t),
{\mathbb{I}})^{3}
\bm{\mathbb{E}}\langle{{\mathscr{B}}}(y,t){{\otimes}}~{{\mathscr{B}}}(\mathbf{z},t)\rangle\nonumber\\&
+\underbrace{\bm{\bm{\alpha}}^{3}U_{a}(x,t)U_{b}(y,t)U_{c}(\mathbf{z},t)
{{\psi}}\big(\bm{\mathrm{Re}}(\bm{\mathfrak{D}},t);{\mathbb{I}})^{3}
\bm{\mathbb{E}}\langle{{\mathscr{B}}}(x,t){{\otimes}}~{{\mathscr{B}}}(y,t){{\otimes}}
~{{\mathscr{B}}}(\mathbf{z},t)\rangle}\nonumber\\&
=U_{a}(x,t)U_{b}(y,t)U_{c}(\mathbf{z},t)
\nonumber\\&+{\bm{\bm{\alpha}}}^{3}U_{a}(x,t)U_{b}(y,t)U_{c}(\mathbf{z},t)
{{\psi}}\big(\bm{\mathrm{Re}}(\bm{\mathfrak{D}},t);{\mathbb{I}})^{3}{\mathsf{C}}\exp(-(|x-y|^{2}\lambda^{-2})
\nonumber\\&+\bm{\bm{\alpha}}^{3}U_{a}(x,t)U_{b}(y,t)U_{c}(\mathbf{z},t)|
{{\psi}}\big(\bm{\mathrm{Re}}(\bm{\mathfrak{D}},t); {\mathbb{I}})^{3}{\mathsf{C}}\exp(-(|x-\mathbf{z}|^{2}\lambda^{-2})\nonumber\\&
+\bm{\bm{\alpha}}^{3}U_{a}(x,t)U_{b}(y,t)U_{c}(\mathbf{z},t)
{{\psi}}\big(\bm{\mathrm{Re}}(\bm{\mathfrak{D}},t);{\mathbb{I}})^{3}
{\mathsf{C}}\exp(-(|y-\mathbf{z}|^{2}\lambda^{-2})\nonumber\\&
\equiv U_{a}(x,t)U_{b}(y,t)U_{c}(\mathbf{z},t)1+\bm{\bm{\alpha}}^{2}\big|
{{\psi}}\big(\bm{\mathrm{Re}}(\bm{\mathfrak{D}},t); {\mathbb{I}})^{3}\nonumber\\&\times{\mathsf{C}}
\big(\exp(-\|x-y\|^{2}\lambda^{-2})+\exp(-\|x-z\|^{2}\lambda^{-2}+\exp(-\|y-\mathbf{z}\|^{2}\lambda^{-2}\big)\nonumber\\&\equiv
{{\mathsf{R}}}_{abc}(x,y,\mathbf{z}1+\bm{\bm{\alpha}}^{2}{{\psi}}\big(\bm{\mathrm{Re}}(\bm{\mathfrak{D}},t);{\mathsf{C}})
^{3}\nonumber\\&\times\mathsf{C}\big(\exp(-\|x-y\|^{2}\lambda^{-2})+\exp(-\|x-\mathbf{z}\|^{2}\lambda^{-2}+\exp(-\|y-\mathbf{z}\|^{2}\lambda^{-2}\big)
\end{align}
\end{proof}
The turbulent contributions to the binary and triple velocity correlations then vanish at large separations $\|x-y\|\gg\lambda$ and/or low Reynolds numbers $\bm{\mathrm{Re}}(\bm{\mathfrak{D}},t))< \bm{\mathrm{Re}}_{c}(\bm{\mathfrak{D}})$, and the flow becomes smooth or laminar again.
\begin{cor}
The turbulent contributions to ${\bm{\mathsf{T}}}(x,y)$ and ${\bm{\mathsf{T}}}(x,y,\mathbf{z})$ vanish
\begin{enumerate}[(a)]
\item $ If \|x-y\|\gg\lambda$ then the tensors decay rapidly due to the Gaussian.
\item If $\bm{\mathrm{Re}}(\bm{\mathfrak{D}},t)<\bm{\mathrm{Re}}(\bm{\mathfrak{D}})$ for all $t>0$ or some $t=t'$ or some $t\in[t_{1},t_{2}]$.
\item If the viscosity $\mu$ tends to infinity or becomes very high.
\end{enumerate}
For $\|x-y\|\gg\lambda$ then $\exp(-\|x-y\|\lambda^{-2})=0$ the binary correlation becomes
\begin{align}
&\bm{\mathsf{T}}_{ab}(x,y)=U_{a}(x,t)U_{b}(y,t)(1+|{\bm{\bm{\alpha}}}|^{2}
{{\psi}}\big(\bm{\mathrm{Re}}(\bm{\mathfrak{D}},t);{\mathbb{I}})^{2}
{\bm{\Xi}}(x,y;\lambda)\nonumber\\&
=U_{a}(x,t)U_{b}(y,t)(1+|{\bm{\bm{\alpha}}}|^{2}
{{\psi}}\big(\bm{\mathrm{Re}}(\bm{\mathfrak{D}},t);{\mathbb{I}})^{2}
\mathsf{C}\exp(-|x-y\|^{2}\lambda^{-2})\nonumber\\&\longrightarrow U_{a}(x,t)U_{b}(y,t)U_{c}(x,t)\equiv{\bm{\mathrm{R}}}_{ab}(x,y)
\end{align}
For $\|x-y\|\gg\lambda$ but $\|x-\mathbf{z}\|,\|y-\mathbf{z}\|\le \lambda$ then the triple correlation becomes
\begin{align}
&\bm{\mathsf{T}}_{abc}(x,y,\mathbf{z})=U_{a}(x,t)U_{b}(y,t)U_{c}(\mathbf{z},t)1+
|{\bm{\bm{\alpha}}}|^{2}{{\psi}}\big(\bm{\mathrm{Re}}(\bm{\mathfrak{D}},t); {\mathbb{I}})^{3}\nonumber\\&
\times{\bm{\Xi}}(x,y;\lambda)+{\bm{\Xi}}(x,\mathbf{z};\lambda)+{\bm{\Xi}}(y,\mathbf{z};\lambda)\big)\nonumber\\&
=U_{a}(x,t)U_{b}(y,t)U_{c}(\mathbf{z},t)1+|{\bm{\bm{\alpha}}}|^{2}
{{\psi}}\big(\bm{\mathrm{Re}}(\bm{\mathfrak{D}},t);{\mathbb{I}})\big\rbrace^{3}\nonumber\\&
\times\mathsf{C}
\big(\exp(-\|x-y\|^{2}\lambda^{-2})+\exp(-\|x-\mathbf{z}\|^{2}\lambda^{-2}+\exp(-\|y-\mathbf{z}\|^{2}\lambda^{-2}\big)\nonumber\\&
\longrightarrow
U_{a}(x,t)U_{b}(y,t)U_{c}(\mathbf{z},t)1+|{\bm{\bm{\alpha}}}|^{2}
{{\psi}}\big(\bm{\mathrm{Re}}(\bm{\mathfrak{D}},t);{\mathbb{I}})^{3}\nonumber\\&
\times\mathsf{C}\exp(-\|x-\mathbf{z}\|^{2}\lambda^{-2}+\exp(-\|y-\mathbf{z}\|^{2}\lambda^{-2}\big)
\end{align}
and so on. If $\|x-y\|\gg\lambda$ and also $\|x-\mathbf{z}\|\gg\lambda$ and $\|y-\mathbf{z}\|\gg \lambda$ then
\begin{align}
&{\bm{\mathsf{T}}}_{abc}(\mathbf{x,y,z})=U_{a}(x,t)U_{b}(y,t)U_{c}(\mathbf{z},t)
1+|{\bm{\bm{\alpha}}}|^{3}{{\psi}}\big(\bm{\mathrm{Re}}(\bm{\mathfrak{D}},t); {\mathbb{I}})^{3}
\nonumber\\&\times{\mathsf{C}}\big(\exp(-\|x-y\|^{2}\lambda^{-2})+\exp(-\|x-\mathbf{z}\|^{2}
\lambda^{-2}+\exp(-\|y-\mathbf{z}\|^{2}\lambda^{-2}\big)\nonumber\\&
\longrightarrow U_{a}(x,t)U_{b}(y,t)U_{c}(\mathbf{z},t)
\end{align}
If $\bm{\mathrm{Re}}(\bm{\mathfrak{D}},t)<\mathrm{Re}_{c}(\bm{\mathfrak{D}})$ for any $t>0$ then the volume-averaged Reynolds number over $\bm{\mathfrak{D}}$ falls below the critical Reynolds value and so ${{\psi}}(\bm{\mathrm{Re}}(\bm{\mathfrak{D}},t),\bm{\mathrm{Re}}(\bm{\mathfrak{D}}))=0$. This indicates that the flow returns to being laminar or nonturbulent so that
\begin{align}
&\bm{\mathsf{T}}_{ab}(\mathbf{x,y})=U_{a}(x,t)U_{b}(y,t)1+{\bm{\bm{\alpha}}}^{2}{{\psi}}\big(\bm{\mathrm{Re}}(\bm{\mathfrak{D}},t);{\mathbb{I}})^{2}
\mathsf{C}\exp(-|x-y\|^{2}\lambda^{-2})\nonumber\\&=U_{a}(x,t)U_{b}(y,t)U_{c}(x,t)
\end{align}
and similarly for ${\bm{\mathsf{T}}}_{abc}(\mathbf{x,y,z})$ and all high-order correlations. Finally, since $\bm{\mathrm{Re}}(\bm{\mathfrak{D}},t)=\frac{\|\mathfrak{A}_{a}(x,t)\|L}{\mu}$ where $Vol[\bm{\mathfrak{D}}]\sim \mathrm{L}^{3}$ and $\|\bm{\mathfrak{U}}_{a}(x,t)$ is the volume-averaged velocity within $\bm{\mathfrak{D}}$, then $\bm{\mathrm{Re}}(\bm{\mathfrak{D}},t)<\bm{\mathrm{Re}}_{c}(\bm{\mathfrak{D}})$ for sufficiently large $\mu$. Also $\bm{\mathrm{Re}}(\bm{\mathfrak{D}},t)=\frac{\|\mathfrak{Q}_{a}(x,t)\|L}{\mu}\rightarrow 0$ as
$\mu\rightarrow\infty$. Hence
\begin{align}
&\lim_{\mu\rightarrow\infty}{\bm{\mathsf{T}}}_{ab}(\mathbf{x,y})=U_{a}(x)U_{b}(y)\\&
\lim_{\mu\rightarrow\infty}{\bm{\mathsf{T}}}_{abc}(\mathbf{x,y,z})=U_{a}(x)U_{b}(y)U_{c}(\mathbf{z})
\end{align}
\end{cor}
\begin{lem}
The time evolution of ${\bm{\mathsf{T}}}_{ab}(\mathbf{x,y};\lambda)$ is
\begin{align}
{\partial}_{t}{\bm{\mathsf{T}}}_{ab}(\mathbf{x,y};t)&=2U_{a}(x,t)U_{b}(y,t)|{{\psi}}\big(\bm{\mathrm{Re}}(\bm{\mathfrak{D}},t);{\mathbb{I}})
{\partial}_{t}{{\psi}}\big(\bm{\mathrm{Re}}(\bm{\mathfrak{D}},t); {\mathbb{I}}){\bm{\Xi}}(x,y;\lambda)\nonumber\\&
{\partial}_{t}U_{a}(x,t))U_{b}(y,t)+({\partial}_{t}U_{b}(x,t))U_{b}(y,t)\nonumber\\&\times
1+|{\bm{\bm{\alpha}}}|^{2}{{\psi}}\big(\bm{\mathrm{Re}}(\bm{\mathfrak{D}},t);{\mathbb{I}})^{3}
{\mathsf{\bm{\Xi}}}(\mathbf{x,y};\lambda)\nonumber\\&
\equiv +2U_{a}(x,t)U_{b}(y,t)|{{\psi}}\big(\bm{\mathrm{Re}}(\bm{\mathfrak{D}},t); {\mathbb{I}})\\&\nonumber
{\partial}_{t}{{\psi}}\big(\bm{\mathrm{Re}}(\bm{\mathfrak{D}},t);{\mathbb{I}}){\mathsf{C}}\exp(-|x-y\|^{2}\lambda^{-2})\nonumber\\&
+{\partial}_{t}U_{a}(x,t))U_{b}(y,t)+({\partial}_{t}U_{b}(x,t))U_{b}(y,t)\nonumber\\&\times
1+|{\bm{\bm{\alpha}}}|^{2}{{\psi}}\big(\bm{\mathrm{Re}}(\bm{\mathfrak{D}},t); {\mathbb{I}})^{3}
{\mathsf{C}}\exp(-|x-y\|^{2}\lambda^{-2})\nonumber\\&
\end{align}
so that ${\partial}_{t}{\bm{\mathsf{T}}}_{ab}(\mathbf{x,y};t)=0$ if $U_{a}(x,t)=U_{a}=consts$ and $\frac{\partial}{\partial t}{{\psi}}\big(\bm{\mathrm{Re}}(\bm{\mathfrak{D}},t);{\mathbb{I}})=0$
\end{lem}
The time evolution of all higher-order terms is also readily estimated.
\subsection{Nth-order moments and non-blowup of the turbulent flow}
It is necessary and important that the turbulent flow ${{\mathscr{U}}}_{a}(x,t)$ does not blow up for any $t>0$ and $x\in\bm{\mathfrak{D}}$. This means that
the N-th order moments $\bm{\mathbb{E}}\big\langle\big|{{\mathscr{U}}}_{a}(x,t)\big|^{N}\big\rangle$ are bounded and finite so that
\begin{align}
\bm{\mathbb{E}}\langle|{{\otimes}}{{\mathscr{U}}}_{a}(x,t)
|^{N}\rangle\equiv\|\!\|{{\otimes}}{{\mathscr{U}}}_{a}(x,t)\|\!\|_{SE_{N}(\bm{\mathfrak{D}}}<\infty
\end{align}
or
\begin{align}
\bm{\mathbb{E}}\langle|{{\otimes}}{{\mathscr{U}}}_{a}(x,t)|^{N}\rangle
\equiv\|\!\|{{\otimes}}{{\mathscr{U}}}_{a}(x,t)\|\!\|_{V_{r}(\mathcal{C}}<K
\end{align}
The $Nth$-order moments ${\bm{\mathfrak{E}}}
\langle|{{\mathscr{U}}}_{a}(x,t)|^{N}\rangle$ or can be estimated as follows
\begin{thm}
Given the turbulent flow ${{\mathscr{U}}}_{a}(x,t)$ as previously defined, the Nth-order moments are then
\begin{align}
&\bm{\mathbb{E}}\langle|{{\mathscr{U}}}_{a}(x,t)|^{N}\rangle\equiv \|\!\|{{\mathscr{U}}}_{a}(x,t)\|\!\|^{N}_{E_{N}(\bm{\mathfrak{D}}))}\nonumber\\&
=\sum_{M=1}^{N}\binom{N}{M}|U_{a}(x,t)|^{N-M}{\bm{\bm{\alpha}}}^{M}
[{{\psi}}(\bm{\mathrm{Re}}(\bm{\mathfrak{D}},t);{\mathfrak{F}})]^{M}\big[\tfrac{1}{2}({\mathsf{C}}+(-1)^{M}{\mathsf{C}})\big]
\end{align}
Then the moments of the turbulent fluid flow are bounded and do not blow up for all $x\in\bm{\mathfrak{D}}$ and $t>0$.
\begin{align}
\bm{\mathbb{E}}\langle|{{\mathscr{U}}}_{a}(x,t)|^{N}\rangle\equiv \|\!\|{{\mathscr{U}}}_{a}(x,t)\|\!\|^{N}_{E_{N}(\bm{\mathfrak{D}})}<\infty
\end{align}
for all $(x,t)\in\bm{\mathfrak{D}}[0\infty)$ iff $\|U_{a}(x,t\|^{N}<\infty$. Hence, the turbulent flow is finite and bounded provided that the underlying NS flow
$U_{a}(x,t)$ is finite and bounded.
\end{thm}
\begin{proof}
Using the binomial expansion
\begin{align}
&{\bm{\mathbb{E}}}\langle|{{\mathscr{U}}}_{a}(x,t)|^{p}\rangle\nonumber\\&
={\bm{\mathbb{E}}}\langle|U_{a}(x,t)+U_{a}(x,t){\bm{\bm{\alpha}}}
{{\psi}}(\bm{\mathrm{Re}}(\bm{\mathfrak{D}},t);{\mathbb{I}})
{{\mathscr{B}}}(x,t)|^{N}\rangle\nonumber\\&
={\bm{\mathbb{E}}}\langle\sum_{M=0}^{N}\binom{\mathrm{N}}{\mathrm{N}}\big|U_{a}(x,t)\big|^{N-M}| U_{a}(x,t)\mathbf{M}{{\psi}}(\bm{\mathrm{Re}}(\bm{\mathfrak{D}},t);{\mathbb{I}}){{\mathscr{B}}}(x,t)
|^{m}\nonumber\\&={\bm{\mathbb{E}}}\langle\sum_{M=0}^{N}\binom{\mathrm{N}}{\mathrm{M}}\big|U_{a}(x,t)\big|^{N-M}|\mathbf{M}|^{M} |U_{a}(x,t)|^{M}{{\psi}}(\bm{\mathrm{Re}}(\bm{\mathfrak{D}},t);{\mathbb{I}})^{M}
|{{\mathscr{B}}}(x,t)|^{M}|\rangle\nonumber\\&=\bm{\mathbb{E}}\langle\sum_{M=0}^{N}\binom{\mathrm{N}}{\mathrm{M}}\big|U_{a}(x,t)\big|^{N}
{\bm{\bm{\alpha}}}^{M}{{\psi}}(\bm{\mathrm{Re}}(\bm{\mathfrak{D}},t);{\mathbb{I}})^{M}
|{{\mathscr{B}}}(x,t)|^{M}\rangle\nonumber\\&\equiv\bm{\mathbb{E}}\langle\sum_{M=0}^{N}\binom{\mathrm{N}}{\mathrm{M}}\big|U_{a}(x,t)\big|^{N}
{\bm{\bm{\alpha}}}^{M}{{\psi}}(\bm{\mathrm{Re}}(\bm{\mathfrak{D}},t);{\mathbb{I}})^{M}
|{{\mathscr{B}}}(x,t)|^{M}\rangle\nonumber\\&\equiv\sum_{M=0}^{N}\binom{\mathrm{N}}{\mathrm{M}}\big|U_{a}(x,t)\big|^{N}|\mathbf{M}|^{M}
{{\psi}}(\bm{\mathrm{Re}}(\bm{\mathfrak{D}},t);{\mathbb{I}})^{M}
{\bm{\mathbb{E}}}\langle|{{\mathscr{B}}}(x,t)|^{M}\rangle\nonumber\\&\equiv\sum_{M=0}^{M}\binom{\mathrm{N}}{\mathrm{M}}
\big|U_{a}(x,t)\big|^{p}{\bm{\bm{\alpha}}}^{q}{{\psi}}(\bm{\mathrm{Re}}(\bm{\mathfrak{D}},t);{\mathbb{I}})^{M}
\|\!\|{{\mathscr{B}}}(x,t)\|\!\|^{M}_{E_{M}(\bm{\mathfrak{D}})}\nonumber\\&
=\sum_{M=1}^{N}\binom{\mathrm{N}}{\mathrm{M}}|U_{a}(x,t)|^{N}|\mathbf{M}|^{M}
{{\psi}}(\bm{\mathrm{Re}}(\bm{\mathfrak{D}},t);{\mathbb{I}})
\big[\tfrac{1}{2}({\mathsf{C}}+(-1)^{M}{\mathsf{C}})]
\end{align}
\end{proof}
\begin{cor}
As a corollary, the equation for the 2nd-order moments follow for $p=2$.
\begin{align}
&{\bm{\mathbb{E}}}\langle\big|{{\mathscr{U}}}_{a}(x,t)|^{2}\rangle
=\sum_{q=1}^{2}\binom{2}{q}|U_{a}(x,t)|^{2}{\bm{\bm{\alpha}}}^{q}{{\psi}}(\bm{\mathrm{Re}}(\bm{\mathfrak{D}},t);{\mathbb{I}}) ^{q}\left[\frac{1}{2}({\mathsf{C}}+(-1)^{q}{\mathsf{C}})\right]\nonumber\\&
=\binom{2}{0}|U_{a}(x,t)|^{2}\mathsf{C}+\binom{2}{1}|U_{a}(x,t)|^{2}{\bm{\bm{\alpha}}}
\left({{\psi}}(\bm{\mathrm{Re}}(\bm{\mathfrak{D}},t);{\mathbb{I}})\right) {\bm{\mathbb{E}}}\langle|{{\mathscr{B}}}(x)|^{2}\rangle\nonumber\\&
+\binom{2}{2}|U_{a}(x,t)|^{2}{\bm{\bm{\alpha}}}^{2}[{{\psi}}(\bm{\mathrm{Re}}(\bm{\mathfrak{D}},t);{\mathbb{I}})]^{2}{\mathsf{C}}\nonumber\\&
=\binom{2}{0}|U_{a}(x,t)|^{2}+\binom{2}{2}|U_{a}(x,t)|^{2}{\bm{\bm{\alpha}}}^{2}
[{{\psi}}(\bm{\mathrm{Re}}(\bm{\mathfrak{D}},t);{\mathbb{I}})]^{2}
\mathsf{C}\nonumber\\&=|U_{a}(x,t)|^{2}+|U_{a}(x,t)|^{2}{\bm{\bm{\alpha}}}^{2}
{{\psi}}(\bm{\mathrm{Re}}(\bm{\mathfrak{D}},t);{\mathbb{I}})^{2}\mathsf{C}
\nonumber\\&=|U_{a}(x,t)|^{2}\left(1+{\bm{\bm{\alpha}}}^{2}[{{\psi}}(\bm{\mathrm{Re}}(\bm{\mathfrak{D}},t);{\mathbb{I}})]^{2}
{\mathsf{C}}\right)
\end{align}
which is exactly (4.15).
\end{cor}
\section{Stochastically averaged Navier-Stokes equations for the turbulent flow ${\mathscr{U}}_{a}(x,t)$ II-incorporation of the pressure gradient term}
The averaged Navier-Stokes equations are considered again, but now incorporating the random fluctuations induced within the pressure gradient term $\nabla_{a}\mathbf{P}(x,t)$. The turbulent velocity ${\mathscr{U}}_{a}(x,t)$ will induce a 'turbulent pressure' $\mathscr{P}(x,t)$. We begin with the
NS equations, defined on $\bm{\mathfrak{D}}\subset\mathbb{R}^{3}$.
\begin{align}
{\partial}_{t}U_{a}(x,t)-\Delta U_{a}(x,t)
+U^{b}(x,t)\nabla_{b}U_{a}(x,t)+\nabla_{a}\mathrm{P}(x,t)=0,~~x\in\mathcal{H},t>0
\end{align}
with $\nabla_{b}U^{b}(\mathrm{x},t)=0$ and $\mathbf{P}(x,t)=0$ and $\nabla_{a}\mathbf{P}(x,t)=0$ on $\partial\mathcal{Q}$. Taking the divergence $\nabla^{a}$ gives
\begin{align}
&{\partial}_{t}\nabla^{a}U_{a}(x,t)-\nabla^{a}\Delta U_{a}(x,t)
+\nabla^{a}U^{b}(\mathrm{x},t)\nabla_{b}U_{a}(x,t)+\nabla^{a}\nabla_{a}\mathbf{P}(x,t)\nonumber\\&
\Delta\mathbf{P}(x,t)+\nabla_{a}\nabla_{b}(U_{a}(x,t)U_{b}(x,t))=0
\end{align}
which is
\begin{align}
-\Delta\mathbf{P}(x,t)=\nabla_{a}\nabla_{b}(U_{a}(x,t)U_{b}(x,t))
\end{align}
Any two solutions of (5.3) can only differ by a harmonic function. Formally, the solution will be of the form
\begin{align}
\mathbf{P}(x,t)=(-\Delta)^{-1}\nabla_{a}\nabla_{b}(U_{a}(x,t)U_{b}(x,t))
\end{align}
\begin{lem}
The pressure gradient and Navier-Stokes equations can be expressed in the form
\begin{align}
&\nabla_{a}\mathbf{P}(x,t)=\frac{1}{4\pi}\int_{\bm{\mathfrak{D}}}
\frac{(x_{a}-y_{a})}{|x-y|^{2}}\nabla^{a}_{(y)}\nabla^{b}_{(y)}
{U}_{a}(y,t){U}_{b}(y,t)d^{3}y\nonumber\\&
\equiv\frac{1}{4\pi}\int_{\bm{\mathfrak{D}}}\frac{(x_{a}-y_{a})}{|x-y|^{2}}
\Delta_{y}U_{a}(y,t)U_{b}(y,t)\delta^{ab}d^{3}y\\&
\end{align}
The Navier-Stokes equations can then be written as
\begin{align}
&{\partial}_{t}U_{a}(x,t)-\Delta U_{a}(x,t)
+{U}^{b}(x,t)\nabla_{b}{U}_{a}(x,t)\nonumber\\&+\left|\frac{1}{4\pi}\int_{\mathbf{H}}
\frac{(x_{a}-y_{a})}{|x-y|^{2}}
\Delta_{y}U_{a}(y,t)U_{b}(y,t)\delta^{ab}d^{3}y\right|=0
\end{align}
\end{lem}
\begin{proof}
The formal solution of the Laplace equation for the pressure is [REF]
\begin{align}
\mathbf{P}(x,t)
=(-\Delta)^{-1}\nabla_{a}\nabla_{b}(U_{a}(x,t)U_{b}(x,t))
\end{align}
where $-\Delta)^{-1}$ is the integral operator.
\begin{align}
(-\Delta)^{-1}{f}(x,t)=\frac{1}{4\pi}\int_{\bm{\mathfrak{D}}}\frac{{f}(y,t)d^{3}y}{|x-y|}
\end{align}
Hence
\begin{align}
\mathbf{P}(x,t)=\frac{1}{4\pi}\int_{\bm{\mathfrak{D}}}
\frac{1}{|x-y|}\nabla^{a}_{(y)}\nabla^{b}_{(y)}
{U}_{a}(y,t){U}_{b}(y,t)d^{3}y
\end{align}
The gradient is then
\begin{align}
\nabla_{a}^{(x)}\mathbf{P}(x,t)&=\frac{1}{4\pi}\int_{\bm{\mathfrak{D}}}
\nabla_{a}^{x}\left(\frac{1}{|x-y|}\right))\nabla^{a}_{(y)}\nabla^{b}_{(y)}
{U}_{a}(y,t){U}_{b}(y,t)d^{3}y\nonumber\\&
=\frac{1}{4\pi}\int_{\bm{\mathfrak{D}}}\frac{(x_{a}-y_{a})}{|x-y|^{2}}\nabla^{a}_{(y)}\nabla^{b}_{(y)}
{U}_{a}(y,t){U}_{b}(y,t)d^{3}y\nonumber\\&
\equiv\frac{1}{4\pi}\int_{\bm{\mathfrak{D}}}
\frac{(x_{a}-y_{a})}{|x-y|^{2}}\Delta_{y}U_{a}(y,t)U_{b}(y,t)\delta^{ab}d^{3}y
\end{align}
and the form (5.7) of the N-S equations follows.
\end{proof}
Note that if $(-\Delta)^{-1}$ and $\nabla_{a}$ are considered in the Fourier space then they respectively correspond to multiplication by
$(2\pi|\mathbf{k}|)^{-2}$ and $2\pi i$ so that they respectively commute. Hence
\begin{align}
\mathbf{P}(x,t)=(-\Delta)^{-1}\nabla_{a}\nabla_{b}U_{a}(x,t)U_{b}(x,t)
\equiv \nabla_{a}\nabla_{b}(-\Delta)^{-1}U_{a}(x,t)U_{b}(x,t)
=\mathbf{T}_{ab}U_{a}(x,t)U_{b}(x,t)
\end{align}
On the whole space it can also be shown that (REF) the following bounds hold. For all $0<r<\infty$
\begin{align}
&\|\mathbf{P}(\bullet,t)\|_{L_{r}(\mathbb{R}^{3})}\le{\mathrm{C}}_{r}\|U_{a}(\bullet,t)\|^{2}_{L_{2r}(\mathbb{R}^{3})}\\&
\|\nabla_{a}\mathbf{P}(\bullet,t)\|_{L_{r}(\mathbb{R}^{3})}\le
{\mathrm{C}}_{r}\|U^{b}(\bullet,t)\nabla_{a})U_{b}(\bullet,t)\|_{L_{r}(\mathbb{R}^{3})}
\end{align}
These are proved via the Calderon-Zygmond Theorem (REF).

Turbulence in the fluid velocity will induce turbulence or random fluctuations in the pressure or the gradient of the pressure via the nonlinear term
\begin{lem}
Given the integral representation of $\nabla_{a}\mathbf{P}(\mathrm{x},t)$, consider a turbulent flow at $y\in\bm{\mathfrak{D}}$ at time $t>0$, again of the form
\begin{align}
{{\mathscr{U}}}_{a}(y,t)=U_{a}(\mathrm{y},t)+{\bm{\bm{\alpha}}}U_{a}(\mathrm{y},t)
{{\psi}}(\mathbf{Re}(\bm{\mathfrak{D}},t),\mathbf{Re}_{c}(\bm{\mathfrak{D}})){{\mathscr{B}}}(y)
\end{align}
Then the turbulent velocity induces a turbulent pressure gradient $\nabla_{a}\mathscr{P}(x,t)$ with expectation
\begin{align}
{\mathbb{E}}\langle \nabla_{a}\mathscr{P}(x,t)\rangle = \nabla_{a}\mathscr{P}(x,t)\big(1+{\bm{\bm{\alpha}}}^{2}
|{{\psi}}(\mathbf{Re}(\bm{\mathfrak{D}},t),\mathbf{Re}_{c}(\bm{\mathfrak{D}}))|^{2}\mathsf{C}\big)
\end{align}
\end{lem}
\begin{proof}
The pressure gradient is
\begin{align}
\nabla_{a}^{(x)}\mathbf{P}(x,t)&=\frac{1}{4\pi}\int_{\bm{\mathfrak{D}}}
\nabla_{a}^{x}\left(\frac{1}{|x-y|}\right))\nabla^{a}_{(y)}\nabla^{b}_{(y)}
{U}_{a}(y,t){U}_{b}(y,t)d^{3}y\nonumber\\&=\frac{1}{4\pi}\int_{\bm{\mathfrak{D}}}\frac{(x_{a}
-y_{a})}{|x-y|^{2}}\nabla^{a}_{(y)}\nabla^{b}_{(y)}{U}_{a}(y,t){U}_{b}(y,t)d^{3}y\nonumber
\end{align}
Hence if $U_{a}(x,t)\rightarrow{\mathscr{U}}_{a}(x,t)$
\begin{align}
\nabla_{a}^{(x)}{\mathscr{P}}(x,t)=\frac{1}{4\pi}\int_{\bm{\mathfrak{D}}}\frac{(x_{a}-y_{a})}{|x-y|^{2}}\nabla^{a}_{(y)}\nabla^{b}_{(y)}
{{\mathscr{U}}}_{a}(y,t){{\otimes}}{{\mathscr{U}}}_{b}(y,t)d^{3}y
\end{align}
Taking the stochastic expectation
\begin{align}
&{\bm{\mathbb{E}}}\langle \nabla_{a}^{(x)}{\mathscr{P}}(x,t)\rangle={\bm{\mathbb{E}}}\langle\frac{1}{4\pi}\int_{\bm{\mathfrak{D}}}\frac{(x_{a}-y_{a})}{|x-y|^{2}}\nabla^{a}_{(y)}\nabla^{b}_{(y)}
{{\mathscr{U}}}_{a}(y,t){{\otimes}}{{\mathscr{U}}}_{b}(y,t)d^{3}y\rangle\nonumber\\&
\equiv=\frac{1}{4\pi}\int_{\bm{\mathfrak{D}}}\frac{(x_{a}-y_{a})}{|x-y|^{2}}{\bm{\mathbb{E}}}\langle \nabla^{a}_{(y)}\nabla^{b}_{(y)}
{{\mathscr{U}}}_{a}(y,t){{\otimes}}{{\mathscr{U}}}_{b}(y,t)\rangle d^{3}y\nonumber\\&
\equiv=\frac{1}{4\pi}\int_{\bm{\mathfrak{D}}}\frac{(x_{a}-y_{a})}{|x-y|^{2}}\nabla^{a}_{(y)}\nabla^{b}_{(y)}
{\bm{\mathbb{E}}}\langle{{\mathscr{U}}}_{a}(y,t){{\otimes}}{{\mathscr{U}}}_{b}(y,t)\rangle d^{3}y\nonumber\\&
\equiv\frac{1}{4\pi}\int_{\bm{\mathfrak{D}}}\frac{(x_{a}-y_{a})}{|x-y|^{2}}\nabla^{a}_{(y)}\nabla^{b}_{(y)}
{\bm{\mathsf{T}}}_{ab}(y,t)d^{3}y\nonumber\\&
=\frac{1}{4\pi}\int_{\bm{\mathfrak{D}}}\frac{(x_{a}-y_{a})}{|x-y|^{2}}\nabla^{a}_{(y)}
\nabla^{b}_{(y)}
\times U_{a}(y,t)U_{b}(y,t)+U_{a}(y,t)U_{b}(y,t)
{\bm{\bm{\alpha}}}^{2}{{\psi}}(\mathbf{Re}(\bm{\mathfrak{D}},t);\mathbb{I})^{2}\mathsf{C}d^{3}y\nonumber\\&
=\frac{1}{4\pi}\int_{\bm{\mathfrak{D}}}\frac{(x_{a}-y_{a})}{|x-y|^{2}}\nabla_{b}\big(\nabla^{b}U_{a}(y,t)\big)
U_{b}(y,t)\delta^{ab}d^{3}y\nonumber\\&+\frac{1}{4\pi}\int_{\bm{\mathfrak{D}}}\frac{(x_{a}-y_{a})}{|x-y|^{2}}(\nabla_{b}\big(\nabla^{b}
U_{a}(y,t))U_{b}(y,t)
+U_{a}(y,t)\underbrace{\nabla^{b}U_{b}(y,t)}\big)|{\bm{\bm{\alpha}}}^{2}|
{{\psi}}(\mathbf{Re}(\bm{\mathfrak{D}},t);\mathbb{I}^{2}\mathsf{C}\delta^{ab}d^{3}y
\nonumber\\&
=\frac{1}{4\pi}\int_{\bm{\mathfrak{D}}}\frac{(x_{a}-y_{a})}{|x-y|^{2}}\nabla_{b}\big(\nabla^{b}U_{a}(y,t)\big)U_{b}(y,t)\delta^{ab}
d^{3}y\nonumber\\&
+\frac{1}{4\pi}\int_{\bm{\mathfrak{D}}}\frac{(x_{a}-y_{a})}{|x-y|^{2}}\nabla_{b}\big(\nabla^{b}U_{a}(y,t)\big)U_{b}(y,t)
\big){\bm{\bm{\alpha}}}^{2}|{{\psi}}(\bm{\mathrm{Re}}(x,t),\mathbb{I})^{2}\mathsf{C}d^{3}y
\nonumber\\&
=\frac{1}{4\pi}\int_{\bm{\mathfrak{D}}}\frac{(x_{a}-y_{a})}{|x-y|^{2}}\nabla_{b}\big(\nabla^{b}U_{a}(y,t)\big)U_{b}(y,t)
\delta^{ab}d^{3}y\nonumber\\&
+\frac{1}{4\pi}{\bm{\bm{\alpha}}}^{2}{{\psi}}(\bm{\mathrm{Re}}(x,t);\mathbb{I})^{2}\mathsf{C}
\int_{\bm{\mathfrak{D}}}\frac{(x_{a}-y_{a})}{|x-y|^{2}}\nabla_{b}\big(\nabla^{b}U_{a}(y,t)\big)U_{b}(y,t)
\delta^{ab}d^{3}y\nonumber\\&
=\frac{1}{4\pi}\int_{\bm{\mathfrak{D}}}\frac{(x_{a}-y_{a})}{|x-y|^{2}}\Delta\big(U_{a}(y,t)\big)U_{b}(y,t)\big)
\delta^{ab}d^{3}y\nonumber\\&
+\frac{1}{4\pi}{\bm{\bm{\alpha}}}^{2}{{\psi}}(\bm{\mathrm{Re}}(x,t);\mathbb{I})^{2}\mathsf{C}
\int_{\bm{\mathfrak{D}}}\frac{(x_{a}-y_{a})}{|x-y|^{2}}\Delta\big(U_{a}(y,t)\big)U_{b}(y,t)\big)d^{3}y\nonumber\\&
=\nabla_{a}\mathscr{P}(x,t)1+{\bm{\bm{\alpha}}}^{2}
{{\psi}}(\bm{\mathrm{Re}}(x,t);\mathbb{I})^{2}\mathsf{C}
\end{align}
\end{proof}
\begin{thm}\textbf{Stochastically~averaged~Navier-Stokes~equations with~induced~pressure~fluctuations}.
Let $U_{a}(x,t)$ satisfy the Navier-Stokes equations within $\bm{\mathfrak{D}}\subset\mathbb{R}^{3}$ so that
\begin{align}
&\frac{\partial}{\partial t}U_{a}(x,t)-\nu \Delta U_{a}(x,t)+U^{b}(x,t)\nabla_{b}U_{a}(x,t)+\nabla_{a}\mathbf{P}(x,t)\nonumber\\&
\equiv\frac{\partial}{\partial t}U_{a}(x,t)-\nu \Delta U_{a}(x,t)+U^{b}(x,t)\nabla_{b}U_{a}(x,t)+
\left|\frac{1}{4\pi}\int_{\bm{\mathfrak{D}}}
\frac{(x_{a}-y_{a})}{|x-y|^{2}}\Delta_{y}U_{a}(y,t)U_{b}(y,t)\delta^{ab}d^{3}y\right|
\end{align}
and as before $\nabla^{b}U_{b}(x,t)=0$ for an incompressible fluid. For $\mathbf{Re}(\bm{\mathfrak{D}},t)>\mathbf{Re}_{c}(\bm{\mathfrak{D}})$ then the flow becomes turbulent so that as before the turbulent flow is the random field
\begin{align}
{{\mathscr{U}}}_{a}(x,t)=U_{a}(x,t)+{\bm{\bm{\alpha}}}U_{a}(x,t)
{{\psi}}(\bm{\mathrm{Re}}(\bm{\mathfrak{D}}),t);{\mathbb{I}}){{\mathscr{B}}}(x)
\end{align}
The turbulent flow then satisfies the following stochastically averaged Navier-Stokes equation
\begin{align}
&{\bm{\mathbb{E}}}\langle \frac{\partial}{\partial t}{{\mathscr{U}}}_{a}(x,t)-\nu \Delta{{\mathscr{U}}}_{a}(x,t)
+{{\mathscr{U}}}^{b}(x,t){{\otimes}}~\nabla_{b}{{\mathscr{U}}}_{a}(x,t)+D{\mathscr{P}}(x,t)\rangle\nonumber\\&
\equiv {\mathds{E}}\langle \frac{\partial}{\partial t}{\mathscr{U}}_{a}(x,t)-\nu \Delta{\mathscr{U}}_{a}(x,t)+{\mathscr{U}}^{b}(x,t)\nabla_{b}{{\mathscr{U}}}_{a}(x,t)\nonumber\\&+
\left|\frac{1}{4\pi}\int_{\bm{\mathfrak{D}}}
\frac{(x_{a}-y_{a})}{|x-y|^{2}}
\Delta_{y}{{\mathscr{U}}}_{a}(y,t){{\otimes}}{{\mathscr{U}}}_{b}(y,t)\delta^{ab}d^{3}y\right|\rangle\nonumber\\&
=\frac{\partial}{\partial t}U_{a}(x,t)-\nu \Delta U_{a}(x,t)+\nabla_{a}\mathbf{P}(x,t)1+{\bm{\bm{\alpha}}}^{2}
{{\psi}}(\bm{\mathrm{Re}}(\bm{\mathfrak{D}}),t);{\mathbb{I}})^{2}{\mathsf{C}}\nonumber\\&
+U^{b}(x,t)\nabla_{b}U_{a}(x,t)1+{\bm{\bm{\alpha}}}^{2}
{{\psi}}(\bm{\mathrm{Re}}(\bm{\mathfrak{D}}),t);{\mathbb{I}})^{2}{\mathsf{C}}\nonumber\\&
=\frac{\partial}{\partial t}U_{a}(x,t)-\nu \Delta U_{a}(x,t)+
\bm{\mathcal{Y}}(t)U^{b}(x,t)\nabla_{b}U_{a}(x,t)+\bm{\mathcal{Y}}(t)\nabla_{a}\mathbf{P}(x,t)
\end{align}
More succinctly, the turbulent flow ${\mathscr{U}}_{a}(x,t)$ is a solution of the stochastically averaged N-S equations where the nonlinear and pressure gradient terms are now modified by the factor $\bm{\mathsf{{\psi}}}_{2}$
\begin{align}
&{\bm{\mathbb{E}}}\langle \frac{\partial}{\partial t}{{\mathscr{U}}}_{a}(x,t)-\nu \Delta{{\mathscr{U}}}_{a}(x,t)
+{{\mathscr{U}}}^{b}(x,t){{\otimes}}~\nabla_{b}{{\mathscr{U}}}_{a}(x,t)+D{\mathscr{P}}(x,t)\rangle\nonumber\\&
=\frac{\partial}{\partial t}U_{a}(x,t)-\nu \Delta U_{a}(x,t)+\bm{\mathcal{Y}}(t)U^{b}(x,t)\nabla_{b}
U_{a}(x,t)+\bm{\mathcal{Y}}(t)\nabla_{a}\mathbf{P}(x,t)
\end{align}
\end{thm}
\begin{proof}
Substituting ${\mathscr{U}}_{a}(x,t)$ into the Navier-Stokes equations and taking the derivatives gives
\begin{align}
\frac{\partial}{\partial t}{\mathscr{U}}_{a}(x,t)&-\nu \Delta {\mathscr{U}}_{a}(x,t)
+{\mathscr{U}}^{b}(x,t)\nabla_{b}{\mathscr{U}}_{a}(x,t)+\nabla_{b}\bm{\mathrm{P}}(x,t)\nonumber\\&
=\overbracket{\frac{\partial}{\partial t}U_{a}(x,t)}+{\bm{\bm{\alpha}}}\frac{\partial}{\partial t}U_{a}(x,t){{\psi}}(\bm{\mathrm{Re}}(\bm{\mathfrak{D}}),t);{\mathbb{I}})
{{\mathscr{B}}}(x)\nonumber\\&+U_{a}(x)
{\bm{\bm{\alpha}}}{{\psi}}(\bm{\mathrm{Re}}(\bm{\mathfrak{D}}),t);{\mathbb{I}})\frac{\partial}{\partial t}{{\mathscr{B}}}(x)\nonumber\\&
-\overbracket{\nu\delta^{ab}\nabla_{a}\nabla_{b}U_{a}(x,t)}-\nu\delta^{ab}{\bm{\bm{\alpha}}}\nabla_{a}\nabla_{b}U_{a}(x,t)
{{\psi}}(\bm{\mathrm{Re}}(\bm{\mathfrak{D}}),t);{\mathbb{I}}){{\mathscr{B}}}(x)\nonumber\\&
-\nu\delta^{ab}{\bm{\bm{\alpha}}}\nabla_{b}U_{a}(x,t){{\psi}}(\bm{\mathrm{Re}}(\bm{\mathfrak{D}}),t);{\mathbb{I}})
\nabla_{a}{{\mathscr{B}}}(x,t)\nonumber\\&
-\nu\delta^{ab}{\bm{\bm{\alpha}}}\times\underbrace{\nabla_{a}U_{a}(x,,t)}_{=0}\times{{\psi}}(\bm{\mathrm{Re}}(\bm{\mathfrak{D}}),t);{\mathbb{I}})
\nabla_{b}{{\mathscr{B}}}(x)\nonumber\\&
-\nu\delta^{ab}{\bm{\bm{\alpha}}}U_{a}(x,t){{\psi}}(\bm{\mathrm{Re}}(\bm{\mathfrak{D}}),t);{\mathbb{I}})\nabla_{a}\nabla^{b}
{{\mathscr{B}}}(x)\nonumber\\&
+\overbracket{U^{b}\nabla_{b}U_{a}(x,t)}+U^{b}(x,t)\nabla_{b}U^{a}(x,t){\bm{\bm{\alpha}}}
{{\psi}}(\bm{\mathrm{Re}}(\bm{\mathfrak{D}}),t);{\mathbb{I}}){{\mathscr{B}}}(x)\nonumber\\&+
+{\bm{\bm{\alpha}}}U^{b}(x,t)\nabla_{b}U_{a}(x,t){{\psi}}(\bm{\mathrm{Re}}(\bm{\mathfrak{D}}),t);{\mathbb{I}})
{{\mathscr{B}}}(x)\nonumber\\&+U^{b}(x,t)|{\bm{\bm{\alpha}}}|^{2}\nabla_{b}U_{a}(x,t)
|{{\psi}}\big(\bm{\mathrm{Re}}(\bm{\mathfrak{D}},t);{\mathbb{I}})|^{2}|{{\mathscr{B}}}(x){{\otimes}}{{\mathscr{B}}}(x)\nonumber\\&
+|{\bm{\bm{\alpha}}}|^{2}U^{b}(x,t)U_{a}(x,t){{\psi}}(\bm{\mathrm{Re}}(\bm{\mathfrak{D}}),t);{\mathbb{I}})^{2}\nabla_{b}{{\mathscr{B}}}(x)
{{\otimes}}{{\mathscr{B}}}(x)\nonumber\\&+\frac{1}{4\pi}\int_{\bm{\mathfrak{D}}}\frac{(x_{a}-y_{a})}{|x-y|^{2}}\nabla^{a}_{(y)}\nabla^{b}_{(y)}
{\mathscr{U}}_{a}(y,t){{\otimes}}{\mathscr{U}}_{b}(y,t)d^{3}y
\end{align}
The pieces of the underlying deterministic Navier-Stokes equations are emphasised with an overbracket and the extra terms represent the random field contributions from turbulence. Now taking the stochastic average ${\bm{\mathbb{E}}}\big\langle\bullet\big\rangle$,or equivalently the stochastic norm, all underbraced (linear) terms vanish upon using (3.1)
\begin{align}
&{\bm{\mathbb{E}}}\langle \frac{\partial}{\partial t}{\mathscr{U}}_{a}(x,t)-\nu \Delta {\mathscr{U}}_{a}(x,t)
+{{\mathscr{U}}}^{b}(x,t)\nabla_{b}{{\mathscr{U}}}_{a}(x,t)+\nabla_{b}\bm{\mathrm{P}}(x,t)\rangle\nonumber\\&
=\overbracket{\frac{\partial}{\partial t}U_{a}(x,t)}+\underbrace{{\bm{\bm{\alpha}}}\frac{\partial}{\partial t}U_{a}(x,t){{\psi}}(\bm{\mathrm{Re}}(\bm{\mathfrak{D}}),t);{\mathbb{I}})^{2}
{\bm{\mathbb{E}}}\langle{{\mathscr{B}}}(x)\rangle}\nonumber\\&+\underbrace{U_{a}(x)
{\bm{\bm{\alpha}}}{{\psi}}(\bm{\mathrm{Re}}(\bm{\mathfrak{D}}),t);{\mathbb{I}}){\bm{\mathbb{E}}}
\langle \frac{\partial}{\partial t}{{\mathscr{B}}}(x)\rangle}\nonumber\\&
-\overbracket{\nu\delta^{ab}\nabla_{a}\nabla_{b}U_{a}(x,t)}-\underbrace{\nu\delta^{ab}{\bm{\bm{\alpha}}}\nabla_{a}\nabla_{b}U_{a}(x,t)
{{\psi}}(\bm{\mathrm{Re}}(\bm{\mathfrak{D}}),t);{\mathbb{I}}){\bm{\mathbb{E}}}\langle{{\mathscr{B}}}(x)\rangle}\nonumber\\&
-\underbrace{\nu\delta^{ab}{\bm{\bm{\alpha}}}\nabla_{b}U_{a}(x,t){{\psi}}(\bm{\mathrm{Re}}(\bm{\mathfrak{D}}),t);{\mathbb{I}})
{\bm{\mathbb{E}}}\langle \nabla_{a}{{\mathscr{B}}}(x,t)\rangle}\nonumber\\&
-\underbrace{\nu\delta^{ab}{\bm{\bm{\alpha}}}\times\underbrace{\nabla_{a}U_{a}(x,,t)}_{=0}\times{{\psi}}(\bm{\mathrm{Re}}(\bm{\mathfrak{D}}),t);{\mathbb{I}})
{\bm{\mathbb{E}}}\langle \nabla_{b}{{\mathscr{B}}}(x)\rangle}\nonumber\\&
-\underbrace{\nu\delta^{ab}{\bm{\bm{\alpha}}}U_{a}(x,t){{\psi}}(\bm{\mathrm{Re}}(\bm{\mathfrak{D}}),t);{\mathbb{I}})
{\bm{\mathbb{E}}}\langle \nabla_{a}\nabla^{b}{{\mathscr{B}}}(x)\rangle}\nonumber\\&
+\overbracket{U^{b}\nabla_{b}U_{a}(x,t)}+\underbrace{U^{b}(x,t)\nabla_{b}U^{a}(x,t){\bm{\bm{\alpha}}}
{{\psi}}(\bm{\mathrm{Re}}(\bm{\mathfrak{D}}),t);{\mathbb{I}}){\bm{\mathbb{E}}}\langle{{\mathscr{B}}}(x)\rangle}\nonumber\\&
+\underbrace{{\bm{\bm{\alpha}}}U^{b}(x,t)\nabla_{b}U_{a}(x,t){{\psi}}(\bm{\mathrm{Re}}(\bm{\mathfrak{D}}),t);{\mathbb{I}})
{\bm{\mathbb{E}}}\langle{{\mathscr{B}}}(x)\rangle}\nonumber\\&+{\bm{\bm{\alpha}}}^{2}U^{b}(x,t)\nabla_{b}U_{a}(x,t)
{{\psi}}\big(\bm{\mathrm{Re}}(\bm{\mathfrak{D}},t);{\mathbb{I}})^{2}{\bm{\mathbb{E}}}
\langle{{\mathscr{B}}}(x){{\otimes}}{{\mathscr{B}}}(x)\rangle\nonumber\\&
+\underbrace{{\bm{\bm{\alpha}}}^{2}U^{b}(x,t)U_{a}(x,t){{\psi}}(\bm{\mathrm{Re}}(\bm{\mathfrak{D}}),t);{\mathbb{I}})^{2}
{\bm{\mathbb{E}}}\langle \nabla_{b}{{\mathscr{B}}}(x)
{{\otimes}}{{\mathscr{B}}}(x)\rangle}\nonumber\\&+\frac{1}{4\pi}\int_{\bm{\mathfrak{D}}}\frac{(x_{a}-y_{a})}{|x-y|^{2}}\nabla^{a}_{(y)}\nabla^{b}_{(y)}
{\bm{\mathbb{E}}}\langle{{\mathscr{U}}}_{a}(y,t){{\otimes}}{{\mathscr{U}}}_{b}(y,t)\rangle d^{3}y
\end{align}
Taking only the expectation of the last term first and using (5.18)
\begin{align}
&{\bm{\mathbb{E}}}\langle \frac{\partial}{\partial t}{{\mathscr{U}}}_{a}(x,t)-\nu \Delta {{\mathscr{U}}}_{a}(x,t)
+{{\mathscr{U}}}^{b}(x,t)\nabla_{b}{{\mathscr{U}}}_{a}(x,t)+\nabla_{b}{\mathscr{P}}(x,t)\rangle\nonumber\\&
=\overbracket{\frac{\partial}{\partial t}U_{a}(x,t)}+\underbrace{{\bm{\bm{\alpha}}}\frac{\partial}{\partial t}U_{a}(x,t){{\psi}}(\bm{\mathrm{Re}}(\bm{\mathfrak{D}}),t);{\mathbb{I}})^{2}
{\bm{\mathbb{E}}}\langle{{\mathscr{B}}}(x)\rangle}\nonumber\\&+\underbrace{U_{a}(x)
{\bm{\bm{\alpha}}}{{\psi}}(\bm{\mathrm{Re}}(\bm{\mathfrak{D}}),t);{\mathbb{I}}){\bm{\mathbb{E}}}
\langle \frac{\partial}{\partial t}{{\mathscr{B}}}(x)\rangle}\nonumber\\&
-\overbracket{\nu\delta^{ab}\nabla_{a}\nabla_{b}U_{a}(x,t)}-\underbrace{\nu\delta^{ab}{\bm{\bm{\alpha}}}\nabla_{a}\nabla_{b}U_{a}(x,t)
{{\psi}}(\bm{\mathrm{Re}}(\bm{\mathfrak{D}}),t);{\mathbb{I}}){\bm{\mathbb{E}}}\langle{{\mathscr{B}}}(x)\rangle}\nonumber\\&
-\underbrace{\nu\delta^{ab}{\bm{\bm{\alpha}}}\nabla_{b}U_{a}(x,t){{\psi}}(\bm{\mathrm{Re}}(\bm{\mathfrak{D}}),t);{\mathbb{I}})
{\bm{\mathbb{E}}}\langle \nabla_{a}{{\mathscr{B}}}(x,t)\rangle}\nonumber\\&
-\underbrace{\nu\delta^{ab}{\bm{\bm{\alpha}}}\times\underbrace{\nabla_{a}U_{a}(x,,t)}_{=0}\times{{\psi}}(\bm{\mathrm{Re}}(\bm{\mathfrak{D}}),t);{\mathbb{I}})
{\bm{\mathbb{E}}}\langle \nabla_{b}{{\mathscr{B}}}(x)\rangle}\nonumber\\&
-\underbrace{\nu\delta^{ab}{\bm{\bm{\alpha}}}U_{a}(x,t){{\psi}}(\bm{\mathrm{Re}}(\bm{\mathfrak{D}}),t);{\mathbb{I}})
{\bm{\mathbb{E}}}\langle \nabla_{a}\nabla^{b}{{\mathscr{B}}}(x)\rangle}\nonumber\\&
+\overbracket{U^{b}\nabla_{b}U_{a}(x,t)}+\underbrace{U^{b}(x,t)\nabla_{b}U^{a}(x,t){\bm{\bm{\alpha}}}
{{\psi}}(\bm{\mathrm{Re}}(\bm{\mathfrak{D}}),t);{\mathbb{I}}){\bm{\mathbb{E}}}\langle{{\mathscr{B}}}(x)\rangle}\nonumber\\&
+\underbrace{{\bm{\bm{\alpha}}}U^{b}(x,t)\nabla_{b}U_{a}(x,t){{\psi}}(\bm{\mathrm{Re}}(\bm{\mathfrak{D}}),t);{\mathbb{I}})
{\bm{\mathbb{E}}}\langle{{\mathscr{B}}}(x)\rangle}\nonumber\\&+{\bm{\bm{\alpha}}}^{2}U^{b}(x,t)\nabla_{b}U_{a}(x,t)
{{\psi}}\big(\bm{\mathrm{Re}}(\bm{\mathfrak{D}},t);{\mathbb{I}})^{2}{\bm{\mathbb{E}}}
\langle{{\mathscr{B}}}(x){{\otimes}}{{\mathscr{B}}}(x)\rangle\nonumber\\&
+\underbrace{{\bm{\bm{\alpha}}}^{2}U^{b}(x,t)U_{a}(x,t){{\psi}}(\bm{\mathrm{Re}}(\bm{\mathfrak{D}}),t);{\mathbb{I}})^{2}
{\bm{\mathbb{E}}}\langle \nabla_{b}{{\mathscr{B}}}(x)
{{\otimes}}{{\mathscr{B}}}(x)\rangle}\nonumber\\&
+\frac{1}{4\pi}\int_{\bm{\mathfrak{D}}}\frac{(x_{a}-y_{a})}{|x-y|^{2}}\Delta\big(U_{a}(y,t)\big)U_{b}(y,t)\big)
\delta^{ab}d^{3}y\nonumber\\&+\frac{1}{4\pi}{\bm{\bm{\alpha}}}^{2}{{\psi}}(\bm{\mathrm{Re}}(x,t);\mathbb{I})^{2}\mathsf{C}
\int_{\bm{\mathfrak{D}}}\frac{(x_{a}-y_{a})}{|x-y|^{2}}\Delta\big(U_{a}(y,t)\big)U_{b}(y,t)\big)d^{3}y
\end{align}
Taking the remaining stochastic expectations leaves only one more additional term involving $\bm{\mathbb{E}}\langle{\mathscr{B}}(x){{\otimes}}{\mathscr{B}}(x)\rangle=\mathsf{C}$, with all other (linear) terms vanishing so that
\begin{align}
{\bm{\mathbb{E}}}\langle{{\mathscr{U}}}_{a}(x,t)&-\nu \Delta{ {\mathscr{U}}}_{a}(x,t)
+{{\mathscr{U}}}^{b}(x,t){{\otimes}}~\nabla_{b}{{\mathscr{U}}}_{a}(x,t)+\nabla_{a}\mathscr{P}(x,t)\rangle\nonumber\\&
=\overbracket{\frac{\partial}{\partial t}U_{a}(x,t)}-\overbracket{\nu \Delta U_{a}(x,t)}+\overbracket{U^{b}\nabla_{b}U_{a}(x,t)}
\nonumber\\& +{\bm{\bm{\alpha}}}^{2}U^{b}(x,t)\nabla_{b}
U_{a}(x,t){{\psi}}(\bm{\mathrm{Re}}(\bm{\mathfrak{D}}),t);{\mathbb{I}})^{2}
{\bm{\mathbb{E}}}\langle{{\mathscr{B}}}(x,t)
{{\otimes}}~{{\mathscr{B}}}(x,t)\rangle
\nonumber\\&
+\frac{1}{4\pi}\int_{\bm{\mathfrak{D}}}\frac{(x_{a}-y_{a})}{|x-y|^{2}}\Delta\big(U_{a}(y,t)\big)U_{b}(y,t)\big)
\delta^{ab}d^{3}y\nonumber\\&+\frac{1}{4\pi}{\bm{\bm{\alpha}}}^{2}{{\psi}}(\bm{\mathrm{Re}}(x,t);\mathbb{I})^{2}\mathsf{C}
\int_{\bm{\mathfrak{D}}}\frac{(x_{a}-y_{a})}{|x-y|^{2}}\Delta\big(U_{a}(y,t)\big)U_{b}(y,t)\big)d^{3}y\nonumber\\&
=\overbracket{\frac{\partial}{\partial t}U_{a}(x,t)}-\overbracket{\nu \Delta U_{a}(x,t)}+\overbracket{U^{b}\nabla_{b}U_{a}(x,t)}
+\overbracket{\nabla_{a}\mathbf{P}(x,t)}
\nonumber\\& +{\bm{\bm{\alpha}}}^{2}U^{b}(x,t)\nabla_{b}
U_{a}(x,t){{\psi}}(\bm{\mathrm{Re}}(\bm{\mathfrak{D}}),t);{\mathbb{I}})^{2}\mathsf{C}\nonumber\\&+\frac{1}{4\pi}{\bm{\bm{\alpha}}}^{2}{{\psi}}(\bm{\mathrm{Re}}(x,t);\mathbb{I})^{2}\mathsf{C}
\int_{\bm{\mathfrak{D}}}\frac{(x_{a}-y_{a})}{|x-y|^{2}}\Delta\big(U_{a}(y,t)\big)U_{b}(y,t)\big)d^{3}y\
\nonumber\\&=\overbracket{\frac{\partial}{\partial t}U_{a}(x,t)}-\overbracket{\nu \Delta U_{a}(x,t)}+\overbracket{U^{b}\nabla_{b}U_{a}(x,t)}
+\overbracket{\nabla_{a}\mathbf{P}(x,t)}
\nonumber\\& +{\bm{\bm{\alpha}}}^{2}U^{b}(x,t)\nabla_{b}
U_{a}(x,t){{\psi}}(\bm{\mathrm{Re}}(\bm{\mathfrak{D}}),t);{\mathbb{I}})^{2}\mathsf{C}\nonumber\\&+{\bm{\bm{\alpha}}}^{2}{{\psi}}(\bm{\mathrm{Re}}(x,t);\mathbb{I})
^{2}\mathsf{C}{\nabla_{a}\mathbf{P}(x,t)}
\nonumber\\&={\frac{\partial}{\partial t}U_{a}(x,t)}-{\nu \Delta U_{a}(x,t)}
+1+{\bm{\bm{\alpha}}}^{2}{{\psi}}(\bm{\mathrm{Re}}(\bm{\mathfrak{D}},t);\mathbb{I})^{2}\mathsf{C}{U^{b}\nabla_{b}U_{a}(x,t)}\nonumber\\&
+1+{\bm{\bm{\alpha}}}^{2}{{\psi}}(\bm{\mathrm{Re}}(\bm{\mathfrak{D}},t);\mathbb{I})^{2}\mathsf{C}{\nabla_{a}\mathbf{P}(x,t)}\nonumber\\&
={\frac{\partial}{\partial t}U_{a}(x,t)}-{\nu \Delta U_{a}(x,t)}
+1+\bm{\mathcal{Y}}(t){U^{b}\nabla_{b}U_{a}(x,t)}+1+\bm{\mathcal{Y}}(t){\nabla_{a}\mathbf{P}(x,t)}\nonumber\\&
\equiv{\frac{\partial}{\partial t}U_{a}(x,t)}-{\nu \Delta U_{a}(x,t)}
+\bm{\mathcal{Y}}(t){U^{b}\nabla_{b}U_{a}(x,t)}+\bm{\mathcal{Y}}(t){\nabla_{a}\mathbf{P}(x,t)}\nonumber\\&
\equiv{\frac{\partial}{\partial t}U_{a}(x,t)}-{\nu \Delta U_{a}(x,t)}
+\bm{\mathcal{Y}}(t){U^{b}\nabla_{b}U_{a}(x,t)}+{\nabla_{a}\mathbf{P}(x,t)}
\end{align}
\end{proof}
\section{Homogeneity and isotropy}
A turbulent flow is essentially homogenous and isotropic (HAI) if it is invariant under translations and rotations in space. However, this is an idealisation and real turbulence is generally not HAI. Boundaries and other physical constraints result in spatial variations of the flow properties. Turbulent flows are also not generally isotropic since any preferred flow direction--for example, along the z-axis of a pipe or tube-- is simply not compatible with the concept of isotropy. However, for large-scale turbulence or turbulent flows in the ocean and atmosphere, there can be regions where turbulence is HAI to a good approximation and there can be scales much less than the physical scale of the system where local homogeneity and isotropy holds. A laboratory idealization of HAI turbulence can be created by passing a laminar or uniform flow through a grid or mesh. The turbulence which forms downstream of the mesh is essentially HAI. However, HAI turbulence is somewhat of a 'manmade' lab idealisation and is generally not realized in nature.

From (1.2)-(1.4) the traditional definition of HAI turbulence is that the derivative of the binary velocity correlation tensor at a point should vanish so that
\begin{align}
\nabla^{b}\bm{R}_{ab}(x,t)=\nabla^{b}\langle \widetilde{U_{a}(x,t)}\widetilde{U_{b}(x,t)}\rangle =0
\end{align}
and for any $\ell>0$
\begin{align}
\langle\widetilde{U_{a}(x+\mathbf{\ell},t)}\widetilde{U_{b}(x+\mathbf{\ell},t)}\rangle = \widetilde{\langle U_{a}(x,t)}\widetilde{U_{b}(x,t)}\rangle
\end{align}
where $\big\langle\bullet\big\rangle$ is the spatial, temporal or ensemble average as discussed in the introduction. Instead, we utilise the following definitions of HAI for the
turbulent flow ${{\mathscr{U}}}_{a}(x,t)$ with respect to the stochastic average or expectation $\bm{\mathbb{E}}\langle\bullet\rangle$.
\begin{defn}
A smooth Navier-Stokes flow will be said to be isotropic if $U_{a}(x,t)$ is isotropic for all $(x,t)\in\bm{\mathfrak{D}}\times[0,T]$ if
\begin{align}
\nabla_{a}U_{b}(x,t)=\nabla_{b}U_{a}(x,t)
\end{align}
\end{defn}
\begin{prop}
The turbulent flow is incompressible if
\begin{align}
{\bm{\mathbb{E}}}\langle \nabla^{a}{{\mathscr{U}}}_{a}(x,t)\rangle=0
\end{align}
which is the case if the underlying NS flow is incompressible $\nabla^{a}U_{a}(x,t)=0$.
\end{prop}
\begin{proof}
\begin{align}
&\nabla_{b}{{\mathscr{U}}}_{a}(x,t)=\nabla^{a}U_{a}(x,t)+{\bm{\bm{\alpha}}}\nabla^{a}U_{a}(x,t)
{{\psi}}(\bm{\mathrm{Re}}(\bm{\mathfrak{D}},t);{\mathbb{I}}))
{{\mathscr{B}}}(x)\nonumber\\&
+{\bm{\bm{\alpha}}}U_{a}(x,t){{\psi}}(\bm{\mathrm{Re}}(\bm{\mathfrak{D}},t);{\mathbb{I}}))\nabla^{a}{{\mathscr{B}}}(x)
\end{align}
Taking the stochastic expectation
\begin{align}
&{\bm{\mathbb{E}}}\langle \nabla_{b}{{\mathscr{U}}}_{a}(x,t)\rangle=\nabla^{a}U_{a}(x,t)+\underbrace{{\bm{\bm{\alpha}}}\nabla^{a}
U_{a}(x,t){{\psi}}(\bm{\mathrm{Re}}(\bm{\mathfrak{D}},t),\mathbb{I})
{\bm{\mathbb{E}}}\langle{{\mathscr{B}}}(x)\rangle}\nonumber\\&
+\underbrace{{\bm{\bm{\alpha}}}U_{a}(x,t){{\psi}}(\bm{\mathrm{Re}}(\bm{\mathfrak{D}},t);
{\mathbb{I}})){\bm{\mathbb{E}}}\langle \nabla^{a}{{\mathscr{B}}}(x)\rangle}=\nabla^{a}U_{a}(x,t)
\end{align}
since the underbraced terms vanish. So the turbulent flow is incompressible if the underlying NS flow is incompressible.
\end{proof}
\begin{prop}
The turbulent flow ${{\mathscr{U}}}_{a}(x,t)$ is \textbf{isotropic} if
\begin{align}
{\bm{\mathbb{E}}}\langle \nabla_{a}{{\mathscr{U}}}_{b}(x,t)\rangle=
{\bm{\mathbb{E}}}\langle \nabla_{b}{{\mathscr{U}}}_{a}(x,t)\rangle
\end{align}
which is the case if $\nabla_{a}U_{b}(x,t)=\nabla_{b}U_{a}(x,t)$; that is, if the underlying Navier-Stokes flow is isotropic.
\end{prop}
\begin{proof}
\begin{align}
&\nabla_{b}{{\mathscr{U}}}_{a}(x,t)=\nabla_{b}U_{a}(x,t)+{\bm{\bm{\alpha}}}\nabla_{b}U_{a}(x,t){{\psi}}(\bm{\mathrm{Re}}(\bm{\mathfrak{D}},t);{\mathbb{I}})){{\mathscr{B}}}(x)\nonumber\\&
+{\bm{\bm{\alpha}}}U_{a}(x,t){{\psi}}(\bm{\mathrm{Re}}(\bm{\mathfrak{D}},t);{\mathbb{I}}))\nabla_{b}{{\mathscr{B}}}(x)\\&
\nabla_{a}{{\mathscr{U}}}_{b}(x,t)=\nabla_{a}U_{b}(x,t)+{\bm{\bm{\alpha}}}\nabla_{a}U_{b}(x,t){{\psi}}(\bm{\mathrm{Re}}(\bm{\mathfrak{D}},t);{\mathbb{I}})){{\mathscr{B}}}(x)\nonumber\\&
+{\bm{\bm{\alpha}}}{{\psi}}(\bm{\mathrm{Re}}(\bm{\mathfrak{D}},t);{\mathbb{I}}))\nabla_{a}{{\mathscr{B}}}(x)
\end{align}
The stochastic expectations are then
\begin{align}
&{\bm{\mathbb{E}}}\langle \nabla_{b}{{\mathscr{U}}}_{a}(x,t)\rangle=\nabla_{b}U_{a}(x,t)
+\underbrace{{\bm{\bm{\alpha}}}\nabla_{b}U_{a}(x,t){{\psi}}(\bm{\mathrm{Re}}(\bm{\mathfrak{D}},t);{\mathbb{I}}))
{\bm{\mathbb{E}}}\langle{{\mathscr{B}}}(x)\rangle}\nonumber\\&
+\underbrace{{\bm{\bm{\alpha}}}U_{a}(x,t){{\psi}}(\bm{\mathrm{Re}}(\bm{\mathfrak{D}},t);{\mathbb{I}})){\bm{\mathbb{E}}}\langle \nabla_{b}{{\mathscr{B}}}(x)\rangle}\\&
\underbrace{{\bm{\bm{\alpha}}}\langle \nabla_{a}{{\mathscr{U}}}_{b}(x,t)\rangle=\nabla_{a}U_{b}(x,t)+{\bm{\bm{\alpha}}}\nabla_{a}
U_{b}(x,t){{\psi}}(\bm{\mathrm{Re}}(\bm{\mathfrak{D}},t);{\mathbb{I}})){\bm{\mathbb{E}}}
\langle{{\mathscr{B}}}(x)\rangle}\nonumber\\&
+\underbrace{{\bm{\bm{\alpha}}}U_{b}(x,t){{\psi}}(\bm{\mathrm{Re}}(\bm{\mathfrak{D}},t);{\mathbb{I}}))
{\bm{\mathbb{E}}}\langle \nabla_{a}{{\mathscr{B}}}(x)\rangle}
\end{align}
and the underbraced terms vanish since ${\bm{\mathbb{E}}}\langle{{\mathscr{B}}}(x)\rangle=0$ and ${\bm{\mathbb{E}}}\langle \nabla_{b}{{\mathscr{B}}}(x)\rangle=0$. Then ${\bm{\mathbb{E}}}\langle \nabla_{a}{{\mathscr{U}}}_{b}(x,t)\rangle=
{\bm{\mathbb{E}}}\langle \nabla_{b}{{\mathscr{U}}}_{a}(x,t)\rangle$
iff $\nabla_{a}U_{b}(x,t)=\nabla_{b}U_{a}(x,t)$
\end{proof}
\begin{prop}
A turbulent flow ${\mathscr{U}}_{a}(x,t)$ within $\bm{\mathfrak{D}}$ is HAI if the following hold:
\begin{enumerate}[(a)]
\item The isotropy condition
\begin{align}
{\bm{\mathbb{E}}}\langle \nabla_{a}{{\mathscr{U}}}_{b}(x,t)\rangle=
{\bm{\mathbb{E}}}\langle \nabla_{b}{{\mathscr{U}}}_{b}(x,t)\rangle
\end{align}
\item The derivative of the Reynolds-type stress tensors vanishes to all orders.
\begin{align}
&\nabla^{b}{\bm{\mathsf{T}}}_{ab}(\mathbf{x,x};t)=\nabla^{b}
{\bm{\mathbb{E}}}\langle{{\mathscr{U}}}_{a}(x,t){{\otimes}}
{{\mathscr{U}}}_{b}(x,t)\rangle=
0\\&
\nabla^{b}{\bm{\mathsf{T}}}_{abc}(\mathbf{x,x,x};t)
=\nabla^{b}{\bm{\mathbb{E}}}\langle{{\mathscr{U}}}_{a}(x,t){{\otimes}}{{\mathscr{U}}}_{b}(x,t){{\otimes}}{{\mathscr{U}}}_{c}(x,t)\rangle=0\\&
\nabla^{b}{\bm{\mathsf{T}}}_{i_{1}...i_{N}}(x...x;t)=\nabla^{b}
{\bm{\mathbb{E}}}\langle\prod_{\alpha=1}^{N}{{\otimes}}{\mathscr{U}}_{i_{\alpha}}(x,t)\rangle=0
\end{align}
and this holds if the underlying NS flow is laminar so that ${\mathscr{U}}_{a}(x,t)=U_{a}=const.$
\end{enumerate}
\end{prop}
\begin{proof}
Equation (6.12) is true from Proposition 6.3. The derivative is
\begin{align}
&\nabla^{b}{\bm{\mathsf{T}}}_{ab}(x,x;t)=\nabla^{b}(U_{a}(x,t)U_{b}(x,t))1+|{\bm{\bm{\alpha}}}|^{2}
{{\psi}}(\bm{\mathrm{Re}}(\bm{\mathfrak{D}},t);{\mathbb{I}}))^{2}{\mathsf{C}}\nonumber\\&
=(\nabla^{b}U_{a}(x,t))U_{b}(x,t)+(\nabla^{b}U_{b}(x,t))U_{a}(x,t)1+|{\bm{\bm{\alpha}}}|^{2}
{{\psi}}(\bm{\mathrm{Re}}(\bm{\mathfrak{D}},t);{\mathbb{I}}))^{2}{\mathsf{C}}\nonumber\\&
=(\nabla^{b}U_{a}(x,t))U_{b}(x,t)1+|{\bm{\bm{\alpha}}}|^{2}
{{\psi}}(\bm{\mathrm{Re}}(\bm{\mathfrak{D}},t);{\mathbb{I}}))^{2}{\mathsf{C}}
\end{align}
since $\nabla^{b}U_{a}(x,t)=0$ from incompressibility. This is zero iff $\nabla^{b}U_{a}(x,t)=0$ or
when $U_{a}(x,t)=U_{a}=const.$ Similarly for the 3-order tensor
${\bm{\mathsf{T}}}_{abc}(\mathbf{x,x,x},;'t)$
\begin{align}
&\nabla^{b}{\bm{\mathsf{T}}}_{ab}(x,x;t)=\nabla^{b}\big((U_{a}(x,t)U_{b}(x,t)U_{c}(x,t)\big)
1+{\bm{\bm{\alpha}}}^{3}{{\psi}}|(\bm{\mathrm{Re}}(\bm{\mathfrak{D}},t); \mathbb{I})^{3}{\mathsf{C}}
\end{align}
Similarly, for the derivative of the 3rd-order tensor
\begin{align}
&\nabla^{b}{\bm{\mathsf{T}}}_{abc}(x,x,x;t)=\nabla^{b}\big(U_{a}(x,t)U_{b}(x,t)U_{c}(x,t)\big)\nonumber\\&
1+|{\bm{\bm{\alpha}}}|^{3}{{\psi}}(\bm{\mathrm{Re}}(\bm{\mathfrak{D}},t);{\mathbb{I}}))^{3}
{\mathsf{C}}^{3}\mathsf{C}\nonumber\\&
=\nabla^{b}(U_{a}(x,t)U_{b}(x,t))U_{c}(x,t)+U_{a}(x,t)U_{b}(x,t)
\nabla^{b}U_{c}(x,t)
\nonumber\\&
1+{\bm{\bm{\alpha}}}^{2}{{\psi}}(\bm{\mathrm{Re}}(\bm{\mathfrak{D}},t);{\mathbb{I}}))^{3}
{\mathsf{C}}^{3}
{\mathsf{C}}\nonumber\\&=
(\nabla^{b}(U_{a}(x,t))U_{b}(x,t))U_{c}(x,t)+\nabla^{b}U_{b}(x,t)U_{a}U_{c}(x,t)
+U_{a}(x,t)U_{b}(x,t)\nabla^{b}U_{c}(x,t)\\&
\times1+{\bm{\bm{\alpha}}}^{2}{{\psi}}(\bm{\mathrm{Re}}(\bm{\mathfrak{D}},t);{\mathbb{I}}))^{3}
{\mathsf{C}}\nonumber\\&=(\nabla^{b}(U_{a}(x,t))U_{b}(x,t))
U_{c}(x,t)+U_{a}(x,t)U_{b}(x,t)\nabla^{b}U_{c}(x,t)\nonumber\\&
1+{\bm{\bm{\alpha}}}^{2}{{\psi}}(\bm{\mathrm{Re}}(\bm{\mathfrak{D}},t);{\mathbb{I}}))^{3}
{\mathsf{C}}
\end{align}
which equals zero if $U_{b}(x,t)=U_{b}$=const. Finally, the derivative of the Nth-order tensor is given by the binomial expansion
\begin{align}
&\nabla^{b}{\bm{\mathsf{T}}}_{i_{1}...i_{N}}(x...x;t)=\nabla^{b}{\bm{\mathbb{E}}}\langle\prod_{\alpha=1}^{N}{{\otimes}}{\mathscr{U}}_{i_{\alpha}}(x,t)\rangle\nonumber\\&=
\nabla^{b}\sum_{M=1}^{N}\binom{N}{M}|U_{a}(x,t)|^{N}||{{\psi}}(\bm{\mathrm{Re}}(\bm{\mathfrak{D}},t);{\mathbb{I}}))^{M}
\big[\frac{1}{2}({\mathsf{C}}+(-1)^{M}{\mathsf{C}}\big]\nonumber\\&
=\sum_{M=1}^{N}\binom{N}{M}\nabla^{b}|U_{a}(x,t)|^{N}|{{\psi}}(\bm{\mathrm{Re}}(\bm{\mathfrak{D}},t);{\mathbb{I}}))^{M}
\big[\frac{1}{2}{\mathsf{C}}+(-1)^{M}{\mathsf{C}}\big]\nonumber\\&
=\sum_{M=1}^{N}\binom{N}{M}N|U_{a}(x,t)|^{N-1}\nabla^{b}U_{a}(x,t)
{{\psi}}(\bm{\mathrm{Re}}(\bm{\mathfrak{D}},t);{\mathbb{I}}))^{M}
\big[\frac{1}{2}({\mathsf{C}}+(-1)^{M}{\mathsf{C}}\big]\nonumber\\&
\end{align}
which again is zero iff $\nabla^{b}U_{a}(x,t)=0$ or $U_{a}(x,t)=U_{a}=const.$.
\end{proof}
\section{Differential equation for ${\bm{\mathsf{T}}}_{ab}(x,t)$}
We now seek a differential equation for the evolution of the Reynolds-like stress
${\bm{\mathsf{T}}}_{ab}(x,t)$. First, the following preliminary lemma establishes the differential equation of the evolution
of ${\bm{\mathsf{T}}}_{ab}(x,t)=U_{a}(x,t)U_{b}(x,t)$
\begin{lem}
Given the tensor ${\bm{\mathsf{T}}}_{ab}(x,t)=U_{a}(x,t)U_{b}(x,t)$, where $U_{a}(x,t)$
evolves by the N-S equations, then the tensor  $\bm{\mathrm{R}}_{ab}(x,t)$ evolves as
\begin{align}
\frac{1}{2}{\partial}_{t}{\bm{\mathsf{T}}}_{ab}(x,t)-\nu \Delta{\bm{\mathrm{R}}}_{ab}(x,t)
+\delta^{cb}\nabla_{b}{\bm{\mathsf{T}}}_{abc}(x,t)+U_{b}(x,t)\nabla_{a}\mathbf{P}(x,t)=0
\end{align}
\end{lem}
\begin{proof}
Since $U_{a}(x,t)$ satisfies the N-S equations for all $x\in\bm{\mathfrak{D}},t>0$,then
\begin{align}
\frac{\partial}{\partial t}U_{a}(x,t)-\nu \Delta U_{a}(x,t)+U^{b}(x,t)\nabla_{b}U_{a}(x,t)+\nabla_{a}P(x,t)=0
\end{align}
Multiplying though on the lhs by $U_{b}(x,t)$ gives
\begin{align}
U_{b}(x,t)\frac{\partial}{\partial t}U_{a}(x,t)-\nu U_{b}(x,t)\Delta U_{a}(x,t)
+U_{b}(x,t)U^{b}(x,t)\nabla_{b}U_{a}(x,t)
+U_{b}(x,t)\nabla_{a}\mathbf{P}(x,t)=0
\end{align}
Then
\begin{align}
U_{b}(x,t){\partial}_{t}U_{a}(x,t)
=\frac{1}{2}{\partial}_{t}(U_{a}(x,t)U_{b}(x,t))
\end{align}
\begin{align}
&\Delta(U_{a}(x,t)U_{b}(x,t))
=U_{a}(x,t)\Delta U_{b}(x,t)+U_{b}(x,t)\Delta U_{a}(x,t)+2\nabla^{b}U_{b}(x,t)\nabla_{b}
U_{a}(x,t)\nonumber\\&=U(x,t)\nabla_{b}\nabla^{b}U_{b}(x,t)+U_{b}(x,t)\nabla_{b}\nabla^{b}U_{a}(x,t)
=U_{b}(x,t)\nabla_{b}\nabla^{b}U_{a}(x,t)\equiv U_{b}(x,t)\Delta U_{a}(x,t)
\end{align}
\begin{align}
&\nabla_{b}(U_{j(x,t)}U^{b}(x,t)U_{a}(x,t))\equiv \nabla_{b}\big(U_{a}(x,t)U_{b}(x,t)U_{c}(x,t)\nonumber\\&
=\delta^{cb}(U_{b}(x,t)U_{c}(x,t)\nabla_{b}U_{a}(x,t)+U_{a}(x,t)U_{b}(x,t)\nabla_{b}U_{c}+U_{a}(x,t)U_{c}(x,t)\nabla_{b}U_{b}(x,t)\big)\nonumber\\&
=\delta^{cb}U_{b}(x,t)U_{c}(x,t)\nabla_{b}U_{a}+U_{a}(x,t)U_{b}(x,t)\nabla^{c}U_{c}(x,t)+U_{a}(x,t)U^{b}(x,t)\nabla_{b}U_{b}(x,t)\nonumber\\&
=U_{b}(x,t)U^{b}(x,t)\nabla_{b}U_{a}(x,t)
\end{align}
using $\nabla^{b}U_{b}(x,t)=0$ and $\nabla_{b}U_{b}(x,t)\equiv\delta^{ab}\delta_{ab}\nabla^{b}U_{b}(x,t)=0$. Then
\begin{align}
&U_{b}(x,t)\frac{\partial}{\partial t}U_{a}(x,t)-\nu U_{b}(x,t)\Delta U_{a}(x,t)
+U_{b}(x,t)U^{b}(x,t)\nabla_{b}U_{a}(x,t)
+U_{b}(x,t)\nabla_{a}\mathbf{P}(x,t)\nonumber\\&=\frac{1}{2}{\partial}_{t}\big(U_{a}(x,t)U_{b}(x,t)\big)
-\nu \Delta\big(U_{a}(x,t)U_{b}(x,t)\big)+\delta^{cb}\nabla_{b}\big (U_{a}(x,t)U_{b}(x,t)U_{c}(x,t)\big)
+U_{b}(x,t)\nabla_{a}P(x,t\nonumber\\&
=\frac{1}{2}{\partial}_{t}{\bm{\mathsf{T}}}_{ab}(x,t)-\nu \Delta{\bm{\mathsf{T}}}_{ab}(x,t)
+\delta^{cb}\nabla_{b}{\bm{\mathsf{T}}}_{abc}(x,t)+U_{b}(x,t)\nabla_{a}P(x,t)=0
\end{align}
\end{proof}
Given (7.1), one can also find a DE for the Reynolds-like stress tensor ${\bm{\mathsf{T}}}_{ab}(x,t)$.
\begin{thm}
Given the turbulent flow ${\mathscr{U}}_{a}(x,t)$ within $\bm{\mathfrak{D}}$, then the tensor\newline ${\bm{\mathsf{T}}}_{ab}(x,t)
=\big\langle {\mathscr{U}}_{a}(x,t){\bm{\bm{\alpha}}}{\mathscr{U}}_{b}(x,t)\big\rangle$ satisfies the following PDE
\begin{align}
&\frac{1}{2}{\partial}_{t}{\bm{\mathsf{T}}}_{ab}(x,t)-\nu \Delta{\bm{\mathsf{T}}}_{ab}(x,t)
+\delta^{cb}\nabla_{b}{\bm{\mathsf{T}}}_{abc}(x,t)+U_{a}(x,t)\nabla_{a}\mathbf{P}(x,t)\nonumber\\&
=\frac{1}{2}{\mathcal{J}}_{2}(t){\partial}_{t}{\bm{\mathsf{T}}}_{ab}(x,t)+
2{\mathcal{J}}(t){\bm{\mathsf{T}}}_{ab}\frac{\partial}{\partial t}\-\nu\bm{\mathcal{Y}}(t)\Delta{\bm{\mathsf{T}}}_{ab}(x,t)
+{\bm{\mathsf{{\psi}}}}_{3}(t)U_{b}(x,t)U^{b}(x,t)\nabla_{a}{\bm{\mathsf{T}}}_{abc}(x,t)
\end{align}
\end{thm}
\begin{proof}
Beginning with
\begin{align}
\frac{1}{2}{\partial}_{t}{\bm{\mathsf{T}}}_{ab}(x,t)-\nu \Delta{\bm{\mathsf{T}}}_{ab}(x,t)
+\delta^{cb}\nabla_{b}{\bm{\mathsf{T}}}_{abc}(x,t)+U_{a}(x,t)\nabla_{a}\mathbf{P}(x,t)
\end{align}
replace ${\bm{\mathsf{T}}}_{ab}(x,t)$ with ${\mathscr{U}}_{a}(x,t){{\otimes}}{\mathscr{U}}_{b}(x,t)$ so that
\begin{align}
&\frac{1}{2}{\partial}_{t}{{\mathscr{U}}}_{a}(x,t){{\otimes}}{{\mathscr{U}}}_{b}(x,t)-\nu \Delta{{\mathscr{U}}}_{a}(x,t)
{{\otimes}}{{\mathscr{U}}}_{b}(x,t)\nonumber\\&+\delta^{cb}\nabla_{b}{{\mathscr{U}}}_{a}(x,t)
{{\otimes}}{{\mathscr{U}}}_{b}(x,t)
{{\otimes}}{{\mathscr{U}}}_{b}(x,t)+{{\mathscr{U}}}_{a}(x,t)\nabla_{a}\mathbf{P}(x,t)
\end{align}
and take the expectation so that
\begin{align}
&\frac{1}{2}{\partial}_{t}{\bm{\mathbb{E}}}\langle{{\mathscr{U}}}_{a}(x,t){{\otimes}}{{\mathscr{U}}}_{b}(x,t)\rangle
-\nu \Delta{\bm{\mathbb{E}}}\langle{{\mathscr{U}}}_{a}(x,t)
{{\otimes}}{{\mathscr{U}}}_{b}(x,t)\rangle\nonumber\\&
+\delta^{cb}\nabla_{b}{\bm{\mathbb{E}}}\langle{{\mathscr{U}}}_{a}(x,t){{\otimes}}{{\mathscr{U}}}_{b}(x,t)
{{\otimes}}{{\mathscr{U}}}_{b}(x,t)\rangle
+{\bm{\mathbb{E}}}\langle{{\mathscr{U}}}_{a}(x,t)\rangle \nabla_{a}\mathbf{P}(x,t)\nonumber\\&
=\frac{1}{2}{\partial}_{t}{\bm{\mathsf{T}}}_{ab}(x,t)-\nu \Delta{\bm{\mathsf{T}}}_{ab}(x,t)
+\delta^{cb}\nabla_{b}{\bm{\mathsf{T}}}_{abc}(x,t)+U_{a}(x,t)\nabla_{a}\mathbf{P}(x,t)
\end{align}
Now
\begin{align}
&{\bm{\mathsf{T}}}_{ab}(x,t)=U_{a}(x,t)U_{b}(x,t)
1+|{\bm{\bm{\alpha}}}|^{2}|{{\psi}}(\mathbf{Re}(\bm{\mathfrak{D}},t),{\mathbb{I}})^{2}{\mathsf{C}}\nonumber\\&
{\bm{\mathsf{T}}}_{abc}(x,t)=U_{a}(x,t)U_{b}(x,t)U_{c}(x,t)
1+{\bm{\bm{\alpha}}}^{3}{{\psi}}(\mathbf{Re}(\bm{\mathfrak{D}},t); \mathbb{I})\big)^{3}{\mathsf{C}}
\end{align}
The derivatives are:
\begin{enumerate}[(a)]
\item The derivative of the stress tensor
\begin{align}
\nabla^{b}{\bm{\mathsf{T}}}_{ab}(x,t)&=\big(\nabla^{b}U_{a}(x,t)\big)U_{b}(x,t)+(\underbrace{\nabla^{b}
U_{b}(x,t))}
U_{a}(x,t)
1+|{\bm{\bm{\alpha}}}|^{2}{{\psi}}(\mathbf{Re}(\bm{\mathfrak{D}},t); {\mathbb{I}})^{2}{\mathsf{C}}\nonumber\\&
=\nabla^{b}U_{a}(x,t)\big)U_{b}(x,t)1+|{\bm{\bm{\alpha}}}|^{2}|
{{\psi}}(\mathbf{Re}(\bm{\mathfrak{D}},t);{\mathbb{I}})
\nonumber\\&
=\nabla^{b}{\bm{\mathsf{T}}}_{ab}(x,t)1+{\bm{\bm{\alpha}}}
{{\psi}}(\mathbf{Re}(\bm{\mathfrak{D}},t); {\mathbb{B}})^{2}{\mathsf{C}}
\end{align}
where the underbraced term vanishes due to the incompressibility of the fluid.
\item The Laplacian of the stress tensor
\begin{align}
&\Delta{\bm{\mathsf{T}}}_{ab}(x,t)\equiv \nabla_{b}\nabla^{b}{\bm{\mathsf{T}}}_{ab}(x,t)\nonumber\\&
=\big(\nabla_{b}\nabla^{b}U^{a}\big)U_{b}(x,t)+\nabla_{b}U_{a}(x,t)
\underbrace{\nabla^{b}U_{b}(x,t)}\times
1+|{\bm{\bm{\alpha}}}^{2}|{{\psi}}(\mathbf{Re}(\bm{\mathfrak{D}},t)
;{\mathbb{I}})^{2}{\mathsf{C}}\nonumber\\&
=\big(\nabla_{b}\nabla^{b}U^{a}\big)U_{b}(x,t)1+|{\bm{\bm{\alpha}}}|^{2}|
{{\psi}}(\mathbf{Re}(\bm{\mathfrak{D}},t);{\mathbb{I}})^{2}{\mathsf{C}}
\nonumber\\&=\Delta{\bm{\mathsf{T}}}_{ab}(x,t)
1+{\bm{\bm{\alpha}}}^{2}{{\psi}}(\mathbf{Re}(\bm{\mathfrak{D}},t);{\mathbb{I}})^{2}{\mathsf{C}}
\end{align}
\item The derivative of the triple tensor is
\begin{align}
&\delta^{cb}\nabla_{b}{\bm{\mathsf{T}}}_{abc}(x,t)=\delta^{cb}\big(\nabla_{b}U_{a}(x,t))\big)U_{b}(x,t)U_{c}(x,t)
+U_{a}(x,t)U_{b}(x,t)\nabla_{b}U_{c}(x,t)\nonumber\\&\times1+{\bm{\bm{\alpha}}}^{2}|
{{\psi}}(\mathbf{Re}(\bm{\mathfrak{D}},t); {\mathbb{I}})^{2}\bm{\mathsf{C}}\nonumber\\&
=\delta^{cb}\big(\nabla_{b}U_{a}(x,t))\big)U_{b}(x,t)U_{c}(x,t)U_{a}(x,t)
U_{b}(x,t)\nabla_{b}U_{c}(x,t)\delta^{cb}\times1+{\bm{\bm{\alpha}}}^{2}{{\psi}}(\mathbf{Re}(\bm{\mathfrak{D}},t);{\mathbb{I}})^{2}{\mathsf{C}}
\nonumber\\&
=\delta^{cb}\big(\nabla_{b}U_{a}(x,t))\big)U_{b}(x,t)U_{c}(x,t)U_{a}(x,t)
U_{b}(x,t)\underbrace{\nabla_{b}U^{b}(x,t)}\times1+{\bm{\bm{\alpha}}}^{2}|{{\psi}}(\mathbf{Re}(\bm{\mathfrak{D}},t);{\mathbb{I}})^{2}{\mathsf{C}}
\nonumber\\&
=\delta^{cb}\big(\nabla_{b}U_{a}(x,t))\big)U_{b}(x,t)U_{c}(x,t)U_{a}(x,t)
U_{b}(x,t)\nonumber\\&
\times\left(1+{\bm{\bm{\alpha}}}\right)^{2}|{{\psi}}(\mathbf{Re}(\bm{\mathfrak{D}},t);{\mathbb{I}})^{2}{\mathsf{C}}\equiv {\bm{\mathsf{T}}}_{abc}(x,t)1+{\bm{\bm{\alpha}}}^{2}|{{\psi}}(\mathbf{Re}(\bm{\mathfrak{D}},t),
{\mathbb{I}})^{2}\bm{\mathsf{C}}
\end{align}
\item Finally the time derivative is
\begin{align}
\frac{\partial}{\partial t}{\bm{\mathsf{T}}}_{ab}(x,t)&=\big(\frac{\partial}{\partial t}U_{a}(x,t)\big)U_{b}(x,t)+
\big(\frac{\partial}{\partial t}U_{b}(x,t)\big)U_{a}(x,t)
1+{\bm{\bm{\alpha}}}^{2}|{{\psi}}(\mathbf{Re}(\bm{\mathfrak{D}},t);{\mathbb{I}})^{2}{\mathsf{C}}\nonumber\\&
+U_{a}(x,t)U_{b}(x,t)2{\bm{\bm{\alpha}}}^{2}
{{\psi}}(\mathbf{Re}(\bm{\mathfrak{D}},t);\mathbb{I}){\frac{\partial}{\partial t}}
{{\psi}}(\mathbf{Re}(\bm{\mathfrak{D}},t);\mathbb{I})\nonumber\\&
=\frac{\partial}{\partial t}{\bm{\mathsf{T}}}_{ab}(x,t)
1+|{\bm{\bm{\alpha}}}^{2}{{\psi}}(\mathbf{Re}(\bm{\mathfrak{D}},t);{\mathbb{I}})^{2}{\mathsf{C}}
+{\bm{\mathsf{T}}}_{ab}(x,t)2{\bm{\bm{\alpha}}}^{2}
{{\psi}}(\mathbf{Re}(\bm{\mathfrak{D}},t);{\mathbb{I}}){\frac{\partial}{\partial t}}
{{\psi}}(\mathbf{Re}(\bm{\mathfrak{D}},t);{\mathbb{I}})
\end{align}
If substitute and set
\begin{align}
&\bm{\mathcal{Y}}(t)={\bm{\bm{\alpha}}}^{2}{{\psi}}(\mathbf{Re}(\bm{\mathfrak{D}},t); {\mathbb{I}})^{2}{\mathsf{C}}\nonumber\\&
{\bm{\mathsf{{\psi}}}}_{3}(t)={\bm{\bm{\alpha}}}^{3}{{\psi}}(\mathbf{Re}(\bm{\mathfrak{D}},t);{\mathbb{I}})^{3}{\mathsf{C}}
\end{align}
then (7.8) follows.
\end{enumerate}
\end{proof}
\subsection{Averaged energy integrals}
The energy integral is defined as follows.
\begin{defn}
As before, let $\bm{\mathfrak{D}}\subset\mathbb{R}^{3}$ and let $\bm{\mathfrak{D}}$ contain a fluid of viscosity $\nu$ and velocity
$U_{a}(x,t)$ satisfying the NS equations. The energy integral with respect to domain $\bm{\mathfrak{D}}$ is then
\begin{align}
&{\bm{\mathcal{E}}}[U_{a}(x,t)]=\frac{1}{2}{\int}_{\bm{\mathfrak{D}}}|U_{a}(x,t)|^{2}d\mu(\bm{\mathfrak{D}})\equiv \frac{1}{2}{\int}_{\bm{\mathfrak{D}}}|U_{a}(x,t)|^{2}d^{3}x\equiv\|U_{a}(\bullet,t)\|_{L_{2}(\bm{\mathfrak{D}})}
\end{align}
\end{defn}
The energy integral is taken to be bounded so that $\exists\bm{\mathfrak{B}}$ such that
${\mathcal{E}}[U_{a}(x,t)]\le \bm{\mathfrak{B}}$ and
\begin{align}
\frac{1}{2}\int_{\bm{\mathfrak{D}}}|U_{a}(x,t)|^{2}d\mu(x)\le {\bm{\mathfrak{B}}}
\end{align}
Suppose we now consider the energy bound on a turbulent flow ${\mathscr{U}}_{a}(x,t)$. Then the averaged energy integral for this turbulent flow should be also be bounded.
\begin{prop}
Let ${\mathscr{U}}_{a}(x,t)$ be a turbulent flow for $x\in\bm{\mathfrak{D}}$ and $t>0$, as previously defined. Then the averaged energy integral is
\begin{align}
&{\bm{\mathbb{E}}}\langle{{\mathscr{U}}}_{a}(x,t)\rangle=\frac{1}{2}{\bm{\mathbb{E}}}\langle\int_{\bm{\mathfrak{D}}}{{\mathscr{U}}}_{a}(x,t){{\otimes}}{{\mathscr{U}}}^{a}(x,t)
\rangle d\mu(x)\nonumber\\&=\frac{1}{2}\int_{\bm{\mathfrak{D}}}{\bm{\mathbb{E}}}\langle{{\mathscr{U}}}_{a}(x,t){{\otimes}}{{\mathscr{U}}}^{a}(x,t)\rangle
d\mu(x)\nonumber\\&={\mathcal{E}}\big[U_{a}(x,t)\big]
1+{\bm{\bm{\alpha}}}^{2}{\mathsf{C}}
{{\psi}}(\mathbf{Re}(\bm{\mathfrak{D}},t);{\mathbb{I}})^{2}\equiv
\bm{\mathcal{Y}}(t){\mathcal{E}}[U_{a}(x,t)]
\end{align}
so that the energy integral is boosted or re-scaled by a factor $\bm{\mathcal{Y}}(t)$.
\end{prop}
\begin{proof}
Given the turbulent flow or random field ${\mathscr{U}}_{a}(x,t)$ then the averaged energy integral is
\begin{align}
&{\mathcal{E}}\langle{{\mathscr{U}}}_{a}(x,t)]\rangle
=\frac{1}{2}{\bm{\mathbb{E}}}\langle\int_{\bm{\mathfrak{D}}}{{\mathscr{U}}}_{a}(x,t){{\otimes}}
{{\mathscr{U}}}^{a}(x,t)d\mu(x)\rangle\nonumber\\&
=\frac{1}{2}\int_{\bm{\mathfrak{D}}}{\bm{\mathbb{E}}}\langle{{\mathscr{U}}}_{a}(x,t){{\otimes}}{{\mathscr{U}}}^{a}(x,t)\rangle
d\mu(x)\nonumber\\&=\frac{1}{2}\int_{\bm{\mathfrak{D}}}U_{a}(x,t)U^{a}(x,t)d\mu(x)
+\frac{1}{2}{\bm{\bm{\alpha}}}^{2}{{\psi}}(\mathbf{Re}(\bm{\mathfrak{D}},t);{\mathbb{I}})^{2} \nonumber\\&\times
\int_{\bm{\mathfrak{D}}}U_{a}(x,t)U^{a}(x,t){\bm{\mathbb{E}}}\langle{{\mathscr{B}}}(x)
{{\otimes}}{{\mathscr{B}}}(x)\rangle d\mu(x)\nonumber\\&=\frac{1}{2}\int_{\bm{\mathfrak{D}}}U_{a}(x,t)U^{a}(x,t)d\mu(x)
+\frac{1}{2}{\bm{\bm{\alpha}}}^{2}{{\psi}}(\mathbf{Re}(\bm{\mathfrak{D}},t); {\mathbb{I}})^{2}
\int_{\bm{\mathfrak{D}}}U_{a}(x,t)U^{a}(x,t)d\mu(x)\nonumber\\&
=\frac{1}{2}\int_{\bm{\mathfrak{D}}}U_{a}(x,t)U^{a}(x,t)d\mu(x)
1+\frac{1}{2}{\bm{\bm{\alpha}}}^{2}{{\psi}}(\mathbf{Re}(\bm{\mathfrak{D}},t); {\mathbb{I}})^{2}
{\mathsf{C}}\nonumber\\&
={{\mathcal{E}}}[U_{a}(x,t)]1+\frac{1}{2}{\bm{\bm{\alpha}}}^{2}{\mathsf{C}}{{\psi}}
(\bm{\mathrm{Re}}(\bm{\mathfrak{D}}),t);\mathbb{I})^{2}\equiv
\bm{\mathcal{Y}}(t){{\mathcal{E}}}[U_{a}(x,t)]\le
\bm{\mathcal{Y}}(t){\bm{\mathfrak{B}}}
\end{align}
\end{proof}
\section{Vorticity and vortex tangles and correlations}
As discussed in the introduction, vorticity and coupled vortex formation over a wide range of length scales, is a salient and universal feature of turbulence
Here, we wish to estimate the vorticity and correlations between vortices with respect to the turbulent flow ${{\mathscr{U}}}_{a}(x,t)$.
\begin{defn}
Given an incompressible fluid velocity ${U}(x,t)\equiv U_{a}(x,t)\bm{e}^{a}$ with $\bm{D}.{U}=0$, the vorticity is defined as its curl so that
\begin{align}
{\bm{\omega}}(x,t)=curl U(x,t)=D\wedge U(x,t)
\end{align}
or equivalently
\begin{align}
{{\bm{\omega}}}_{a}(x,t)={\epsilon}_{abc}\nabla^{b}U^{c}(x,t)
\end{align}
The vorticity is the measure of the rotation of the fluid locally about a particular point, but not the rotation of the fluid as a whole.
\end{defn}
\begin{lem}
The PDE for the vorticity is given by
\begin{align}
{\partial}_{t}{\bm{\omega}}(x,t)=\nu \Delta{\bm{\omega}}(x,t)+U(x,t).
D{\bm{\omega}}(x,t)-{\bm{\omega}}(x,t).DU(x,t)=0
\end{align}
with $D.{\bm{\omega}}(x,t)=0$ or
\begin{align}
\frac{\partial}{\partial t}{\bm{\omega}}_{a}(x,t)-\eta \Delta{\bm{\omega}}_{a}(x)+U_{a}(x)\nabla_{b}{\bm{\omega}}_{a}(x,t)
-{\bm{\omega}}^{b}(x,t)\nabla_{b}U_{a}(x,t)=0
\end{align}
The time-independent equation is
\begin{align}
-\eta \Delta{\bm{\omega}}_{a}(x)+{\bm{\omega}}_{a}(x)\nabla_{b}{\bm{\omega}}_{a}(x)-{\bm{\omega}}_{b}(x)\nabla_{b}
U_{a}(x)=0
\end{align}
\end{lem}
The term ${\mathbf{\omega}}.D\overrightarrow{U}$ is the \textbf{vortex stretching term} and is generally difficult to deal with. The vorticity vanishes on a boundary
$\partial\bm{\mathfrak{D}}$ if the no-slip BC is implemented such that $U_{\partial\bm{\mathfrak{D}}}=0$
\begin{lem}
The vorticity and velocity can also be related by a Poisson equation so that
\begin{align}
-\Delta U(x)=D\wedge{\bm{\omega}}(x)
\end{align}
which has the solution
\begin{align}
U(x)=\frac{1}{4\pi}{\int}_{\mathbf{R}^{3}}\frac{\nabla_{(y)}\wedge{\bm{\omega}}(y)d^{3}y}{|x-y|}
\end{align}
or
\begin{align}
{U}_{a}(x)=\frac{1}{4\pi}{\int}_{\bm{\mathfrak{D}}^{3}}\epsilon_{abc}\nabla^{b}_{(y)}\left(\frac{1}{|x-y|}\right){\mathrm{V}^{c}(y)d^{3}y}
\end{align}
Integrating by parts gives a Bio-Savart-type law of the form
\begin{align}
U_{a}(x)=-\frac{1}{4\pi}\int_{\bm{\mathfrak{D}}^{3}}\epsilon_{abc}\left(\frac{|x^{b}-y^{b}|}{|x-y|^{3}}\right){\mathrm{V}^{c}(y)d^{3}y}
\end{align}
\end{lem}
\subsection{Stochastic vorticity}
\begin{prop}
Given the turbulent flow ${{\mathscr{U}}}_{a}(x,t)$ as previously defined, the turbulent or stochastic vorticity is
\begin{align}
&{\mathscr{W}}_{a}(x,t)=\varepsilon_{abc}\nabla^{b}{{\mathscr{U}}}^{c}(x,t)\nonumber\\&
=\epsilon_{abc}\nabla^{b}U^{c}(x,t)+\epsilon_{abc}\bm{\bm{\alpha}}\nabla^{b}U^{c}(x,t){{\psi}}
(\bm{\mathrm{Re}}(\bm{\mathfrak{D}},t);{\mathbb{I}})^{2}
{{\mathscr{B}}}(x,t)\nonumber\\&+\epsilon_{abc}{{\psi}}
(\bm{\mathrm{Re}}(\bm{\mathfrak{D}},t);{\mathbb{I}})^{2}
\nabla^{b}{{\mathscr{B}}}(x,t)
\end{align}
Taking the stochastic expectation
\begin{align}
&{\psi}_{a}(x,y;t)={\mathbb{E}}\langle{\mathscr{W}}_{a}(x,t)\rangle
=\varepsilon_{abc}{\bm{\mathbb{E}}}\langle \nabla^{b}{{\mathscr{U}}}^{c}(x,t)\rangle\nonumber\\&
=\epsilon_{abc}\nabla^{b}U^{c}(x,t)+\epsilon_{abc}{\bm{\bm{\alpha}}}\nabla^{b}U^{c}(x,t)
{{\psi}}(\bm{\mathrm{Re}}(\bm{\mathfrak{D}},t);{\mathbb{I}})^{2}
{\bm{\mathbb{E}}}\langle{{\mathscr{B}}}(x,t)\rangle\nonumber\\&+\epsilon_{abc}{\bm{\bm{\alpha}}}U^{c}(x,t)
{{\psi}}(\bm{\mathrm{Re}}(\bm{\mathfrak{D}},t);{\mathbb{I}})^{2}
{\bm{\mathbb{E}}}\langle \nabla^{b}{{\mathscr{B}}}(x,t)\rangle={\bm{\omega}}_{a}(x,t)
\end{align}
\end{prop}
\begin{prop}
Given $(\mathbf{x,y,z})\in\bm{\mathfrak{D}}$, the binary and triple vorticity correlations are
\begin{align}
&{\psi}_{ab}(\mathbf{x,y};t)={\bm{\mathbb{E}}}\langle{\mathscr{W}}_{\ell}(x,t)
{{\otimes}}{\mathscr{W}}_{p}(y,t)\rangle\nonumber\\&
=\epsilon_{ilm}\epsilon_{jpq}{\bm{\mathbb{E}}}\langle \nabla^{l}{{\mathscr{U}}}^{m}(x,t)
{{\otimes}} \nabla^{p}{{\mathscr{U}}}^{q}(y,t)\rangle
\end{align}
and
\begin{align}
&{\psi}_{abc}(x,y;t)={\bm{\mathbb{E}}}\langle{\mathscr{W}}(x,t){{\otimes}}{\mathscr{W}}(y,t)
{{\otimes}}{\mathscr{W}}(z,t)\rangle\nonumber\\&=\epsilon_{ilm}\epsilon_{jpq}\epsilon_{krs}
{\bm{\mathbb{E}}}\langle \nabla^{l}{\mathscr{U}}^{m}(x,t){{\otimes}} \nabla^{p}{{\mathscr{U}}}^{q}(\mathsf{y},t)
{{\otimes}} \nabla^{r}{{\mathscr{U}}}^{s}(\mathsf{z},t)\rangle
\end{align}
The N-point correlation is then
\begin{align}
&{\bm{\mathbb{E}}}\langle{\mathscr{W}}(x_{1},t){{\otimes}}...{{\otimes}}{\mathscr{W}}(x_{N},t)\rangle=\epsilon_{i_{1}j_{1}k_{1}}\times...
\times\epsilon_{i_{1}j_{1}k_{1}}{\bm{\mathbb{E}}}\langle(\nabla^{j_{1}}{\mathscr{U}}^{k_{1}}{{\otimes}}...
{{\otimes}} \nabla^{j_{1}}{{\mathscr{U}}}^{k_{1}}\rangle
\end{align}
The correlations at a single point $x$ then arise in the limit as $(y,z)\rightarrow x $ so that
\begin{align}
&{\psi}_{a}(x,y;t)={\bm{\mathbb{E}}}\langle{\mathscr{W}}(x,t){{\otimes}}{\mathscr{W}}(x,t)\rangle=\lim_{x,y\rightarrow x}{\bm{\mathbb{E}}}\langle{\mathscr{W}}(x,t){{\otimes}}\mathscr{\mathscr{W}}(y,t)\rangle\nonumber\\&
=\lim_{x,y\rightarrow x}\epsilon_{ilm}\epsilon_{jpq}{\bm{\mathbb{E}}}\langle \nabla^{l}{\mathscr{U}}^{m}(x,t){{\otimes}} \nabla^{p}{{\mathscr{U}}}^{q}(y,t)\rangle
=\epsilon_{ilm}\epsilon_{jpq}{\bm{\mathbb{E}}}\langle \nabla^{l}{\mathscr{U}}^{m}(x,t){{\otimes}} \nabla^{p}{\mathscr{U}}^{q}(x,t)\rangle
\end{align}
and at 3rd order
\begin{align}
&\lim_{x,y\rightarrow x}{\bm{\mathbb{E}}}\langle{\mathscr{W}}(x,t){{\otimes}}{\mathscr{W}}(y,t){{\otimes}}{\mathscr{W}}(z,t)\rangle
\nonumber\\&=\lim_{x,y\rightarrow x}\epsilon_{ilm}\epsilon_{jpq}\epsilon_{krs}{\bm{\mathbb{E}}}\big[\nabla^{l}{{\mathscr{U}}}^{m}(x,t){{\otimes}} \nabla^{p}{{\mathscr{U}}}^{q}(y,t){{\otimes}} \nabla^{r}
{{\mathscr{U}}}^{s}(z,t)\big]\nonumber\\&
=\lim_{x,y\rightarrow x}\epsilon_{ilm}\epsilon_{jpq}\epsilon_{krs}\big \|\!\big\|\nabla^{l}{{\mathscr{U}}}^{m}(x,t)\nabla^{p}D{{\mathscr{U}}}^{q}(y,t){{\mathscr{U}}}^{s}(z,t)\big\|\!\big\|_{E_{1}(\mathfrak{D})}
\end{align}
\end{prop}
The N-point correlation is then
\begin{align}
&\lim_{x,y\rightarrow x}{\bm{\mathbb{E}}}\langle{\mathscr{W}}(x_{1},t){{\otimes}}...{{\otimes}}{\mathscr{W}}(x_{N},t)\rangle
=\mathcal{E}_{i_{1}j_{1}k_{1}}\times...\times\mathcal{E}_{i_{1}j_{1}k_{1}}
{\bm{\mathbb{E}}}\langle \nabla^{j_{1}}{{\mathscr{U}}}^{k_{1}}{{\otimes}}...
\lim_{x,y\rightarrow x}\times \nabla^{j_{1}}{{\mathscr{U}}}^{k_{1}}\rangle\nonumber\\&
=\lim_{x,y\rightarrow x}{\bm{\mathbb{E}}}\langle\mathcal{E}_{i_{1}j_{1}k_{1}}\times...\times\mathcal{E}_{i_{1}j_{1}k_{1}} |\nabla^{j_{1}}{{\mathscr{U}}}^{k_{1}}(x,t){{\otimes}}...{{\otimes}} \nabla^{j_{1}}{{\mathscr{U}}}^{k_{1}}(x,t)|\rangle
\end{align}
The circulation is defined as follows
\begin{defn}$(\bm{\mathrm{Definition~of~circulation}})$\newline
Let $\Im\in\bm{\mathfrak{D}}$ be an open curve or closed loop in $\bm{\mathfrak{D}}$ such that $\Im\in\bm{\mathfrak{D}}$ and let $U_{a}(x,t)$ be a NS flow within
$\bm{\mathfrak{D}}$ with $x\in\Im$. The circulation $\mathrm{C}[\Im,t]$ is then the line integral
\begin{align}
\bm{\mathrm{C}}[\Im,t]={\int}_{\Im}U_{a}(x,t)dx^{a}
\end{align}
or for a closed loop or knot $\mathfrak{K}\in\bm{\mathfrak{D}}$
\begin{align}
\bm{\mathrm{C}}[\mathfrak{K},t]={\oint}_{\mathfrak{K}}U_{a}(x,t)dx^{a}
\end{align}
\end{defn}
For an incompressible flow
\begin{align}
\nabla_{a}\bm{\mathrm{C}}[\Im,t]={\int}_{\mathcal{S}}\nabla_{a}U_{a}(x,t)dx^{a}=0,~~~~
\nabla_{a}\bm{\mathrm{C}}[\mathfrak{K},t]={\oint}_{\mathfrak{K}}\nabla_{a}U_{a}(x,t)dx^{a}=0
\end{align}
The time evolution of the circulations is then
\begin{align}
&{\partial}_{t}\bm{\mathrm{C}}[\Im,t]={\partial}_{t}{\int}_{\Im}U_{a}(x,t)dx^{a}={\int}_{\Im}{\partial}_{t}U_{a}(x,t)dx^{a}\\&
{\partial}_{t}\bm{\mathrm{C}}[\mathfrak{K},t]={\partial}_{t}{\oint}_{\mathfrak{K}}U_{a}(x,t)dx^{a}={\oint}_{\mathfrak{K}}
{\partial}_{t}U_{a}(x,t)dx^{a}
\end{align}
\begin{defn}
Kelvin's circulation theorem for a closed curve or loop $\mathfrak{K}\in\bm{\mathfrak{D}}$ states that for an ideal barytropic fluid with $P=P(\rho)$
\begin{align}
{{\bm{\mathrm{D}}}}_{m}\bm
{\mathrm{C}}[\mathfrak{K},t]=\frac{\partial}{\partial t}+U^{b}\nabla_{b}{\mathrm{C}}[\mathfrak{K},t] =0
\end{align}
\end{defn}
\begin{prop}$(\bm{\mathrm{Stochastic~circulation~for~a~turbulent~flow}})$\newline
Let $x\in\Im\in\bm{\mathfrak{D}}$ be a loop or knot. As before, let the turbulent fluid flow within $\bm{\mathfrak{D}}$ be the random
field ${{\mathscr{U}}}_{a}(x,t)$. Then the stochastic circulation is defined as the line integral
\begin{align}
&{{\mathscr{C}}}(\Im,t)=\int_{\Im}{{\mathscr{U}}}_{a}(x,t)dx^{a}\nonumber\\&
=\int_{\Im}U_{a}(x,t)dx^{a}+{\bm{\bm{\alpha}}}{{\psi}}
\big(|\bm{\mathrm{Re}}(\bm{\mathfrak{D}},t);{\mathbb{I}}\big)
\int_{\Im}U_{a}(x,t){{\mathscr{B}}}(x,t)dx^{a}\\&
=\int_{\Im}U_{a}(x,t)dx^{a}+{\bm{\bm{\alpha}}}{{\psi}}
\big(|\bm{\mathrm{Re}}(\bm{\mathfrak{D}},t);{\mathbb{I}}\big)
\int_{\Im}{{\mathscr{B}}}_{a}(x,t)dx^{a}\nonumber\\&
={\mathrm{C}}(\Im)+{\bm{\bm{\alpha}}}{{\psi}}\big(|\bm{\mathrm{Re}}(\bm{\mathfrak{D}},t);{\mathbb{I}}\big)
\int_{\Im}{{\mathscr{R}}}_{a}(x,t)dx^{a}
\end{align}
where ${\mathscr{R}}_{a}(x,t)=U_{a}(x,t){\mathscr{B}}(x)$.
For a closed loop
\begin{align}
{{\mathscr{C}}}(\mathfrak{K},t)&=\int_{\mathfrak{K}}{{\mathscr{U}}}_{a}(x,t)dx^{a}\nonumber\\&={\oint}_{\mathfrak{K}}U_{a}(x,t)dx^{a}+{\bm{\bm{\alpha}}}{{\psi}}\big(|\bm{\mathrm{Re}}(\bm{\mathfrak{D}},t);{\mathbb{I}}\big)
{\oint}_{\mathfrak{K}}U_{a}(x,t){{\mathscr{B}}}(x)dx^{a}\nonumber\\&
={\oint}_{\mathfrak{K}}U_{a}(x,t)dx^{a}+{\bm{\bm{\alpha}}}{{\psi}}\big(|\bm{\mathrm{Re}}(\bm{\mathfrak{D}},t);{\mathbb{I}}\big)
{\oint}_{\mathfrak{K}}{{\mathscr{R}}}_{a}(x,t)dx^{a}\nonumber\\&
={\mathrm{C}}(\mathfrak{K},t)+{\bm{\bm{\alpha}}}{{\psi}}\big(|\bm{\mathrm{Re}}(\bm{\mathfrak{D}},t);{\mathbb{I}}\big)
{\oint}_{\mathfrak{K}}{{\mathscr{B}}}_{a}(x,t)dx^{a}
\end{align}
The stochastic average or expectation is then
\begin{align}
&{\bm{\mathsf{C}}}(\mathfrak{K},t)={\bm{\mathbb{E}}}\langle{{\mathscr{C}}}(\mathfrak{K},t)\rangle=\int_{\mathfrak{K}}{{\mathscr{U}}}_{a}(x,t)dx^{a}\nonumber\\&
={\oint}_{\mathfrak{K}}U_{a}(x,t)dx^{a}+{\bm{\bm{\alpha}}}{{\psi}}\big(|\bm{\mathrm{Re}}(\bm{\mathfrak{D}},t);{\mathbb{I}}\big)
{\oint}_{\mathfrak{K}}U_{a}(x,t){\bm{\mathbb{E}}}\langle{{\mathscr{B}}}(x)\rangle dx^{a}\nonumber\\&
={\oint}_{\mathfrak{K}}U_{a}(x,t)dx^{a}+{\bm{\bm{\alpha}}}{{\psi}}\big(|\bm{\mathrm{Re}}(\bm{\mathfrak{D}},t);\mathbb{I}\big)
{\oint}_{\mathfrak{K}}{\bm{\mathbb{E}}}\langle{{\mathscr{B}}}_{a}(x)dx^{a}\rangle\nonumber\\&
={\mathrm{C}}(\mathfrak{K},t)+{\bm{\bm{\alpha}}}{{\psi}}\big(|\bm{\mathrm{Re}}(\bm{\mathfrak{D}},t);{\mathbb{I}}\big)
{\oint}_{\mathfrak{K}}{\bm{\mathbb{E}}}\langle{{\mathscr{B}}}_{a}(x)\rangle dx^{a}={\mathrm{C}}(\mathfrak{K},t)
\end{align}
since again ${\bm{\mathbb{E}}}\big\langle{{\mathscr{B}}}(x)\big\rangle=0$.
\end{prop}
Equation (8.26) or (8.27) also vanishes as the averaged Reynolds number $\mathrm{R}(\mathscr{W},t)$ within $\mathscr{W}$ is reduced to or below the critical
Reynolds number.
\subsection{Vortex correlations or 'tangles' within $\bm{\mathfrak{D}}$}
A set of N knots or curves $\Im_{1},...\Im_{N}\in\bm{\mathfrak{D}}\subset\mathbb{R}^{3}$ can correlated or 'tangled', which is a key characteristic of turbulence.
(And also of quantum/superfluid turbulence in He3 say.) A correlation of two curves $\Im_{1},\Im_{2}\in\bm{\mathfrak{D}}$ can be denoted $\Im_{1}\bigcap\Im_{2}$; a correlation of three curves or knots is then $\Im_{1}\bigcap\Im_{2}\bigcap\Im_{3}$ and so on. For N curves or knots
\begin{align}
{{\Im}}_{1}\bigcap{\Im}_{2}\bigcap...\bigcap{\Im}_{N-1}\bigcap{\Im}_{N}
\end{align}
Similarly, for N closed curves or knots $\mathfrak{K}_{1},...\mathfrak{K}_{N}\in\bm{\mathfrak{D}}\subset\mathbb{R}^{3}$ one has
\begin{align}
{\mathfrak{K}}_{1}\bigcap{\mathfrak{K}}_{2}\bigcap...\bigcap{\mathfrak{K}}_{N-1}\bigcap{\mathfrak{K}}_{N}
\end{align}
\begin{prop}
Let ${{\mathscr{U}}}_{a}(x,t)$ and ${{\mathscr{U}}}_{b}(y,t)$ be two turbulent flows at $(x,y)\in\bm{\mathfrak{D}}\subset\mathbb{R}^{3}$ at some $t>0$. Then
\begin{align}
&{{\mathscr{U}}}_{a}(x,t)=U_{a}(x,t)+{\bm{\bm{\alpha}}}U_{a}(x,t)
{{\psi}}\big(|\bm{\mathrm{Re}}(\bm{\mathfrak{D}},t);{\mathbb{I}}\big){{\mathscr{B}}}(x)\nonumber\\&=
U_{a}(x,t)+{\bm{\bm{\alpha}}}{{\psi}}\big(|\bm{\mathrm{Re}}(\bm{\mathfrak{D}},t);\mathbb{I}\big){{\mathscr{B}}}(x)\\&
{{\mathscr{U}}}_{b}(y,t)=U_{a}(y,t)+{\bm{\bm{\alpha}}}U_{a}(y,t){\bm{\bm{\alpha}}}{{\psi}}\big(|\bm{\mathrm{Re}}(\bm{\mathfrak{D}},t);{\mathbb{I}}\big){{\mathscr{B}}}(y)\nonumber\\&
=U_{a}(y,t)+{\bm{\bm{\alpha}}}{{\psi}}\big(|\bm{\mathrm{Re}}(\bm{\mathfrak{D}},t);\mathbb{I}\big){{\mathscr{B}}}(y)
\end{align}
If $x\in\Im_{1}\subset\bm{\mathfrak{D}}$ and $y\in\Im_{2}{\prime}\subset\bm{\mathfrak{D}}$ are two curves or knots then the circulations with respect to the underlying NS
flows $U_{a}(x,t)$ and $U_{a}(x,t)$ are
\begin{align}
{\bm{\mathrm{C}}}(\Im_{1};t)={\int}_{\Im_{1}}U_{a}(x,t)dx^{a},~~~~{\bm{\mathrm{C}}}(\Im_{2};t)
={\int}_{\Im_{2}}U_{b}(y,t)dy^{b}
\end{align}
The corresponding stochastic circulations are then
\begin{align}
{{\mathscr{C}}}(\Im_{1};t)&={\int}_{\Im_{1}}{{\mathscr{U}}}_{a}(x,t)dx^{a}={\int}_{\Im_{1}}U_{a}(x,t)dx^{a}
+{\int}_{\Im_{1}}U_{a}(x,t){\bm{\bm{\alpha}}}{{\psi}}\big(|\bm{\mathrm{Re}}(\bm{\mathfrak{D}},t);{\mathbb{I}}\big){{\mathscr{B}}}(x)
dx^{a}\nonumber\\&=
{\int}_{\Im_{1}}U_{a}(x,t)dx^{a}+{\bm{\bm{\alpha}}}{{\psi}}\big(|\bm{\mathrm{Re}}(\bm{\mathfrak{D}},t);{\mathbb{I}}\big)
{\int}_{\Im_{1}}{{\mathscr{B}}}(x)dx^{a}\nonumber\\&
\equiv{\bm{\mathrm{C}}}(\Im_{1};t)+{\bm{\bm{\alpha}}}{{\psi}}\big(|\bm{\mathrm{Re}}(\bm{\mathfrak{D}},t);{\mathbb{I}}\big)
{\int}_{\Im_{1}}{{\mathscr{B}}}(x)dx^{a}
\end{align}
and
\begin{align}
{{\mathscr{C}}}(\Im_{2};t)&={\int}_{\Im_{1}}{{\mathscr{U}}}_{b}(y,t)dy^{a}={\int}_{\Im_{1}}U_{b}(y,t)dy^{b}
+{\int}_{\Im_{1}}U_{b}(y,t){\bm{\bm{\alpha}}}{{\psi}}\big(|\bm{\mathrm{Re}}(\bm{\mathfrak{D}},t);\mathbb{I}\big)
{{\mathscr{B}}}(y)dy^{b}\nonumber\\&=
{\int}_{\Im_{1}}U_{a}(x,t)dx^{b}+{\bm{\bm{\alpha}}}{{\psi}}\big(|\bm{\mathrm{Re}}(\bm{\mathfrak{D}},t);{\mathbb{I}}\big)
{\int}_{\Im_{1}}{{\mathscr{B}}}(y)dy^{b}\nonumber\\&
{\bm{\mathrm{C}}}(\Im_{2};t){\bm{\bm{\alpha}}}{{\psi}}\big(|\bm{\mathrm{Re}}(\bm{\mathfrak{D}},t);{\mathbb{I}}\big)
{\int}_{\Im_{1}}{{\mathscr{B}}}(y)dy^{b}
\end{align}
with stochastic expectations
\begin{align}
{\bm{\mathbb{E}}}\langle{{\mathscr{C}}}(\Im_{1};t)\rangle&
={\bm{\mathbb{E}}}\langle{\int}_{\Im_{1}}{{\mathscr{U}}}_{a}(x,t)dx^{a}
\rangle\nonumber\\&=
{\int}_{\Im_{1}}U_{a}(x,t)dx^{a}+{\int}_{\Im_{1}}U_{a}(x,t){\bm{\bm{\alpha}}}{{\psi}}
\big(|\bm{\mathrm{Re}}(\bm{\mathfrak{D}},t);{\mathbb{I}}\big){\bm{\mathbb{E}}}\langle{{\mathscr{B}}}(x)\rangle dx^{a}\nonumber\\&=
{\int}_{\Im_{1}}U_{a}(x,t)dx^{a}+{\bm{\bm{\alpha}}}{{\psi}}\big(|\bm{\mathrm{Re}}(\bm{\mathfrak{D}},t);{\mathbb{I}}\big)
{\int}_{\Im_{1}}{\bm{\mathbb{E}}}\langle{{\mathscr{B}}}(x)\rangle dx^{a}\nonumber\\&
={\bm{\mathrm{C}}}(\Im_{1};t)+{\bm{\bm{\alpha}}}{{\psi}}\big(|\bm{\mathrm{Re}}(\bm{\mathfrak{D}},t);{\mathbb{I}}\big)
{\int}_{\Im_{1}}{\bm{\mathbb{E}}}\langle{{\mathscr{B}}}(x)\rangle dx^{a}={\bm{\mathrm{C}}}(\Im_{1};t)
\end{align}
and
\begin{align}
{\bm{\mathbb{E}}}\langle{{\mathscr{C}}}(\Im_{2};t)\rangle&={\bm{\mathbb{E}}}\langle{\int}_{\Im_{1}}{{\mathscr{U}}}_{b}(y,t)dy^{ab}
\rangle\nonumber\\&=
{\int}_{\Im_{1}}U_{b}(y,t)dy^{b}+{\int}_{\Im_{1}}U_{a}(y,t){\bm{\bm{\alpha}}}{{\psi}}\big(|\bm{\mathrm{Re}}(\bm{\mathfrak{D}},t);
{\mathbb{I}}\big){\bm{\mathbb{E}}}\langle{{\mathscr{B}}}(y)\rangle dy^{a}\nonumber\\&=
{\int}_{\Im_{1}}U_{b}(y,t)dx^{b}+{\bm{\bm{\alpha}}}{{\psi}}\big(|\bm{\mathrm{Re}}(\bm{\mathfrak{D}},t);{\mathbb{I}}\big)
{\int}_{\Im_{1}}{\bm{\mathbb{E}}}\langle{{\mathscr{R}}}_{b}(y)\rangle dy^{b}\nonumber\\&
={\mathbf{C}}(\Im_{2};t)+{\bm{\bm{\alpha}}}{{\psi}}\big(|\bm{\mathrm{Re}}(\bm{\mathfrak{D}},t);{\mathbb{I}}\big)
{\int}_{\Im_{1}}{\bm{\mathbb{E}}}\langle{{\mathscr{R}}}(y)\rangle dy^{b}={\bm{\mathrm{C}}}(\Im_{2};t)
\end{align}
The correlation of the stochastic vortices is then
\begin{align}
&{\bm{\mathsf{F}}}(\Im_{1},\Im_{2};t)={\bm{\mathbb{E}}}\langle{{\mathscr{C}}}(\Im_{1};t){{\otimes}}{{\mathscr{C}}}(\Im_{2};t)\rangle\nonumber\\&
={\mathbb{E}}\left\langle{\int}\!\!{\int}_{\Im_{1},\Im_{2}}{{\mathscr{U}}}_{a}(x,t){{\otimes}}{{\mathscr{U}}}_{b}(y,t)dx^{a}dy^{b}\right\rangle
\equiv{\int}\!\!{\int}_{\Im_{1},\Im_{2}}{\bm{\mathbb{E}}}\langle{{\mathscr{U}}}_{a}(x,t){{\otimes}}{{\mathscr{U}}}_{b}(y,t)\rangle dx^{a}dy^{b}\nonumber\\&
={\int}\!\!{\int}_{\Im_{1},\Im_{2}}U_{a}(x,t)U_{b}(y,t)1+{\bm{\bm{\alpha}}}
{{\psi}}\big(|\bm{\mathrm{Re}}(\bm{\mathfrak{D}},t);\mathbb{I}\big)
\mathsf{C}\exp(-\|x-y\|^{2}\lambda^{-2})
dx^{a}dy^{b}\nonumber\\&={\int}\!\!{\int}_{\Im_{1},\Im_{2}}U_{a}(x,t)U_{b}(y,t)dx^{a}dy^{b}
\end{align}
\begin{figure}[htb]
\begin{center}
\includegraphics[height=2.0in,width=3.5in]{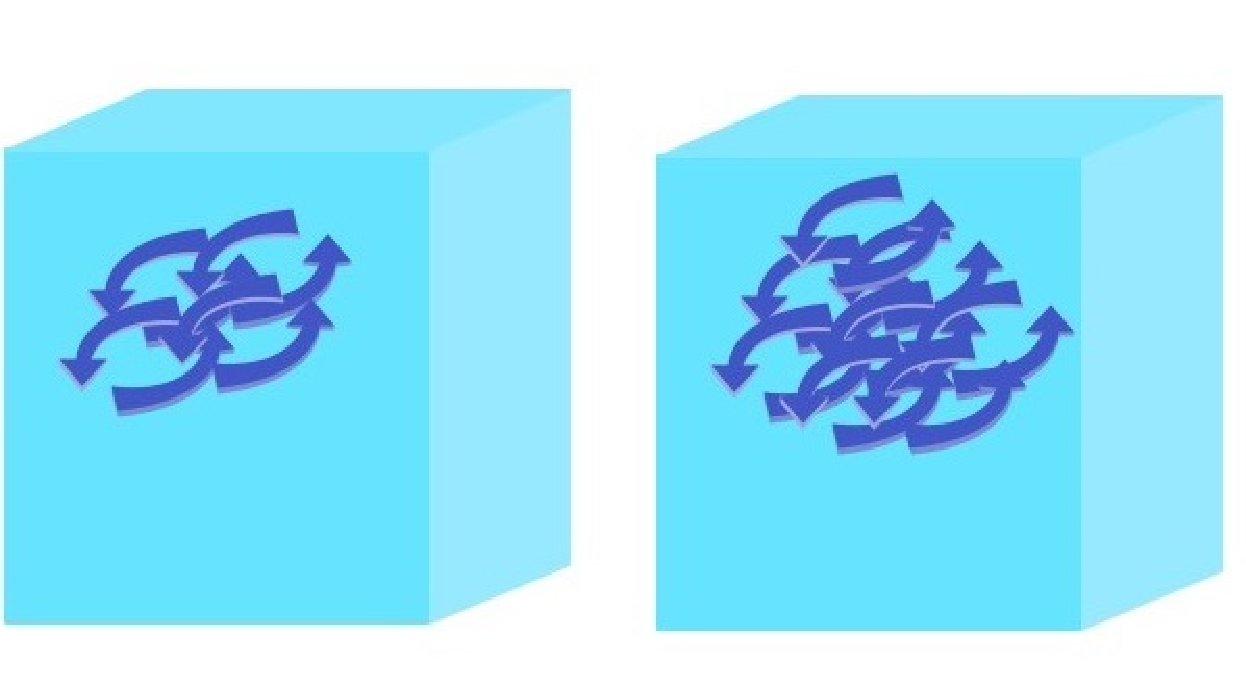}
\caption{Evolution of vortex tangles or correlations within $\bm{\mathfrak{D}}$ with \newline increasing $\bm{\mathrm{Re}}(\bm{\mathfrak{D}},t)$.}
\end{center}
\end{figure}
For two closed knots or loops $(\mathfrak{K}_{1},\mathfrak{K}_{2})\in\bm{\mathfrak{D}}$
\begin{align}
&{\mathbf{F}}(\mathfrak{K}_{1},\mathfrak{K}_{2};t)={\bm{\mathbb{E}}}\langle{{\mathscr{C}}}(\mathfrak{K}_{1};t){{\otimes}}{{\mathscr{C}}}(\mathfrak{K}_{2};t)\rangle\nonumber\\&
={\bm{\mathbb{E}}}\left\langle{\oint}\!\!{\oint}_{\mathfrak{K}_{1},\mathfrak{K}}{{\mathscr{U}}}_{a}(x,t){{\otimes}}{{\mathscr{U}}}_{b}(y,t)dx^{a}dy^{b}\right\rangle
\equiv{\oint}\!\!{\oint}_{\mathfrak{K}_{1},\mathfrak{K}}{\bm{\mathbb{E}}}\langle{{\mathscr{U}}}_{a}(x,t)
{{\otimes}}{{\mathscr{U}}}_{b}(y,t)\rangle dx^{a}dy^{b}\nonumber\\&
={\oint}\!\!{\oint}_{\mathfrak{K}_{1},\mathfrak{K}}U_{a}(x,t)U_{b}(y,t)1+|{\bm{\bm{\alpha}}}^{2}
{{\psi}}\big(|\bm{\mathrm{Re}}(\bm{\mathfrak{D}},t);{\mathbb{I}}\big)\mathsf{C}\exp(-\|x-y\|^{2}\lambda^{-2})
dx^{a}dy^{b}\nonumber\\&={\oint}\!\!{\oint}_{\mathfrak{K}_{1},\mathfrak{K}}U_{a}(x,t)U_{b}(y,t)dx^{a}dy^{b}\nonumber\\&
+{\oint}\!\!{\oint}_{\mathfrak{K}_{1},\mathfrak{K}}
U_{a}(x,t)U_{b}(y,t)|{\bm{\bm{\alpha}}}^{2}
{{\psi}}\big(|\bm{\mathrm{Re}}(\bm{\mathfrak{D}},t);{\mathbb{I}}
\big)
{\mathsf{C}}\exp(-\|x-y\|^{2}\lambda^{-2})dx^{a}dy^{b}\nonumber\\&
={\bm{\mathrm{C}}}(\mathfrak{K}_{1},t){\bm{\mathrm{C}}}(\mathfrak{K},t)\nonumber\\&
+{\oint}\!\!{\oint}_{\mathfrak{K}_{1},\mathfrak{K}}U_{a}(x,t)U_{b}(y,t){\bm{\bm{\alpha}}}^{2}
{{\psi}}\big(|\bm{\mathrm{Re}}(\bm{\mathfrak{D}},t);{\mathbb{I}}\big)
{\mathsf{C}}\exp(-\|x-y\|^{2}\lambda^{-2})
dx^{a}dy^{b}
\end{align}
\end{prop}
An immediate corollary is that the vortices become uncorrelated or "untangled" for large separations when $\|x-y\|\gg \lambda$ since the exponential term decays rapidly.
\begin{cor}
when $\|x-y\|\gg \lambda$.
\end{cor}
Another corollary is that the vortices are uncorrelated when the volume-averaged Reynold's number within $\bm{\mathfrak{D}}$ falls below the critical value $\bm{\mathrm{R}}_{c}(\bm{\mathfrak{D}})$ and smooth
or laminar flow is restored.
\begin{cor}
If $\bm{\mathrm{Re}}(\bm{\mathfrak{D}},t)<\bm{\mathrm{Re}}_{c}(\bm{\mathfrak{D}})$ for some $t>0$ or all $t>0$ then the functional ${\psi}$ vanishes so that ${{\psi}}\big(|\bm{\mathrm{Re}}(\bm{\mathfrak{D}},t);{\mathbb{I}}\big)=0$.
Then
\begin{align}
&{\bm{\mathsf{F}}}(\Im_{1},\Im_{2};t)={\bm{\mathbb{E}}}\langle{{\mathscr{C}}}(\Im_{1};t){{\otimes}}{{\mathscr{C}}}(\Im_{2};t)\rangle\nonumber\\&
={\int}\!\!\!\!{\int}_{\Im_{1}\Im_{2}}U_{a}(x,t)U_{b}(y,t)dx^{a}dy^{b}
\equiv{\bm{\mathrm{C}}}(\Im_{1},t){\bm{\mathrm{C}}}(\Im_{2},t)
\end{align}
\end{cor}The higher-order correlations are similarly defined
\begin{prop}
The 3rd-order correlation for tangle of 3 vortices is
\begin{align}
&{\bm{\mathsf{F}}}(\Im_{1},\Im_{2},\Im_{3};t)=
{\bm{\mathbb{E}}}\langle{{\mathscr{C}}}(\Im_{1};t){{\otimes}}{{\mathscr{C}}}(\Im_{2};t)\rangle
\nonumber\\&{\int\!\!\!\!\int\!\!\!\!\int}_{\Im_{1},\Im_{2},\Im_{3}}{\bm{\mathbb{E}}}\langle{{\mathscr{U}}}_{a}(x,t){{\otimes}}~{{\mathscr{U}}}_{b}(y,t){{\otimes}}
~{{\mathscr{U}}}_{c}(z,t)\rangle dx^{a}dy^{b}dz^{c}\nonumber\\&
+|{\bm{\bm{\alpha}}}|^{3}{\int\!\!\!\!\int\!\!\!\!\int}_{\Im_{1},\Im_{2},\Im_{3}}U_{a}(x,t)U_{b}(y,t)U_{c}(z,t){{\psi}}\big(|\bm{\mathrm{Re}}(\bm{\mathfrak{D}},t);{\mathbb{I}}\big)^{3}
{\mathsf{C}}\exp(-\|y-z\|^{2}\lambda^{-2})dx^{a}dy^{b}dz^{c}
\nonumber\\&
+|{\bm{\bm{\alpha}}}|^{3}{\int\!\!\!\!\int\!\!\!\!\int}_{\Im_{1},\Im_{2},\Im_{3}}U_{a}(x,t)U_{b}(y,t)U_{c}(z,t){{\psi}}\big(|\bm{\mathrm{Re}}(\bm{\mathfrak{D}},t);{\mathbb{I}}\big)^{3}
{\mathsf{C}}\exp(-\|x-y\|^{2}\lambda^{-2})dx^{a}dy^{b}dz^{c}
\nonumber\\&
+|{\bm{\bm{\alpha}}}|^{3}{\int\!\!\!\!\int\!\!\!\!\int}_{\Im_{1},\Im_{2},\Im_{3}}U_{a}(x,t)U_{b}(y,t)U_{c}(z,t){{\psi}}\big(|\bm{\mathrm{Re}}(\bm{\mathfrak{D}},t);{\mathbb{I}}\big)^{3}
{\mathsf{C}}\exp(-\|x-z\|^{2}\lambda^{-2})dx^{a}dy^{b}dz^{c}
\end{align}
The proof simply utilises the results of Section 3 for velocity correlations.

If $\|x-y\|\gg\lambda,\|x-\mathbf{z}\|\gg \lambda,\|y-\mathbf{z}\|\gg\lambda$ then this reduces to
\begin{align}
&{\bm{\mathsf{F}}}(\Im_{1},\Im_{2},\Im_{3};t)={\bm{\mathbb{E}}}\langle{{\mathscr{C}}}(\Im_{1};t){{\otimes}}{{\mathscr{C}}}(\Im_{2};t)
{{\otimes}}{{\mathscr{C}}}(\Im_{3};t)\rangle\nonumber\\&=
{\bm{\mathbb{E}}}\left\langle{\int\!\!\!\!\int\!\!\!\!\int}_{\Im_{1},\Im_{2},\Im_{3}}{{\mathscr{U}}}_{a}(x,t)
{{\otimes}}{{\mathscr{U}}}_{b}(y,t){{\otimes}}{{\mathscr{U}}}_{c}(z,t)dx^{a}dy^{b}dz^{c}\right\rangle\nonumber\\&=
{\int\!\!\!\!\int\!\!\!\!\int}_{\Im_{1},\Im_{2},\Im_{3}}{\bm{\mathbb{E}}}\langle{{\mathscr{U}}}_{a}(x,t)
{{\otimes}}{{\mathscr{U}}}_{b}(y,t){{\otimes}}{{\mathscr{U}}}_{c}(z,t)\rangle dx^{a}dy^{b}dz^{c}\nonumber\\&
\longrightarrow{\int\!\!\!\!\int\!\!\!\!\int}_{\Im_{1},\Im_{2},\Im_{3}}U_{a}(x,t)U_{b}(y,t)U_{z}(z,t)dx^{a}dy^{b}dz^{c}
\end{align}
\end{prop}
\subsection{Hopf-like functional integrals for turbulent flows}
Hopf constructed a functional integral to describe a turbulent flow \textbf{[88-91]}. By applying the Navier-Stokes equation to the moment-generating functional for
velocity, a nonlinear differential equation describing a single flow realization is transformed to a linear functional integro-differential equation governing
an ensemble of flows. However, this theory remains underdeveloped mathematically and is hampered by a lack of solutions or methods to obtain solutions.

To see how the Hopf functional integral arises, consider the Burgers equation in 1D, which is the NS equation with the pressure term dropped so that
\begin{align}
\frac{\partial}{\partial t}U(x,t)
+U({x},t)\frac{\partial}{\partial x}U(x,t)
=\nu\frac{\partial^{2}}{\partial x^{2}},~~x~\in\mathbb{R},t>0
\end{align}
As mentioned, this equation has been extensively studied. The Hopf characteristic functional for the velocity is then
\begin{align}
{H}[{{f}}(x)]=\left\langle\exp\left(i\int_{-\infty}^{\infty}{f}(x)U(x)dx\right)\right\rangle
\end{align}
where ${f}$ is a 'test function' and $\langle\bullet\rangle$ is a generic 'ensemble average'. Then $H[{f}(x)]$ satisfies the
Hopf functional differential equation
\begin{align}
&\frac{\partial}{\partial t}{H}[{f}(x)]={\mathcal{L}}{H}[{f}(x)]\nonumber\\&
=\frac{i}{2}\int{\psi}(x)\frac{\partial}{\partial x}\frac{\delta^{2}H}{\delta{\psi}(x)^{2}}dx+\nu\int{\psi}(x)\frac{\partial^{2}}{\partial x^{2}}
\frac{\delta H}{\delta{\psi}(x)}dx
\end{align}
where $\delta/\delta{\psi}(x)$ is a functional derivative. It is required that ${H}[{f}(x)]=1$ at
${f}(x)=0$ and that positive-definite condition is imposed
\begin{align}
\sum_{\ell=1}^{N}\sum_{k=1}^{N}{H}[{f}_{c}(x)-{{\psi}}_{\ell}(x)]
\mathrm{C}_{c}\mathrm{C}_{p}=\left\langle\left|\sum_{k=1}^{N}{\mathrm{C}}_{c}\exp\left(i\int U(x,t)f(x)dx\right)\right|^{2}\right\rangle
\end{align}
for constants $\mathrm{C}_{c},\mathrm{C}_{l}$ and $N=1,2,3...$. Although this FDE is linear, there is no known general method to find solutions. However,
if the nonlinear term in the Burgers equation is very small or set to zero then the Burgers' equations reduces to the form of the heat equation
\begin{align}
\frac{\partial}{\partial t}U(x,t)=\nu\frac{\partial^{2}}{\partial x^{2}}U(x,t)
\end{align}
The HFDE then reduces to
\begin{align}
\frac{\partial}{\partial t}{H}[{f}(x)]=\nu\int{f}(x)\frac{\partial^{2}}{\partial x^{2}}
\frac{\delta{H}}{\delta{f}(x)}dx
\end{align}
which has a solution in terms of the heat kernel
\begin{align}
{H}[{f}(x),t]=H_{o}\left[\frac{1}{\sqrt{4\pi\nu t}}
\int\exp\left( -\frac{|x-y|^{2}}{4\pi\nu}\right){{\psi}}(y)dy\right]
\end{align}
This is the Hopf-Titt solution \textbf{[91]}. It is also well known that a Cole-Hopf transform takes the Burgers equation to the heat equation form. By analogy, one can tentatively propose a Hopf-like functional integral for the random field ${\mathscr{U}}_{a}(x,t)$.
\begin{prop}
Let $\bm{\mathfrak{D}}\subset\mathbb{R}^{3}$ and let $U_{a}(x,t)$ be a smooth underlying flow evolving by the NS equations. If $\Im\in\bm{\mathfrak{D}}$ is
a curve or knot then the circulation is $\mathbf{C}[U]=\int_{\Im}U_{a}(x,t)dx^{a}$. Define a Hopf functional or "Wilson line" as
\begin{align}
{H}[U]=\exp\left(i\int_{\Im}~{f}(x,t)U_{a}(x,t)dx^{a}\right)=\exp(\mathbf{C}[U])
\end{align}
Now given the turbulent flow or random field ${\mathscr{U}}_{a}(x,t)$ then define a stochastic Hopf-like integral as
\begin{align}
&{\mathscr{H}}[{{\mathscr{U}}}_{a}]=\exp(\left(i{\int}_{\Im}\left({f}(x,t)U_{a}(x,t)
+{\bm{\bm{\alpha}}}{f}(x,t)U_{a}(x,t){{\psi}}\big(\mathbf{Re}(\bm{\mathfrak{D}},t);
{\mathbb{I}}\big){{\mathscr{B}}}(x)dx^{a}
\right)\right)\nonumber\\&
=\exp\left(i{\int}_{\Im}\left({f}(x,t)U_{a}(x,t)dx^{a}
+{\bm{\bm{\alpha}}}{f}(x,t){\int}_{\Im}U_{a}(x,t)
{{\psi}}\big(\mathbf{Re}(\bm{\mathfrak{D}},t);{\mathbb{I}}\big){{\mathscr{B}}}(x)dx^{a}
\right)\right)\nonumber\\&
=\exp\left(i{\int}_{\Im}{f}(x,t)U_{a}(x,t)dx^{a}\right)\exp(\left(
{i\bm{\bm{\alpha}}}{f}(x,t){\int}_{\Im}U_{a}(x,t)
{{\psi}}\big(\mathbf{Re}(\bm{\mathfrak{D}},t);{\mathbb{I}}\big){{\mathscr{B}}}(x)dx^{a}
\right)\nonumber\\&
={H}[U]\exp\left(
{i\bm{\bm{\alpha}}}{{f}}(x,t){\int}_{\Im}U_{a}(x,t){{\psi}}
\big(\mathbf{Re}(\bm{\mathfrak{D}},t);{\mathbb{I}}\big){{\mathscr{B}}}(x)dx^{a}\right)
\equiv\nonumber\\&
=H[U]\exp\left(
{i\bm{\bm{\alpha}}}{f}(x,t){\int}_{\Im}
{{\psi}}\big(\mathbf{Re}(\bm{\mathfrak{D}},t);{\mathbb{I}}\big)
{{\mathscr{B}}}(x)U_{a}(x,t)dx^{a}\right)
\end{align}
The stochastic expectation is then
\begin{align}
&\mathbb{H}[{{\mathscr{U}}}_{a};\Im]
={\bm{\mathbb{E}}}\langle{\mathscr{H}}[{{\mathscr{U}}}_{a}]\rangle\nonumber\\&
=H[U;\Im]\times{\bm{\mathbb{E}}}\left\langle
\exp\left({i\bm{\bm{\alpha}}}{{{f}}}(x,t){\int}_{\Im}
{{\psi}}\big(\mathbf{Re}(\bm{\mathfrak{D}},t);{\mathbb{I}}\big)
{{\mathscr{B}}}(x)U_{a}(x,t)dx^{a}\right)\right\rangle
\end{align}
For a close curve or knot $\mathfrak{K}$
\begin{align}
&\mathbb{H}[{{\mathscr{U}}}_{a};\mathfrak{K}]
={\bm{\mathbb{E}}}\langle{\mathscr{H}}[{{\mathscr{U}}}_{a}]\rangle\nonumber\\&
={H}[U;\mathfrak{K}]\times{\bm{\mathbb{E}}}\left\langle
\exp\left({i\bm{\bm{\alpha}}}{{{f}}}(x,t){\oint}_{\Im}
{{\psi}}\big(\mathbf{Re}(\bm{\mathfrak{D}},t);{\mathbb{I}}\big)
{{\mathscr{B}}}(x)U_{a}(x,t)dx^{a}\right)\right\rangle
\end{align}
\end{prop}
\begin{cor}
When the average Reynolds number within $\bm{\mathfrak{D}}$ is below the critical value then the turbulence subsides and $\big\lbrace{{\psi}}\big(\mathbf{Re}(\bm{\mathfrak{D}},t);{\mathbb{I}}\big)
\big\rbrace=0$. Then the Hopf-like integral reduces to
\begin{align}
&{\bm{\mathbb{H}}}[{{\mathscr{U}}}_{a};\Im]={H}[U;\Im]\nonumber\\&
{\bm{\mathbb{H}}}[{{\mathscr{U}}}_{a};\mathcal{O}]={H}[U;\mathcal{O}]\nonumber
\end{align}
\end{cor}
It should be possible to expand the stochastic integral out as a 'cluster expansion' or Van Kampen-type expansion of cumulants (REF).
This is involved however, and may be the subject of a future article.
\clearpage
\appendix
\renewcommand{\theequation}{\Alph{section}.\arabic{equation}}
\section{\textbf{Proof of Lemma}}
\begin{proof}
To prove (2.53)
\begin{align}
&\bm{\mathbb{E}}\langle{\mathscr{G}}(x){{\otimes}}~{\mathscr{G}}(y)\rangle=
{\int}_{\mathbf{K}^{3}}d^{3}k\bm{\Phi}(k)\exp(ik_{a}(x-y)^{a})\nonumber\\&=
\bm{\bm{\alpha}}{\int}_{\mathbf{K}^{3}}d^{3}k\exp(ik_{a}(x-y)^{a})=\bm{\bm{\alpha}} \delta^{3}(x-y)
\end{align}
To prove (2.54)
\begin{align}
&\bm{\mathbb{E}}\langle{\mathscr{G}}(x){{\otimes}}~{\mathscr{G}}(y)
\rangle=\int_{\mathbf{K}^{3}}d^{3}\mathbf{k}\bm{\Phi}(\mathbf{k})\exp(ik_{a}(x-y)^{a})\nonumber\\&
=\bm{\bm{\alpha}}\int_{\mathbf{K}^{3}}d^{3}\mathbf{k}\frac{\exp(-\tfrac{1}{4}\ell^{2}\mathbf{k}^{2})}{\mathbf{k}^{2}}\exp(ik_{a}(x-y)^{a})\nonumber\\&
=\bm{\bm{\alpha}}\int_{\mathbf{K}^{3}}d^{3}\mathbf{k}\frac{\exp(-\tfrac{1}{4}\ell^{2}\mathbf{k}^{2}+ik_{a}(x-y)^{a})}{\mathbf{k}^{2}}\nonumber\\&
\equiv \bm{\bm{\alpha}}\int_{-\infty}^{\infty}\int d\mathbf{k} \mathbf{k}^{2} d\omega\frac{\exp\big(-\frac{1}{4}\ell^{2}\mathbf{k}^{2}+ik_{a}(x-y)^{a}\big)}{\mathbf{k}^{2}}\nonumber\\&
\equiv \bm{\bm{\alpha}}\int_{-\infty}^{\infty}\int d\mathbf{k} d\omega\exp\big(-\tfrac{1}{4}\ell^{2}\mathbf{k}^{2}+ik_{a}(x-y)^{a}\big)
\end{align}
The angular integral $\int d\Omega$ just contributes a constant factor which an be absorbed into $\bm{\bm{\alpha}}$ giving
\begin{align}
\bm{\mathbb{E}}\langle{\mathscr{G}}(x){{\otimes}}~{\mathscr{G}}(y)
\rangle=\equiv\overline{\bm{\bm{\alpha}}}\int_{-\infty}^{\infty}d\mathbf{k} \exp\big(-\tfrac{1}{4}\ell^{2}k_{a}^{2}+ik_{a}(x-y)^{a}\big)\nonumber\\&
\end{align}
The square is then completed on the integrand
\begin{align}
-\tfrac{1}{4}\ell^{2}k_{a}^{2}-ik_{a}(x-y)^{a}=-ak_{a}^{2}+bk_{a}=-a\left(k_{a}-\frac{b}{2a}\right)^{2}+\frac{b}{4a}
\end{align}
\begin{align}
&{\bm{\mathbb{E}}}\langle{\mathscr{G}}(x){{\otimes}}~{\mathscr{G}}(y)
\rangle=\equiv\overline{\bm{\bm{\alpha}}}{\int}_{-\infty}^{\infty}dk \exp\big(-\tfrac{1}{4}\ell^{2}k_{a}^{2}+ik_{a}(x-y)^{a}\big)\nonumber\\&
=\overline{\bm{\bm{\alpha}}}{\int}_{-\infty}^{\infty}dk \exp\left(k_{a}-\frac{(x-y)_{a}}{\ell^{2}}\right)^{2}\exp\left(-\frac{(x-y)_{a}(x-y)^{a})}{\ell^{2}}\right)\nonumber\\&
\equiv\overline{\bm{\bm{\alpha}}}\exp\left(-\frac{(x-y)_{a}(x-y)^{a})}{\ell^{2}}\right)\int_{-\infty}^{\infty}dk \exp\left(k_{a}-\frac{(x-y)_{a}}{\ell^{2}}\right)^{2}\nonumber\\&
\equiv \overline{\bm{\bm{\alpha}}}\exp\left(-\frac{\|x-y\|^{2}}{\ell^{2}}\right)\int_{-\infty}^{\infty}dk \exp\left(k_{a}-\frac{(x-y)_{a}}{\ell^{2}}\right)^{2}
\end{align}
Setting $\overline{k_{a}}=k_{a}-\frac{(x-y)_{a}}{\ell^{2}}$ then $d\overline{k}=dk$ so that
 \begin{align}
&\bm{\mathbb{E}}\langle{\mathscr{G}}(x){{\otimes}}~{\mathscr{G}}(y)
\rangle
=\overline{\mathfrak{K}}\exp\left(-\frac{\|x-y\|^{2}}{\ell^{2}}\right)\int_{-\infty}^{\infty}dk \exp\left(k_{a}-\frac{(x-y)_{a}}{\ell^{2}}\right)^{2}\nonumber\\&=
\overline{\mathfrak{K}}\exp\left(-\frac{\|x-y\|^{2}}{\ell^{2}}\right)\int_{-\infty}^{\infty}d\overline{k}\exp(\overline{k}^{2})
\end{align}
The basic Gaussian integral and this constant factor can again be absorbed into $\overline{\bm{\bm{\alpha}}}$ so that
\begin{align}
{\mathbb{E}}\langle{\mathscr{G}}(x){{\otimes}}~{\mathscr{G}}(y)
\rangle
={\mathbb{E}}\langle{\mathscr{G}}(x){{\otimes}}~{\mathscr{G}}(y)=\exp\left(-\frac{\|x-y\|^{2}}{\ell^{2}}\right)
\end{align}
\end{proof}
\section{\textbf{Derivatives of BF fields}}
To derive the spatial derivatives $\nabla_{a}{{\mathscr{B}}}(x,t)),\nabla_{a}\nabla_{b}-+-(x,t))$ one first requires the following preliminary lemmas.
\begin{lem}
Let $(x,y)\in\mathbb{Q}\in\bm{\mathbb{R}}^{3}$ then
\begin{align}
&\nabla_{a}^{(x)}\|x-y\|^{2}\equiv\frac{\partial}{\partial x_{a}}\|x-y\|^{2}=6(x-y)_{a}\nonumber\\&
\nabla_{b}^{(y)}\|x-y\|^{2}\equiv\frac{\partial}{\partial y_{b}}\|x-y\|^{2}=-6(x-y)_{b}=6(y-x)_{b}
\end{align}
and the second derivative is
\begin{align}
\nabla_{a}^{x}\nabla_{b}^{(y)}\|x-y\|=18\delta_{ab}
\end{align}
\begin{proof}
The derivative is quite obvious but can still be carried out in detail
\begin{align}
&\nabla_{a}^{(x)}\|x-y\|^{2}=\nabla_{a}^{(x)}(x-y)_{a}(x-y)^{a}=\nabla_{a}^{(x)}(x-y)_{a}(x-y)^{a}\nonumber\\&
=\nabla_{a}^{(x)}(x_{a}-y_{a})(x^{a}-y^{a})=\nabla_{a}^{(x)}(x_{a}x^{a}-x_{a}y^{a}-x^{a}y_{a}-y_{a}y^{a})
\nonumber\\&=\nabla_{a}^{(x)}(x^{a}x^{a}\delta_{ii}-x^{a}y^{a}\delta_{ii}-x^{a}y^{a}\delta_{ii}-y^{a}y^{a}\delta_{ii})\nonumber\\&
=2x^{a}\nabla_{a}^{(x)}x^{a}\delta_{ii}-2y^{a}\nabla_{a}x^{a}\delta_{ii}-2y^{a}\nabla^{a} x^{a}\delta_{ii}=2\nabla_{a}^{(x)}(x^{a}-y^{a})\delta_{ii}=6(x-y)_{a}
\end{align}
since $div x=\nabla_{a}^{(x)}x^{a}=3$ in $\mathbf{R}^{3}$.
\end{proof}
Similarly, $\nabla_{b}^{(y)}\|x-y\|^{2}\equiv\frac{\partial}{\partial y_{b}}\|x-y\|^{2}=-6(x-y)_{b}=6(y-x)_{b}$. Then (B2) follows immediately.
\end{lem}
\begin{lem}
Given the Gaussian kernel ${\mathsf{\bm{\Xi}}}(x,y;\lambda)=\mathsf{C}\exp(-\|x-y\|\lambda^{-2})$, the first derivatives are
\begin{align}
&{\mathsf{\Phi}}_{a}^{(x)}(x,y;\lambda)=\nabla_{a}^{(x)}{\mathsf{\bm{\Xi}}}(x,y;\lambda)=-\frac{6(x-y)_{a}}{\lambda^{2}}{{\mathsf{\bm{\Xi}}}}(x,y;\lambda)\\&
{\mathsf{\Phi}}_{a}^{(y)}(x,y;\lambda)=\nabla_{b}^{(y)}{\mathsf{\bm{\Xi}}}(x,y;\lambda)=+\frac{6(x-y)_{b}}{\lambda^{2}}{\mathsf{\bm{\Xi}}}(x,y;\lambda)
\end{align}
and the second derivative is
\begin{align}
&{\mathsf{\Phi}}_{ab}(x,t;\lambda)=\nabla_{a}^{(y)}\nabla_{b}^{(x)}{{\psi}}(x,y;\lambda)=\frac{18\delta_{ab}}{\lambda^{2}}
{\mathsf{\bm{\Xi}}}(x,y;\ell)+\frac{6(x-y)_{a}}{\lambda^{2}}
\nabla_{a}^{(x)}{\mathsf{\Phi}}(x,y;\lambda)\nonumber\\&
=\frac{18\delta_{ab}}{\lambda^{2}}{\mathsf{\Phi}}(x,y;\ell)-\frac{36(x-y)_{a}(x-y)_{b}}{\lambda^{2}}{\mathsf{\Phi}}(x,y;\lambda)
\end{align}
Then also
\begin{align}
\lim_{y\rightarrow x}\nabla_{a}^{x}{\mathsf{\bm{\Xi}}}(x,y;\lambda)=-\lim_{y\rightarrow x}\frac{6(x-y)_{a}}{\lambda^{2}}{\mathsf{\bm{\Xi}}}(x,y;\lambda)=0
\end{align}
\end{lem}
\begin{proof}
Taking the log of both sides of the Gaussian kernel then
\begin{align}
\log {\mathsf{\bm{\Xi}}}(x,y;\lambda)=\log\alpha-\frac{\|x-y\|^{2}}{\lambda^{2}}
\end{align}
so that
\begin{align}
\frac{\nabla_{a}^{(x)}{\mathsf{\bm{\Xi}}}(x,y;\lambda)}{{\mathsf{\bm{\Xi}}}(x,y;\lambda)}=-\frac{\nabla_{a}^{(x)}\|x-y\|}{\lambda^{2}}
=-\frac{6(x-y)_{a}}{\lambda^{2}}
\end{align}
and the derivative is
\begin{align}
\nabla_{a}^{(x)}{\mathsf{\bm{\Xi}}}(x,y;\lambda)=-\frac{\nabla_{a}^{(x)}\|x-y\|}{\lambda^{2}}
=-\frac{6(x-y)_{a}}{\lambda^{2}}{\mathsf{\bm{\Xi}}}(x,y;\lambda)
\end{align}
The 2nd derivative is
\begin{align}
\nabla_{a}^{(x)}\nabla_{b}^{(y)}{\mathsf{\bm{\Xi}}}(x,;\lambda)=\frac{36(x-y)_{a}(x-y)_{b}}{\lambda^{4}}
{\mathsf{\bm{\Xi}}}(x,y;\lambda)+\frac{18\delta_{ab}}{\lambda^{2}}
{\mathsf{\bm{\Xi}}}(x,y;\lambda)
\end{align}
In the limit as $y\rightarrow x$
\begin{align}
&\lim_{y\uparrow x}\nabla_{a}^{(x)}{\mathsf{\bm{\Xi}}}(x,y;\lambda)=-\lim_{y\uparrow x}\frac{6(x-y)_{a}}{\lambda^{2}}{\mathsf{\bm{\Xi}}}(x,y;\lambda)=0\\&
\lim_{y\uparrow x}\nabla_{a}^{(x)}\nabla_{b}^{(y)}{\mathsf{\bm{\Xi}}}(x,y;\lambda)=\lim_{y\uparrow x}\left(\frac{36(x-y)_{a}(x-y)_{b}}{\lambda^{4}}{\mathsf{\bm{\Xi}}}(x,y;\lambda)+\frac{18\delta_{ab}}{\lambda^{2}}
{\mathsf{\bm{\Xi}}}(x,y;\lambda)\right)=
\frac{18\alpha}{\lambda^{2}}\delta_{ab}
\end{align}
\end{proof}
\section{\textbf{Stochastic Integration of random fields}}
The integral of a GRSF is defined as the limit of a Riemann sum of the field over the partition of a domain.
\begin{prop}
Let ${\bm{\mathfrak{D}}}\subset{\mathbb{R}}^{n}$ be a (closed) domain with boundary $\partial{\bm{\mathfrak{D}}}$ and $x=(x_{1},...,x_{n})\subset{\bm{\mathfrak{D}}}$. Let $\bm{\mathfrak{D}}=\bigcup_{q=1}^{M}\bm{\mathfrak{D}}_{1}$ be a partition of $\bm{\mathfrak{D}}$ with $\bm{\mathfrak{D}}_{q}\bigcap\bm{\mathfrak{D}}_{q}=\varnothing$ if $p\ne q$. Let $x^{(q)}\in\bm{\mathfrak{D}}_{q}$ for all $q=1...M$. Note $x^{(q)}\equiv (x_{1}^{(q)},...x_{n}^{(q)})$. Then $x^{(1)}\in\bm{\mathfrak{D}}_{1},x^{(2)}\in\bm{\mathfrak{D}}_{2},
...,x^{(M)}\in\bm{\mathfrak{D}}$. Let $\mathrm{Vol}(\bm{\mathfrak{D}}_{q})$ be the volume of the partition $\bm{\mathfrak{D}}_{q}$ so that $\mathrm{vol}(\bm{\mathfrak{D}})=\sum_{q}^{M}\mathrm{Vol}(\bm{\mathfrak{D}})$. Similarly, if $\partial\bm{\mathfrak{D}}$ is the surface or boundary of
${\bm{\mathfrak{D}}}$ then let
\begin{equation}
\bm{\mathfrak{D}}=\bigcup_{\xi=1}^{M}\partial\bm{\mathfrak{D}}{q}
\end{equation}
be a partition of $\partial\bm{\mathfrak{D}}$ with $\partial\bm{\mathfrak{D}}_{\xi}$ be a partition of the boundary or surface into N constituents. Let $\bm{x}^{(q)}\in\partial\bm{\mathfrak{D}}_{q}$ for all $q=1...M$. Note $x^{(q)}\equiv (x_{1}^{(q)},...x_{n}^{(q)})$. Then $x^{(1)}\in\partial\bm{\mathfrak{D}}_{1}, x^{(2)}\in\partial\bm{\mathfrak{D}}_{2},...,x^{(H)} \in\partial\bm{\mathfrak{D}}$. Let $\|{\mathscr{W}}_{q}\|\equiv\mu(\partial\bm{\mathfrak{D}}_{q})$ be the surface area of the partition $\partial\bm{\mathfrak{D}}_{q}$ so that $\mu\bm{\mathfrak{D}}=\bigcap_{q=1}^{M}A(\partial\bm{\mathfrak{D}}_{q}$. The total volume and area of $\bm{\mathfrak{D}}$ is
\begin{align}
\mathrm{Vol}(\bm{\mathfrak{D}})=\sum_{q=1}^{M}(\bm{\mathfrak{D}}_{q})
=\sum_{q=1}^{M}v(\bm{\mathfrak{D}}_{q})
\end{align}
\begin{align}
\mathrm{Area}(\partial\bm{\mathfrak{D}})=\sum_{q=1}^{M}\mathrm{Area}(\partial\bm{\mathfrak{D}}_{q})=\sum_{q=1}^{M}\mathrm{Area}(\bm{\mathfrak{D}}_{q})
\end{align}
Given the probability triplet $(\Omega,{\psi},\mathbb{P})$ then a Gaussian random field on $\bm{\mathfrak{D}}$ for all $x\in\bm{\mathfrak{D}}$ is ${\mathscr{R}}:\omega\times\bm{\mathfrak{D}}\rightarrow {\mathbb{R}}$ and ${\mathscr{R}}(x^{q},\omega)\in{\mathscr{W}}_{q}$ exists for all $x^{(q)}\in \bm{\mathfrak{D}}_{q}$ and $\omega\in\Omega$. The stochastic volume integral and the stochastic surface integral are
\begin{align}
&\int_{\bm{\mathfrak{D}}}{\mathscr{R}}(x;\omega)d U_{n}(x)=\lim_{all~v(\mathscr{W}_{q})\uparrow 0}\sum_{q=1}^{M}{\mathscr{R}}(x^{(q)};\omega)\mathrm{Vol}(\bm{\mathfrak{D}}_{q})\\&
\int_{\partial\bm{\mathfrak{D}}}{\mathscr{R}}(x;\omega)d U_{n-1}(x)
=\lim_{all~\mathrm{Area}(\partial\bm{\mathfrak{D}}_{q})\uparrow 0}\sum_{q=1}^{M}{\mathscr{R}}(x^{(q)};\omega)Area(\partial\bm{\mathfrak{D}}_{q})
\end{align}
When a Gaussian random field is integrated, it is the limit of a linear combination of Gaussian random variables/fields so it is again Gaussian.
\end{prop}
Next, the stochastic expectations or averages are defined.
\begin{prop}
Since
\begin{equation}
{\bm{\mathbb{E}}}\langle\bullet\rangle
=\int_{\bm{\mathfrak{D}}}(\bullet)d{\mathbb{P}}(\omega)
\end{equation}
the expectation of the volume integral is as follows.
\begin{align}
&{\bm{\mathbb{E}}}\langle\int_{\mathbf{Q}}{\mathscr{R}}(x;\omega)d U_{n}(x)\rangle
\equiv\int\!\!\int_{\bm{\mathfrak{D}}}{{\mathscr{R}}}(x;\omega)dmU_{n}(x) d \mathbb{P}(\omega)\\&
=\lim_{M\uparrow\infty}\lim_{all~\mathrm{Vol}(\bm{\mathfrak{D}}_{q})\uparrow 0}\int_{\Omega}\sum_{q=1}^{M}
{{\mathscr{R}}}(x^{(q)};\omega)\mathrm{Vol}(\bm{\mathfrak{D}}_{q})d\bm{\mathbb{P}}(\omega)=0
\end{align}
which vanishes for GRSFs since $\bm{\mathbb{E}}\big\langle{{\mathscr{R}}}(x^{(q)})\big\rangle=0$. Similarly, for the stochastic surface integrals
\begin{align}
{\bm{\mathbb{E}}}[\int_{\bm{\mathfrak{D}}}{\mathscr{R}}(x;\omega)
d^{n-1}x]
=\lim_{all~area({{\mathscr{R}}}_{q})\uparrow 0}\mathbb{E}\left\langle\sum_{q=1}^{M}
{{\mathscr{R}}}(x^{(q)};\omega)A(\partial{{\mathscr{R}}}_{q})\right\rangle
\end{align}
or
\begin{align}
&{\bm{\mathbb{E}}}\left\langle\int_{\partial{\bm{\mathfrak{D}}}}{\mathscr{R}}(x;\omega)d U_{n-1}(x)
\right\rangle\equiv\int_{\Omega}\int_{\partial{\bm{\mathfrak{D}}}}{\mathscr{R}}(x;\omega)
d^{n-1}x d\bm{\mathsf{P}}(\omega)\\&=\lim_{all~\mu(\partial{\bm{\mathfrak{D}}}_{q})\uparrow 0}\int_{\Omega}\sum_{q=1}^{M}{\mathscr{R}}(x^{(q)};\omega)
Area(\partial{\bm{\mathfrak{D}}}_{q})d{\mathbb{P}}(\omega)=0
\end{align}
\end{prop}
Given an integral (or summation) over a random field or stochastic process, the Fubini theorem states that the expectation of the integral or sum over a random field is equivalent to the integral or sum of the expectation of the field.

\begin{thm}
Let ${\mathscr{R}}(x)$ be a random field not necessarily Gaussian, existing for all $x\in\bm{\mathfrak{D}}$ with expectation $\bm{\mathbb{E}}
\langle {\mathscr{R}}(x)\rangle $, not necessarily zero. Then
\begin{equation}
{\bm{\mathbb{E}}}\langle\int_{\bm{\mathfrak{D}}}{{\mathscr{R}}}(x)d\mu(x)
\rangle\equiv \int_{{\bm{\mathfrak{D}}}}{\bm{\mathbb{E}}}\langle
{\mathscr{R}}(x)\rangle dU_{n}(x)
\end{equation}
For a set of N random fields ${{\mathscr{R}}}_{q}(x)$
\begin{equation}
{\bm{\mathbb{E}}}\langle\sum_{q=1}^{N}{{\mathscr{R}}}_{q}(x)\rangle=
\sum_{q=1}^{N}\bm{\mathbb{E}}\langle{{\mathscr{R}}}_{q}(x)\rangle
\end{equation}
\end{thm}
It is also possible to define a 'mollifier' or convolution integral.

\begin{prop}
Let $(x,y)\in {\bm{\mathfrak{D}}}$ and let $\mathbf{K}(x-y)$ be a smooth function of $(x,y)$ that will depend on the separation $\|x-{y}\|$. Given the GRSF ${\mathscr{R}}(y)$ define the volume and surface integral convolutions
\begin{align}
&{\mathscr{J}}(x)=\mathsf{\bm{\Xi}}(x-y){\mathscr{R}}(y)\equiv\int_{\bm{\mathfrak{D}}}\mathbf{K}(x-y){{\otimes}} {\mathscr{R}}(y)dU_{n}(y),~(x,y)
\in\bm{\mathfrak{D}}\\& {\mathscr{J}}(x)={\mathsf{\bm{\Xi}}}(x-y){\mathscr{R}}(y)\equiv\int_{\bm{\mathfrak{D}}}\mathbf{K}(x-y){{\otimes}}
{\mathscr{R}}(y)dU_{n-1}(y),(x,y)\in\partial{\bm{\mathfrak{D}}}
\end{align}
\end{prop}

For example, if $\mathsf{\bm{\Xi}}(x-y)$ is a Gaussian function then the random field ${\mathscr{R}}(x)$ can be 'Gaussian smoothed' on a scale $\ell$.
\begin{align}
&{\mathscr{J}}(x)={\mathsf{\bm{\Xi}}}(x-y){\mathscr{R}}(y)\equiv \mathsf{C}\int_{{\bm{\mathfrak{D}}}}\exp\left(-\frac{\|x-y\|^{2}}{\ell^{2}}\right){{\otimes}}{\mathscr{R}}(y)
dU_{n}(y),
~(x,y)\in{\mathbf{H}}\\&{\mathscr{J}}(x)= {\mathsf{\bm{\Xi}}}(x-y){\mathscr{R}}(y)\equiv \mathsf{C}\int_{{\bm{\mathfrak{D}}}}\exp\left(-\frac{\|x-y\|^{2}}{\ell^{2}}\right){{\otimes}}
{\mathscr{R}}(y)dU_{n-1}(y),(x,y)\in\partial{\bm{\mathfrak{D}}}
\end{align}
For the heat kernel $h(x-y,t)$ the convolution give a time-dependent GRSF that solves the heat equation $\square{\mathscr{R}}(x,t)=0$ for random initial boundary data.
\begin{align}
&\mathscr{J}(x,t)=h(x-y,t)* {\mathscr{J}(y)}\equiv \int_{{\bm{\mathfrak{D}}}}\frac{1}{(4\pi t)^{n/2}}\exp\left(-\frac{\|x-y\|^{2}}{4t}\right)\big({\mathscr{J}(y)}dU_{n}(y),
~(x,y)\in{\mathbf{Q}}
\end{align}
The result is extended to include GRSFs ${\mathscr{R}}(x)$.

\begin{prop}
Let ${\mathscr{R}}(x)$ be a regulated GRSF defined for all $x\in{\bm{\mathfrak{D}}}$ such that $\bm{\mathbb{E}}\langle{\mathscr{R}}(x)\rangle=0$ and
\begin{equation}
\mathbb{E}\langle\big|{\mathscr{R}}(x)\big|^{\ell}\rangle=
\frac{1}{2}[\mathsf{C}^{\ell/2}+(-1)^{\ell}\mathsf{C}^{\ell/2}]
\end{equation}
so that $\bm{\mathbb{E}}\langle\big|{\mathscr{R}}(x)\big|^{\ell}\rangle=0$ for all odd p. Then
\begin{align}
{\bm{\mathbb{E}}}\langle
|\int_{{\bm{\mathfrak{D}}}}{\mathscr{R}}(x)d\mu(x)|^{\ell}\rangle&~\le
\frac{1}{2}[\mathsf{C}^{\ell/2}+(-1)^{\ell}\mathsf{C}^{\ell/2}]|\int_{{\bm{\mathfrak{D}}}}d\mu(x)|^{\ell}
\nonumber\\&=\frac{1}{2}[\alpha^{\ell/2}+(-1)^{\ell}\alpha^{\ell/2}]\mathrm{Vol}[{\bm{\mathfrak{D}}}]^{\ell}
\end{align}
\end{prop}
}
\end{document}